\documentclass{article}
\usepackage[modulo]{lineno}
\usepackage{ifpdf}
\usepackage{graphicx}
\usepackage{times}
\usepackage{amssymb}
\usepackage{amsmath}
\usepackage{latexsym}
\usepackage{proof}
\usepackage{epsfig}
\usepackage{verbatim}	
\usepackage{float}	
\usepackage[usenames]{color} 
\usepackage{stmaryrd}

\ifpdf
\usepackage[%
  pdftitle={Instructions for use of the document class
    elsart},%
  pdfauthor={Simon Pepping},%
  pdfsubject={The preprint document class elsart},%
  pdfkeywords={instructions for use, elsart, document class},%
  pdfstartview=FitH,%
  bookmarks=true,%
  bookmarksopen=true,%
  breaklinks=true,%
  colorlinks=true,%
  linkcolor=blue,anchorcolor=blue,%
  citecolor=blue,filecolor=blue,%
  menucolor=blue,pagecolor=blue,%
  urlcolor=blue]{hyperref}
\else
\usepackage[%
  breaklinks=true,%
  colorlinks=true,%
  linkcolor=blue,anchorcolor=blue,%
  citecolor=blue,filecolor=blue,%
  menucolor=blue,pagecolor=blue,%
  urlcolor=blue]{hyperref}
\fi

\makeatletter
\def\elsartstyle{%
    \def\normalsize{\@setfontsize\normalsize\@xiipt{14.5}}
    \def\small{\@setfontsize\small\@xipt{13.6}}
    \let\footnotesize=\small
    \def\large{\@setfontsize\large\@xivpt{18}}
    \def\Large{\@setfontsize\Large\@xviipt{22}}
    \skip\@mpfootins = 18\p@ \@plus 2\p@
    \normalsize
}
\@ifundefined{square}{}{\let\Box\square}
\makeatother

\newcommand{\QlSRN}{\textsf{QlSRN}}
\newcommand{\LAL}{\textsf{LAL}}
\newcommand{\PRNQ}{\QlSRN}
\newcommand{\PRN}{\SRN}
\newcommand{\QlSRNVnames}{\mathbb{V}_{\PRNQ}} 

\newcommand{\ElTens}{Elementary tensors}

\newcommand{\BEmbed}[2]{\mathtt{Eb}^{#1}[#2]}
\newcommand{\LEmbed}[3]{\mathtt{El}^{#1}_{#2}[#3]}
\newcommand{\EEmbed}[4]{\mathtt{Ee}^{#1}_{#2;#3}[#4]}

\newcommand{\Coerc}{\mathtt{Coerce}}
\newcommand{\Iter}[5]{\mathtt{It}_{#1,#2}[#3,#4,#5]} 
\newcommand{\comp}[5]{\circ^{#1,#2}_{#3,#4}[#5]} 
\newcommand{\Wcomp}[5]{\bullet^{#1,#2}_{#3,#4}[#5]} 
\newcommand{\Zero}{0}                        
\newcommand{\Sucz}{\mathtt{s}_0}             
\newcommand{\Suco}{\mathtt{s}_1}             
\newcommand{\preds}{\mathtt{p}}              
\newcommand{\zero}[2]{\mathtt{z}^{#1,#2}}    
\newcommand{\sucz}{\mathtt{s}^{0,1}_0}       
\newcommand{\suco}{\mathtt{s}^{0,1}_1}       
\newcommand{\pred}{\mathtt{p}^{0,1}}         
\newcommand{\proj}[3]{\pi^{#1,#2}_{#3}}      
\newcommand{\bran}{\mathtt{c}^{0,3}}         
\newcommand{\rec}[3]{\mathtt{r}^{#1,#2}[#3]} 
\newcommand{\Suc}[1]{\mathtt{s}_{#1}}        
\newcommand{\USucc}{\mathtt{Ss}}
\newcommand{\BSucc}[1]{\mathtt{Ws}#1}
\newcommand{\BSuccZ}{\mathtt{Ws0}}
\newcommand{\BSuccO}{\mathtt{Ws1}}

\newcommand{\MkCompact}{\mathtt{MkC}}
\newcommand{\StepMkCompZ}{\mathtt{SMkC}_0}
\newcommand{\StepMkCompO}{\mathtt{SMkC}_1}
\newcommand{\BaseMkComp}{\mathtt{BMkC}}
\newcommand{\BInttoUInt}{\mathtt{W2S}}
\newcommand{\UInttoList}{\mathtt{S2L}}

\newcommand{\BCconftoBInt}{\mathtt{FC2W}}
\newcommand{\BCconftoFConf}{\mathtt{C2FC}}
\newcommand{\BInttoBCconf}{\mathtt{W2C}}
\newcommand{\ListstoConf}{\mathtt{L2C}}
\newcommand{\NextConf}{\mathtt{PC2C}}
\newcommand{\HeadsandTails}{\mathtt{C2PC}}
\newcommand{\Pred}{\mathtt{P}}
\newcommand{\BaseP}{\mathtt{BaseP}}
\newcommand{\StepP}{\mathtt{StepP}}
\newcommand{\Branch}{\mathtt{B}}

\newcommand{\TransFunc}{\mathtt{C2C}}
\newcommand{\BaseTransFunc}{\mathtt{Ba}\HeadsandTails}
\newcommand{\StepTransFunc}{\mathtt{St}\HeadsandTails}

\newcommand{\wght}[2]{\operatorname{wg}_{#2}(#1)}   
\newcommand{\srtw}[2]{\llbracket#1\rrbracket_{#2}} 

\newcommand{\BIntT}{\mathbf{W}}
\newcommand{\bcConfT}{\mathbf{C}}
\newcommand{\FbcConfT}{\mathbf{FC}}
\newcommand{\SbcConfT}{\mathbf{SC}}
\newcommand{\ST}{\mathbf{T}}
\newcommand{\SU}{\mathbf{U}}
\newcommand{\SV}{\mathtt{V}}

\newcommand{\bcConf}[1]{\mathtt{\lan\!\lan}#1\mathtt{\ran\!\ran}}

\newcommand{\WALL}{\textsf{WALT}} %
\newcommand{\WALT}{\WALL} %
\newcommand{\PT}{$\Lambda$} 
\newcommand{\PTV}{$\Lambda_{\texttt{V}}$} 
\newcommand{\PTT}{$\Lambda^{\texttt{T}}$} 

\newcommand{\Nat}{\mathbb{N}}

\newcommand{\DSpace}[1]{\textsf{DSpace}[#1]}
\newcommand{\ICC}{\textsf{ICC}}
\newcommand{\FP}{\textsf{FP}}
\newcommand{\LLL}{\textsf{LLL}}
\newcommand{\SF}{\textsf{System F}}
\newcommand{\SRN}{\textsf{SRN}}
\newcommand{\LLC}{\textsf{LLC}}
\newcommand{\DLAL}{\textsf{DLAL}}

\newcommand{\BoolT}{\mathbf{B}}
\newcommand{\ListT}{\mathbf{L}}
\newcommand{\UIntT}{\mathbf{N}}
\newcommand{\AlphT}{\mathbf{\Sigma}}
\newcommand{\CAlph}[1]{\overline{#1}}
\newcommand{\StatT}{\mathbf{S}}
\newcommand{\CStat}[1]{\overline{#1}}

\newcommand{\GG}{\gamma^\gamma}
\newcommand{\BB}{\beta^\beta}
\newcommand{\Nil}{\mathbf{nil}}

\newcommand{\BNum}[1]{\overline{\overline{#1}}}
\newcommand{\UNum}[1]{\overline{#1}}
\newcommand{\Intg}[1]{\mathsf{#1}}

\newcommand{\Id}{I}
\newcommand{\eseq}{\_\,}
\newcommand{\seq}[3]{\vec{#1}_{[#2;#3]}}
\newcommand{\gseq}[1]{\vec{#1}}
\newcommand{\seqel}[2]{\vec{#1}_{(#2)}}

\newcommand{\LTape}{\bot}
\newcommand{\RTape}{\top}

\newcommand{\PreTape}[3]{\texttt{T}[#1;#2;#3]}

\newcommand{\ConfT}{\mathbf{C}}
\newcommand{\Conf}[3]{\mathtt{C}[[#1];#2;[#3]]}

\newcommand{\PreConfT}{\mathbf{P}}
\newcommand{\TypeT}[2]{\mathbf{U}^{#1}_{#2}}
\newcommand{\TypeU}[2]{\mathbf{V}^{#1}_{#2}}
\newcommand{\PreConf}[5]
{\texttt{P}
 [\lan#1,#2\ran;#3;\lan#4,#5\ran]}

\newcommand{\WALLTrF}{\ol{\TrF}}
\newcommand{\WALLTM}{\ol{\TM}}
\newcommand{\Lookup}[1]{\Delta_{#1}}
\newcommand{\LookupRow}[1]{\Delta^{#1}}
\newcommand{\LookupCol}[2]{\Delta^{#1,#2}}
\newcommand{\LookupEl}[3]{\Delta^{#1,#2,#3}}
\newcommand{\CtoP}{\mathtt{C2P}}
\newcommand{\SCtoP}[1]{\mathtt{SC2P}[#1]}
\newcommand{\BCtoP}[2]{\mathtt{BC2P}[#1,#2]}
\newcommand{\PtoC}{\mathtt{P2C}}

\newcommand{\LtoC}{\mathtt{L2C}}
\newcommand{\LtoN}{\mathtt{L2N}}

\newcommand{\NSucc}{\UIntT\mathtt{Succ}}
\newcommand{\NSuccE}{\UIntT\mathtt{Succ}_{\odot}}
\newcommand{\NCoerce}{\UIntT\mathtt{Coerce}}
\newcommand{\NDiagE}{\UIntT\nabla_{\odot}}
\newcommand{\NSum}{+}
\newcommand{\NMult}{\times}
\newcommand{\NMultE}{\times_{\odot}}
\newcommand{\NSquare}{\mathtt{Sq}}

\newcommand{\LCoerce}{\ListT\mathtt{Coerce}}
\newcommand{\MLCoerce}{\mathtt{It}\ListT\mathtt{Coerce}}
\newcommand{\LSucc}{\ListT\mathtt{Succ}}
\newcommand{\LPush}{\ListT\mathtt{Push}}
\newcommand{\LSuccE}{\ListT\mathtt{Succ}_\odot}
\newcommand{\LPushE}{\ListT\mathtt{Push}_\odot}
\newcommand{\LDiagE}{\ListT\nabla_{\odot}}

\newcommand{\ACoerce}{\AlphT\mathtt{Coerce}}
\newcommand{\ADiag}{\AlphT\nabla}
\newcommand{\ADiagE}{\AlphT\nabla_{\odot}}

\newcommand{\TrF}{\delta}
\newcommand{\TM}{\mathbf{M}}
\newcommand{\TMPoly}[2]{p^{#1}(#2)}
\newcommand{\NPoly}[2]{\ol{p}[#1,#2]}

\newcommand{\MoveL}{\Leftarrow}
\newcommand{\MoveR}{\Rightarrow}
\newcommand{\DoNotMove}{\Downarrow}

\newcommand{\ctoc}{\mathtt{c2c}}

\newcommand{\lan}{\langle}
\newcommand{\ran}{\rangle}
\newcommand{\llan}{\langle\!\langle}
\newcommand{\rran}{\rangle\!\rangle}
\newcommand{\elan}{\langle\!\lbrace}
\newcommand{\eran}{\rbrace\!\rangle}
\newcommand{\subs}[2]{\{^{#1}/_{#2}\}}  %
\newcommand{\li}{\multimap}   %
\newcommand{\ten}{\otimes}   %
\newcommand{\liv}{-\!\!\bullet\ }   %
\newcommand{\mliv}{-\!\!\bullet\!}   %
\newcommand{\ol}[1]{\overline{#1}}        
\newcommand{\nocc}[2]{\operatorname{no}({#1, #2})}
  
\newcommand{\FV}[1]{\operatorname{FV}(#1)}              
\newcommand{\wdth}[2]{\operatorname{w}_{#1}(#2)}        
\newcommand{\dpth}[1]{\operatorname{d}({#1})}        
\newcommand{\epttopt}[1]{[#1]^{\circ}} 
\newcommand{\etp}{\epttopt} 
\newcommand{\nf}[1]{\operatorname{nf}(#1)} 
\newcommand{\nfl}[2]{\operatorname{nf}_{#1}(#2)} 

\newcommand{\ca}{\textit{ta}} 

\newcommand{\dom}[1]{\operatorname{Dom}(#1)}
\newcommand{\bs}{\backslash}
\newcommand{\ta}[2]{#1\!:\!#2}
\newcommand{\red}{\rightarrow_{w}}      
\newcommand{\redl}[1]{\stackrel{#1}{\longrightarrow_{w}}} 
\newcommand{\canstra}[1]{\stackrel{#1}{\Longrightarrow}} 
\newcommand{\round}[1]{\stackrel{#1}{\leadsto}} 
\newcommand{\size}[1]{|#1|}      
\newcommand{\psz}[2]{\operatorname{s}_{#1}(#2)}
\newcommand{\degree}[1]{\partial{(#1)}}
\newcommand{\polynom}[2]{p_{#1}{#2}}
\newcommand{\st}[2]{\operatorname{t}_{#1}(#2)}

\usepackage{theorem}

\floatstyle{boxed}
\newfloat{Figure}{ht}{lop}
\newfloat{Table}{ht}{lop}

 \newtheorem{theorem}{Theorem}
 \newtheorem{proposition}{Proposition}
 \newtheorem{lemma}{Lemma}
 \newtheorem{corollary}{Corollary}
 \newtheorem{fatto}{Fact}
 \newtheorem{definition}{Definition}[section]

\newenvironment{proof}{\noindent{\bfseries Proof.}}{\hbox{}\hfill$\Box$\vspace{\baselineskip}}

\begin{document}
\title{Weak Affine Light Typing:\\  
Polytime intensional expressivity, soundness and completeness}
\author{Luca Roversi
\footnote{Dipartimento di Informatica,
          Universit\`a di Torino,
          Corso Svizzera 185 --- Torino --- Italy.
}
\footnote{
\textit{e-mail}:\texttt{roversi@di.unito.it}.
\textit{home page}:\texttt{http://www.di.unito.it/\~{}rover}.
}
}
\date{\today}
\maketitle
\begin{abstract}

\textbf{Ridefinite Quasi-linear come Composition-linear}

Weak affine light typing (\WALT) assigns light affine linear formulae as types to a subset of $\lambda$-terms in \SF. \WALL\ is poly-time sound: if a $\lambda$-term $M$ has type in \WALT, $M$ can be evaluated with a polynomial cost in the dimension of the derivation that gives it a type. In particular, the evaluation can proceed under any strategy of a rewriting relation, obtained as a mix of both call-by-name/call-by-value $\beta$-reductions.
\WALL\ is poly-time complete since it can represent any poly-time Turing machine.
\WALL\ \textit{weakens}, namely \textit{generalizes}, the notion of stratification of deductions common to some \textsf{Light Systems} --- we call as such those logical systems, derived from Linear logic, to characterize \FP, the set of Polynomial functions --- .
A weaker stratification allows to define a compositional \textit{embedding} of the \textit{Quasi}-linear fragment \QlSRN\ of Safe recursion on notation (\SRN) into \WALT. \QlSRN\ is \SRN, which is a recursive-theoretical system characterizing \FP, where only the composition scheme is restricted to \textit{linear safe variables}.
So, the expressivity of \WALT\ is stronger, as compared to the known \textsf{Light Systems}.
In particular, using the types, the \textit{embedding} puts in evidence the stratification of normal and safe arguments hidden in \QlSRN: the less an argument is impredicative, the deeper, in a formal, proof-theoretical sense, gets its representation in \WALT.
\end{abstract} 
\tableofcontents


\linenumbers

\section{Introduction}
\label{section:Introduction}
Implicit computational complexity (\ICC) explores machine-independent characterizations of complexity classes without any explicit reference to resource usage bounds, which, instead, result from restricting suitable computational structures. \ICC\ systems originate from recursion theory \cite{Cobham65,Bellantoni92CC,Leivant94RRII,Leivant95RRI,Leivant99RRIII,LeivantMarion00RRIV}, structural proof-theory and linear logic \cite{Girard:1998-IC,Lafont02SLL}, rewriting systems or functional programming \cite{Huet:JACM-80,Dershowitz:TCS-92,Jones:TCS-99,Leivant93RR,Leivant94ic}, type systems \cite{hofmann97csl,hofmann99linear,hofmann99thesis,hofmann00safe,Bellantoni00APAL,Bellantoni01MSS} \ldots.
\par
This work is mainly concerned with the theoretical aspects of \ICC\ whose essential goal is to support the evidence that the notions of the known complexity classes are natural concepts. Classically, a complexity class is defined in terms of some specific computational model. \ICC\ aims to show that such computational models have mathematical counterparts, independent from them.
Here, we approach \ICC\ from a type-theoretical point of view.
\par
We start from generalizing the structural proof-theoretical design principles, used for Light linear logic (\LLL) \cite{Girard:1998-IC} and Light affine logic (\LAL)
\cite{Asperti:1998-LICS,Roversi:1999-CSL,Asperti02TOCL}. The reason is that, so far, such principles look quite restrictive. Indeed, we know that 
\LAL\ is polynomially strongly normalizable: the normalization of every of its derivations is polynomial under every rewriting strategy
\cite{Terui:2001-LICS,Terui:2002-AML}.
This limits the intensional expressiveness of \LAL, hence of \LLL,
as witnessed by the difficulty to relate the computational behavior of \LAL\ to the one of \ICC\ systems based on principles other than structural proof-theory.
In this direction, the only known relation is in \cite{MurawskiOng00}. There, a compositional, and intuition preserving, embedding of a \textit{fragment}
of \textit{Safe recursion on notation} (\SRN) \cite{Bellantoni92CC} into \LAL\ is given. The fragment can only use the \textit{safe arguments} linearly, and is $\DSpace{ln}$-complete \cite{Neergaard:APLAS-04}.
\par
We introduce \textit{Weak Affine Light Typing} (\WALL) as a typing system for pure $\lambda$-terms.
It generalizes a basic design principle of \LAL\ and gives an extension of the formulae of \LAL\ as types to standard $\lambda$-terms that belong to a fragment of  \SF\ \cite{GLT:PT}.
The distinguishing feature of \WALL, as compared to \LAL, is its \textit{weaker}, hence, \textit{more liberal} definition of the deductions that can be duplicated in the course of the normalization. Recall that any deduction $\Pi$ of \LAL\ that, eventually, will be duplicated by a cut elimination step has a conclusion of type $!A$, and \textit{must be defined} in a way that it depends on \textit{at most} a single assumption of type $!B$. \WALL\ weakens this constraint. Any deduction $\Pi$ of \WALL\ that, eventually, will be duplicated by a normalization step has conclusion of type $!A$, and it \textit{may depend} on an \textit{arbitrary} number of assumptions, \textit{one} of which \textit{must} be of type $!B$, while the \textit{others must} have type $\$C_i$ --- here we adopt $\$$ to name the ``paragraph'' modality of \LAL\ --- .
Before $\Pi$ gets duplicated, it \textit{must} evaluate to a deduction $\Pi'$ that depends on \textit{at most} a single assumption of type $!B$.
The correct duplication of the weaker form of duplicable deductions is obtained by extending the language of formulae of \LAL. In particular, \WALL\ builds formulae with two linear implications $\li, \liv$, two modalities $\$, !$, and a universal quantification.
The new implication $\liv$ denotes the linear functions whose arguments are the assumptions, with type $\$C_i$, of the deductions with conclusion of type $!A$.
Intuitively, if a term $M$ has type $\$A\liv B$, then a necessary condition to fully evaluating $M\,N$ is that $N$ becomes a closed term.
\par
\WALL\ \textit{is poly-time sound}. Every $\lambda$-term typable by \WALL\ can be evaluated with a polynomial cost under any rewriting strategy of a rewriting relation $\red$: a mix of the standard call-by-name and call-by-value $\beta$-reduction. The bound can be read from the structure of any deduction of \WALL, but $\red$ evaluates any typable $\lambda$-term completely ignoring the types. So, \WALL\ is a framework where the program part, represented by a typable $\lambda$-term, and the complexity specification part, represented by the corresponding deduction, are completely separate, like in \DLAL\ \cite{BaillotTerui:2004-LICS} and \cite{Cop-DLag-Ron:EALCBV-04}.
\par
\WALL\ \textit{is more expressive than} \LAL.
Let \QlSRN\ be \SRN, where the composition scheme uses safe variables \textit{linearly}: the safe variables in two, or more, functions being composed, must be different. Then,
there exists an interpretation map $\srtw{\,}{}$ from $\QlSRN$ to $\WALT$, such that, for every $f(n_1,\ldots,n_k,s_1,\ldots,s_l)\in\QlSRN$, with $k$ \textit{normal} and $l$ \textit{safe} arguments, we can prove that:
(i) if $f(n_1,\ldots,n_k,s_1,\ldots,s_l)=n$, then 
$\srtw{f(n_1,\ldots,n_k,s_1,\ldots,s_l)}{}$ reduces to $\srtw{n}{}$, using $\red$, and
(ii) $\srtw{f(n_1,\ldots,n_k,s_1,\ldots,s_l)}{}$ has type $\$^{m}\BIntT$, since
$\srtw{f}{}$ has type 
$\overbrace{\$\BIntT\liv\ldots\liv\$\BIntT}^{k}
 \liv
 \overbrace{\$^{m}\BIntT\liv\ldots\liv\$^{m}\BIntT}^{l}\liv\$^{m}\BIntT$,
for some $m\geq 1$, the type $\BIntT$ being the one for binary words in \WALT.
\par
Point (i) is the obvious behavior we expect from the embedding and shows that
\WALT\ is strongly more expressive that the known systems derived as restrictions of Linear logic to characterize \FP.
Point (ii) links $m$ to the complexity of the definition of $f$,
$m$ depending on the number of nested linear safe compositions and of safe recursive schemes that define $f$. Moreover, the types explicitly show the layered structure of the normal and safe arguments hidden in \QlSRN. The type of a safe argument is $m$ $\$$-modality occurrences deep because a safe argument can be used in the course of recursive unfolding to produce the result. Orthogonally, the depth of the type of every normal argument is limited to $1$. This allows to give to the normal arguments the necessary ``replication power'' required to duplicate syntactic structure in the course of an unfolding.
This underlines a radical difference between the approach to the implicit characterization of \FP, through Light linear logic-like systems, and the approach of the recursive ones.
The formers say that the weaker is the possibility of a word to replicate structure, behaving it as an iterator, the deeper is its type. The latter, are based exactly on the reversed idea, though this cannot be formally stated in terms of any typing information inside \QlSRN.
\par
\WALL\ \textit{is poly-time complete}. It can represent and simulate every poly-time Turing machine. The result must be explicitly reproved. Indeed, we cannot take advantage of the existing proofs of poly-time completeness for \LAL\
\cite{Roversi:1999-CSL,Asperti02TOCL}
because of the call-by-name/call-by-value rewriting notion that \WALL\ induces on the $\lambda$-terms.
\par
\textbf{Outline.}
Section~\ref{section:Weak Affine Light Typing WALL} formally introduces \WALT, gives some intuitions, and proves its structural properties useful to get the subject reduction, given in Section~\ref{section:Dynamic properties of WALL} after the definition of the hybrid call-by-name/call-by-value rewriting system.
Section~\ref{section:Polytime soundness} is about the poly-time soundness of \WALL.
Section~\ref{section:Quasi-linear safe recursion on notation} formally defines \QlSRN\ in the style of \cite{Beckmann:96-AML}.
Section~\ref{section:Programming combinators in WALT} develops the combinators, of \WALT, required to embed \QlSRN\ into \WALT. In particular, Subsection~\ref{subsection:Iterators}, details the intuition about how we implement the virtual machine that interprets the recursive scheme of \QlSRN, hence of \SRN.
Section~\ref{section:Algorithmic expressivity of WALL} formally develops the embedding.
Section~\ref{section:Conclusions and further work} delineates some possible research directions. Appendix~\ref{section:Completeness} is about the poly-time completeness, while Appendix~\ref{section:Details about the proofs} details some of the proofs.
\par
\textbf{Acknowledgements.}
My gratitude goes to Harry Mairson and Peter M{\o}ller Neergaard who deeply read and helpfully commented \cite{Roversi:2002-FALL}, the root of this work, and all those researchers that, in the last years, thanks to their results, indirectly helped me to write this work in a more accessible way than \cite{Roversi:2002-FALL}, hopefully.
Also, I want to thank the anonymous referees as well as Marco Gaboardi and Luca Vercelli that helped me to improve early versions of this work.
\section{Weak Affine Light Typing (\WALL)}
\label{section:Weak Affine Light Typing WALL}
\WALL\ gives the formulae that belong to the language, generated by the following grammar:
\small
\begin{align*}
A &::= L
	\mid !A
	\mid \$A
\\
L &::=  \alpha
	\mid A\li A
	\mid \$A\liv A
	\mid \forall \alpha.L
\end{align*}
\normalsize
as types to a subset of \PT, the set of $\lambda$-terms, generated by
$M  ::=  x
    \mid (\bs x.M) 
    \mid (MM)$.
\par
\textbf{Notations and definitions.}
$A$ is the start symbol.
A \textit{modal} formula has form $!A$ or $\$A$, and, in particular, $!A$ is $!$-modal, while $\$A$ is $\$$-modal.
$L$ generates \textit{linear}, or \textit{non modal},
formulae. Notice that the linear formulae are closed under the substitution of linear formulae for a universally quantified variable. Also, the universal quantification cannot hide a modal formula by means of the quantifier.
Generic formulae are ranged over by $A, B, C$. Linear ones by $L, L'$.
$M\{^{N_1}/_{x_1} \cdots {}^{N_m}/_{x_m}\}$ denotes the usual capture free simultaneous substitution of every $N_i$ for the corresponding $x_i$, with $1\leq i \leq m$. If $N_1, \ldots, N_m$ are all equal, the substitution is denoted as $M\subs{N_1}{x_1\ldots x_m}$. Parentheses are left-associative, so $((\cdots((MM)M)\cdots)M)$
shortens to $MMM\cdots M$. 
A sequence of abstractions 
$(\bs x_1.\ldots(\bs x_m.M)\ldots)$ is shortened by 
$\bs x_1 \ldots x_m.M$, for any $m$. 
An abstraction $\bs x.M$ binds every free occurrence of $x$ in $M$. Given a term $M$, the set of its free variables is $\FV{M}$. A \textit{closed term} has no free variables. \PTV\ is the set of the $\lambda$-terms which are \textit{values}, generated by
$V ::=  x
  \mid (\bs x.M)$, where $M$ is in \PT.
The \textit{cardinality of a free variable in a term}
is $\nocc{x}{M}$ and counts the number of \emph{free occurrences of $x$ in $M$}:
\small
\begin{align*}
\nocc{x}{x}&=1 &
\nocc{x}{y}&=0 &(x \not\equiv y)\\
\nocc{x}{\bs x.M}&=0 &
\nocc{x}{\bs y.M}&=\nocc{x}{M} &(x \not\equiv y)\\
\nocc{x}{MN} &=\nocc{x}{M}+\nocc{x}{N}
\end{align*}
\normalsize
The \textit{size of a term}
$\size{M}$ gives the \emph{dimension}  of $M$ as expected:
$\size{x}= 1, 
\size{\bs x.M}=\size{M}+1, 
\size{MN}=\size{M}+\size{N}+1$.
\begin{Figure}
\begin{center}
\begin{tabular}{c}
 \infer[A]
  {
   \Gamma,\ta{x}{L};\Delta;{\mathcal E}\vdash \ta{x}{L}
  }
  {
  }
\\
\\
 \infer[C]
  {
   \Gamma; 
   \Delta; 
   {\mathcal E}\sqcup
   \{(\Theta_x,\Theta_y;\{\ta{z}{A}\})\}
   \vdash \ta{M\{^{z}/_{x} {}^{z}/_{y}\}}{B}
  }
  {
   \Gamma; \Delta; 
   {\mathcal E},
   (\Theta_x;\{\ta{x}{A}\}),
   (\Theta_y;\{\ta{y}{A}\})
   \vdash \ta{M}{B}
  }
\\
\\
 \infer[\li I]
  {\Gamma; \Delta; {\mathcal E}
   \vdash \ta{\bs x.M}{L\li B}}
  {\Gamma,\ta{x}{L}; \Delta; {\mathcal E}\vdash \ta{M}{B}}
\qquad
 \infer[\li I_{\$}]
  {\Gamma; \Delta; {\mathcal E}
   \vdash \ta{\bs x.M}
             {\$A\li B}}
  {\Gamma;\Delta,\ta{x}{A}; {\mathcal E}\vdash \ta{M}{B}}
\\
\\
 \infer[\li E]
  {\Gamma_M, \Gamma_N;
   \Delta_M, \Delta_N;
   {\mathcal E}_M
   \sqcup
   {\mathcal E}_N
   \vdash \ta{MN}{B}
  }
  {\Gamma_M; 
   \Delta_M;
   {\mathcal E}_M
   \vdash \ta{M}{A\li B}
   &
   \Gamma_N; 
   \Delta_N;
   {\mathcal E}_N
   \vdash \ta{N}{A}
   &
   A\not\equiv !C, \text{ for any } C
   }
\\
\\
 \infer[\li I_!]
  {
   \Gamma; 
   \Delta; 
   {\mathcal E}\sqcup
   \{(\Theta;\emptyset)\}
   \vdash \ta{\bs x.M}{\ !A\li B}
  }
  {
   \Gamma; \Delta; 
   {\mathcal E},
   (\Theta;\{\ta{x}{A}\})
   \vdash \ta{M}{B}
  }
\\
\\
 \infer[\li E_!]
  {\Gamma_M, \Gamma_N;
   \Delta_M, \Delta_N;
   {\mathcal E}_M
   \sqcup
   {\mathcal E}_N
   \vdash \ta{MN}{B}
  }
  {
   \Gamma_M;
   \Delta_M;
   {\mathcal E}_M
   \vdash \ta{M}{\,!A\li B}
   &
   \Gamma_N; 
   \Delta_N;
   {\mathcal E}_N
   \vdash \ta{N}{\,!A}
   &
   {\mathcal E}_M\subseteq\{(\emptyset;\Phi_1),\ldots,(\emptyset;\Phi_n)\}
}
\\
\\
 \infer[\liv I]
  {
   \Gamma;\Delta;{\mathcal E}\sqcup\{(\Theta;\emptyset)\}
   \vdash \ta{\bs x.M}{\$A\liv B}
  }
  {
   \Gamma;\Delta;
   {\mathcal E},(\Theta,\ta{x}{A}; \emptyset)
   \vdash \ta{M}{B}
  }
\\
\\
\infer[\liv E]
  {
   \Gamma_M; 
   \Delta;
   {\mathcal E}_M
   \sqcup
   {\mathcal E}_N
   \vdash \ta{MN}{B}
  }
  {
   \Gamma_M; 
   \Delta;
   {\mathcal E}_M
   \vdash \ta{M}{\$A\liv B}
   &
   \emptyset; 
   \emptyset;
   {\mathcal E}_N
   \vdash \ta{N}{\$A}
   &
   {\mathcal E}_N\subseteq\{(\Theta;\emptyset)\}
   }
\\
\\
\infer[\$]
{
 \Gamma';
 \$\Delta',\Delta;
 \{(\$\Theta';\emptyset)\}
 \sqcup
 \{(\Theta_1;\Phi_1)\}
 \sqcup\ldots\sqcup
 \{(\Theta_m;\Phi_m)\}
 \vdash \ta{M}{\$B}
}
{
 \Gamma;
 \Delta';
 \{(\Theta';\emptyset)\}
 \vdash \ta{M}{B}
 &
 \Gamma\subseteq
 \Delta\cup\bigcup_{i=1}^{m}\Theta_i
 \cup\bigcup_{i=1}^{m}\Phi_i
 &
 \Theta_i\neq\emptyset\text{ iff }\Phi_i=\emptyset
}
\\
\\
\infer[!]
 {\Gamma';
  \Delta;
  \{(\$\Theta';\emptyset)\}\sqcup\{(\Theta;\Phi)\}
  \vdash \ta{M}{\ !B}
 }
 {
  \Gamma;
  \emptyset;
  \{(\Theta';\emptyset)\}
  \vdash \ta{M}{B}
  &
  \Gamma
  \subseteq
  \Theta\cup\Phi
  &
  \Theta\neq\emptyset \Rightarrow \dom{\Phi}\cap\FV{M}\neq\emptyset
 }
\\
\\
 \infer[\forall I]
  {\Gamma; \Delta; {\mathcal E}\vdash \ta{M}{\forall \alpha.L}}
  {
  \Gamma; \Delta; {\mathcal E}\vdash \ta{M}{L}
  &
  \text{$\alpha$ not free in $\Gamma, \Delta$ and ${\mathcal E}$}
  }
  
\qquad
 \infer[\forall E]
  {\Gamma; \Delta; {\mathcal E}\vdash \ta{M}{L\{L'/\alpha\}}}
  {\Gamma; \Delta; {\mathcal E}\vdash \ta{M}{\forall\alpha.L}}
\end{tabular}
\end{center}
\normalsize
\caption{Weak Affine Light Typing}
\label{figure:Weak LAL as a type assignment}
\end{Figure}

\par
\textbf{The type assignment.}
Figure~\ref{figure:Weak LAL as a type assignment}
defines Weak Affine Light Typing (\WALL). \WALL\ is a deductive system that deduces judgments
$\Gamma; \Delta; {\mathcal E} \vdash \ta{M}{A}$, where $M$ is a $\lambda$-term.
If we call type assignment any pair $\ta{x}{A}$ where $x$ is a variable and $A$ a type, meaning that $A$ is a type for $x$, then $\Gamma$ and $\Delta$
are sets of type assignment and
${\mathcal E}$ is a set of pairs $(\Theta;\Phi)$ such that both $\Theta$ and $\Phi$ are sets of type assignments as well. Namely, the judgments assign a type $A$ to a $\lambda$-term $M$ from four sets of assumptions, in analogy to
\LLL\
\cite{Girard:1998-IC}, \LLC\
\cite{Terui:2001-LICS,Terui:2002-AML}, and Dual Light Affine Logic (\DLAL) \cite{BaillotTerui:2004-LICS}.
\par
\textbf{Notations and definitions.} 
Given a set of type assignments $\{\ta{x_1}{A_1},\ldots,\ta{x_n}{A_n}\}$,
$\dom{\{\ta{x_1}{A_1},\ldots,\ta{x_n}{A_n}\}}
=\{x_1,\ldots,x_n\}$ denotes the \textit{domain} of such a set.
In general,
$\Gamma$ denotes a set of \emph{linear type assignments} $\ta{x}{L}\in\Gamma$, and we call $x$ \emph{linear}.
$\Delta$ denotes a set of \emph{linear partially discharged} type assignments to variables that we call \emph{linear partially discharged}.
${\mathcal E}$ denotes a set of \emph{partially discharged contexts}.
$\mathcal E$ is either empty or it contains pairs $(\Theta_1;\Phi_1),\ldots,(\Theta_n;\Phi_n)$ where, for every $i\in\{1,\ldots,n\}$, the following four points hold:
(i) $\Theta_i$ is a set of \emph{elementary partially discharged} type assignments to variables that we simply call \emph{elementary};
(ii) $\Phi_i$ is either empty or it is a singleton $\ta{x}{A}$. We call $x$ \emph{polynomially partially discharged}, or simply \emph{polynomial};
(iii) only one between $\Phi_1,\ldots,\Phi_n$ can be $\emptyset$;
(iv) the domains of any two $\Phi_i$ and $\Phi_j$ are distinct, with $i\neq j$.
\par
For every 
${\mathcal E}=\bigcup^{n}_{i=1} \{(\Theta_i;\Phi_i)\}$,
$\dom{\mathcal E}$ is 
$(\bigcup^{n}_{i=1} \dom{\Theta_i})\cup
	(\bigcup^{n}_{i=1} \dom{\Phi_i})$.
In every of the rules of \WALL\ the domain of two sets of type assignments $\Phi_M$ and $\Phi_N$ may intersect when $\Phi_M$ and $\Phi_N$ are part of two partially discharged contexts ${\mathcal E}_M$ and ${\mathcal E}_N$ that belong to two distinct premises of a rule.
This observation justifies the definition of
${\mathcal E}_M\sqcup{\mathcal E}_N$
that merges ${\mathcal E}_M$ and ${\mathcal E}_N$, preserving the structure of a partially discharged context:
\small
\begin{align*}
&
{\mathcal E}_M
\sqcup
{\mathcal E}_N
=
\\
&\quad
\{(\Theta_M,\Theta_N;\Phi) \mid
	(\Theta_M;\Phi) \in {\mathcal E}_M
	\text{ and }
	(\Theta_N;\Phi) \in {\mathcal E}_N
\}\cup\\
&\quad
\{(\Theta_M;\Phi_M) \mid
	(\Theta_M;\Phi_M) \in {\mathcal E}_M
	\text{ and there is no }
	\Theta_N \text{ such that }
	(\Theta_N;\Phi_M) \text{ in } {\mathcal E}_N
\}
\cup\\
&\quad
\{(\Theta_N;\Phi_N) \mid
	(\Theta_N;\Phi_N) \in {\mathcal E}_N
	\text{ and there is no }
	\Theta_M \text{ such that }
	(\Theta_M;\Phi_N) \text{ in } {\mathcal E}_M
\}
\end{align*}
\normalsize
The sequence ${\mathcal E},(\Theta;\Phi)$ denotes that 
$(\Theta;\Phi)\not\in{\mathcal E}$. Also, 
${\mathcal E} \sqcup \{(\emptyset;\emptyset)\} 
= {\mathcal E} \sqcup \emptyset =  {\mathcal E}$.
In every other cases, the domain of two sets of type assignments that belong to two distinct premises of a rule of \WALL\ \textit{must} be disjoint.
\PTT\ is the subset of \textit{typeable} elements $M$ of \PT, namely,
those for which a deduction $\Pi$ with conclusion 
$\Gamma;\Delta;{\mathcal E}\vdash\ta{M}{A}$ exists, denoted by
$\Pi\rhd\Gamma;\Delta;{\mathcal E}\vdash\ta{M}{A}$.
Finally,
$\Pi'\preceq\Pi$ denotes that $\Pi'$ is a \textit{subdeduction} of $\Pi$,
while
$\Pi(R,\Pi_1,\ldots,\Pi_m)$, with $0\leq m\leq2$, denotes a deduction $\Pi\rhd\Gamma;\Delta;{\mathcal E}\vdash\ta{M}{A}$, whose conclusion is the rule $R$, and such that the premises of $R$ are the conclusions of $\Pi_i\preceq\Pi$, with $1\leq i\leq m$. The notation $\Pi(R)$ highlights that $R$ is the last rule of $\Pi$.
\par
\textbf{Intuition.}
\WALL\ controls the number of normalization steps of its deductions by means of a \textit{weak} stratification.
``Stratification'' means that every deduction $\Pi$ of \WALL\ can be thought of as it was organized into levels, so that the logical rules of $\Pi$ may be at different depths.
The normalization preserves the levels. Namely, if the instance of a rule $R$ in $\Pi$ is at depth $d$, then it will keep to be at depth $d$ after any number of normalization steps that, of course, do not erase it.
The only duplication allowed is of deductions $\Pi$ that have undergone an instance $R$ of the $!$ rule, namely the conclusion of $\Pi$ has a $!$-modal type, introduced by $R$. Ideally, the $!$ rule defines a, so called, $!$-box around the deduction that proves its premise. The $!$-box may depend on more than one assumption, so generalizing the $!$-box of \LAL, that, in the context of \WALL, takes form:
\small
$$\infer[]
 {\emptyset;
  \emptyset;
  \{(\emptyset;\Phi)\}
  \vdash \ta{M}{\ !B}
 }
 {
  \Phi;
  \emptyset;
  \emptyset
  \vdash \ta{M}{B}
  &
  \Phi\subseteq\{\ta{x}{A}\}
 }
$$
\normalsize
A first immediate consequence of generalizing the $!$-boxes is that every elementary partially discharged assumption they may depend on can only be replaced, as effect of the normalization, by the conclusion of $\$$-boxes of \WALL\ which exclusively depend on elementary partially discharged assumptions as well. Otherwise, we could build $!$-boxes with an arbitrary number of $!$-modal assumptions, immediately getting deductions that normalize with an elementary cost. This justifies the name \emph{elementary partially discharged} type assignments.
The correct substitution discipline for the elementary partially discharged assumptions is obtained by introducing the linear arrow $\liv$. The rule $\liv I$ fully discharges them, while $\liv E$ forces the application of a function with type $\$A\liv B$ to arguments which, if they normalize to a $\$$-box, such a box can only depend on elementary partially discharged assumptions.
\par
Of course, not every $!$-box of \WALL, with conclusion of type $!A$, can be replaced for the argument of a function with type $!A\li B$. Such a replacement can occur only if the $!$-box gets normalized to another $!$-box with at most one $!$-modal assumption. Otherwise, we would again loose the main property of the duplicable objects inherited from \LAL\ which ensures the polynomial bound on the normalization cost.
\par
Summing up, \WALL\ allows to type $\lambda$-terms more liberally than \LAL, while keeping the same normalization principles: the stratification is never canceled, and only deductions that, eventually, depend on at most one free variable may be effectively duplicated as effect of the normalization. 
This is why \WALL\ does not enjoy a full normalizing procedure, the analogous of the cut elimination for a corresponding sequent calculus formulation. For example, the deduction 
$\Pi_{(\bs x.yxx)(wz)}$:
\small
\[
\infer[\li E_!]
{
\ta{w}{A\li A},\ta{z}{A};\emptyset;
\{(\ta{w}{A\li A};\ta{z}{A})\}
\vdash\ta{(\bs x.yxx)(wz)}{A}
}
{\begin{array}{l}
\Pi_{\bs x.yxx}\rhd
 \ta{y}{\, !A\li!A\li A};
 \emptyset;\emptyset\vdash\ta{\bs x.yxx}{\, !A\li A}
 \\
 \Pi_{wz}\rhd
 \emptyset;\emptyset;
  \{(\ta{w}{A\li A};\ta{z}{A})\}\vdash\ta{wz}{\, !A}
 \end{array}
}
\]
\normalsize
where $\Pi_{\bs x.yxx}$, and $\Pi_{wz}$ are, respectively, the two following deductions:
\[
\scriptsize
 \infer[\li I_{!}]
 {\ta{y}{\, !A\li!A\li A};
  \emptyset;\emptyset\vdash\ta{\bs x.yxx}{\, !A\li A}
 }
 {
  \infer[\li E_!]
  {\ta{y}{\, !A\li!A\li A};
   \emptyset;\{(\emptyset;\ta{x}{A})\}
   \vdash\ta{yxx}{A}}
  {
   \infer[\li E_!]
   {\ta{y}{\, !A\li!A\li A};
    \emptyset;\{(\emptyset;\ta{x}{A})\}
    \vdash\ta{yx}{\, !A\li A}}
   {
    \infer[A]
    {\ta{y}{\, !A\li!A\li A};
     \emptyset;\emptyset
     \vdash\ta{y}{\, !A\li!A\li A}
    }
    {}
    &
    \infer[!]
    {\emptyset;\emptyset;\{(\emptyset;\ta{x}{A})\}
    \vdash\ta{x}{\, !A}}
    {\infer[A]
     {\ta{x}{A};\emptyset;\emptyset
     \vdash\ta{x}{A}}
     {}
    }
   }
    &
    \infer[!]
    {\emptyset;\emptyset;\{(\emptyset;\ta{x}{A})\}
    \vdash\ta{x}{\, !A}}
    {\infer[A]
     {\ta{x}{A};\emptyset;\emptyset
     \vdash\ta{x}{A}}
     {}
    }
  }
 }
\]
\small
\[
 \infer[!]
 {\emptyset;\emptyset;
  \{(\ta{w}{A\li A};\ta{z}{A})\}\vdash\ta{wz}{\, !A}
 }
 {
  \infer[\li E]
  {\ta{w}{A\li A},\ta{z}{A};
   \emptyset;\emptyset\vdash\ta{wz}{A}
  }
  {
   \infer[A]
   {\ta{w}{A\li A};
    \emptyset;\emptyset\vdash\ta{w}{A\li A}
   }
   {}
   &
   \infer[A]
   {\ta{z}{A};
    \emptyset;\emptyset\vdash\ta{z}{A}}
   {}
  }
 }
\]
\normalsize
cannot normalize despite both the concluding $\li E_!$ in
$\Pi_{(\bs x.yxx)(wz)}$, and the last instance of $\li I_{!}$ in
$\Pi_{\bs x.yxx}$ seem to form a ``detour''. In fact, it is not a detour because replacing the conclusion of $\Pi_{wz}$ for the two occurrences of $A$ in $\Pi_{\bs x.yxx}$ would produce a deduction with a wrong instance of $\li E$.
In analogy to \cite{Ronchi-Roversi:1997-STUDIA-LOGICA},
the rules that introduce the linear implications $\li,\liv$ may simultaneously introduce a modal connective. In this way we use the $\lambda$-abstraction to denote the occurrence of an assumption of a $\$$, or of a $!$-box inside the $\lambda$-terms.
As a conclusion of this informal description of \WALL, we observe the following:

\begin{fatto}[\WALL\ types a subset of \SF.]
\label{fact:WALL types a subset of SF}
Let $\Pi\rhd\Gamma;\Delta;{\mathcal E}\vdash\ta{M}{A}$ be given.
Call $\Gamma_{F}$ and $A_F$ the set of type assumptions and the type that we can obtain from every $\ta{x}{B}$ of $\Gamma;\Delta;{\mathcal E}$, and from $A$ by both: 
(i) replacing the intuitionistic arrow $\Longrightarrow$ of \SF\ for every
occurrence of both $\li$, and $\liv$, and
(ii) erasing every occurrence of $!$ and $\$$.
Then, $\Gamma_F\vdash\ta{M}{A_F}$ can be deduced in \SF.
\end{fatto}
\textbf{Measures and structural properties.}
\par
\textit{Level or depth of deductions and terms}.
The \textit{level} or \textit{depth} $\dpth{\Pi}$ of a deduction $\Pi$ is
the maximal depth of every of its subdeductions:
\small
\begin{align*}
\dpth{\Pi(A)}&=0\\
\dpth{\Pi(R,\Pi')}&=\dpth{\Pi'}+1
&(R \in\{!,\$\})\\
\dpth{\Pi(R,\Pi')}&=\dpth{\Pi'}
&(R \text{ with a single premise, } R\not\in\{!,\$\})\\
\dpth{\Pi(R,\Pi',\Pi'')}&=\max\{\dpth{\Pi'},\dpth{\Pi''}\}
&(R \text{ with two premises})
\end{align*}
\normalsize
If $\Pi\rhd\Gamma;\Delta;{\mathcal E}\vdash\ta{M}{A}$, 
then $M$ has depth $\dpth{\Pi}$, namely, a term inherits the depth of the \textit{considered} deduction that types it.
For example, the deduction $\Pi_{wy(\bs x.x)}$:
\scriptsize
\[
\infer[\li E]
{\ta{w}{\,!\alpha\li\$!\BB\li\gamma};
  \emptyset;
  \{(\emptyset;\{\ta{y}{\alpha}\})\}
  \vdash\ta{wy(\bs x.x)}{\gamma}
}
{
 \infer[\li E_!]
  {\ta{w}{\,!\alpha\li\$!\BB\li\gamma};
   \emptyset;
   \{(\emptyset;\{\ta{y}{\alpha}\})\}
   \vdash\ta{wy}{\,\$!\BB\li\gamma}
  }
  {
   \infer[A]
   {\ta{w}{\,!\alpha\li\$!\BB\li\gamma};
    \emptyset;
    \emptyset
    \vdash\ta{w}{\,!\alpha\li\$!\BB\li\gamma}
   }
   {
   }
   &
   \infer[!]
   {\emptyset;
    \emptyset;
    \{(\emptyset;\{\ta{y}{\alpha}\})\}
    \vdash\ta{y}{\,!\alpha}
   }
   {
    \infer[A]
    {\ta{y}{\alpha};
     \emptyset;
     \emptyset
     \vdash\ta{y}{\alpha}
    }
    {}
   }
  }
 &
 \infer[\$]
 {\emptyset;
  \emptyset;
  \emptyset
  \vdash\ta{\bs x.x}{\$!\BB}
 }
 {
  \infer[!]
  {\emptyset;
   \emptyset;
   \emptyset
   \vdash\ta{\bs x.x}{\,!\BB}
  }
  {
   \infer[\li I]
   {\emptyset;
    \emptyset;
    \emptyset
    \vdash\ta{\bs x.x}{\BB}
   }
   {
    \infer[A]
    {\ta{x}{\beta};
     \emptyset;
     \emptyset
     \vdash\ta{x}{\beta}
    }
    {
    }
   }
  }
 }
}
\]
\normalsize
where $\BB$ abbreviates $(\beta\li\beta)$, has 
$\dpth{\Pi_{wy(\bs x.x)}}=2$ because we cross one instance of $!$ and one instance of $\$$, going from the conclusion to the rightmost axiom.
\par
\textit{Partial size of a deduction at a given depth.}
$\psz{d}{\Pi}$ is the \textit{partial size of a deduction} $\Pi$, at depth $d\leq\dpth{\Pi}$. $\psz{d}{\Pi}$ is defined by induction on the last rule of $\Pi$:
\small
\begin{align*}
\psz{0}{\Pi(A)} &= 1
\\
\psz{d}{\Pi(A)} &= 0 & (d\geq 1)
\\
\psz{0}{\Pi(C,\Pi')} &= \psz{0}{\Pi'}
\\
\psz{d}{\Pi(C,\Pi')} 
  &= \psz{d}{\Pi'}+1 & (d\geq 1)
\\
\psz{d}{\Pi(R,\Pi')} &= \psz{d}{\Pi'} 
  & (R\in\{\forall I,\forall E\},d\geq 0)
\\
\psz{0}{\Pi(R,\Pi')} 
  &= \psz{0}{\Pi'}+1 & R\in\{\li I
                            ,\li I_{!}
                            ,\li I_{\$}
                            ,\liv I
                            \}
\\
\psz{d}{\Pi(R,\Pi')} 
  &= \psz{d}{\Pi'} 
  & (R\in\{\li I
          ,\li I_{!}
          ,\li I_{\$}
          ,\liv I
         \}, d\geq 1)
\\
\psz{d}{\Pi(R,\Pi',\Pi'')} 
  &= \psz{d}{\Pi'}
   + \psz{d}{\Pi''}
   + 1 
   & (R\in\{\li E
           ,\li E_!
           ,\liv E
           \}, d\geq0)
\\
\psz{0}{\Pi(R)} 
  &= 0
  & (R\in\{!, \$\})
\\
\psz{d}{\Pi(R,\Pi')} 
  &= \psz{d-1}{\Pi'}
  & (R\in\{!, \$\}, d\geq 1)
\end{align*}
\normalsize
The partial size at a depth outside the interval $0,\ldots,\dpth{\Pi}$ is taken equal to $0$.
The partial size over estimates the intuitive notion of partial size at a given level.
In particular, $\psz{d}{\Pi_M(R)}$ counts the instances of the rules $A, \li I, \li I_!, \li I_{\$}, \liv I$ in $\Pi_M(R)$, which do not contract any variables in $M$ after $d$ instances of $!$, and $\$$ from $R$.
On the contrary, instances of $\li E, \li E_!, \liv E$ and $C$, used to contract variables of $M$, are always counted.
Also, we observe that we do not count the contractions at level $0$ because they cannot exist there.
Reconsidering $\Pi_{wy(\bs x.x)}$ above, we get:
(i) $\psz{0}{\Pi_{wy(\bs x.x)}}=3$,
(ii) $\psz{1}{\Pi_{wy(\bs x.x)}}=3$, and
(iii) $\psz{2}{\Pi_{wy(\bs x.x)}}=4$.
(i) holds because we count the occurrences of $\li E$, $\li E_!$, and the single occurrence of $A$ outside the scope of the instances of $!$, and $\$$. 
(ii) holds because we count the single axiom with conclusion $\ta{y}{\alpha};\emptyset;\emptyset\vdash\ta{y}{\alpha}$, plus the occurrences of $\li E$ and $\li E_!$.
(iii) holds because we count the axiom with conclusion $\ta{x}{\beta};\emptyset;\emptyset\vdash\ta{x}{\beta}$, the rule $\li I$ below it, and the occurrences of $\li E$ and $\li E_!$.
\par
\textit{Width of a deduction at a given depth.}
$\wdth{d}{\Pi}$ is the \emph{width} of a deduction $\Pi$ at depth $d\leq \dpth{\Pi}$.  $\wdth{d}{\Pi}$ is defined by induction on the last rule of $\Pi$:
\small
\begin{align*}
\wdth{0}{\Pi}&= 0 \\
\wdth{d}{\Pi(A)}&= 0 & (d\geq 1)\\
\wdth{1}{\Pi(C,\Pi')}&= \wdth{1}{\Pi'}+1\\
\wdth{d}{\Pi(C,\Pi')}&= \wdth{d}{\Pi'} & (d> 1)\\
\wdth{d}{\Pi(R,\Pi')}&= \wdth{d}{\Pi'}
& (R\in\{\forall I, \forall E, \li I,\li I_{\$},\li I_{!},\liv I \}, d\geq 1)\\
\wdth{1}{\Pi(R,\Pi',\Pi'')}&= \wdth{1}{\Pi'}+\wdth{1}{\Pi''}+1
& (R\in\{\li E, \li E_!,\liv E\})\\
\wdth{d}{\Pi(R,\Pi',\Pi'')}&= \wdth{d}{\Pi'}+\wdth{d}{\Pi''}
& (R\in\{\li E, \li E_!,\liv E \}, d> 1)\\
\wdth{d}{\Pi(R,\Pi')}&= \wdth{d-1}{\Pi'}
& (R\in\{\$, !\}, d\geq 1)
\end{align*}
\normalsize
The width at a depth outside the interval $0,\ldots,\dpth{\Pi}$ is taken equal to $0$.
We observe that the first clause states that no variable contraction, by means of instances of $C, \li E, \li E_!$, and $\liv E$, can exist at level $0$. So, we do not count those rules as part of the width at that level.
We are essentially interested to observe the width at level one, where, relatively to the context, the substitutions --- hence the possible duplications --- may occur in the course of the normalization. To preserve the overall complexity bounds on the normalization the number of these duplications, namely the width that regulates them, cannot be too big.
For example, let us assume $\Pi_{(\bs x.x)(\bs fy.f(fy))}\rhd \vdash\ta{(\bs x.x)(\bs fy.f(fy))}{\UIntT\li\UIntT}$,
$\UIntT$ being $\forall \alpha.!(\alpha\li\alpha)\li\$(\alpha\li\alpha)$.
Then, 
$\wdth{1}{\Pi_{(\bs x.x)(\bs fy.f(fy))}}
=\wdth{1}{\Pi_{\bs x.x}}+\wdth{1}{\Pi_{\bs fy.f(fy)}}+1
=\wdth{1}{\Pi_{\bs y.f(fy)}}+1
=\wdth{1}{\Pi_{\bs y.f_1(f_2 y)}}+2
=\wdth{0}{\Pi_{\bs y.f_1(f_2 y)}}+2
=2$, counting $\li E$, and $C$.

\begin{lemma}[Width and size of deductions and terms.]
For every deduction $\Pi$:
\begin{enumerate}
\item 
\label{lemma:Width of a deduction vs. its partial size}
$\wdth{d}{\Pi}\leq\psz{d}{\Pi}$, with $0\leq d\leq \dpth{\Pi}$.
\item 
\label{lemma:Partial size vs. absolute size of a term}
If $\Pi\rhd\Gamma;\Delta;{\mathcal E}\vdash\ta{M}{A}$, then
$\size{M}\leq\sum_{d=0}^{\dpth{\Pi}}\psz{d}{\Pi}$.
\end{enumerate}
\end{lemma}
The two points hold by structural induction on $\Pi$. In particular the size of a term cannot be greater than the size of a deduction that gives it a type because the instances of the $C$ rule disappear in the $\lambda$-terms.
\par
The structural properties, here below, are the preliminary steps to prove the substitution property (Lemma~\ref{lemma:subst-vs-wght} of Section~\ref{section:Dynamic properties of WALL}), which, in turn, serves to show that \WALL\ enjoys the subject reduction (Theorem~\ref{theorem:substitution-property} of Section~\ref{section:Dynamic properties of WALL}) with respect to a suitable notion of reduction, and with the wanted polynomial bound. Essentially, the structural properties here below say that the assumptions of a deduction can be weakened, or deleted, fix bounds on the number of occurrences of a variable in a typeable term, and highlight the structure of the subdeductions that introduce a variable in the type assignments and contexts of a judgment.
\begin{lemma}[Structural properties.]
\label{lemma:structural-properties-WALL}
Let $\Pi(R)\rhd \Gamma;\Delta;{\mathcal E}\vdash\ta{M}{B}$, and
${\mathcal E}=
\{(\Theta_0;\emptyset),(\Theta_1;\Phi_1),$ $\ldots,(\Theta_m;\Phi_m)\}$.
\begin{enumerate}
\item 
\label{lemma:structural-properties-WALL-(-1)}
For every linear type $L$, there exists 
$\Pi'(R)\rhd
 \Gamma\subs{L}{\alpha};
 \Delta\subs{L}{\alpha};
 {\mathcal E}\subs{L}{\alpha}\vdash\ta{M}{B\subs{L}{\alpha}}$ such that
$\dpth{\Pi}=\dpth{\Pi'}$, $\psz{d}{\Pi}=\psz{d}{\Pi'}$ and $\wdth{d}{\Pi}=\wdth{d}{\Pi'}$, for every $d\leq\dpth{\Pi}$.

\item 
\label{lemma:structural-properties-WALL-0-nuovo}
Then, $\nocc{x}{M}\leq 1$, for every variable $x$ of $\dom{\Gamma}\cup\dom{\Delta}\cup(\bigcup^{m}_{i=0}\dom{\Theta_i})$.

\item 
\label{lemma:structural-properties-WALL-1}
For every $d\leq \dpth{\Pi}$ and ${\mathcal E}'$,
there is
$\Pi'\rhd
 \Gamma';
 \Delta';
 {\mathcal E}\sqcup{\mathcal E}'\vdash\ta{M}{A}$, 
such that 
$\Gamma\subseteq\Gamma'$,
$\Delta\subseteq\Delta'$,
$\dpth{\Pi}=\dpth{\Pi'}$,
$\wdth{d}{\Pi}=\wdth{d}{\Pi'}$, and
$\psz{d}{\Pi}=\psz{d}{\Pi'}$.

\item 
\label{lemma:structural-properties-WALL-4}
For every 
$x\in\dom{\Gamma}\cup\dom{\Delta}\cup(\bigcup^{m}_{i=0}\dom{\Theta_i})\cup(\bigcup^{m}_{i=1}\dom{\Phi_i})$,
if $\nocc{x}{M}=0$, then there exists
$\Pi'\rhd\Gamma';\Delta';{\mathcal E}'\vdash\ta{M}{A}$ such that
$x\not\in\dom{\Gamma'}\cup\dom{\Delta'}\cup{\mathcal E}'$ with the same depth, width and size of $\Pi$.

\item 
\label{lemma:structural-properties-WALL-5}
For every linear partially discharged type assignment $\ta{x}{A}\in\Delta$, there is
$\Pi'(R')\preceq\Pi$ introducing $\ta{x}{A}$ such that
$R'\in\{A,\$\}$ and
$\Pi'(R')\rhd
\Gamma';\Delta',\ta{x}{A};{\mathcal E}'\vdash\ta{N}{C}$,
for some
$\Gamma', \Delta', {\mathcal E}'$, and $C$. If $R'\equiv A$, then $\nocc{x}{M}=0$.

\item 
\label{lemma:structural-properties-WALL-6}
For every elementary partially discharged type assignment $\ta{x}{A}\in\Theta_i$, there is
$\Pi'(R')\preceq\Pi$ introducing $\ta{x}{A}$ such that
$R'\in\{A,\$, !\}$ and
$\Pi'(R')\rhd
\Gamma';\Delta';{\mathcal E}',(\Theta',\ta{x}{A};\Phi')\vdash\ta{N}{C}$,
for some
$\Gamma', \Delta', {\mathcal E}', \Theta', \Phi'$, and $C$. If $R'\equiv A$, then $\nocc{x}{M}=0$.

\item 
\label{lemma:structural-properties-WALL-2}
For every polynomial partially discharged assignment $\ta{x}{A}\in\Phi_1\cup\ldots\cup\Phi_m$, there are $n\geq 1$ and $q_1,\ldots,q_n\geq 0$ such that $\wdth{1}{\Pi(R)}\geq q_1+\ldots+q_n$ and the following three points hold:
\begin{enumerate}
\item 
there is $M'$ such that $M$ can be written as 
$M'\subs{x}{x^1_1\ldots x^1_{q_1}\ldots\ldots x^n_1\ldots x^n_{q_n}}$;

\item 
for every $1\leq i\leq n$,
there is
$\Pi'_i(R_i)\rhd
\Gamma_i;\Delta_i;
{\mathcal E}_i,
(\Theta^i_1;\{\ta{x^i_1}{A}\}),
\ldots,
(\Theta^i_{q_i};$ $\{\ta{x^i_{q_i}}{A}\})
\vdash\ta{P_i}{C_i}$, subdeduction of $\Pi$,
with $R_i\in\{A, \$, !\}$,
that introduces $\ta{x^i_1}{A},\ldots,\ta{x^i_{q_i}}{A}$;

\item 
$q_1+\ldots+q_n-1$ instances of $C, \li E, \li E_!, \liv E$ are required in the tree with the conclusion of $\Pi$ as root and the conclusions of all the deductions
$\Pi'_1,\ldots, \Pi'_n$ as leaves to contract $x^1_1\ldots x^1_{q_1}\ldots\ldots x^n_1\ldots x^n_{q_n}$ to $x$.
\end{enumerate}

\item 
\label{lemma:structural-properties-WALL-3.1}
If $M\in$ \PTV, and $B$ is $!A$, for some $A$, then $R\in\{!\}$
and $\FV{M}\subseteq\dom{\mathcal E}$.

\item 
\label{lemma:structural-properties-WALL-3.3}
If $M\in$ \PTV, and $B$ is $\$A$, for some $A$, then $R\in\{\$, C\}$
and $\FV{M}\subseteq\dom{\Delta}\cup\dom{{\mathcal E}}$.
\end{enumerate}
\end{lemma}
Point \ref{lemma:structural-properties-WALL-(-1)} holds because the substitution of linear types for a variable of a type $A$ cannot change the nature of $A$: it remains linear if it was as such before the substitution, or, in the other case, modal.
Point \ref{lemma:structural-properties-WALL-0-nuovo} holds because, by definition, the domains of linear, linear partially discharged and elementary partially discharged type assignments that belong to distinct premises must be disjoint. This implies that $\nocc{x}{M}$ cannot be greater that $1$. We also admit weakening on the type assignments of the rules $A, !$, and $\$$. So, a variable name may also not be occurring in $M$.
Point \ref{lemma:structural-properties-WALL-1} holds by using the weakening implicit in the rules $A, !, \$$.
Point \ref{lemma:structural-properties-WALL-4} holds by omitting the use of weakening implicit in the rules $A, !, \$$.
Points \ref{lemma:structural-properties-WALL-5} and \ref{lemma:structural-properties-WALL-6} hold by simply inspecting the rules and observing that the only rules that introduce linear partially discharged and elementary partially discharged type assignments are $A, !$, and $\$$. In particular, $A$ can only introduce them as fake assumptions.
Points \ref{lemma:structural-properties-WALL-2}, \ref{lemma:structural-properties-WALL-3.1} and \ref{lemma:structural-properties-WALL-3.3}, by structural induction on $\Pi$.
\par
\textbf{Notation.} 
The definition of partially discharged context justifies to
shorten
$\Pi\rhd\Gamma;\Delta;{\mathcal E}
 \sqcup\{(\Theta;\{\ta{x}{A}\})\}\vdash\ta{M}{B}$, or
$\Pi\rhd\Gamma;\Delta;{\mathcal E},
 \{(\Theta;\{\ta{x}{A}\})\}\vdash\ta{M}{B}$
by means of
$\Pi\rhd\Gamma;\Delta;{\mathcal E}
 \sqcup\{(\Theta;\ta{x}{A})\}\vdash\ta{M}{B}$, or
$\Pi\rhd\Gamma;\Delta;{\mathcal E},
 (\Theta;\ta{x}{A})\vdash\ta{M}{B}$.
\section{Dynamic properties}
\label{section:Dynamic properties of WALL}
\textbf{What we already know.}
\cite{Asperti:1998-LICS,Asperti02TOCL,BaillotTerui:2004-LICS,Cop-DLag-Ron:EALCBV-04} remark the independence between the normalization of the deductions of some deductive system, derived from Linear logic for characterizing some computational class, and the standard $\beta$-reduction of usual $\lambda$-terms. Unsurprisingly, we have analogous phenomenon with \WALL. We can observe, indeed, that the deduction 
$\Pi\rhd\ta{y}{\, !A\li!A\li A},\ta{w}{A\li!A},\ta{z}{A};
\emptyset;\emptyset\vdash\ta{(\bs x.yxx)(wz)}{A}$ exists, 
but we cannot build the one giving type to the $\beta$-reduct $y(wz)(wz)$ of $(\bs x.yxx)(wz)$. This because
$\ta{y}{\, !A\li!A\li A},\ta{w}{A\li!A},\ta{z}{A};
\emptyset;\emptyset\vdash\ta{y(wz)}{\,!A\li A}$ and
$\ta{w}{A\li!A},\ta{z}{A};\emptyset;\emptyset\vdash\ta{wz}{\,!A}$
would require an instance of $\li E_!$ where the domain of the linear type assignments in its two assumptions intersect.
The problem persists even when the $\beta$-redex contains only linear variables and no $!$-modal types at all. For example, consider the following deduction $\Pi_{(\bs x.M)(wz)}$, where $x\in\FV{M}$:
\scriptsize
\[
\infer[\li E]
{
 \ta{w}{C\li \$A},\ta{z}{C};\emptyset;\emptyset
 \vdash\ta{(\bs x.M)(wz)}{\$B}
}
{
 \infer[\li I_{\$}]
 {\emptyset;\emptyset;\emptyset\vdash\ta{\bs x.M}{\$A\li\$B}}
 {
  \emptyset;\ta{x}{A};\emptyset\vdash\ta{M}{\$B}
 }
 &
 \infer[\li E]
 {\ta{w}{C\li\$A},\ta{z}{C};\emptyset;\emptyset\vdash\ta{wz}{\$A}}
 {
  \infer[A]
  {\ta{w}{C\li\$A};\emptyset;\emptyset\vdash\ta{w}{C\li\$A}}
  {}
  &
  \infer[A]
  {\ta{z}{C};\emptyset;\emptyset\vdash\ta{z}{C}}
  {}
 }
}
\]
\normalsize
The $\beta$-reduction $(\bs x.M)(wz)\rightarrow_{\beta}M\subs{wz}{x}$ would correspond to eliminate  the sequence of rules $\li I_{\$}$ and $\li E$ in $\Pi_{(\bs x.M)(wz)}$. Such an elimination would leave us with the conclusion of 
$\Pi'(\li E)\rhd
\ta{w}{C\li\$A},\ta{z}{C};\emptyset;\emptyset\vdash\ta{wz}{\$A}$ that must be plugged into the partially discharged assumption $\ta{x}{A}$ of $\emptyset;\ta{x}{A};\emptyset\vdash\ta{M}{\$B}$. But this is structurally illegal, since the conclusion of $\Pi'(\li E)$ and the partially discharged assumption $\ta{x}{A}$ live at different depths.
\par
\textbf{The restriction on the contexts of $\liv E$.}
In Section~\ref{section:Weak Affine Light Typing WALL} we have intuitively described how $\liv$ forces the correct substitution discipline, relatively to the elementary partially discharged assumptions. Let us assume, for a moment, to relax $\liv E$ to $\liv E'$\footnote{This rule was used in an earlier version of this work and the associated counterexample was pointed out by an anonymous referee.} as follows:
\small
\[
\infer[\liv E']
  {
   \Gamma_M,\Gamma_N; 
   \Delta;
   {\mathcal E}_M
   \sqcup
   {\mathcal E}_N
   \vdash \ta{MN}{B}
  }
  {
   \Gamma_M; 
   \Delta;
   {\mathcal E}_M
   \vdash \ta{M}{\$A\liv B}
   &
   \Gamma_N; 
   \emptyset;
   {\mathcal E}_N
   \vdash \ta{N}{\$A}
   &
   {\mathcal E}_N\subseteq\{(\Theta;\emptyset)\}
   }
\]
\normalsize
with an arbitrary $\Gamma_N$. Then, we could write the following deduction:
\scriptsize
\[
\infer[\li E]
{
 \Gamma_M;\Delta_M;{\mathcal E}_M\sqcup\{(\emptyset;\ta{W}{D})\}
 \vdash\ta{(\bs x.M(xI))(\bs z.N)}{B}
}
{
 \infer[\li I]
 {
  \Gamma_M;\Delta_M;{\mathcal E}_M\vdash\ta{\bs x.M(xI)}{(C\li \$A)\li B}
 }
 {
  \infer[\liv E']
  {
   \Gamma_M,\ta{x}{C\li\$A};\Delta_M;{\mathcal E}_M\vdash\ta{M(xI)}{B}
  }
  {
   \Gamma_M;\Delta_M;{\mathcal E}_M\vdash\ta{M}{\$A\liv B}
   &
\infer[\li E]
{
\ta{x}{C\li\$A};\emptyset;\emptyset\vdash\ta{xI}{\$A}
}
{
\infer[A]
{\ta{x}{C\li\$A};\emptyset;\emptyset\vdash\ta{x}{C\li\$A}}
{}
&
\Pi_I\rhd \emptyset;\emptyset;\emptyset\vdash \ta{I}{C}
}
  }  
 }
 &
 \infer[\li I]
 {
 \emptyset;\emptyset;\{(\emptyset;\ta{w}{D})\}\vdash\ta{\bs z.N}{C\li\$A}
 }
 {
  \infer[\$ ]
  {
   \ta{z}{C};\emptyset;\{(\emptyset;\ta{w}{D})\}\vdash\ta{N}{\$A}
  }
  {
   \Pi_N\rhd \ta{w}{D};\emptyset;\emptyset\vdash\ta{N}{A}
  }
 }
}
\]
\normalsize
where $C$ stands for $\alpha\li\alpha$, and $I$ for $\bs x.x$.
Indifferently using the $\beta$-reduction, or its call-by-value version, namely the one where $(\bs x.M)N\rightarrow_{\beta_v} M\subs{N}{x}$ only if $N$ is a variable or a $\lambda$-abstraction, we would rewrite $(\bs x.M(xI))(\bs z.N)$ to $MN$, typeable  with:
\small
\[
  \infer[\liv E']
  {
   \Gamma_M;\Delta_M;{\mathcal E}_M\sqcup\{(\emptyset;\ta{w}{D})\}
   \vdash\ta{MN}{B}
  }
  {
   \Gamma_M;\Delta_M;{\mathcal E}_M\vdash\ta{M}{\$A\liv B}
   &
     \infer[\$ ]
     {
      \emptyset;\emptyset;\{(\emptyset;\ta{w}{D})\}\vdash\ta{N}{\$A}
     }
     {
      \Pi_N\rhd \ta{w}{D};\emptyset;\emptyset\vdash\ta{N}{A}
     }
  }  
\]
\normalsize
Namely, a relaxed set of assumptions in $\liv E'$ would allow to generate a $\$$-box that depends on a polynomial assumption. That $\$$-box could be used for building a $!$-box, here inside $M$, with more than one polynomial assumption.
\par
\textbf{The side condition on the rule $!$.}
We focus on the condition $\Theta\neq\emptyset \Rightarrow \dom{\Phi}\cap\FV{M}\neq\emptyset$ of the rule $!$. It is justified by our goal to control the duplication at the level of the deductions directly inside the syntax of the typeable $\lambda$-terms. Let us assume to have a typeable term $(\lambda x.M)N$ where $\nocc{x}{M}>1$ and let $N$ be a value such that $\FV{N}\subseteq\{y\}$, for some $y$. In principle we are in front of a redex. But this is true only if the type of $y$ is $!$-modal. A relaxed version of $!$ like:
\small
\[
\infer[!']
 {\Gamma';
  \Delta;
  \{(\$\Theta';\emptyset)\}\sqcup\{(\Theta;\Phi)\}
  \vdash \ta{M}{\ !B}
 }
 {
  \Gamma;
  \emptyset;
  \{(\Theta';\emptyset)\}
  \vdash \ta{M}{B}
  &
  \Gamma
  \subseteq
  \Theta\cup\Phi
  &
  \Theta\neq\emptyset \Rightarrow \Phi\neq\emptyset
 }
\]
\normalsize
for example, would allow to derive:
\small
\[
\infer[!']
 {\Gamma';
  \Delta;
  \{(\{\ta{x}{A}\};\ta{y}{C})\}
  \vdash \ta{\bs w.M}{\,!B}
 }
 {
  \ta{x}{A},\ta{y}{C};
  \emptyset;
  \emptyset
  \vdash \ta{\bs w.M}{B}
  &
  y\not\in\FV{M}
 }
\]
\normalsize
where $\bs w.M$ depends on an elementary partially discharged assumption that cannot be duplicated in case the whole $\bs w.M$ is. In a sense, the current precondition on the rule $!$ assures that it depends on elementary partially discharged assumptions only when necessary.
\par
\textbf{There is also $\li E_!$.}
The rule $\li E_!$ prevents the existence of exponential free variables in the term $M$. To see why, let us assume to drop such a constraint and to consider the following derivation with a relaxed version $\li E'_!$ of $\li E_!$:
\small
\[
\infer[\li E'_!]
{
 \emptyset;\emptyset;\{(\{\ta{x}{C}\};\emptyset)\}\vdash \ta{(\bs y.M)I}{\,!B}
}
{
 \infer[\li I_!]
 {
 \emptyset;\emptyset;\{(\{\ta{x}{C}\};\emptyset)\}\vdash\ta{\bs y. M}{\,!A\li !B}
 }
 {
  \infer[!]
  {
   \emptyset;\emptyset;\{(\{\ta{x}{C}\};\ta{y}{A})\}\vdash\ta{M}{\,!B}
  }
  {
   \Pi_M\rhd\ta{x}{C},\ta{y}{A};\emptyset;\emptyset\vdash\ta{M}{B}
   &
   x,y \in\FV{M}
  }
 }
 &
 \infer[!]
 {
  \emptyset;\emptyset;\emptyset\vdash\ta{I}{\,!A}
 }
 {
  \Pi_{I}\rhd \emptyset;\emptyset;\emptyset\vdash\ta{I}{A}
 } 
}
\]
\normalsize
where $A\equiv\alpha\li\alpha$ and $I\equiv\bs x.x$. Reasonably, $(\bs y.M)I$ could be reduced to $M\subs{I}{y}$, corresponding to:
\small
\[
  \infer[!]
  {
   \emptyset;\emptyset;\{(\{\ta{x}{C}\};\emptyset)\}\vdash\ta{M\subs{I}{y}}{\,!B}
  }
  {
   \Pi_{M\subs{I}{y}}\rhd\ta{x}{C};\emptyset;\emptyset
   \vdash\ta{M\subs{I}{y}}{B}
   &
   x \in\FV{M}
  }
\]
\normalsize
with an illegal application of our current rule $!$.
\par
\textbf{The formal counterpart.}
All the above observations imply what follows.
\begin{lemma}[Substitution property.]
\label{lemma:subst-vs-wght}
Let $N$ be a value of \PTV, and $x, x_1, \ldots, x_n$ belong to $\FV{M}$.
\begin{enumerate}

\item
\label{lemma:subst-vs-wght-01}
If
$\Pi_M\rhd
\Gamma_M,\ta{x}{L};
\Delta_M;
{\mathcal E}_M
\vdash \ta{M}{B}$, 
and
$\Pi_N\rhd
\Gamma_N; 
\Delta_N;
{\mathcal E}_N
\vdash \ta{N}{L}$,
then there exists
$\Pi_{M\subs{N}{x}}
\rhd
\Gamma_M, \Gamma_N; 
\Delta_M, \Delta_N; 
{\mathcal E}_M\sqcup{\mathcal E}_N
\vdash \ta{M\subs{N}{x}}{B}$ such that:
\begin{enumerate}
	\item 
	\label{lemma:subst-vs-wght-01-b}
	$\dpth{\Pi_{M\subs{N}{x}}}=\max\{\dpth{\Pi_{M}},\dpth{\Pi_{N}}\}$;
	
	\item 
	\label{lemma:subst-vs-wght-01-a}
	$\wdth{d}{\Pi_{M\subs{N}{x}}}=\wdth{d}{\Pi_{M}}+\wdth{d}{\Pi_{N}}$, for every $d\geq 0$;
	
	\item
	\label{lemma:subst-vs-wght-01-c}
	$\psz{0}{\Pi_{M\subs{N}{x}}}<\psz{0}{\Pi_{M}}+\psz{0}{\Pi_{N}}$;
	
	\item
	\label{lemma:subst-vs-wght-01-d}
	$\psz{d}{\Pi_{M\subs{N}{x}}}=\psz{d}{\Pi_{M}}+\psz{d}{\Pi_{N}}$, for every $d\geq 1$.
\end{enumerate}

\item
\label{lemma:subst-vs-wght-03}
If 
$\Pi_M\rhd
\Gamma_M;
\Delta_M,\ta{x}{A}; 
{\mathcal E}_M
\vdash \ta{M}{B}$,
and
$\Pi_N\rhd
\Gamma_N; 
\Delta_N; 
{\mathcal E}_N
\vdash \ta{N}{\$A}
$, 
then there exists
$\Pi_{M\subs{N}{x}}\rhd
\Gamma_M, \Gamma_N; 
\Delta_M, \Delta_N;
{\mathcal E}_M\sqcup{\mathcal E}_N
\vdash \ta{M\subs{N}{x}}{B}$ such that:
\begin{enumerate}
	\item 
	\label{lemma:subst-vs-wght-03-b}
	$\dpth{\Pi_{M\subs{N}{x}}}=\max\{\dpth{\Pi_{M}},\dpth{\Pi_{N}}\}$;
	
	\item 
	\label{lemma:subst-vs-wght-03-a}
	$\wdth{d}{\Pi_{M\subs{N}{x}}}=\wdth{d}{\Pi_{M}}+\wdth{d}{\Pi_{N}}$, for every $d\geq 0$;
	
	\item
	\label{lemma:subst-vs-wght-03-c}
	$\psz{0}{\Pi_{M\subs{N}{x}}}=\psz{0}{\Pi_{M}}$;
	
	\item
	\label{lemma:subst-vs-wght-03-d}
	$\psz{d}{\Pi_{M\subs{N}{x}}}\leq\psz{d}{\Pi_{M}}+\psz{d}{\Pi_{N}}$, 
	for every $d\geq 1$.
\end{enumerate}

\item
\label{lemma:subst-vs-wght-05}
If
$\Pi_M\rhd
\Gamma_M;
\Delta_M; 
{\mathcal E}_M,(\Theta_M,\{\ta{x}{A}\};\emptyset)
\vdash \ta{M}{B}
$,
and
$\Pi_N\rhd
\emptyset;
\emptyset;
{\mathcal E}_{N}
\vdash \ta{N}{\$A}
$, with ${\mathcal E}_{N}\subseteq\{(\Theta_N;\emptyset)\}$,
then there exists
$\Gamma_M;\Delta_M;
{\mathcal E}_M,
\{(\Theta_M;\emptyset)\}\sqcup{\mathcal E}_{N}
\vdash \ta{M\subs{N}{x}}{B}$ such that:
\begin{enumerate}
	\item 
	\label{lemma:subst-vs-wght-05-b}
	$\dpth{\Pi_{M\subs{N}{x}}}=\max\{\dpth{\Pi_{M}},\dpth{\Pi_{N}}\}$;
	
	\item 
	\label{lemma:subst-vs-wght-05-a}
	$\wdth{d}{\Pi_{M\subs{N}{x}}}=\wdth{d}{\Pi_{M}}+\wdth{d}{\Pi_{N}}$, for every $d\geq 0$;
	
	\item
	\label{lemma:subst-vs-wght-05-c}
	$\psz{0}{\Pi_{M\subs{N}{x}}}=\psz{0}{\Pi_{M}}$;
		
	\item
	\label{lemma:subst-vs-wght-05-d}
	$\psz{d}{\Pi_{M\subs{N}{x}}}\leq\psz{d}{\Pi_{M}}+\psz{d}{\Pi_{N}}$,
	for every $d\geq 1$.
\end{enumerate}

\item
\label{lemma:subst-vs-wght-02}
Let
$\Pi_M\rhd
\Gamma_M;\Delta_M;{\mathcal E}_M,(\emptyset;\ta{x}{A})
\vdash \ta{M}{B}$,
$\nocc{x}{M} = 1$,
and
$\Pi_N\rhd
\Gamma_N;\Delta_N;{\mathcal E}_{N}\vdash \ta{N}{\ !A}$, with
${\mathcal E}_{N}\subseteq\{(\Theta_N;\ta{y}{C})\}$.
Then, there exists
$\Gamma_M,\Gamma_N;\Delta_M,\Delta_N;
{\mathcal E}_M\sqcup{\mathcal E}_{N}
\vdash 
\ta{M\subs{N}{x}}{B}$ such that:
\begin{enumerate}
	\item 
	\label{lemma:subst-vs-wght-02-b}
	$\dpth{\Pi_{M\subs{N}{x}}}=\max\{\dpth{\Pi_{M}},\dpth{\Pi_{N}}\}$;
	
	\item 
	\label{lemma:subst-vs-wght-02-a}
	$\wdth{d}{\Pi_{M\subs{N}{x}}}\leq\wdth{d}{\Pi_{M}}$, for every $0\leq d\leq 1$;
	
	\item
	\label{lemma:subst-vs-wght-02-c}
	$\psz{0}{\Pi_{M\subs{N}{x}}}=\psz{0}{\Pi_{M}}$;
		
	\item
	\label{lemma:subst-vs-wght-02-d}
	$\psz{d}{\Pi_{M\subs{N}{x}}}\leq\psz{d}{\Pi_{M}}+\psz{d}{\Pi_{N}}$,
	for every $d\geq1$.
\end{enumerate}

\item
\label{lemma:subst-vs-wght-04}
Let
$\Pi_M\rhd
\Gamma_M;\Delta_M;{\mathcal E}_M,
(\emptyset;\ta{x}{A})
\vdash \ta{M}{B}$, $\nocc{x}{M}> 1$, and
$\Pi_N\rhd
\Gamma_N;\Delta_N;{\mathcal E}_{N}\vdash \ta{N}{\ !A}$, with
${\mathcal E}_{N}\subseteq\{(\emptyset;\ta{y}{C})\}$.
Then, there exists
$\Gamma_M, \Gamma_N;\Delta_M, \Delta_N;
{\mathcal E}_M\sqcup{\mathcal E}_N
\vdash 
\ta{M\subs{N}{x}}{B}$ such that:
\begin{enumerate}
	\item 
	\label{lemma:subst-vs-wght-04-b}
	$\dpth{\Pi_{M\subs{N}{x}}}=\max\{\dpth{\Pi_{M}},\dpth{\Pi_{N}}\}$;
	
	\item 
	\label{lemma:subst-vs-wght-04-a}
	$\wdth{d}{\Pi_{M\subs{N}{x}}}
	\leq
	\wdth{d}{\Pi_{M}}$ for $0\leq d\leq 1$;
	
	\item
	\label{lemma:subst-vs-wght-04-c}
	$\psz{0}{\Pi_{M\subs{N}{x}}}=\psz{0}{\Pi_{M}}$;
		
	\item
	\label{lemma:subst-vs-wght-04-d}
	$\psz{d}{\Pi_{M\subs{N}{x}}}\leq\psz{d}{\Pi_{M}}+\nocc{x}{M}\psz{d}{\Pi_{N}}$,
	for every $d\geq 1$.
\end{enumerate}
\end{enumerate}
\end{lemma}
The lemma above can be proved by induction on the derivation $\Pi_M$, using Lemma~\ref{lemma:structural-properties-WALL}.
\par
As a first observation, let us look at point~\ref{lemma:subst-vs-wght-04} of Lemma~\ref{lemma:subst-vs-wght} here above. It says that the width at level 0 and 1 cannot increase. The reason are the requirements on $N$ which must be a value with a $!$-modal type and with at most a single free variable. This means that the deduction that gives the type to $N$ can contain instances of the rules $C, \li E, \li E_!$, and $\liv E$ --- those which may increase the width --- only at levels at least 2.
On the other side, point~\ref{lemma:subst-vs-wght-02} of Lemma~\ref{lemma:subst-vs-wght} says that when a polynomial variable occurs free only once in $M$ we can relax a little bit the conditions on $N$: it must still be a value with a $!$-modal type, but it may depend on more than a single variable.
\par
The subject reduction above suggests how to restrict the $\beta$-reduction on the $\lambda$-terms so that those ones which are typeable by \WALT\ enjoy the subject reduction.
\begin{definition}[Rewriting \PT.]
\label{definition:Rewriting PT}
\begin{itemize}
\item
\textbf{Generic rewriting relation.}
The relation $\red\subseteq$\PT$^2$ is the contextual closure of the rewriting relation
$\blacktriangleright\subseteq$\PT$^2$, such that
$(\bs x.M)N\blacktriangleright M\subs{N}{x}$ if, and only if:
\small
\begin{align}
\label{align:red-1}
\text{either } & \nocc{x}{M}=0\\
\label{align:red-2}
\text{or     } & \nocc{x}{M}=1 \text{ and } N\in\text{\PTV}\\
\label{align:red-3}
\text{or     } & \nocc{x}{M}>1 \text{ and } N\in\text{\PTV},
                                            \FV{N}\subseteq\{y\},
					    \text{ for some } y
\end{align}
\normalsize
$\red^+$ is the transitive closure of $\red$, while
$\red^*$ is the reflexive and transitive closure of $\red$.
$M$ is in $\red$-normal form, and we write $\nf{M}$, if $\red$ cannot rewrite $M$ anymore.
\item 
\textbf{Rewriting relation by depth.}
Assume $\Pi\rhd\Gamma;\Delta;{\mathcal E}\vdash\ta{M}{A}$. Let
$M\red N$ by means of the reduction of a redex 
$(\bs x.P)Q\red P\subs{Q}{x}$. If $(\bs x.P)Q$ is at depth $d$ in $\Pi$, then we write $M\redl{d} N$.
$M$ is in $\red$-normal form at depth $d$, and we write $\nfl{d}{M}$, if $\redl{d}$ cannot rewrite $M$ anymore.
\end{itemize}
\end{definition}
\textbf{Subject reduction.}
We prove the subject reduction (Theorem~\ref{theorem:substitution-property} below) in two steps. First we show that it holds at depth $0$. Then, we extend the result to any depth, observing that depth $d>0$, of any deduction $\Pi$ is, in fact, the depth $0$ of every subdeduction of $\Pi$ whose conclusion is at depth $d$.

\begin{lemma}[Subject reduction at depth $0$.]
\label{lemma:substitution-property}
Let us assume $(\bs x.M)N$ be at depth $0$ in
$\Pi_{(\bs x.M)N}\rhd\Gamma;\Delta;{\mathcal E}
\vdash\ta{(\bs x.M)N}{A}$.
If $(\bs x.M)N\red M\subs{N}{x}$, then there exists
$\Pi_{M\subs{N}{x}}\rhd\Gamma;\Delta;{\mathcal E}\vdash
\ta{M\subs{N}{x}}{A}$ such that:
\begin{enumerate}
\item 
\label{lemma:substitution-property-1}
$\dpth{\Pi_{M\subs{N}{x}}}\leq\dpth{\Pi_{(\bs x.M)N}}$;

\item 
\label{lemma:substitution-property-2}
$\wdth{d}{\Pi_{M\subs{N}{x}}}\leq\wdth{d}{\Pi_{(\bs x.M)N}}$, with $d\leq 1$;

\item
\label{lemma:substitution-property-3}
$\psz{0}{\Pi_{M\subs{N}{x}}}<\psz{0}{\Pi_{(\bs x.M)N}}$;

\item
\label{lemma:substitution-property-4}
$\psz{d}{\Pi_{M\subs{N}{x}}}\leq
\psz{d}{\Pi_{(\bs x.M)N}}+\nocc{x}{M}\psz{d}{\Pi_{N}}
$, for every $d\geq 1$.
\end{enumerate}
\end{lemma}
The lemma can be proved proceeding by cases on the definition of $\red$, using Lemma~\ref{lemma:subst-vs-wght}.

\begin{theorem}[Subject reduction.]
\label{theorem:substitution-property}
Let us assume 
$\Pi_{M}\rhd\Gamma;\Delta;{\mathcal E}\vdash\ta{M}{A}$.
Let us assume also $M\redl{d}N$ by means of the reduction of a redex
$(\bs x.P)Q\red P\subs{Q}{x}$ at depth $d\leq \dpth{\Pi_{M}}$ in $\Pi_{M}$.
Then, there exists
$\Pi_{N}\rhd\Gamma;\Delta;{\mathcal E}\vdash
\ta{N}{A}$ such that:
\begin{enumerate}
\item 
\label{theorem:substitution-property-1}
$\dpth{\Pi_{N}}\leq\dpth{\Pi_{M}}$, namely 
the reduction cannot increase the global depth.

\item 
\label{theorem:substitution-property-2}
$\wdth{i}{\Pi_{N}}\leq\wdth{i}{\Pi_{M}}$, for every $0\leq i\leq d+1$, namely
the reduction cannot increase the width at depth $0,1,\ldots,d+1$.

\item
\label{theorem:substitution-property-3}
$\psz{i}{\Pi_{N}}=\psz{i}{\Pi_{M}}$,
for every $0\leq i< d$, namely
the reduction cannot alter the size at depth $0,1,\ldots,d-1$.

\item
\label{theorem:substitution-property-4}
$\psz{d}{\Pi_{N}}<\psz{d}{\Pi_{M}}$, namely the reduction
strictly consumes structure at the depth it occurs.

\item
\label{theorem:substitution-property-5}
$\psz{i}{\Pi_{N}}\leq
\psz{i}{\Pi_{M}}+\nocc{x}{P}\psz{i}{\Pi_{Q}}$,
for every $d< i\leq \dpth{\Pi_{M}}$, namely
the reduction may increase the dimension at depth $d+1,d+2,\ldots,\dpth{\Pi_M}$ but not too much. The bound is given by the partial size of $\Pi_Q$ and by the number of occurrences of $x$ in $P$.
\end{enumerate}
\end{theorem}
It can be proved proceeding by induction on $\Pi_M$, using Lemma~\ref{lemma:substitution-property}.
\section{Polytime soundness}
\label{section:Polytime soundness}
We want to prove that \WALL\ is poly-time strongly normalizing, mixing ideas, observations and terminology from
\cite{Asperti:1998-LICS,Terui:2001-LICS,Asperti02TOCL,MollerMairson:2002-ICC,Terui:2002-AML}. We first prove the poly-time weak normalization by showing the existence of a canonical strategy, composed by normalization rounds at a given depth, that normalizes every deduction $\Pi_M$ in a time which is bounded by a polynomial in the dimension of $\Pi_M$. Then, the poly-time strong normalization follows by showing that the canonical strategy is the worst one.
\subsection{Weak polytime soundness}
\label{subsection:Weak polytime soundness}
\textbf{Rounds at level $d$.}
Given a deduction $\Pi_{M}\rhd\Gamma;\Delta;{\mathcal E}\vdash \ta{M}{A}$, a \textit{round at level $d\leq\dpth{\Pi_{M}}$ from $M$ to $N_{n}$} is a sequence
$M\equiv N_0 \redl{d} N_1 \redl{d}\cdots \redl{d}N_{n-1} \redl{d} N_{n}$
of reduction of redexes, for some $n\geq 0$, abbreviated as $M\round{d} N_{n}$. A \textit{complete round from $M$} is every round such that $\nfl{d}{N_{n}}$.
Notice that the rounds from $M$ to $N_{n}$ are not unique since every $N_{i}$ may have many redexes at level $d$ that we can reduce in any order.
\begin{corollary}[Behavior of every round.]
\label{corollary:Behaviour of every round}
Let $\Pi_{M}\rhd\Gamma;\Delta;{\mathcal E}\vdash \ta{M}{A}$ be given.
\begin{enumerate}
\item
There is at least one complete round $M\round{d}N_{n}$ from $M$.
\item 
Every complete round $M\round{d}N_{n}$ from $M$ is such that
$n\leq \psz{d}{\Pi_{M}}$. Namely, the complete rounds from $M$ are strongly normalizing in, at most, $\psz{d}{\Pi_{M}}$ steps.
\item 
For every complete round $M\round{d}N_{n}$ from $M$, the derivation $\Pi_{N_{n}}$ is identical to $\Pi_{M}$, at every level $i\leq d-1$. Namely, nothing changes in $\Pi_{M}$ in the course of the round at the levels $0, 1, \ldots, d-1$.
\end{enumerate}
\end{corollary}
Corollary~\ref{corollary:Behaviour of every round} follows from Theorem~\ref{theorem:substitution-property} which proves that a normalization step at level $d$ in $\Pi_{M}$, that corresponds to the reduction of a redex, strictly shrinks the dimension at that level, while preserving the structure at $0, 1, \ldots, d-1$.
\par
\textbf{Canonical strategy.}
Given a deduction $\Pi_{M}\rhd\Gamma;\Delta;{\mathcal E}\vdash \ta{M}{A}$, a \textit{canonical strategy from $M$ to $N_{d}$} is a sequence
$M\equiv N_0 \round{0} N_1 \round{1}\cdots \round{d-2} N_{d-1} \round{d-1} N_{d}$
of \textit{complete rounds}, abbreviated as $M\canstra{0,d-1} N_{d}$.
Notice that we say ``\textit{a} canonical'' instead of ``\textit{the} canonical'' because the complete rounds that define any canonical strategy are not unique.
A \textit{complete canonical strategy from $M$} is every canonical strategy such that $\nf{N_{d}}$. Beware that we do not require $\nfl{d}{N_{d}}$ only, but the full $\nf{N_{d}}$.

\begin{corollary}[Behavior of a canonical strategy.]
\label{corollary:Behaviour of a canonical strategy}
Let $\Pi_{M}\rhd\Gamma;\Delta;{\mathcal E}\vdash \ta{M}{A}$ be given.
\begin{enumerate}
\item 
In every canonical strategy $M\equiv N_0 \round{0} N_1 \round{1}\cdots \round{d-2} N_{d-1} \round{d-1} N_{d}$ from $M$, for every $0< i\leq d$, we have $\nfl{j}{N_{i}}$  with $0\leq j< i$. Namely, every $\Pi_{N_i}$, with $0< i\leq d$, is normal at level $0, 1, \ldots, i-1$.

\item
There is at least one complete canonical strategy $M\canstra{0,d-1} N_{d}$ from $M$.

\item 
Every complete canonical strategy $M\canstra{0,d-1} N_{d}$ from $M$ is such that $d\leq\dpth{\Pi_{M}}+1$. Namely, every complete strategy from $M$ is strongly normalizing in, at most, $\dpth{\Pi_{M}}+1$ complete rounds.
\end{enumerate}
\end{corollary}
Corollary~\ref{corollary:Behaviour of a canonical strategy} is a consequence of Corollary~\ref{corollary:Behaviour of every round} which says that every complete round terminates, and of Theorem~\ref{theorem:substitution-property}, which assures that the maximal depth of $\Pi_{M}$ cannot increase as the normalization proceeds.
\par
\textbf{How to get the polynomial bound.}
Let us assume both that $\Pi_{M}\rhd\Gamma;\Delta;{\mathcal E}\vdash \ta{M}{A}$ and that
we \textit{shall} be able to prove the following proposition:

\begin{proposition}[Bounding the size of the result of every complete round.]
\label{proposition:Bounding the result of every complete round}
Every complete round
$N_{i-1}\round{i-1}N_{i}$ from $N_{i-1}$, with $0\leq i\leq d-1$, implies $\size{\Pi_{N_{i}}}\leq \polynom{i-1}{(\size{\Pi_{N_{i-1}}})}$, where $\polynom{i-1}{}$ is a polynomial of maximal degree $\degree{\polynom{i-1}{}}$.
\end{proposition}

\begin{theorem}[Bounding the result size of any complete canonical strategy.]
\label{theorem:Bounding the result size of any complete canonical strategy}
Let $M\equiv N_0\canstra{0,\dpth{M}}N_{\dpth{M}+1}$ be a complete canonical strategy from $M$. Then $\size{\Pi_{N_{\dpth{M}+1}}}$
$\leq \polynom{}{(\size{\Pi_{M}})}$, such that 
$\degree{\polynom{}{}} = \prod_{i=0}^{\dpth{\Pi_{M}}}\degree{\polynom{i}{}}$.
\end{theorem}
To prove it, we observe that, from 
Proposition\ref{proposition:Bounding the result of every complete round}:
\small
\begin{align*}
&\size{\Pi_{N_{\dpth{M}+1}}}\leq 
\polynom{\dpth{M}}{(\size{\Pi_{N_{\dpth{M}}}})},
\size{\Pi_{N_{\dpth{M}}}}\leq
\polynom{\dpth{M}-1}{(\size{\Pi_{N_{\dpth{M}-1}}})},
\ldots,
\size{\Pi_{N_{2}}}\leq \polynom{1}{(\size{\Pi_{N_{1}}})},
\size{\Pi_{N_{1}}}\leq \polynom{0}{(\size{\Pi_{N_{0}}})}
\end{align*}
\normalsize
which implies
$\size{\Pi_{N_{\dpth{M}+1}}}
\leq
\polynom{\dpth{M}}
        {(\polynom{\dpth{M}-1}
                  {(\ldots
                     \polynom{1}{(\polynom{0}{(\size{\Pi_M})})}
                    \ldots)})}
$, where: 
\small
$$\degree{\polynom{\dpth{M}}
        {(\polynom{\dpth{M}-1}
                  {(\ldots
                     \polynom{1}{(\polynom{0}{(\size{\Pi_M})})}
                    \ldots)})}
}=\prod_{i=0}^{\dpth{\Pi_M}}\degree{\polynom{i}{}}
\enspace .
$$
\normalsize
\begin{corollary}[The normalization of \WALT\ is poly-step.]
\label{corollary:A polynomial bound on the number of the normalization steps}
There is a $k$ such that $\Pi_M$ normalizes in a number of steps which is $O(\size{\Pi_M}^{k^{\dpth{\Pi_M}}})$.
\end{corollary}
To prove it, let $k$ be the maximal of the values among $\degree{\polynom{0}{}},\ldots,\degree{\polynom{\dpth{\Pi_M}}{}}$ of Theorem~\ref{theorem:Bounding the result size of any complete canonical strategy}. Then, the number of redexes we have to reduce in the course of the complete canonical strategy is bounded by:
\small
\begin{align*}
\sum_{i=0}^{\dpth{\Pi_M}} \size{\Pi_{M}}^{k^{i}}
\leq
\sum_{i=0}^{k^{\dpth{\Pi_M}}} \size{\Pi_{M}}^{i}
=
\frac{\size{\Pi_M}^{k^{\dpth{\Pi_M}}+1}-1}{\size{\Pi_M}-1} \in 
O(k^{\dpth{\Pi_M}})
\enspace ,
\text{ using } \sum_{i=0}^{n} x^i = \frac{x^{n+1}-1}{x-1}
\enspace .
\end{align*}
\normalsize
\vspace{-.5cm}
\begin{theorem}[Weak poly-time normalization.]
\label{theorem:Weak poly-time normalization}
Every $\Pi_M$ normalizes in a time bounded by a polynomial in $\size{\Pi_M}$, whose degree depends on $\dpth{\Pi_M}$.
\end{theorem}
To prove it, we use Corollary~\ref{corollary:A polynomial bound on the number of the normalization steps} and the known fact that a single $\beta$-reduction of standard $\lambda$-calculus, of which $\red$ is a special case, can be implemented by a Turing machine with a quadratic overhead in the dimension of the term being reduced \cite{Asperti:1998-LICS,Terui:2001-LICS,Terui:2002-AML}.

\subsubsection{Proving Proposition~\ref{proposition:Bounding the result of every complete round}.}
Essentially we have to prove two facts. One is that every normalization step does not produce too many copies of the deductions that need to be replicated.
The other fact is that the deductions of \WALT, forming them a subsystem of \SF, are essentially acyclic. We start focusing on the first property.

\begin{corollary}[Subject reduction iterated by a round.]
\label{corollary:Subject reduction iterated by a round}
Let us assume 
$\Pi_{M}\rhd\Gamma;\Delta;$ ${\mathcal E}\vdash\ta{M}{A}$.
Let us assume also that in a given round
$M\equiv N_0 \redl{d} N_1 \redl{d}\cdots \redl{d}N_{n-1} \redl{d} N_{n}$,
for some $n\geq 0$, every step $N_{j} \redl{d} N_{j+1}$, with $0\leq j\leq n-1$,
rewrites a redex $(\bs x.P_{j})Q_{j} \redl{d} P_{j}\subs{Q_{j}}{x}$.
Then, $\Pi_{N_{n}}\rhd\Gamma;\Delta;{\mathcal E}\vdash\ta{N_n}{A}$ is such that:
\begin{enumerate}
\item 
\label{corollary:Subject reduction iterated by a round-1}
$\dpth{\Pi_{N_n}}\leq\dpth{\Pi_{M}}$.

\item 
\label{corollary:Subject reduction iterated by a round-2}
$\wdth{i}{\Pi_{N_n}}\leq\wdth{i}{\Pi_{M}}$, for every $0\leq i\leq d+1$.

\item
\label{corollary:Subject reduction iterated by a round-3}
$\psz{i}{\Pi_{N_n}}=\psz{i}{\Pi_{M}}$, for every $0\leq i< d$.

\item
\label{corollary:Subject reduction iterated by a round-4}
$\psz{d}{\Pi_{N_n}}<\psz{d}{\Pi_{M}}$.

\item
\label{corollary:Subject reduction iterated by a round-5}
$\psz{i}{\Pi_{N_n}}\leq
\psz{i}{\Pi_{M}}+\wdth{d+1}{\Pi_{M}}(\sum_{j=0}^{n-1}\psz{i}{\Pi_{Q_j}})$,
for every $d< i\leq \dpth{\Pi_{M}}$.
\end{enumerate}
\end{corollary}
All its points follow by applying Theorem~\ref{theorem:substitution-property} on every step of the round. We develop some details of point~\ref{corollary:Subject reduction iterated by a round-5}.
Theorem~\ref{theorem:substitution-property}, applied  to every
$(\bs x.P_{j})Q_{j} \redl{d} P_{j}\subs{Q_{j}}{x}$, with $0\leq j\leq n-1$,
implies
$\psz{i}{\Pi_{N_{j+1}}}
\leq\psz{i}{\Pi_{N_j}}+\nocc{x}{P_j}\psz{i}{\Pi_{Q_j}}
\leq\psz{i}{\Pi_{N_j}}+\wdth{d+1}{\Pi_{N_j}}\psz{i}{\Pi_{Q_j}}
$, for every $d< i\leq \dpth{\Pi_{M}}$, because, by definition,
$\nocc{x}{P_j}\leq\wdth{d+1}{\Pi_{N_j}}$.
So, point~\ref{corollary:Subject reduction iterated by a round-2} here above implies:
\small
\begin{align*}
\label{align:Subject reduction iterated by a round-5}
\psz{i}{\Pi_{N_{n}}}
\leq\psz{i}{\Pi_{N_0}}+\sum_{j=0}^{n-1}\wdth{d+1}{\Pi_{N_j}}\psz{i}{\Pi_{Q_j}}
\leq\psz{i}{\Pi_{N_0}}+\wdth{d+1}{\Pi_{N_0}}\sum_{j=0}^{n-1}\psz{i}{\Pi_{Q_j}}
\enspace .
\end{align*}
\normalsize
Now, let us assume we \textit{shall} be able to prove:
\begin{proposition}[Bounding the size of the substituted arguments.]
\label{proposition:Bounding the size of the substituted arguments}
Let us suppose the assumptions of Corollary~\ref{corollary:Subject reduction iterated by a round} hold.
Then, $\sum_{j=0}^{n-1}\psz{i}{\Pi_{Q_j}}\leq\psz{i}{\Pi_M}$, for every 
$d< i\leq \dpth{\Pi_{M}}$.
\end{proposition}
Proposition~\ref{proposition:Bounding the result of every complete round}
is directly implied by Corollary~\ref{corollary:Subject reduction iterated by a round} and Proposition~\ref{proposition:Bounding the size of the substituted arguments}
by assuming that
$M\equiv N_0 \round{d} N_{n}$ is complete and observing:
\small
\[
\psz{i}{\Pi_{N_{n}}}
\leq\psz{i}{\Pi_{N_0}}+\wdth{d+1}{\Pi_{N_0}}\psz{i}{\Pi_{N_0}}
\leq\psz{i}{\Pi_{N_0}}+\psz{i}{\Pi_{N_0}}^2
\leq 2\size{\Pi_{N_0}}^2
\enspace ,
\]
\normalsize
where the two terms $M, N_n$ in the round $M\equiv N_0 \round{d} N_{n}$ here above coincide to $N_{i-1}, N_i$, respectively, of Proposition~\ref{proposition:Bounding the result of every complete round}. Consequently, we get Theorem~\ref{theorem:Weak poly-time normalization} about the weak poly-time normalization of \WALT.

\subsubsection{Proving Proposition~\ref{proposition:Bounding the size of the substituted arguments}.}
This amounts to check the absence of cycles.
For doing this, we trace how the copies of the same deduction compose in the course of a normalization.
The main tool for tracing are the \textit{substitution traces}, or simply, \textit{traces}. They record how we compose (sub-)deductions that conclude by a modal rule as the normalization proceeds. We are interested to them since they determine the size growth at the levels deeper than the one a normalization step takes place.
\par
\textbf{(Substitution) traces.}
Every \textit{(substitution) trace} is a set of sequences of deductions of \WALT, defined as follows.
The empty set of sequences is a \textit{trace}.
For every $\Pi_M(!)\rhd\Gamma,\Delta,{\mathcal E}\vdash\ta{M}{\,!A}$ the singleton $\{\st{\Phi}{\Pi_{M}(!)}\}$ is a \textit{trace}, $\Phi$ containing the single polynomial variable in $\dom{\mathcal E}\cap\FV{M}$, if any. Otherwise, $\Phi$ is $\emptyset$.
For every
$\Pi_M(\$)\rhd\Gamma,\Delta,{\mathcal E}\vdash\ta{M}{\$A}$,
the set
$\biguplus_{x\in{\mathcal F}}\{\st{\{x\}}{\Pi_{M}(\$)}\}$,
is a \textit{trace}, where ${\mathcal F}$ is the set of polynomial variables in $\dom{\mathcal E}\cap\FV{M}$, if any. Otherwise, the trace is a singleton
$\{\st{\emptyset}{\Pi_{M}(\$)}\}$.
Finally, for any ${\mathcal F}', {\mathcal F}''$, let
$t_1=\biguplus_{x\in{\mathcal F}'\cup{\mathcal F}''}
 \{\st{\{x\}}{\Pi_{P^x_1},\ldots,\Pi_{P^x_{m_x}}}\}$
and
$t_2=\{\st{\Phi}{\Pi_{Q_1},\ldots,\Pi_{Q_n}}\}$ be traces.
Then, we obtain a trace by \textit{plugging $t_2$ (on top of some of the sequences) in $t_1$}:
\small
\[
(\biguplus_{x\in{\mathcal F}'}
  \{\st{\Phi}{\Pi_{P^x_1},\ldots,\Pi_{P^x_{m_x}},
              \Pi_{Q_1}  ,\ldots,\Pi_{Q_n}}\}
 )\uplus
 (\biguplus_{y\in{\mathcal F}''}
 \{\st{\{y\}}{\Pi_{P^y_1},\ldots,\Pi_{P^y_{m_y}}}\})
\enspace .
\]
\normalsize
\textbf{Intuition about traces and initial traces, introduced here below.}
Initial traces can be thought of as traces assigned to a deduction we want to normalize. Those which are not initial can be thought of as built stepwise and associated to the deductions as far as the normalization proceeds.
\par
\textbf{Assigning (initial) traces to a deduction.}
Let $\Pi_{M}\rhd\Gamma;\Delta;{\mathcal E}\vdash \ta{M}{A}$.
An \textit{assignment of traces to $\Pi_M$} is a map $\ca(\Pi_M)=\bigcup_{\Pi_{P_0}\preceq\Pi_M}\{\ca(\Pi_{P_0})\}$ such that, for every 
subdeduction $\Pi_{P_0}(R)\rhd\Gamma_0;\Delta_0;{\mathcal E}_0\vdash \ta{P_0}{B}$ of $\Pi_M$, $\ca(\Pi_{P_0}(R))$ yields a trace, as follows:
\begin{enumerate}
\item
\label{enum:initial-trace-1}
$\ca(\Pi_{P_0}(R))=\{\st{\Phi}{\Pi_{P_0},\Pi_{P_1},\ldots,\Pi_{P_m}}\}$ for some $m\geq 0$, if $R\equiv!$, $B\equiv !A_0$, and
\begin{itemize}
\item 
for every $1\leq i\leq m-1$, $P_i$ is a value,
and
$\Pi_{P_i}(!)
\rhd
\Gamma_i;\Delta_i;{\mathcal E}_i,\{(\emptyset;\ta{x_i}{A_{i}})\}
\vdash\ta{P_i}\,{!A_{i-1}}$;

\item 
$\Pi_{P_m}(!)
\rhd
\Gamma_m;\Delta_m;{\mathcal E}_m\sqcup\{(\emptyset;\Phi_m)\}
\vdash\ta{P_m}\,{!A_{m-1}}$.
\end{itemize}

\item
\label{enum:initial-trace-2}
$\ca(\Pi_{P_0}(R))=
(\biguplus_{i=1}^p
\{\st{\Phi_{m_i}}{\Pi_{P_0},\Pi_{P_{i 1}},\ldots,\Pi_{P_{i m_i}}}\}
)\uplus
(\biguplus_{j=1}^q
\{\st{\Phi'_j}{\Pi_{P_0}}\})$, for some $p,q\geq 0$, $n_1,\ldots,n_p\geq 1$,
if $R\equiv\$$, $B\equiv \$A_0$, and
\begin{itemize}
\item 
$\Pi_{P_0}(\$)
\rhd
\Gamma_0;\Delta_0;{\mathcal E}_0,
(\emptyset;\ta{x^0_{1 1}}{A_{1}}),
\ldots,
(\emptyset;\ta{x^0_{1 n_1}}{A_{1}}),
\ldots\ldots
(\emptyset;\ta{x^0_{p 1}}{A_{p}}),
\ldots,
(\emptyset;\ta{x^0_{p n_p}}{A_{p}})
\vdash \ta{P_0}{\$A}$;

\item 
for every $1\leq i\leq p$, and $m_1,\ldots,m_p\geq 1$:
\begin{itemize}
\item 
all the variables $x_{i1},\ldots,x_{in_i}$ are contracted to the same variable $x_i$ in $\Pi_M$;
\item 
$P_{i m_i}$ is a value,  and
$\Pi_{P_{i m_i}}(!)\rhd
\Gamma_{i m_i};
\Delta_{i m_i};
{\mathcal E}_{i m_i}\sqcup\{(\emptyset;\Phi_{m_i})\}
\vdash\ta{P_{i m_i}}\,{!A_{i m_i-1}}$;
\item 
$A_{i}\equiv A_{i 0}$;
\item 
for every $1\leq j\leq m_i-1$,
$P_{i j}$ is a value, and
$\Pi_{P_{i j}}(!)
\rhd
\Gamma_{i j};
\Delta_{i j};
{\mathcal E}_{i j},\{(\emptyset;\ta{y_{i j}}{A_{i j}})\}
\vdash\ta{P_{ij}}\,{!A_{i j-1}}$;
\end{itemize}

\item 
if we say that $\{z_1,\ldots,z_q\}$ is the set of all the polynomial variables in 
$\dom{{\mathcal E}_0}\cap\FV{P_0}$, namely the set such that
$\{z_1,\ldots,z_q\}
\cap
\{x^{0}_{1{1}},\ldots,x^{0}_{1{n_1}},
\ldots$ $\ldots
  x^{0}_{p{1}},\ldots,x^{0}_{p{n_p}}
\}=\emptyset$, then, for every $1\leq j\leq q$,
$\Phi'_j$ contains the single polynomial variable $z_j$ in $\dom{{\mathcal E}_{0}}\cap\FV{P_{0}}$, if any. Otherwise, there is a unique $\Phi'_1=\emptyset$.
\end{itemize}

\item 
$\ca(\Pi_{P_0}(R))=\emptyset$ if $R\not\in\{\$, !\}$.
\end{enumerate}
The assignment $\ca(\Pi_M)$ is \textit{initial}, if
$m=0$ in clause~\ref{enum:initial-trace-1},
and $p=0$ in clause~\ref{enum:initial-trace-2}. Namely, every initial assignment assigns at least $\{\st{\emptyset}{\Pi_{P_0}}\}$ to
$\ca(\Pi_{P_0}(R))$, with $R\in\{\$, !\}$.
Moreover, $\Pi_{P_{i1}},\ldots, \Pi_{P_{i m_i}}$ in the definition of
$\ca(\Pi_{P_0}(\$))$ represent deductions whose subject is a value, with a modal type, replaced for $x_{i1},\ldots,x_{in_i}$.
\par
As consequence, the following lemma shows that we can transform a trace assignment by means of the normalization steps.
In particular, it says how a reduction step can modify the trace of a subdeduction that concludes by an instance of the rule $\$$, or $!$. The trace is modified by plugging the trace, associated to the deduction, which is argument of the $\beta$-redex, on top of the modal rules occurring in the deduction which is the body of the function in the $\beta$-redex.

\begin{lemma}[Stepwise transformation of trace assignments.]
\label{lemma:Step wise transformation of trace assignments}
Let $\Pi_{M}\rhd\Gamma;\Delta;{\mathcal E}\vdash \ta{M}{A}$, and 
$\ca(\Pi_M)$ be an assignment of traces to $\Pi_M$.
Also, let us suppose that,
for some $0\leq d\leq \dpth{\Pi_M}$, 
$M\redl{d}N$ by means of $(\bs x.P)Q\redl{d}P\subs{Q}{x}$, 
where $\nocc{x}{P}\geq 1$,
$\Pi_Q\rhd\Gamma_Q;\Delta_Q;{\mathcal E}_Q\vdash\ta{Q}{\,!A}$, and $\ca(\Pi_M)$ is such that, for $\Pi_Q\preceq\Pi_M$, we have $\ca(\Pi_Q)=\{\st{\Phi_Q}{\Pi_{Q_1},\ldots,\Pi_{Q_{m}}}\}$.
The reduction step induces $\ca(\Pi_N)$ from $\ca(\Pi_M)$ as follows:
\begin{enumerate}
\item
\label{lemma:Step wise transformation of trace assignments-01}
$\ca(\Pi_N)$ contains $\ca(\Pi_{P'\subs{Q\subs{y}{z}}{y}})=
\{\st{\{y\}}{\Pi_{P_1},\ldots,\Pi_{P_{n}},\Pi_{Q_1}\ldots,\Pi_{Q_{m}\subs{y}{z}}}\}$,
for every
$\Pi_{P'}(!)\preceq\Pi_M$ such that: 
(i) $\ca(\Pi_M)$ contains $\ca(\Pi_P'(!))=\{\st{\{y\}}{\Pi_{P_1},\ldots,\Pi_{P_{n}}}\}$, 
(ii) $y$ is contracted to $x$ in $\Pi_P$ by occurrences of the rules $C, \li E, \li E_!, \liv E$, and 
(iii) $\Phi_Q=\{z\}$.
Otherwise, if $\Phi_Q=\emptyset$, then
$\ca(\Pi_{P'\subs{Q}{y}})=
\{\st{\emptyset}{\Pi_{P_1},\ldots,\Pi_{P_{n}},\Pi_{Q_1}\ldots,\Pi_{Q_{m}}}\}$.

\item 
\label{lemma:Step wise transformation of trace assignments-02}
$\ca(\Pi_N)$ contains 
$\ca(\Pi_{P'\subs{Q\subs{x_1}{z}}{x_1}
            \ldots
            \subs{Q\subs{x_m}{z}}{x_m}
         })=
{\mathcal S}
\uplus
(\biguplus_{i=1}^{m}
 \{\st{\{x_i\}}{\Pi_{P^i_1},\ldots,\Pi_{P^i_{n_i}},
            \Pi_{Q_1}\ldots,$ $\Pi_{Q_{m}\subs{x_i}{z}}}\})$,
for every 
$\Pi_{P'}(\$)\preceq\Pi_M$ such that:
(i) $\ca(\Pi_M)$ contains
$\ca(\Pi_{P'}(\$))=
{\mathcal S}
\uplus
(\biguplus_{i=1}^{m}\{\st{\{x_i\}}{\Pi_{P^i_1},\ldots,\Pi_{P^i_{n_i}}}\})
$, 
(ii) the variables $x_1,\ldots,x_{m}$ are contracted to $x$ in $\Pi_M$ by occurrences of the rules $C, \li E, \li E_!, \liv E$, and 
(iii) $\Phi_Q=\{z\}$, for some $m\geq 1$ and $n_1,\ldots,n_m\geq 1$. Otherwise, if $\Phi_Q=\emptyset$, we have that
$\ca(\Pi_{P'\subs{Q}{x_1}
            \ldots
            \subs{Q}{x_m}
         }
    )=
{\mathcal S}
\uplus
(\biguplus_{i=1}^{m}
 \{\st{\emptyset}{\Pi_{P^i_1},\ldots,\Pi_{P^i_{n_i}},
                  \Pi_{Q_1}\ldots,\Pi_{Q_{m}}}\})$.
\end{enumerate}
$\ca(\Pi_N)$ is identical to $\ca(\Pi_M)$ everywhere else.
\end{lemma}
To prove Lemma~\ref{lemma:Step wise transformation of trace assignments} we start observing that $Q$ must be a value.
Lemma~\ref{lemma:structural-properties-WALL}, point~\ref{lemma:structural-properties-WALL-3.1}, implies that the last rule of $\Pi_{Q}$ is an instance of $!$. So, $x$ is a polynomial variable in $\Pi_P$.
In relation to $x$, Lemma~\ref{lemma:structural-properties-WALL}, point~\ref{lemma:structural-properties-WALL-2}, implies the existence of $n\geq 1$ subdeductions of $\Pi_P$, namely of $\Pi_M$, such that,
for some $q_1,\ldots,q_n\geq 0$,
$\wdth{1}{\Pi_{P}}\geq q_1+\ldots+q_n \geq\nocc{x}{P}$, and, 
for every $1\leq k\leq n$,
$\Pi_{P'_k}(R_k)\rhd
\Gamma_k;\Delta_k;
{\mathcal E}_k,
(\Theta^k_1;\ta{x^k_1}{A}),
\ldots,
(\Theta^k_{q_k};\ta{x^k_{q_k}}{A})
\vdash\ta{P_k}{C_k}$,
with $R_k\in\{A, \$, !\}$, is a subdeduction $\Pi_P$. Namely, every
$\Pi_{P'_k}$ introduces $x^k_1,\ldots,x^k_{q_k}$ that will be contracted to $x$.
Let ${\mathcal K}$ be the maximal subset of $\{1,\ldots,n\}$ such that, for every $k\in{\mathcal K}$, $R_k\in\{\$, !\}$.
The points~\ref{lemma:subst-vs-wght-02} and~\ref{lemma:subst-vs-wght-04} of Lemma~\ref{lemma:subst-vs-wght} imply that we can build the derivations
$\Pi_{P'_k\subs{Q\subs{x^k_1}{z}}{x^k_1}
          \ldots
          \subs{Q\subs{x^k_{q_k}}{z}}{x^k_{q_k}}
     }(R_k)$ of $\Pi_N$.
Now, if $R_k\equiv !$, then $q_k=1$ and
we define
\\
$\ca(\Pi_{P'_k\subs{Q\subs{x^k_1}{z}}{x^k_1}}(!))$ as in point~\ref{lemma:Step wise transformation of trace assignments-01} here above, by identifying $y$ with $x^k_1$.
Otherwise, if $R_k\equiv \$$,
we define \\
$\ca(\Pi_{P'_k\subs{Q\subs{x^k_1}{z}}{x^k_1}
          \ldots
          \subs{Q\subs{x^k_{q_k}}{z}}{x^k_{q_k}}
     })(\$)$ 
as in point~\ref{lemma:Step wise transformation of trace assignments-02} by identifying every $x_1,\ldots,x_m$ with $x^k_1,\ldots,x^k_{q_k}$.

\begin{theorem}[Nature of the elements of the sequences in a trace.]
\label{theorem:Nature of the elements of the sequences in a trace}
Let $\Pi_{M}\rhd\Gamma;\Delta;$ ${\mathcal E}\vdash \ta{M}{A}$ with $\ca(\Pi_{M})$ initial.
Also, let $M\equiv N_0\redl{d}\ldots\redl{d}N_m$ be a round, for some $m\geq 0$.
For every $\Pi_P\preceq\Pi_{N_m}$:
\begin{enumerate}
\item 
\label{theorem:Nature of the elements of the sequences in a trace-01}
Every element of every sequence in $\ca(\Pi_P)$ is a subdeduction of $\Pi_{N_0}$.

\item 
\label{theorem:Nature of the elements of the sequences in a trace-02}
Every pair of deductions $\Pi_{Q'}, \Pi_{Q''}$ that occur in a given sequence of $\ca(\Pi_P)$ are distinct subdeductions of $\Pi_{N_0}$.
\end{enumerate}
\end{theorem}
To prove the first point here above we proceed by induction on $m$.
If $m=0$, the statement holds by definition of (initial) trace.
If $m>0$, then the statement holds by Lemma~\ref{lemma:Step wise transformation of trace assignments} where the step $N_{m-1}\redl{0}N_m$ defines $\ca(\Pi_{N_m})$ using the trace of $\ca(\Pi_{N_{m-1}})$ which, by induction, contains subdeductions of $\Pi_{N_0}$.
\par
For the second point, let us suppose $\ca(\Pi_{N_m})$ contained a sequence of subdeductions of $\Pi_{N_0}$ in which the two occurrences $\Pi_{Q'}, \Pi_{Q''}$ are, in fact, the same occurrence $\Pi_{\tilde Q}$ such that $\Pi_{\tilde Q}\preceq\Pi_{N_0}$.
This means that, in the course of the normalization, we have a sequence
$\Pi_{\tilde Q}\equiv\Pi_{Q_0},\Pi_{Q_1}\ldots,\Pi_{Q_{m-1}},\Pi_{Q_m}\equiv\Pi_{\tilde Q}$, with $m\geq 0$, such that the conclusion of $\Pi_{Q_{i}}$ is plugged into an assumption of $\Pi_{Q_{i-1}}$, for every $1\leq i\leq m$. This would mean to have a cycle in the normalization procedure, contradicting 
Fact~\ref{fact:WALL types a subset of SF}, saying that \WALT\ is a subsystem of \SF, which is strongly normalizing.
\par
\textbf{Proposition~\ref{proposition:Bounding the size of the substituted arguments} as a corollary of Theorem~\ref{theorem:Nature of the elements of the sequences in a trace} and Corollary~\ref{corollary:Behaviour of every round}.}
Let us assume that the hypothesis of Theorem~\ref{theorem:Nature of the elements of the sequences in a trace} hold, $M\equiv N_0\round{d}N_m$ being a complete round.
For every $\Pi_P$ such that both $\Pi_{P}\preceq\Pi_{N_m}$ and $\ca(\Pi_{P})\neq\emptyset$,
we can say that every sequence $\st{\Phi}{\Pi_{P_1},\ldots,\Pi_{P_n}}$ that belongs to $\ca(\Pi_{P})$ only contains \textit{distinct} instances of subdeductions of $\Pi_{N_0}$, thanks to the two points of Theorem~\ref{theorem:Nature of the elements of the sequences in a trace}. Moreover, $m\leq\psz{d}{\Pi_{N_0}}$ thanks to Corollary~\ref{corollary:Behaviour of every round}. This, by the definition of partial size implies $\sum_{j=1}^{m}\psz{i}{\Pi_{P_j}}\leq\psz{d}{\Pi_{N_0}}$, for every $d<i\leq\dpth{\Pi_{N_0}}$, which is the statement of Proposition~\ref{proposition:Bounding the size of the substituted arguments}.
So, Theorem~\ref{theorem:Weak poly-time normalization} holds and \WALT\ is weakly poly-time normalizable, at least.
\subsection{Strong polytime soundness}
\label{subsection:Strong polytime soundness}
Our goal is to see that every normalization strategy which is not canonical cannot be worst in terms of the number of steps it can perform.
We assume to have some $\Pi_{N_0}$, and we observe the differences between the complete round $N_0\round{d}N_m$ and any other reduction sequence defined as follows as a \textit{perturbation} of the round:
\small
$$
\begin{array}{rcccccccl}
N_0&\redl{d}&N_1&\redl{d}\ldots\redl{d}&N_{i}&\redl{d}&N_{i+1}&\redl{d}\dots\redl{d}&N_{m}
\\
&&&&\downarrow_{d+1}&&&&
\\
&&&&N'_{i}&\redl{d}&N'_{i+1}&\redl{d}\dots\redl{d}&N'_{m}
\enspace .
\end{array}
$$
\normalsize
We call \textit{perturbation step at level $d+1$ of the complete round $N_0\round{d}N_{m}$} the step $\downarrow_{d+1}$; it stands for
a single step $N_{i}\redl{d+1}N'_{i}$ that reduces a redex
$(\bs x.P)Q\redl{d+1}P\subs{Q}{x}$ which we may assume be in some $\Pi_{N''}\preceq\Pi_{N_i}$.
Theorem~\ref{theorem:substitution-property} implies that the perturbation step at $d+1$ does not modify the size at depth $d$. Both $\Pi_{N_i}$ and 
$\Pi_{N'_i}$ coincide ad level $d$, and it is correct to keep reducing $N'_i$ to $N'_m$ in the same number of steps we need from $N_i$ to $N_m$.
However, the perturbation produces $\psz{d+1}{\Pi_{N'_i}}<\psz{d+1}{\Pi_{N_i}}$.
The consequence is that, given the two complete rounds $N_m\round{d+1}N_{n_1}$, and $N'_m\round{d+1}N'_{n_2}$, we have $n_1\geq n_2$.
In particular, if \textit{none} of the rewriting steps in $N_{i}\round{d}N_m$, or in $N'_{i}\round{d}N'_m$, produces copies of $\Pi_{N''}$, which would mean to replicate $(\bs x.P)Q$, then $n_1= n_2+1$. Indeed, reducing $(\bs x.P)Q$ as a perturbation of $N_0\round{d}N_{m}$ implies that we do not have to reduce it in the course of $N'_m\round{d+1}N'_{n_2}$. Otherwise, if we do not perturbate $N_0\round{d}N_{n}$, we shall reduce $(\bs x.P)Q$ in the course of $N_m\round{d+1}N_{n_1}$, increasing $n_1$ exactly by one step, as compared to $n_2$.
If, on the contrary, at least one of the rewriting steps in $N_{i}\round{d}N_m$ replicates $\Pi_{N''}$, so generating many copies of $(\bs x.P)Q$, then $n_1= n_2+k$, with $k>1$. It is enough to observe that $N_m\round{d+1}N_{n_1}$ will contain a step for every copy of $(\bs x.P)Q$, while these steps will not belong to $N'_m\round{d+1}N'_{n_2}$ because it will already contain copies of $P\subs{Q}{x}$.
\par
So every \textit{arbitrary sequence of perturbations}, at any level, of any round in any canonical strategy reduces in advance redexes that, instead, would be first replicated, and then reduced, by the canonical strategy itself. So the canonical strategy is the worst one.
\section{Quasi-linear safe recursion on notation (\QlSRN)}
\label{section:Quasi-linear safe recursion on notation}
We define the fragment \QlSRN\ of the Safe recursion on notation (\SRN) that we shall be able to embed into \WALT. Recall that \QlSRN\ is \SRN\ where the composition scheme is restricted to linear safe arguments only. To introduce \QlSRN, we follow \cite{Beckmann:96-AML}.
\par
\textbf{The signature of \QlSRN.}
Let $\Sigma_{\QlSRN}=\cup_{k,l\in\Nat} \Sigma^{k,l}_{\QlSRN}$ be the
signature of \QlSRN.
$\Sigma_{\QlSRN}$ contains the \textit{base functions} and it is closed under the schemes called \textit{linear safe composition} and \textit{safe recursion}.
For every $k,l\in\Nat$, the \textit{base functions} are
the \textit{zero} $\zero{k}{l}\in\Sigma^{k,l}_{\QlSRN}$,
the \textit{successors} $\sucz, \suco$, and
the \textit{predecessor} $\pred\in\Sigma^{0,1}_{\QlSRN}$,
the \textit{projection} $\proj{k}{l}{i}\in\Sigma^{k,l}_{\QlSRN}$, with $1\leq i\leq k+l$, and the \textit{branching} $\bran\in\Sigma^{0,3}_{\QlSRN}$.
\par
For every $k,l,l',l_1,\ldots,l_{l'}\in\Nat$,
the \textit{linear safe composition} is
$\comp{k}
      {\sum_{i=1}^{l'}l_i}
	{k'}
	{l'}
	{f,g_1,\ldots,g_{k'}
	,h_1,\ldots,h_{l'}
	}\in\Sigma^{k,\sum_{i=1}^{l'}l_i}_{\QlSRN}$
if
$f\in\Sigma^{k',l'}_{\QlSRN}$,
$g_1,\ldots,g_{k'}\in\Sigma^{k,0}_{\QlSRN}$, and
$h_{i}\in\Sigma^{k,l_i}_{\QlSRN}$, with
$i\in\{1,\ldots,l'\}$,
while the \textit{safe recursion} is
$\rec{k+1}{l}{g,h_0,h_1}\in\Sigma^{k+1,l}_{\QlSRN}$
if
$g\in\QlSRN^{k,l}$, and
$h_0,h_1\in\Sigma^{k+1,l+1}_{\QlSRN}$.
\par
\textbf{Quasi-linear safe recursion on notation (\QlSRN).}
Let $\QlSRNVnames$ be a denumerable set of
\emph{names of variables}, disjoint from $\Sigma_{\QlSRN}$.
\QlSRN\ is the set of \textit{Safe recursive functions on notation with quasi-linear safe arguments with signature $\Sigma_{\QlSRN}$}, or, simply \textit{Quasi-linear safe recursion}. \QlSRN\ is such that
$\QlSRNVnames\subset\QlSRN$, and
for every $k,l\in\Nat$,
if
$f\in\Sigma_{\QlSRN}^{k,l}$, and
$t_1,\ldots,t_{k},u_1,\ldots,u_{l}\in\QlSRN$,
then $f(t_1,\ldots,t_{k},u_1,\ldots,u_{l})\in\QlSRN$. As usual, a term is \textit{closed} if it does not contain variables of $\QlSRNVnames$.
\par
\textbf{Notations and terminology.}
$x, y, z \ldots$ denote elements of $\QlSRNVnames$.
$t, u, v \ldots$ denote elements of $\QlSRN$.
For every $f\in\Sigma^{k,l}_{\QlSRN}$, $k$ and $l$ are \textit{normal} and \textit{safe} arity of $f$, respectively.
For every $k, l\in\Nat$, such that $l-k\geq1$,
$\seq{t}{k}{l}$ denotes a \textit{non empty} sequence
$t_k,\ldots,t_l$ of $l-k+1$ terms in \PRN.
$\seq{t}{k}{l}(i)$, with $k\leq i\leq l$, denotes the element $t_i$
of $\seq{t}{k}{l}$.
\par
\textbf{An equational theory on \QlSRN.}
The definition of the equational theory exploits that every natural number $n$ can be written, uniquely, as $\sum^{m}_{j=0}2^{m-j} \nu_{m-j}$.
So, assuming to abbreviate the base functions $\sucz, \suco$ as
$\Sucz, \Suco$, respectively, we can follow \cite{MurawskiOng00} and say that $0$ is equivalent to $\zero{0}{0}$, and $n\geq 1$ to
$\mathtt{s}_{\nu_0}(\ldots(\mathtt{s}_{\nu_{m-1}}(\Suco\,\zero{0}{0}))\ldots)$.
Notice that we could have expressed $n$ as $\sum^{m}_{j=0}2^{j} \nu_{j}$, but our choice makes proofs simpler. Then, the \textit{equational theory} is as follows.
\textit{Zero} is constantly equal to $0$:
$\zero{k}{l}(\seq{x}{1}{k},\seq{x}{k+1}{k+l})=\Zero$
for any $k,l\in\Nat$.
The \textit{predecessor} erases the most significant bit of any number greater than $0$:
for every $i\in\{0,1\}$, $\pred(\Zero)= \Zero$, and $\pred(\Suc{i}(y))=y$.
We shall use $\preds$ as an abbreviation of $\pred$.
The \textit{conditional} has three arguments. If the first is zero, then the result is the second argument. Otherwise, it is the third one:
for every $i\in\{0,1\}$, $\bran(\Zero,y_0,y_1) = y_0$, and $\bran(\Suc{i}(y),y_0,y_1) = y_1$.
The \textit{projection} chooses one argument, out of a given tuple, as a result: for every $1\leq i\leq k+l$, $\proj{k}{l}{i}(\seq{x}{1}{k},\seq{x}{k+1}{k+l}) = x_i$.
The \textit{linear composition} uses the safe arguments linearly. This means that it splits the sequence of safe arguments into as many sub-sequences as required by the safe arity of every $h_i$ function,
used to calculate the safe arguments of $f$:
\small
\begin{alignat*}{1}
&
\comp{k}
     {\sum_{i=1}^{l'} l_i}
     {k'}
     {l'}
     {f,g_1,\ldots,g_{k'},h_1,\ldots,h_{l'}}
\!\!
\begin{array}[t]{l}
(\seq{x}{1}{k}
,\seq{x}{k+1}{k+l_1},\seq{x}{k+1+l_1}{k+l_1+l_2}
,\ldots,\seq{x}{k+1+\sum_{i=1}^{l'-1}l_i}
               {k+1+\sum_{i=1}^{l'}l_i})
\end{array}
\\
&
\qquad
\qquad
= f\!\!
  \begin{array}[t]{l}
  (g_1(\seq{x}{1}{k})
  ,\ldots
  ,g_{k'}(\seq{x}{1}{k})
  ,h_1(\seq{x}{1}{k},\seq{x}{k+1}{k+l_1})
  ,\ldots
  ,h_{l'}(\seq{x}{1}{k},\seq{x}
                            {k+1+\sum_{i=1}^{l'-1}l_i}
                            {k+1+\sum_{i=1}^{l'}l_i}))
\enspace .
  \end{array}
\end{alignat*}
\normalsize
The \textit{recursion} iterates either the function $h_0$, or $h_1$, as many times as
the length of its first argument. The choice between $h_0$, and $h_1$ depends on the
least significant digit of the first argument.
The base of the iteration is a function $g$:
\small
\begin{alignat*}{1}
\rec{k+1}{l}{g,h_0,h_1}
 (\Zero
 ,\seq{x}{1}{k}
 ,\seq{x}{k+1}{k+l}
 )
& =
 g(\seq{x}{1}{k},\seq{x}{k+1}{k+l})
 \\
\rec{k+1}{l}{g,h_0,h_1}
(\Suc{i}(x)
,\seq{x}{1}{k}
,\seq{x}{k+1}{k+l}
)
&=
h_i(x
   ,\seq{x}{1}{k}
   ,\seq{x}{k+1}{k+l}
   ,\rec{k+1}{l}{g,h_0,h_1}(x
                           ,\seq{x}{1}{k}
                           ,\seq{x}{k+1}{k+l}))
\enspace .
\end{alignat*}
\normalsize
We notice once more that the recursion evaluates its safe arguments with no restrictions at all.
\section{Programming combinators in \WALT}
\label{section:Programming combinators in WALT}
To embed \QlSRN\ into \WALT, inductively, we are going to program some combinators in \WALT. They will represent the base functions, and both the composition and recursive schemes of \QlSRN. This requires to find the correct data-types that allow to capture the call-by-value nature that \QlSRN\ inherits from \SRN, once \SRN\ is taken as rewriting system, and not ``only'' as equational theory \cite{Beckmann:96-AML}. In particular, the recursive scheme will be implemented by an iteration scheme in \WALT, whose behavior will be intuitively illustrated before its formal definition is given in Subsection~\ref{subsection:Iterators}.
\subsection{Basic data-types in \WALT}
\label{subsection:Basic data-types in WALT}
We define a set of Data-types --- (unary) strings, (binary) words, booleans, two kinds of tensors, and lists --- whose canonical constructors can be typed in \WALT.
\par
\textbf{Notations and definitions.}
If $X$ is any finite set, $\size{X}$ is its cardinality.
$\$^nA$ denotes $\$\cdots\$A$ with $n\geq 0$ occurrences of $\$$. An analogous meaning holds for $!^nA$. $(\li^{n}_{i=1} A_i)$ abbreviates $A_1\li\cdots\li A_n$, while $(\mliv^{n}_{i=1} A_i)$ abbreviates $A_1\liv\cdots\liv A_n$. If useful, $B^A$ shortens $A\li B$, for any $A, B$. Finally, the $\lambda$-term identity  $\bs x.x$ is $\Id$.
\par
\textbf{(Unary) Strings.}
We call \textit{Unary string}, or simply \textit{strings}, the terms identified as \textit{Church numerals}, the reason being that the Church numerals are, generally, used to encode the integers $\Nat$ in unary notation. The type $\UIntT$ of strings is
$\UIntT  \equiv \forall \alpha.!(\alpha\li\alpha)\li\$(\alpha\li\alpha)$
whose \textit{constructors} have the standard form:
\small
\begin{align*}
\UNum{0} &\equiv \bs fy.y\\
\UNum{m} &\equiv \bs fy.(f(\cdots(f\, y)\cdots))
	 & (m\geq 1 \text{ occurrences of } f)
\end{align*}
\normalsize
The \textit{successor on the strings} is the usual term, up to a $\beta$-expansion:
\small
\begin{align*}
\USucc   &\equiv \bs nf.(\bs zx.f(z\,x))(n\,f)
\enspace .
\end{align*}
\normalsize
First,  $\USucc$ develops the iteration of $n$ applied to $f$, and, then, it applies the fully unfolded iteration to $x$.
The presence of a $\beta$-expansion when a string is used to iterate some step function seems a kind of constant design property that \WALT\ induces on the $\lambda$-terms it gives a type to.
\begin{proposition}[Typing the strings.]
\label{proposition:Typing rules relative to strings}
Rules derivable in \WALT:
\small
\[
\infer[]
{
\emptyset;\emptyset;\emptyset
\vdash \ta{\UNum{n}}{\UIntT}
}
{n\geq0}
\qquad
\infer[]
{
\emptyset;\emptyset;\emptyset
\vdash\ta{\USucc}{\UIntT\li\UIntT}
}
{}
\]
\normalsize
\end{proposition}

\begin{proposition}[Dynamics of the successor on strings.]
\label{proposition:Dynamics of the successor on strings}
For every $n\in\Nat$, $\USucc\, \UNum{n}\red^+\UNum{n+1}$.
\end{proposition}
\par
\textbf{Tensors.}
We use the \textit{Tensors} to represent tuples of $\lambda$-terms. The tensor type symbol $\bigotimes$ is used as follows
$\bigotimes^{m}_{i=1} A_i\equiv
\forall \alpha.((\li^{m}_{i=1} A_i)\li\alpha)\li\alpha$, with $m\geq 1$,
and the type constructors coincide to the usual definition of tuples in the $\lambda$-calculus:
\small
\begin{align*}
\lan M_1\ldots M_m\ran
	&= \bs z.z\,M_1\,\ldots\,M_m
	&(m\geq 1)\\
\bs\lan x_1\ldots x_m\ran.M
	&= \bs w.w(\bs x_1\ldots x_m. M)
	&(m\geq 1)
\end{align*}
\normalsize
\begin{proposition}[Typing the tensors.]
\label{proposition:Typing tuples}
\normalsize
Rules derivable in \WALT:
\small
\[
\infer[\otimes I_{}]
  {
   \Gamma_1\ldots\Gamma_m;
   \Delta_1\ldots\Delta_m;
   {\mathcal E}_1\sqcup\cdots\sqcup{\mathcal E}_m
   \vdash \ta{\lan M_1,\ldots,M_m\ran}
             {\bigotimes^{m}_{i=1} A_i}
  }
  {
   \Gamma_1; \Delta_1;
   {\mathcal E}_1
   \vdash \ta{M_1}{A_1}
   &\ldots&
   \Gamma_m; \Delta_m;
   {\mathcal E}_m
   \vdash \ta{M_m}{A_m}
  }
\]
\[
 \infer[\li I_{\otimes}]
  {\Gamma; \Delta; {\mathcal E}
   \vdash \ta{\bs \lan x_1 \ldots x_m \ran.M}
             {(\bigotimes^{m}_{i=1} A_i)\li B}}
  {\Gamma,
   \ta{x_1}{A_1}
   ,\ldots,
   \ta{x_m}{A_m}
   ; \Delta; {\mathcal E}\vdash \ta{M}{B}}
\]
\normalsize
\end{proposition}

\begin{proposition}[Dynamics of the tensors.]
\label{proposition:Dynamics relative to tensors}
For every $M_{1}\,\ldots\,M_{m}$
$(\bs \lan x_1\cdots x_m\ran.M)\lan M_1,\ldots, M_m\ran
\red^+\\
(\bs x_1 \ldots x_m.M)\,M_1\ldots M_m$.
\end{proposition}
\par
\textbf{Booleans.}
We call \textit{Booleans} the terms that, applied to a tuple, project out one its components. The type $\BoolT_m$ of a space of booleans with $m$ elements is
$\BoolT_m\equiv \forall \alpha.
          (\bigotimes^{m}_{i=1}\alpha)\li\alpha$, with $m\geq 1$.
The type constructor is a projection
$\pi^{m}_{i}\equiv\bs \lan x_0 \ldots x_{m-1}\ran. x_i$
with $m\geq 1$, and $0\leq i\leq m-1$,
defined on tensor tuples of terms.
\begin{proposition}[Typing the booleans.]
\label{proposition:Typing rules relative to booleans}
\normalsize
Rules derivable in \WALL:
\small
\[
\infer[]
{
\emptyset;\emptyset;\emptyset
\vdash \ta{\pi^n_i}{\BoolT_n}
}
{n\geq 1 & 0\leq i\leq n-1}
\]
\normalsize
\end{proposition}

\begin{proposition}[Dynamics of the booleans.]
\label{proposition:Dynamics relative to booleans}
For every $n\geq 1$ and $0\leq i\leq n-1$,
$\pi^{m}_{i}\, \lan M_0,\ldots,M_{m-1}\ran\red^+ M_i$.
\end{proposition}
\par
\textbf{(Binary) Words.}
We call \textit{Binary words}, or simply \textit{words}, the terms that allow to encode the integers $\Nat$ in \textit{binary} notation. The type of the words is
$\BIntT\equiv\forall\alpha.!(\alpha\li\alpha)\li!(\alpha\li\alpha)\li\$(\alpha\li\alpha)$.
The \textit{canonical constructors of $\BIntT$} are:
\small
\begin{align}
\nonumber
\BNum{0}&=\bs 01y.y \\
\label{align:word-constructor}
\BNum{2^m+2^{m-1}\cdot\nu_{m-1}+\cdots+2^0\cdot\nu_{0}}
	&=\bs 01y.\nu_0(\cdots(\nu_{m-1}(1\, y)\cdots)
\end{align}
\normalsize
where $m\geq 0$ and $\nu_{0\leq i\leq m-1}\in\{0,1\}$. In particular, observe that, for any $n\geq 0$, the least significant bit of $\BNum{2n+\nu}$ coincides to $\nu$, and that a word is a Church numeral built using the two successors names $0$, and $1$. Two combinators $\BSuccZ$, and $\BSuccO$ that yield the \textit{successors} of any word exist:
\small
\begin{align*}
\BSuccO&=\bs n01.(\bs zy.1(zy))(n01)\\
\BSuccZ&=\bs n.\MkCompact\,(\bs 01.(\bs zy.0(zy))(n01))\\
\MkCompact
	&\equiv
          \bs n01.
		(
		\bs z.
		\bs y.
		(\bs\lan x\,y\ran.y)
		(z(\BaseMkComp\, y))
		)
		(n(\StepMkCompZ\, 0)(\StepMkCompO\, 1))
	  \\
\BaseMkComp&\equiv\bs y.\lan \pi_0^2, y\ran\\
\StepMkCompZ
	&=\equiv
          \bs x.
	  \bs \lan p\,r\ran.
		(\bs \lan p_1\,p_2\ran.
		 \lan p_1,p_2\lan \bs x.x,x\ran r\ran
		)
		(p\lan\lan\pi_0^2,\pi_0^2\ran,\lan\pi_1^2,\pi_1^2\ran
		\ran)
	\\
\StepMkCompO
	&\equiv\bs x.\bs \lan p\,r\ran. \lan\pi_1^2,x r\ran
\end{align*}
\normalsize
$\BSuccO$ has the form we expect, namely it generalizes the form of the successor on strings.
$\BSuccZ$ uses $\MkCompact$ to erase every occurrence of the symbol $0$ to the right hand side of the most significant bit of a word, as in \cite{MurawskiOng00}. This allows to preserve the requirements on \eqref{align:word-constructor}, where any $\BNum{n}$, with $n\neq 0$, must have $1$ as its most significant bit.
The words have a single \textit{predecessor}:
\small
\begin{align*}
\Pred
&\equiv\bs n.\bs 01. (\bs z y. \pi^2_1(z\, (\BaseP\, y)))
                    (n\, (\StepP\, 0)
			 (\StepP\, 1))
\\
\StepP
&\equiv\bs x. \bs \lan u\, v\ran.\lan x, uv\ran
\\
\BaseP
&\equiv\bs x.\lan \bs x.x, x\ran
\end{align*}
\normalsize
which is completely linear \cite{Roversi:1999-CSL,Asperti02TOCL}: all variable names occur once.
Finally, it is also useful to define a term that \textit{discriminates} words:
\small
\begin{align*}
\Branch
&\equiv\bs n.\bs ab.\bs 01.
  (\bs w. \bs z_1 z_2.
   w\,\pi^2_0\, \lan z_1, z_2 \ran
  )(n\,(\bs x.\pi^2_1)\,(\bs x.\pi^2_1))(a\,0\,1)(b\,0\,1)
\end{align*}
\normalsize
that, applied to three words, if the first one is $\BNum{0}$, it gives the second one. Otherwise it yields the third word.

\begin{proposition}[Typing the words.]
\label{proposition:Typing rules relative to words}
Rules derivable in \WALL:
\small
\[
\infer[]
{
\emptyset;\emptyset;\emptyset
\vdash \ta{\BNum{n}}{\BIntT}
}
{n\geq0}
\qquad
\infer[]
{\emptyset;\emptyset;\emptyset
\vdash\ta{M}{\BIntT\li\BIntT}
}
{M\in\{\BSuccZ,\BSuccO\}}
\]
\[
\infer[]
{\emptyset;\emptyset;\emptyset
\vdash\ta{\MkCompact}{\BIntT\li\BIntT}}
{}
\qquad
\infer[]
{\emptyset;\emptyset;\emptyset
 \vdash
 \ta{\BaseMkComp}{\alpha\li(\BoolT_2\ten\alpha)}}
{}
\qquad
\infer[]
{
\emptyset;\emptyset;\emptyset
\vdash
\ta
{M}
{\,!(\alpha\li\alpha)\li!((\BoolT_2\ten\alpha)\li(\BoolT_2\ten\alpha))}
}
{M\in\{\StepMkCompZ,\StepMkCompO\}}
\]
\[
\infer[]
{
\emptyset;\emptyset;\emptyset
 \vdash
 \ta{\Pred}
    {\BIntT\li\BIntT}
}
{}
\qquad
\infer[]
{
\emptyset;\emptyset;\emptyset
 \vdash
 \ta{\BaseP}
    {\alpha\li((\alpha\li \alpha)\ten \alpha)}
}
{}
\]
\[
\infer[]
{
\emptyset;\emptyset;\emptyset
 \vdash
 \ta{\StepP}
    {(\alpha\li\alpha)\li
     ((\alpha\li\alpha)\ten\alpha)\li
     ((\alpha\li\alpha)\ten\alpha)}
}
{}
\qquad
\infer[]
{
\emptyset;\emptyset;\emptyset
 \vdash
 \ta{\Branch}
    {\BIntT\li\BIntT\li\BIntT\li\BIntT}
}
{}
\]
\normalsize
\end{proposition}

\begin{proposition}[Dynamics of combinators relative to the words.]
\label{proposition:Dynamics of combinators relative to words}
For every $n, a, b\in\Nat$:
\small
\begin{align}
\BSuccO\, \BNum{n}&\red^+\BNum{2n+1}
\label{cxred:bsuco}
\\
\BSuccZ\, \BNum{0}&\red^+\BNum{0}
\label{cxred:bsucz}
\\
\BSuccZ\, \BNum{n}&\red^+\BNum{2n}
 \qquad\qquad\qquad\qquad (n > 0)
\nonumber
\\
\Pred\, \BNum{2n}&\red^+ \BNum{n}
\label{cxred:pred0}
\\
\Pred\, \BNum{2n+1} &\red^+ \BNum{n}
 \nonumber
\\
\Branch\,\BNum{0}
       \,\BNum{a}
       \,\BNum{b}
 &\red^+\BNum{a}
\label{cxred:branch0}
\\
\Branch\,\BNum{2n+i}
       \,\BNum{a}
       \,\BNum{b}
 &\red^+\BNum{b}
 \qquad
 (n\geq 0 \text{ and if }
 n =   0 \text{ then }
 i = 1)
 \nonumber
\end{align}
\normalsize
\end{proposition}
\eqref{cxred:bsuco} shifts (to the right) its argument and adds 1 as a new least significant digit. \eqref{cxred:bsucz} right-shifts its argument, but the added digit is 0.
\eqref{cxred:pred0} calculates the predecessor on a word, that amounts to erase the least significant bit.
\eqref{cxred:branch0} chooses between two words, depending on the value of its first argument.
\par
\textbf{\ElTens.}
We use the \textit{\ElTens} to represent tuples of $\lambda$-terms. The relative type is
$\bigodot^{m}_{i=1} A_i
\equiv
\forall \alpha.(\liv^{m}_{i=1} A_i\liv\alpha)\li\alpha$,
with $m\geq 1$,
and the type constructors coincide to the usual definition of tuples in the $\lambda$-calculus:
\small
\begin{align}
\label{align:etensor-constr}
\elan M_1\ldots M_m\eran
	&\equiv \bs z.z\,M_1\,\ldots\,M_m
	&(m\geq 1)\\
\nonumber
\bs\elan x_1\ldots x_m\eran.M
	&\equiv \bs w.w(\bs x_1\ldots x_m. M)
	&(m\geq 1)
\enspace .
\end{align}
\normalsize
Intuitively, the type of the variable $z$ in \eqref{align:etensor-constr} requires that every component $M_i$ \textit{only} depends on \textit{elementary discharged} free names.
\begin{proposition}[Typing the elementary tensor.]
\label{proposition:Typing rules relative to the elementary tensor}
Rules derivable in \WALL:
\small
\[
 \infer[\odot I_{}]
  {
   \emptyset; \emptyset;
   \{(\Theta_1,\ldots,\Theta_m;\emptyset)\}
   \vdash \ta{\elan M_1,\ldots,M_m\eran}
             {\bigodot^{m}_{i=1} A_i}
  }
  {
   \emptyset; \emptyset;
   \{(\Theta_1;\emptyset)\}
   \vdash \ta{M_1}{A_1}
   &\ldots&
   \emptyset; \emptyset;
   \{(\Theta_m;\emptyset)\}
   \vdash \ta{M_m}{A_m}
  }
\]
\[
 \infer[\li I_{\odot}]
  {
   \Gamma;
   \Delta;
   {\mathcal E},\{(\Theta;\emptyset)\}
   \vdash \ta{\bs \elan x_1,\ldots,x_m \eran.M}
             {(\bigodot^{m}_{i=1}\$A_i)\li B}
  }
  {
   \Gamma;
   \Delta;
   {\mathcal E},
   (\Theta,
    \ta{x_1}{A_1},\ldots,\ta{x_m}{A_m}; \emptyset)
   \vdash \ta{M}{B}
  }
\]
\normalsize
\end{proposition}

\begin{proposition}[Dynamics for the elementary tensor.]
\label{proposition:Dynamics relative to the elementary tensor}
For every $M_{1},$
$\ldots,M_{m}$,\\
$(\bs \elan x_1\cdots x_m\eran.M)\elan M_1,\ldots, M_m\eran
\red^+
(\bs x_1 \ldots x_m.M)\,M_1\ldots M_m$.
\end{proposition}
\par
\textbf{Lists.}
We use $\ListT\,\$A$ as the type of a \textit{list of elements of type} $\$A$, which is
$\ListT\, \$A
\equiv \forall \alpha.
	   !(\$A\liv\alpha\li\alpha)\li\$(\alpha\li\alpha)$.
Observe that the type is derived from the one
$\forall \alpha.!(A\li\alpha\li\alpha)\li\$(\alpha\li\alpha)$ we could expect.
Our choice induces list constructors fruitfully usable in a call-by-value context, like \WALT\ is. Also, we remark that $\$A$ is an argument of the arrow $\liv$ as consequence of the way the combinator $\ListstoConf$ uses every element of a list. We shall see later on how $\ListstoConf$ maps a list of elements to an (initial) configuration, to define an iterator in \WALT. The canonical constructors of the lists are:
\small
\begin{align*}
\Nil    &\equiv\bs cx. x \\
[M_1,\ldots,M_m]
	&\equiv\bs cx. c\,M_1( \ldots (c\,M_m\,x) \ldots)
	&(m\geq 1)\enspace.
\end{align*}
\normalsize
\vspace{-.7cm}
\begin{proposition}[Typing the lists.]
\label{proposition:Typing rules relative to lists}
Rules derivable in \WALL:
\small
\[
\infer[]
{
\emptyset;\emptyset;\emptyset
\vdash \ta{\Nil}{\ListT\, \$A}
}
{}
\qquad
\qquad
\infer[]
{
\emptyset;
\emptyset;
{\mathcal E}_{1}\sqcup\ldots\sqcup{\mathcal E}_{m}
\vdash \ta{[M_1, \ldots, M_m]}{\ListT\, \$A}
}
{
\emptyset;
\emptyset;
{\mathcal E}_{i}
\vdash \ta{M_i}{\$A}
&
{\mathcal E}_{i}\subseteq\{(\Theta_i;\emptyset)\}
&
i\in\{1,\ldots,m\}
}
\]
\normalsize
\end{proposition}
\subsection{Core combinators}
\label{subsection:Core combinators}
The core combinators constitute an intermediate step that simplifies the definition of an iterator and of a composition in \WALT.
We start showing how we can embed the arguments and the result of terms into a suitable number of boxes. Since we have two kinds of implications, and we can transform the standard linear implication into an eager one, we have three kinds of embedding functors.
\par
\textbf{Basic embedding.}
For every $n\geq 1$ the \textit{basic embedding} is $\BEmbed{n}{M}\equiv\bs x. M x$.
Its purpose is to take a term $M$, representing a function with a single linear argument, and to transform it into a term that represents a function with a single eager argument.
\par
\textbf{Linear embedding.}
For every $n,p\geq 0$, the \textit{linear embedding} is
$\LEmbed{n}{p}{M}\equiv\bs x_1 \ldots x_{p}. M\,x_1 \ldots x_{p}$.
\par
\textbf{Eager embedding.}
For every $n, p, q\geq 0$, the \textit{eager embedding} is:
\small
\begin{align*}
\EEmbed{n}{p}{q}{M}
&\equiv
\bs w_1\ldots w_p z_1\ldots z_q.
\\
&
\quad
(\bs w_1\ldots w_p.M w_1\ldots w_p z_1\ldots z_q)
(\BEmbed{1}{\Coerc^n} w_1)
\ldots
(\BEmbed{1}{\Coerc^n} w_p)
\enspace .
\end{align*}
\normalsize
\begin{proposition}[Typing the embeddings.]
\label{proposition:Typing the embeddings}
Rules derivable in \WALT:
\small
\[
\infer[]
{
\Gamma;\Delta;{\mathcal E}
 \vdash
 \ta{\BEmbed{n}{M}}
    {\$^n L\liv \$^{m+n} A}
}
{\emptyset;\emptyset;\emptyset
 \vdash
 \ta{M}
    {L\li \$^{m} A}
 & m\geq 0
 & n\geq 1
 }
\qquad
\infer[]
{
\Gamma;\Delta;{\mathcal E}
 \vdash
 \ta{\LEmbed{n}{p}{M}}
    {(\li^{p}_{i=1}\$^n L_i)\li\$^{m+n} A
    }
}
{\emptyset;\emptyset;\emptyset
 \vdash
 \ta{M}{(\li^{p}_{i=1} L_i)\li\$^{m} A} & m,n,p\geq 0}
\]
\[
\infer[]
{
\Gamma;\Delta;{\mathcal E}
 \vdash
 \ta{\EEmbed{n}{p}{q}{M}}
    {(\liv^{p}_{i=1}\$^n \BIntT)\liv
     (\liv^{q}_{j=1}\$^{m+n} L_j)\liv
     \$^{m+n} A
    }
}
{\emptyset;\emptyset;\emptyset
 \vdash
 \ta{M}
   {(\liv^{p}_{i=1} \$\BIntT)\liv
    (\liv^{q}_{j=1}\$^{m} L_j)\liv
    \$^{m} A}
 & m\geq 1 & n,p,q\geq 0
}
\]
\normalsize
\end{proposition}
We observe that $\Gamma, \Delta$, and ${\mathcal E}$ in the rules that give type to both $\LEmbed{n}{p}{M}$, and $\EEmbed{n}{p}{q}{M}$ can be \textit{not empty} only if $n\geq 1$.

\begin{proposition}[Dynamics of the embeddings.]
\label{proposition:Dynamics of the embeddings}
For every $n,p,q\geq 0$, and values
$M_{1}, \ldots, M_{p},M'_{1}, \ldots, M'_{q}$:
\small
\begin{align*}
\BEmbed{}{M}\,M_{1}
&\red^* M\,M_{1}
\\
\LEmbed{n}{p}{M}\,M_{1}\,\ldots\,M_{p}
&\red^* M\,M_{1}\,\ldots\,M_{p}
\\
\EEmbed{n}{p}{q}{M}\,M_{1}\,\ldots\,M_{p}\,M'_{1}\,\ldots\,M'_{q}
&\red^* M\,M_{1}\,\ldots\,M_{p}\,M'_{1}\,\ldots\,M'_{q}
\enspace .
\end{align*}
\normalsize
\end{proposition}
\par
\textbf{Coercion.}
The \textit{coerce} function takes an instance of a binary word and reconstructs it inside a box. It is
$\Coerc\equiv\bs n. (\bs z. z\,\BNum{0})(n\, \BSuccZ\, \BSuccO)$.
To our purposes, $\Coerc$ must be iterated to reconstruct the given word into some given number of boxes:
\small
\begin{align*}
\Coerc^{0}   &\equiv\bs x. x\\
\Coerc^{1}   &\equiv\Coerc\\
\Coerc^{m+1} &\equiv\bs x.\LEmbed{1}{1}{\Coerc^{m}}(\Coerc^1\,x)
&(m\geq 1)
\end{align*}
\normalsize
\vspace{-.5cm}
\begin{proposition}[Typing the coercions.]
\label{proposition:Typing the coercions}
Rule derivable in \WALT:
\small
\[
\infer[]
{\emptyset;\emptyset;\emptyset
\vdash\ta{\Coerc^{m}}{\BIntT\li\$^{m}\BIntT}}
{m\geq 0}
\]
\normalsize
\end{proposition}
\begin{proposition}[Dynamics of the coercion.]
\label{proposition:Dynamics of the coercion}
For every $m\geq 0$,
$\Coerc^{m}\ \BNum{n}\red^+\BNum{n}$.
\end{proposition}
So, $\Coerc^{m}$ is the identity on the given argument, but the type of the result changes, getting a modal type.
\par
\textbf{Diagonals.}
Every \textit{diagonal} replicates a word instance inside some boxes.
The (\textit{standard}) \textit{diagonal} puts together the copies of the given input by means of a tensor. Every copy is generated from scratch, by iterating the successors on words. The parameter $n$ indicates the number of copies to generate. For every $n\geq 1$, the result is one box deep:
\small
\begin{align*}
\nabla_{n}
	&\equiv\bs w.\!\!\begin{array}[t]{l}
	            (\bs z.
		     z\,\overbrace{\lan
		                   \BNum{0},\ldots,\BNum{0}
				   \ran}^{n}
		    )
		    (w\, \!\!\begin{array}[t]{l}
		         (\bs \lan x_1\ldots x_{n}\ran.
			  \lan
			  \BSuccZ\, x_1,
		          \ldots,
		          \BSuccZ\, x_{n}
			  \ran
		         )\\
		         (\bs \lan x_1\ldots x_{n}\ran.
			  \lan
			  \BSuccO\, x_1,
		          \ldots,
		          \BSuccO\, x_{n}
			  \ran
			 )
	             )\enspace .
	                  \end{array}
		    \end{array}
\end{align*}
\normalsize
We have a second version of diagonal, the \textit{elementary diagonal}, that puts together the copies of the given input by means of an elementary tensor constructor. Every copy is generated from scratch, by iterating the successors on words. The parameter $n$ indicates the number of copies to generate. For every $m,n\geq 1$, the result is contained into a single box, but every component of the elementary tensor, in the result, is $m$ boxes deep:
\small
\begin{align*}
\nabla^{m}_{n}
	&\equiv\bs w.\!\!\begin{array}[t]{l}
	            (\bs z.
		     z\,\overbrace{\elan
		                   \BNum{0},\ldots,\BNum{0}
				   \eran}^{n}
		    )
		    (w\, \!\!\begin{array}[t]{l}
		         (\bs \elan x_1\ldots x_{n}\eran.
			  \elan
			  \BEmbed{m}{\BSuccZ}\, x_1,
		          \ldots,
		          \BEmbed{m}{\BSuccZ}\, x_{n}
			  \eran
		         )\\
		         (\bs \elan x_1\ldots x_{n}\eran.
			  \elan
			  \BEmbed{m}{\BSuccO}\, x_1,
		          \ldots,
		          \BEmbed{m}{\BSuccO}\, x_{n}
			  \eran
			 )
	             )\enspace .
			 \end{array}
		    \end{array}
\end{align*}
\normalsize
\vspace{-.5cm}
\begin{proposition}[Typing the diagonals.]
\label{proposition:Typing the diagonals}
Rules derivable in \WALT:
\small
\[
\infer[]
{\emptyset;\emptyset;\emptyset
\vdash
\ta{\nabla^{}_{n}}
   {\BIntT\li
    \$(\bigotimes_{i=1}^{n}\BIntT)
   }
}
{n\geq 1}
\qquad\qquad
\infer[]
{\emptyset;\emptyset;\emptyset
\vdash
\ta{\nabla^{m}_{n}}
   {\BIntT\li
    \$(\bigodot_{i=1}^{n}\$^{m}\BIntT)
   }
}
{m\geq 1
 &
 n\geq 1}
\]
\normalsize
\end{proposition}

\begin{proposition}[Dynamics of the diagonals.]
\label{proposition:Dynamics of the diagonals}
For every $m, n\geq 1$, both
$\nabla^{}_{n}\, \BNum{a}
\red^+\overbrace{\lan\BNum{a},\ldots,\BNum{a}\ran}^{n}$, and
$\nabla^{m}_{n}\, \BNum{a}
\red^+\overbrace{\elan\BNum{a},\ldots,\BNum{a}\eran}^{n}$.
\end{proposition}
\par
\textbf{Recasting combinators.}
We define a term that maps any \textit{word into a string} as long as the word:
\small
\begin{align*}
\BInttoUInt&\equiv\bs nf.(\bs zy. z\,y)(n\,f\,f)
\enspace .
\end{align*}
\normalsize
Another useful term maps a \textit{string to a list} with as many copies of a given \textit{closed} term as the string's length. The closed term is the first argument, while the string is its second one:
\small
\begin{align*}
\UInttoList
	&\equiv
        \bs knc.
	  (\bs zx. z\,(\bs f.x)\,\Id)(n (\bs lf. c\,k\,(l\,\Id)) )
\enspace .
\end{align*}
\normalsize
\begin{proposition}[Typing the recasting combinators.]
\label{proposition:Typing the recasting combinators}
Rules derivable in \WALT:
\small
\[
\infer[]
{
\emptyset;\emptyset;\emptyset
\vdash\ta{\BInttoUInt}{\BIntT\li\UIntT}
}
{}
\qquad\qquad
\infer[]
{
\emptyset;
\emptyset;
\emptyset
\vdash
\ta{\UInttoList}{\$^2A\liv\UIntT\li\ListT\, \$A}
}
{}
\]
\normalsize
\end{proposition}
The will to type a term like $\UInttoList$ required to generalize the rule $!$ of \LAL\ as in \WALT.

\begin{proposition}[Dynamics of the recasting combinators.]
\label{proposition:Dynamics of the recasting combinators}
For every $m\geq 0$ and every \textit{closed value} $M$:
\small
\begin{align}
\BInttoUInt\, \BNum{0}&\red^+ \UNum{0}
\label{cxred:w2s}
\\
\BInttoUInt\,\BNum{2^m+2^{m-1}\cdot\nu_{m-1}+\cdots+2^0\cdot\nu_{0}}
 &\red^+\UNum{m+1}
 &(\nu_i\in\{0,\ldots,m-1\})
\nonumber
\\
\UInttoList\ M\, \UNum{m}
 &\red^+\underbrace{[M,\ldots, M]}_{m}
\label{cxred:uinttoist}
\end{align}
\normalsize
\end{proposition}
\eqref{cxred:w2s} transforms a word with $m$ digits into a string of the same length.
\eqref{cxred:uinttoist} builds a list as much long as the value of the second argument. The list contains copies of the first argument, which must be a \textbf{closed} term.
\subsection{Iterators}
\label{subsection:Iterators}
We shall define combinators to build an iterator scheme in \WALT, the goal being the simulation of the recursive scheme in \QlSRN.
We want to give some intuitions about how the iterator works. So, we assume to have a recursively defined function $f(0,a) = g(a)$, and $f(n,a) = h(n-1,a,f(n-1,a))$, with $n\geq1$. Then, we show how simulating its top-down recursive unfolding:
\small
\begin{align*}
f(n,a) =  h(n-1,a,f(n-1,a))
= \ldots =  h(n-1,a,h(n-2,a,\ldots h(0,a,g(a))\ldots))
\end{align*}
\normalsize
by a bottom-up reconstruction that iterates some transition functions on suitable configurations and pre-configurations.
The reconstruction requires to assume $H, G$ be the interpretations of $h, g$, respectively, in \WALT. Moreover, for simplicity, we assume
the unary strings $\UNum{n}, \UNum{a}$ represent $n, a$ in \WALT.
What we are going to say, though, keeps holding with $f$ of arbitrary arity and with words as its arguments, instead of strings.
The \textit{main problem} to reconstruct the unfolding above is the need to replicate $\UNum{a}$. To see how overcoming that problem, we start by assuming that, in \WALT, we can develop sequences of computations like the following one:
\small
\begin{align}
\label{align:intro-step1}
&
\llan  G\UNum{a}
	, [\underbrace{\UNum{0},\ldots,\UNum{0}}_{n+1}]
	, [\underbrace{\UNum{a},\ldots,\UNum{a}}_{n+1}] \rran
\red^{*}
\\
\label{align:intro-step2}
\llan G\UNum{a}
     ,\lan \UNum{0},\underbrace{[\UNum{1},\ldots,\UNum{1}]}_{n}\ran
     ,\lan \UNum{a},\underbrace{[\UNum{a},\ldots,\UNum{a}]}_{n}    \ran     \rran
\red^{*}
&
\llan H\,\UNum{0}\,\UNum{a}\,(G\UNum{a})
    ,\underbrace{[\UNum{1},\ldots,\UNum{1}]}_{n}
    ,\underbrace{[\UNum{a},\ldots,\UNum{a}]}_{n} \rran
\red^{*}
\\
\label{align:intro-step3}
\llan H\,\UNum{0}\,\UNum{a}\,(G\UNum{a})
     ,\lan \UNum{1},\underbrace{[\UNum{2},\ldots,\UNum{2}]}_{n-1}\ran
     ,\lan \UNum{a},\underbrace{[\UNum{a},\ldots,\UNum{a}]}_{n-1}    \ran     \rran
\red^{*}
&
\llan H\,\UNum{1}\,\UNum{a}\,(H\,\UNum{0}\,\UNum{a}\,(G\UNum{a}))
    ,\underbrace{[\UNum{1},\ldots,\UNum{1}]}_{n-1}
    ,\underbrace{[\UNum{a},\ldots,\UNum{a}]}_{n-1} \rran
\red^{*}\ldots
\end{align}
\normalsize
The right hand column contains \textit{configurations}, the topmost being the \textit{initial} one.
The left hand column contains \textit{pre-configurations}.
Every pre-configuration comes from its preceding configuration by
(i) separating head and tail of every list, and storing them as the two components of a same pair,
(ii) only on the leftmost list, simultaneously to the separation, the successor is mapped on the tail.
\par
Every configuration, other than the initial one, is obtained from its preceding pre-configuration by the application of an instance of $H$ to the first element of every pair, and to the first element of the whole pre-configuration, which accumulates the partial result of the bottom-up reconstruction.
\par
Everything works correctly if the formula
$\$\UIntT\liv\$^{m}\UIntT\liv\$^{m}\UIntT$, for some $m$,
becomes the type of the term of \WALT\ that represents $f$.
Such a type says that $f$ becomes a term of \WALT\ that eagerly evaluates its two arguments. $\$\UIntT$ is the type of $\UNum{n}$, here representing a normal argument.
$\$^{m}\UIntT$ is the type of both $\UNum{a}$, here representing a safe argument, and of the result. It is crucial that these two types coincide, otherwise we could not interpret any recursive scheme. We could obtain such a coincidence only by generating $[\UNum{a},\ldots,\UNum{a}]$, in the initial configuration, using $\UInttoList$ above.
The peculiarity of $\UInttoList$ is that, having $\$^2\UIntT\liv\UIntT\li\ListT\, \$\UIntT$, as its type, we can look at $\UInttoList$ as it was a kind of dereliction: one of the $\$$-modalities in the type of its first argument is absorbed by the $\$$ rule hidden in the structure of the resulting list. This behavior is obtained by making an essential use of the rule $!$ where $\UNum{a}$ is an exponential assumption, namely a value that, eventually, the context will supply. The ``disappearing'' $\$$ modality allows to implement the bottom-up reconstruction through a combinator $\ctoc$ that takes a configuration at a \textit{given level} and yields another one at the \textit{same} level.
\par
Now, we move to the technical parts, where we set the relevant data-types.
\par
\textbf{Configurations.}
For every $\Intg{k}\geq 1$, the type of the configurations is:
\small
\begin{align*}
\bcConfT[\$A_1\ldots \$A_{\Intg{k}};\$B]
&\equiv
\forall \alpha_1\ldots\alpha_{\Intg{k}}.
(\li_{i=1}^{\Intg{k}}!(\$A_i\liv\alpha_i\li\alpha_i))\li
\\
&\qquad
\$((\li_{i=1}^{\Intg{k}}\alpha_i)
   \li
   \forall \gamma.
   ((\$B\liv(\li_{i=1}^{\Intg{k}}\alpha_i)\li\gamma)
       \li\gamma))
\end{align*}
\normalsize
such that
$\{\alpha_1,\ldots,\alpha_{\Intg{k}},\gamma\}\cap\FV{B}=\emptyset$.
We shall use the following canonical instance of the type of the configurations:
\small
\begin{align*}
\bcConfT[1+\Intg{n};\Intg{s};m]
&\equiv
\bcConfT[\overbrace{\$\BIntT\ldots\$\BIntT}^{\Intg{n}+1}\,
         \overbrace{\$^{m}\BIntT\ldots\$^{m}\BIntT}^{\Intg{s}};
	 \$^{m}\BIntT]
\end{align*}
\normalsize
for some given $\Intg{n},\Intg{s} \geq 0$ and $m \geq 1$,
whose canonical realizers are given by the following scheme:
\small
\begin{align*}
&
\lan\!\lan \BNum{r},
           [\BNum{a_1}   ,\ldots,\BNum{a_{\Intg{r}}} ],
           {[\BNum{n_{11}},\ldots,\BNum{n_{1\Intg{r}}}]}
	   ,\ldots
	   ,[\BNum{n_{\Intg{n}1}},\ldots,\BNum{n_{\Intg{n}\Intg{r}}}],
	   {[\BNum{s_{11}},\ldots,\BNum{s_{1\Intg{r}}}]}
	   ,\ldots
	   ,[\BNum{s_{\Intg{s}1}},\ldots,\BNum{s_{\Intg{s}\Intg{r}}}]
	   \ran\!\ran
\equiv
\\
&\qquad
\bs d_0 d_1 \ldots d_{\Intg{n}}
        e_1 \ldots e_{\Intg{s}}.
\bs w_0 w_1\ldots w_{\Intg{n}}
        z_1\ldots z_{\Intg{s}}.
\\
&\qquad
\begin{array}[t]{ll}
\bs x.
 x\,
 \BNum{r}\,
 (d_0\,\BNum{a_1}
   (\cdots(d_0\,\BNum{a_{\Intg{r}}}\,w_0)\cdots))\,
 \\\phantom{\bs x.x\,\BNum{r}\,}
 (d_1\,\BNum{n_{11}}
   (\cdots(d_1\BNum{n_{1\Intg{r}}}\,w_1)\cdots))
 \ldots
 (d_{\Intg{n}}\,\BNum{n_{\Intg{n}1}}
   (\cdots
     (d_{\Intg{n}}\,\BNum{n_{\Intg{n}\Intg{r}}}
          \,w_{\Intg{n}})\cdots))\,
 \\\phantom{\bs x.x\,\BNum{r}\,}
 (e_1\BNum{s_{11}}
   (\cdots(e_1\BNum{s_{1\Intg{r}}}z_1)\cdots))
 \ldots
 (e_{\Intg{s}}\BNum{s_{\Intg{s}1}}
   (\cdots
         (e_{\Intg{s}}\BNum{s_{\Intg{s}\Intg{r}}}
	    \,z_{\Intg{s}})\cdots))
\end{array}
\end{align*}
\normalsize
Essentially, the scheme is a tuple whose first element is a word, and all the remaining elements are lists of words, all with the same length.

\begin{proposition}[Typing the configurations]
\label{proposition:Typing the configurations}
Let $\Intg{n}, \Intg{s}, \Intg{r}\geq 0$, and $m\geq 1$.
A rule derivable in \WALT:
\label{proposition:Relating XWALL and WALL-1}
\begin{center}
\small
\begin{tabular}{c}
\infer[]
{\emptyset;\emptyset;\emptyset\vdash
 \!\!
 \begin{array}[t]{l}
 \lan\!\lan
         \BNum{r}
        ,[\BNum{a_{1}},\ldots,\BNum{a_{\Intg{r}}}]
 \\\phantom{\lan\!\lan\BNum{r}}
        ,[\BNum{n_{11}},\ldots,\BNum{n_{1\Intg{r}}}]
	,\ldots
        ,[\BNum{n_{\Intg{n}1}},\ldots,\BNum{n_{\Intg{n}\Intg{r}}}]
 \\\phantom{\lan\!\lan\BNum{r}}
	,[\BNum{s_{11}},\ldots,\BNum{s_{1\Intg{r}}}]
	,\ldots
        ,[\BNum{s_{\Intg{s}1}},\ldots,\BNum{s_{\Intg{s}\Intg{r}}}]
 \ran\!\ran
 \!:\!\bcConfT[1+\Intg{n};\Intg{s};m]
 \end{array}
}
{
\begin{array}{ll}
\emptyset;\emptyset;\emptyset
\vdash\ta{\BNum{a_i}}{\$\BIntT} & i\in\{1,\ldots,\Intg{r}\}
\\
\emptyset;\emptyset;\emptyset
\vdash\ta{\BNum{r}}{\$^{m}\BIntT}
\\
\emptyset;\emptyset;\emptyset
\vdash\ta{\BNum{n_{ij}}}{\$\BIntT} & i\in\{1,\ldots,\Intg{n}\},\
j\in\{1,\ldots,\Intg{r}\}
\\
\emptyset;\emptyset;\emptyset
\vdash\ta{\BNum{s_{ij}}}{\$^{m}\BIntT} & i\in\{1,\ldots,\Intg{s}\},\
j\in\{1,\ldots,\Intg{r}\}
\end{array}
}
\end{tabular}
\normalsize
\end{center}
\end{proposition}
\textbf{Final configurations.}
For every $\Intg{k}\geq 1$, the type of the final configurations is:
\small
\begin{align*}
\FbcConfT[\$A_1\ldots \$A_{\Intg{k}};\$B]
&\equiv
\forall \alpha_1\ldots\alpha_{\Intg{k}}\gamma.
(\li_{i=1}^{\Intg{k}}!(\$A_i\liv\alpha_i\li\alpha_i))\li
\\
&\qquad
\$((\li_{i=1}^{\Intg{k}}\alpha_i)
   \li
   (\$B\liv(\li_{i=1}^{\Intg{k}}\alpha_i)\li\gamma)
       \li\gamma)
\end{align*}
\normalsize
such that
$\{\alpha_1,\ldots,\alpha_{\Intg{k}},\gamma\}\cap\FV{B}=\emptyset$.
The difference with the type of the configurations is the extrusion of the universal quantifier on $\gamma$.
We shall use the following canonical instance of the type of the final configurations,
for some given $\Intg{n},\Intg{s} \geq 0$, and $m \geq 1$:
\small
\begin{align*}
\FbcConfT[1+\Intg{n};\Intg{s};m]
&\equiv
\bcConfT[\overbrace{\$\BIntT\ldots\$\BIntT}^{\Intg{n}+1}\,
         \overbrace{\$^{m}\BIntT\ldots\$^{m}\BIntT}^{\Intg{s}};
	 \$^{m}\BIntT]
\end{align*}
\normalsize
The canonical realizer of a final configuration has the same form as a canonical realizer of the configurations.
The final configurations are introduced as a necessary step to extract, with the correct typing, the first component $\BNum{r}$ of a configuration, that will represent the result of an iteration.
\par
\textbf{pre-Configurations.}
For every $\Intg{k}\geq 1$, the type of the pre-configurations is:
\small
\begin{align*}
\SbcConfT[\alpha_1\ldots\alpha_{\Intg{k}},\delta;
          \$A_1\ldots \$A_{\Intg{k}};\$B]
&\equiv
\forall \gamma.
((\$B\liv
 (\li_{i=1}^{\Intg{k}}\ST[\alpha_i,\delta;\$A_i])\li\gamma
)\li\gamma)
\end{align*}
\normalsize
where:
\small
\begin{align*}
\SU[\alpha,\delta;\$A]
&\equiv
(\$A\li\alpha\li\alpha)
\li
(\$A\li \$A)
\li
\$A
\liv
((\delta\li\delta)\li\alpha)
\\
\ST[\alpha,\delta;\$A]
&\equiv
\forall \beta. ((\SU[\alpha,\delta;\$A]\li\beta)\li\beta)
\end{align*}
\normalsize
such that $\{\alpha_1,\ldots,\alpha_{\Intg{k}},\delta,\gamma\}\cap\FV{B}=\emptyset$.
We shall use the following canonical instance of the type of the pre-configurations:
\small
\begin{align*}
\SbcConfT
[\alpha_0\ldots\alpha_{\Intg{n}+\Intg{s}},\delta;m]
&\equiv
\SbcConfT
[\alpha_0\ldots\alpha_{\Intg{n}+\Intg{s}},\delta;
 \overbrace{\$\BIntT\ldots\$\BIntT}^{1+\Intg{n}},
 \overbrace{\$^{m}\BIntT\ldots\$^{m}\BIntT}^{\Intg{s}};
 \$^{m}\BIntT
]
\end{align*}
\normalsize
for some $\Intg{n},\Intg{s} \geq 0$ and $m \geq 1$,
Given $\Intg{r}\geq 0$, the realizers of the canoncal pre-configurations are given by the following scheme:
\small
\begin{align*}
&
\begin{array}[t]{lll}
\lan\!\lan &\BNum{r} ,
	    \lan \BNum{a_1}, [\BNum{a_2}
                              ,\ldots,
			      \BNum{a_{\Intg{r}}} ]\ran,&\\
           &\quad
            {\lan \BNum{n_{11}},[\BNum{n_{12}}
	                         ,\ldots,
				 \BNum{n_{1\Intg{r}}}]\ran}
            ,\ldots,
            \lan
	    \BNum{n_{\Intg{n}1}},[\BNum{n_{\Intg{n}2}}
	                          ,\ldots,
				  \BNum{n_{\Intg{n}\Intg{r}}}]\ran,\\
           &\quad
            {\lan \BNum{s_{11}},[\BNum{s_{12}}
	                         ,\ldots,
				 \BNum{s_{1\Intg{r}}}]\ran}
            ,\ldots,
            \lan
	    \BNum{s_{\Intg{s}1}},[\BNum{s_{\Intg{s}2}}
	                          ,\ldots,
				  \BNum{s_{\Intg{s}\Intg{r}}}]\ran
	   \quad
	   \ran\!\ran_{\scriptsize
	         \begin{array}[t]{l}
	          c_{1} \ldots c_{\Intg{r}}\\
	          d_{11}\ldots d_{1\Intg{r}}
		    \cdots\cdots
	          d_{\Intg{n}1}\ldots d_{\Intg{n}\Intg{r}}\\
	          e_{11}\ldots e_{1\Intg{r}}
		    \cdots\cdots
	          e_{\Intg{s}1}\ldots e_{\Intg{s}\Intg{r}}\\
	          w_{0} w_{1} \ldots w_{\Intg{n}}
		        z_{1} \ldots z_{\Intg{s}}\\
	         \end{array}
		       \normalsize}\equiv
\end{array}
\\
&
\bs x.x\,\BNum{r}
      (\bs t.t\,
	     c_1\,
	     M_0\,
             \BNum{a_1}\,
	     (\bs f.c_2\BNum{a_2}
               (\ldots(c_{\Intg{r}}\BNum{a_{\Intg{r}}}
	        w_0)\ldots)))
\\
&
\phantom{\bs x.x\,\BNum{r}}
    (\bs t.t\,
     d_{11}\,
     M_1\,
     \BNum{n_{11}}\,
     (\bs f.d_{12}\BNum{n_{12}}
        (\ldots(d_{1\Intg{r}}\BNum{n_{1\Intg{r}}}
	  w_1)\ldots)))
\cdots
    (\bs t.t\,
     d_{\Intg{n}1}\,
     M_{\Intg{n}}\,
     \BNum{n_{\Intg{n}1}}\,
     (\bs f.d_{\Intg{n}2}
       \BNum{n_{\Intg{n}2}}
        (\ldots
	 (d_{\Intg{n}\Intg{r}}
	   \BNum{n_{\Intg{n}\Intg{r}}}
	     w_{\Intg{n}}
	 )\ldots
	)))
\\
&\phantom{\bs x.x\,\BNum{r}}
    (\bs t.t\,
     e_{11}\,
     N_{1}\,
     \BNum{s_{11}}\,
     (\bs f.e_{12}
       \BNum{s_{12}}
        (\ldots
	  (e_{1\Intg{r}}
	   \BNum{s_{1\Intg{r}}}
	   z_1
	  )\ldots
	)
      )
     )
\cdots
    (\bs t.t\,
     e_{\Intg{s}1}\,
     N_{\Intg{s}}\,
     \BNum{s_{\Intg{s}1}}\,
     (\bs f.e_{\Intg{s}2}
       \BNum{s_{\Intg{s}2}}
        (\ldots(e_{\Intg{s}\Intg{r}}
	        \BNum{s_{\Intg{s}\Intg{r}}}
		z_{\Intg{s}}
	       )
	 \ldots
	)))
\end{align*}
\normalsize
\textbf{up} to any choices of the \textbf{closed values} $M_i$, with $i\in\{0,\ldots,\Intg{n}\}$, and $N_i$, with $j\in\{1,\ldots,\Intg{s}\}$.
Essentially, every realizer of a canonical pre-configuration is a tuple. The first element is a word, and all the others are pairs. Every pair contains a word and a list of words. All the lists in the same pre-configuration have the same length.

\begin{proposition}[Typing the pre-configurations]
\label{proposition:Typing the pre-configurations}
\label{proposition:Relating XWALL and WALL-2}
Let $\Intg{n}, \Intg{s}, \Intg{r}\geq 0$, and $m\geq 1$.
A rule derivable in \WALL:
\small
\[\infer[]
{
\begin{array}[t]{l}
 \begin{array}[t]{l}
 \ta{c_1,\ldots,c_{\Intg{r}}}
    {\$\BIntT\li\alpha_0\li\alpha_0},\\
 \ta{d_{11},\ldots,d_{1\Intg{r}}}
    {\$\BIntT\li\alpha_1\li\alpha_1},
 \\\cdots\cdots,
 \ta{d_{\Intg{n}1},\ldots,d_{\Intg{n}\Intg{r}}}
    {\$\BIntT\li\alpha_{\Intg{n}}\li\alpha_{\Intg{n}}},\\
 \ta{e_{11},\ldots,e_{1\Intg{r}}}
    {\$^{m}\BIntT\li\alpha_{1+\Intg{n}}\li\alpha_{1+\Intg{n}}},
 \\\cdots\cdots,
 \ta{e_{\Intg{s}1},\ldots,e_{\Intg{s}\Intg{r}}}
    {\$^{m}\BIntT\li
      \alpha_{\Intg{s}+\Intg{n}}\li
      \alpha_{\Intg{s}+\Intg{n}}},
 \\
 \ta{w_0}{\alpha_0},
 \ta{w_1}{\alpha_1},
  \ldots,\ta{w_{\Intg{n}}}{\alpha_{\Intg{n}}},
 \ta{z_1}{\alpha_{1+\Intg{n}}},
  \ldots,\ta{z_{\Intg{s}}}{\alpha_{\Intg{s}+\Intg{n}}}
 ;\emptyset;\emptyset
 \end{array}\\
 \qquad
 \vdash\lan\!\lan\!\!\!\!
 \begin{array}[t]{l}
  \
  \BNum{r},  \lan\BNum{a_1},[\BNum{a_2},\ldots,\BNum{a_{\Intg{r}}}]\ran,
  \\\phantom{\BNum{r}}
  \lan\BNum{n_{11}},[\BNum{n_{12}},\ldots,\BNum{n_{1\Intg{r}}}],\ran
  \ldots
  ,\lan\BNum{n_{\Intg{n}1}},
       [\BNum{n_{\Intg{n}2}},
       \ldots,
       \BNum{n_{\Intg{n}\Intg{r}}}]
  \ran,
  \\\phantom{\BNum{r}}
  \lan\BNum{s_{11}},[\BNum{s_{12}},\ldots,\BNum{s_{1\Intg{r}}}]\ran
  ,\ldots
  ,\lan\BNum{s_{\Intg{s}1}}
       ,[\BNum{s_{\Intg{s}2}},\ldots,\BNum{s_{\Intg{s}\Intg{r}}}]\ran
\ \
  \ran\!\ran\!_{\scriptsize
                \begin{array}[t]{l}
                c_1\ldots c_{\Intg{r}}\\
                d_{11}\ldots d_{1\Intg{r}}
                \cdots
		d_{\Intg{n}1}\ldots d_{\Intg{n}\Intg{r}}\\
                e_{11}\ldots e_{1\Intg{r}}
                \cdots
		e_{\Intg{s}1}\ldots e_{\Intg{s}\Intg{r}}\\
		w_{0} w_{1}\ldots w_{\Intg{n}}z_{1}\ldots z_{\Intg{s}}
                \end{array}
		\normalsize
               }
 \end{array}
\\\qquad
:\!\SbcConfT[\alpha_{0}\ldots\alpha_{\Intg{n}+\Intg{s}},\delta;m]
\end{array}
}
{
\begin{array}{ll}
\emptyset;\emptyset;\emptyset
\vdash\ta{\BNum{a_i}}{\$\BIntT} & i\in\{1,\ldots,\Intg{r}\}
\\
\emptyset;\emptyset;\emptyset
\vdash\ta{\BNum{r}}{\$^{m}\BIntT}
\\
\emptyset;\emptyset;\emptyset
\vdash\ta{\BNum{n_{ij}}}{\$\BIntT} & i\in\{1,\ldots,\Intg{n}\},\
j\in\{1,\ldots,\Intg{r}\}
\\
\emptyset;\emptyset;\emptyset
\vdash\ta{\BNum{s_{ij}}}{\$^{m}\BIntT} & i\in\{1,\ldots,\Intg{s}\},\
j\in\{1,\ldots,\Intg{r}\}
\end{array}
}
\]
\normalsize
\end{proposition}
\textbf{Transition function.}
\label{subsubsection:Transition function}
$\TransFunc_{1+\Intg{n};\Intg{s}}[F]$ is a transition function that maps configurations to configurations. It composes the two combinators $\HeadsandTails_{1+\Intg{n};\Intg{s}}[F]$ and
$\NextConf_{1+\Intg{n};\Intg{s}}[F']$. The first produces a pre-configuration, from a given configuration, while the second goes in the opposite direction. $F$ and $F'$ are parameters that will be instantiated by combinators. $1+\Intg{n}$ and $\Intg{s}$ represent the normal and the safe arities, so pointing to the future use we shall make of the transition function to represent \QlSRN\ functions with normal and safe arguments.
The definitions are here below, where $\BaseTransFunc$ and $\StepTransFunc$, are a base and a step function, respectively, used to extract pairs on lists:
\small
\begin{align*}
\TransFunc_{1+\Intg{n};\Intg{s}}[F,F']
\equiv&
\bs x.
\bs d_0 d_1 \ldots d_{\Intg{n}}
        e_1 \ldots e_{\Intg{s}}.
\\
&
(\bs b.
 \bs w_0 w_1 \ldots w_{\Intg{n}}
         z_1 \ldots z_{\Intg{s}}.
 \NextConf_{1+\Intg{n};\Intg{s}}[F']
   (b\, w_0\, w_1\ldots w_{\Intg{n}}\,z_1\ldots z_{\Intg{s}})
)
\\
&\qquad
(\HeadsandTails_{1+\Intg{n};\Intg{s}}[F]\,
   x\,d_0\,d_1\ldots d_{\Intg{n}}\,e_1\ldots e_{\Intg{s}})
\\
\StepTransFunc^{m}[G]
&\equiv
        \bs catx.x\, c\, \LEmbed{m}{1}{G}\, a\,
	  (\bs f.t\,(\bs cgal. c(g\,a)(l\,\Id)))
\\
\BaseTransFunc^{m}
&\equiv
\bs yx. x\, (\bs xy.y)\, \LEmbed{m}{1}{\Id}\, \BNum{0}\, (\bs f.y)
\end{align*}
\begin{align*}
\HeadsandTails_{1+\Intg{n};\Intg{s}}[G]
\equiv&
\bs x.
\bs d_0 d_1\ldots d_{\Intg{n}}
        e_1\ldots e_{\Intg{s}}.
\\
&
(\bs b. \bs w_0 w_1\ldots w_{\Intg{n}}
                z_1\ldots z_{\Intg{s}}.
	        b \begin{array}[t]{l}
	           (\BaseTransFunc^{1}\, w_0)
		   \\
	           (\BaseTransFunc^{1}\, w_{1})
		     \ldots (\BaseTransFunc^{1}\, w_{\Intg{n}})
		   \\
		   (\BaseTransFunc^{1}\, z_{1})
		     \ldots (\BaseTransFunc^{1}\, z_{\Intg{s}})
                  \end{array}
\\
&
)(x \begin{array}[t]{l}
    (\StepTransFunc^{1}[G]\, d_0)
    (\StepTransFunc^{1}[\Id]\,d_{1})
       \ldots(\StepTransFunc^{1}[\Id]\,d_{\Intg{n}})\\
    \phantom{(\StepTransFunc^{1}[F]\, d_0)}
    (\StepTransFunc^{m}[\Id]\,e_{1})
       \ldots(\StepTransFunc^{m}[\Id]\,e_{\Intg{s}}))
    \end{array}
\\
\NextConf_{1+\Intg{n};\Intg{s}}[F']
&\equiv
\bs x. x(\bs r t_0 t_1\ldots t_{\Intg{n}}
                   t'_1\ldots t'_{\Intg{s}}.
         t'_{\Intg{s}}
	 (\ldots(t'_1(t_{\Intg{n}}(
	          \ldots
	          (t_1(t_0\,H))
	          \ldots
	         )))
         \ldots)
        )
\\
&\text{ where }
H=\begin{array}[t]{l}
  \bs d_0 f_0 n_0 n^t_0.
  \\
  \bs d_1 f_1 n_1 n^t_1.
  \ldots
  \bs d_{\Intg{n}} f_{\Intg{n}} n_{\Intg{n}} n^t_{\Intg{n}}.
  \\
  \bs e_1 g_1 e_1 s^t_1.
  \ldots
  \bs e_{\Intg{s}} g_{\Intg{s}} s_{\Intg{s}} s^t_{\Intg{s}}.
  \\
  \bs x. x\,
         (F'\,n_0\,n_1\ldots n_{\Intg{n}}
	         \,s_1\ldots s_{\Intg{s}}\,r)\,
	 (n^t_0\,\Id)\,
	 (n^t_1\,\Id)\ldots (n^t_{\Intg{n}}\,\Id)
   \\\phantom{\bs x. x\,
             (F'\,n_0\,n_1\ldots n_{\Intg{n}}
	             \,s_1\ldots s_{\Intg{s}}\,r)
                \,(n^t_0\,\Id)\,}
	 (s^t_1\,\Id)\ldots (s^t_{\Intg{s}}\,\Id)
  \enspace .
  \end{array}
\end{align*}
\normalsize
$\HeadsandTails$ uses lists as they were stacks: it pops the head of the stack, keeping head and tail in a pair.

\begin{proposition}[Typing the transition function]
\label{proposition:Typing the transition function}
Let $\Intg{n}, \Intg{s}\geq 0$, and $m\geq 1$.
Rules derivable in \WALT:
\small
\[
\infer[]
{\emptyset
;\emptyset
;\emptyset
\vdash
\ta{\BaseTransFunc^m}{\alpha\li\ST[\alpha,\delta;\$^m\BIntT]}}
{}
\]
\[
\infer[]
{\emptyset
;\emptyset
;\emptyset
\vdash
\ta{\StepTransFunc^m[G]}
   {(\$^m\BIntT \li\alpha\li\alpha)\li
     \$^m\BIntT\liv
     \ST[\alpha,\delta;\$^m\BIntT]\li
     \ST[\alpha,\delta;\$^m\BIntT]}
}
{\emptyset
;\emptyset
;\emptyset
\vdash
\ta{G}{\BIntT\li \BIntT}
}
\]
\[
\infer[]
{
\begin{array}[t]{l}
\emptyset;\emptyset;\emptyset\vdash
\HeadsandTails_{1+\Intg{n};\Intg{s}}[G]:
\bcConfT[1+\Intg{n};\Intg{s};m]\li
(\li^{\Intg{n}}_{i=0}
  !(\$\BIntT\li\alpha_i\li\alpha_i))
\\
\qquad\qquad\qquad\qquad\qquad
\li
(\li^{\Intg{n}+\Intg{s}}_{j=\Intg{n}+1}
  !(\$^{m}\BIntT\li\alpha_j\li\alpha_j))
\\\qquad\qquad\qquad\qquad\qquad\quad
\li
\$(
(\li^{\Intg{n}}_{i=0}\alpha_i)\li
(\li^{\Intg{n}+\Intg{s}}_{j=\Intg{n}+1}\alpha_j)
\li\SbcConfT[\alpha_0\ldots\alpha_{\Intg{n}+\Intg{s}},\delta;m])
\end{array}
}
{
\emptyset
;\emptyset
;\emptyset
\vdash
\ta{G}{W\li W}
}
\]
\[
\infer[]
{\begin{array}{l}
 \emptyset;\emptyset;\emptyset\vdash
 \NextConf_{1+\Intg{n};\Intg{s}}[F']:
 \SbcConfT[\alpha_0\ldots\alpha_{\Intg{n}+\Intg{s}},\delta;m]\li
 \\
 \qquad\qquad\qquad\qquad\qquad\qquad
  (\$^{m}\BIntT\liv
   (\li^{\Intg{n}}_{i=0}\alpha_i)\li
   (\li^{\Intg{n}+\Intg{s}}_{j=\Intg{n}+1}\alpha_j)\li
   \gamma
  )\li
  \gamma
 \end{array}
}
{\emptyset;\emptyset;\emptyset;\vdash
 \ta{F'}
    {
     \$\BIntT\liv
     (\liv^{\Intg{n}}_{i=1}\$\BIntT)\liv
     (\liv^{\Intg{s}}_{i=1}\$^{m}\BIntT)\liv
     \$^{m}\BIntT\liv
     \$^{m}\BIntT
    }
}
\]
\[\infer[]
{\emptyset;\emptyset;\emptyset;\vdash
 \ta{\TransFunc_{1+\Intg{n};\Intg{s}}[F,F']}
    {\bcConfT[1+\Intg{n};\Intg{s};m]\li
     \bcConfT[1+\Intg{n};\Intg{s};m]}
}
{\begin{array}{l}
 \emptyset;\emptyset;\emptyset
 \vdash
 \ta{F}{\BIntT\li \BIntT}
 \\
 \emptyset;\emptyset;\emptyset
 \vdash
 \ta{F'}
    {\$\BIntT\liv
     (\liv^{\Intg{n}}_{i=1}\$\BIntT)\liv
     (\liv^{\Intg{s}}_{i=1}\$^{m}\BIntT)\liv
     \$^{m}\BIntT\liv
     \$^{m}\BIntT}
 \end{array}
}
\]
\normalsize
\end{proposition}

\begin{proposition}[Dynamics of the transition function]
\label{proposition:Dynamics of the transition function}
Let $\Intg{n}, \Intg{s} \geq 0$, and $m\geq 1$.
For every
$\BNum{r},$ \\
 $[\BNum{a_{1}},\ldots,\BNum{a_{\Intg{r}}}]
 ,[\BNum{n_{11}},\ldots,\BNum{n_{1\Intg{r}}}]
 ,\ldots
 ,[\BNum{n_{\Intg{n}1}},\ldots,\BNum{n_{\Intg{n}\Intg{r}}}]
 ,[\BNum{s_{11}},\ldots,\BNum{s_{1\Intg{r}}}],
\ldots,[\BNum{s_{\Intg{s}1}},
\ldots,\BNum{s_{\Intg{s}\Intg{r}}}]
$, if $F\,\BNum{a_i}\red^+\BNum{a'_i}$, when
$i\in\{2,\ldots,\Intg{r}\}$, for some $a'_i$, we have:
\small
\begin{align}
\label{proposition:Dynamics of the transition function-C2PC}
&
\HeadsandTails_{1+\Intg{n};\Intg{s}}[F]\,
\lan\!\lan
 \BNum{r},
 [\BNum{a_{1}},\ldots,\BNum{a_{\Intg{r}}}]
 \!\!\!
 \begin{array}[t]{l}
  ,[\BNum{n_{11}},\ldots,\BNum{n_{1\Intg{r}}}]
  ,\ldots
  ,[\BNum{n_{\Intg{n}1}},\ldots,\BNum{n_{\Intg{n}\Intg{r}}}]\\
  ,[\BNum{s_{11}},\ldots,\BNum{s_{1\Intg{r}}}]
  ,\ldots
  ,[\BNum{s_{\Intg{s}1}},\ldots,\BNum{s_{\Intg{s}\Intg{r}}}]
  \ran\!\ran
  \red^+
 \end{array}
\\&
\nonumber
\bs d_0 d_1\ldots d_{\Intg{n}}e_1\ldots e_{\Intg{s}}.
\bs w_0 w_1\ldots w_{\Intg{n}}z_1\ldots z_{\Intg{s}}.
\\&
\nonumber
(\!\!
\begin{array}[t]{l}
\lan\!\lan
\BNum{r},
\lan\BNum{a_{1}},[\BNum{a'_{2}},\ldots,\BNum{a'_{\Intg{r}}}]\ran\\
\phantom{\lan\!\BNum{r},}
    ,\lan\BNum{n_{11}},[\BNum{n_{12}},\ldots,\BNum{n_{1\Intg{r}}}]\ran
    ,\ldots
    ,\lan\BNum{n_{\Intg{n}1}},
       [\BNum{n_{\Intg{n}2}},
           \ldots,\BNum{n_{\Intg{n}\Intg{r}}}]\ran\\
    \phantom{\lan\!\BNum{r},}
    ,\lan\BNum{s_{11}},
       [\BNum{s_{12}},\ldots,\BNum{s_{1\Intg{r}}}]\ran
    ,\ldots
    ,\lan\BNum{s_{\Intg{s}1}},
        [\BNum{s_{\Intg{s}2}},
	   \ldots,\BNum{s_{\Intg{s}\Intg{r}}}]\ran
    \ran\!\ran_{\scriptsize
	         \begin{array}[t]{l}
	          c_{1} \ldots c_{\Intg{r}}\\
	          d_{11}\ldots d_{1\Intg{r}}
		    \cdots\cdots
	          d_{\Intg{n}1}\ldots d_{\Intg{n}\Intg{r}}\\
	          e_{11}\ldots e_{1\Intg{r}}
		    \cdots\cdots
	          e_{\Intg{s}1}\ldots e_{\Intg{s}\Intg{r}}\\
	          w_{0} w_{1} \ldots w_{\Intg{n}}
		        z_{1} \ldots z_{\Intg{s}}\\
	         \end{array}
		\normalsize
	       }
\end{array}
\\
&
\nonumber
  )\{
     \begin{array}[t]{l}
     ^{d_0}\!/_{c_1}\ldots^{d_0}\!/_{c_{\Intg{r}}}
     \\
     ^{d_{1}}\!/_{d_{11}}
            \ldots
     ^{d_{1}}\!/_{d_{1\Intg{r}}}
     \cdots
     ^{d_{\Intg{n}}}\!/_{d_{\Intg{n}1}}
            \ldots
     ^{d_{\Intg{n}}}\!/_{d_{\Intg{n}\Intg{r}}}
     \
     ^{e_{1}}\!/_{e_{11}}
            \ldots
     ^{e_{1}}\!/_{e_{1\Intg{r}}}
     \cdots
     ^{e_{\Intg{s}}}\!/_{e_{\Intg{s}1}}
            \ldots
     ^{e_{\Intg{s}}}\!/_{e_{\Intg{s}\Intg{r}}}
     \\
     ^{w_{0}}/_{w_{0}}\ldots^{w_{\Intg{n}}}/_{w_{\Intg{n}}}
     \
     ^{z_{1}}/_{z_{1}}\ldots^{z_{\Intg{s}}}/_{z_{\Intg{s}}}
     \}
     \end{array}
\end{align}
\normalsize
Moreover, if we assume
$F'$ be such that
$F'\,\BNum{a}_{1}
   \,\BNum{n_{11}}\ldots\BNum{n_{\Intg{n}1}}
   \,\BNum{s_{11}}\ldots\BNum{s_{\Intg{s}1}}
   \,\BNum{r} \red^+\BNum{r'}$, for some $r'$,
then, we also have the two following reduction sequences:
\small
\begin{align}
\label{proposition:Dynamics of the transition function-PC2C}
&
\bs d_0 d_1\ldots d_{\Intg{n}}e_1\ldots e_{\Intg{s}}.
\bs w_0 w_1\ldots w_{\Intg{n}}z_1\ldots z_{\Intg{s}}.
\\
\nonumber
&
(\NextConf_{1+\Intg{n};\Intg{s}}[F']\,
    \lan\!\lan
    \BNum{r},
    \lan\BNum{a_{1}},[\BNum{a_{2}},\ldots,\BNum{a_{\Intg{r}}}]\ran
\\
\nonumber
&\phantom{\NextConf_{1+\Intg{n};\Intg{s}}[F']\,\lan\!\lan\BNum{r},}
    ,\lan\BNum{n_{11}},[\BNum{n_{12}},\ldots,\BNum{n_{1\Intg{r}}}]\ran
    ,\ldots
    ,\lan\BNum{n_{\Intg{n}1}},
         [\BNum{n_{\Intg{n}2}},
	     \ldots,\BNum{n_{\Intg{n}\Intg{r}}}]\ran
\\
\nonumber
&\phantom{\NextConf_{1+\Intg{n};\Intg{s}}[F']\,\lan\!\lan\BNum{r},}
    ,\lan\BNum{s_{11}},
         [\BNum{s_{12}},\ldots,\BNum{s_{1\Intg{r}}}]
     \ran
    ,\ldots
    ,\lan\BNum{s_{\Intg{s}1}},
         [\BNum{s_{\Intg{s}2}},\ldots,\BNum{s_{\Intg{s}\Intg{r}}}]
     \ran
    \ran\!\ran_{\scriptsize
	         \begin{array}[t]{l}
	          c_{1} \ldots c_{\Intg{r}}\\
	          d_{11}\ldots d_{1\Intg{r}}
		    \cdots\cdots
	          d_{\Intg{n}1}\ldots d_{\Intg{n}\Intg{r}}\\
	          e_{11}\ldots e_{1\Intg{r}}
		    \cdots\cdots
	          e_{\Intg{s}1}\ldots e_{\Intg{s}\Intg{r}}\\
	          w_{0} w_{1} \ldots w_{\Intg{n}}
		        z_{1} \ldots z_{\Intg{s}}
                 \end{array}
		\normalsize}
\\
\nonumber
&
  )\{
     \begin{array}[t]{l}
     ^{d_0}\!/_{c_1}\ldots^{d_0}\!/_{c_{\Intg{r}}}
     \\
     ^{d_{1}}\!/_{d_{11}}
            \ldots
     ^{d_{1}}\!/_{d_{1\Intg{r}}}
     \cdots
     ^{d_{\Intg{n}}}\!/_{d_{\Intg{n}1}}
            \ldots
     ^{d_{\Intg{n}}}\!/_{d_{\Intg{n}\Intg{r}}}
     \
     ^{e_{1}}\!/_{e_{11}}
            \ldots
     ^{e_{1}}\!/_{e_{1\Intg{r}}}
     \cdots
     ^{e_{\Intg{s}}}\!/_{e_{\Intg{s}1}}
            \ldots
     ^{e_{\Intg{s}}}\!/_{e_{\Intg{s}\Intg{r}}}
     \\
     ^{w_{0}}/_{w_{0}}\ldots^{w_{\Intg{n}}}/_{w_{\Intg{n}}}
     \
     ^{z_{1}}/_{z_{1}}\ldots^{z_{\Intg{s}}}/_{z_{\Intg{s}}}
     \}
     \qquad\qquad\qquad\qquad
     \qquad\qquad\qquad\qquad
     \red^+
     \end{array}
\\
&
\nonumber
   \begin{array}[t]{l}
    \lan\!\lan
    \BNum{r'},
    [\BNum{a_{2}},\ldots,\BNum{a_{\Intg{r}}}]
    ,[\BNum{n_{12}},\ldots,\BNum{n_{1\Intg{r}}}]
    ,\ldots
    [\BNum{n_{\Intg{n}2}},\ldots,\BNum{n_{\Intg{n}\Intg{r}}}]
    ,[\BNum{s_{12}},\ldots,\BNum{s_{1\Intg{r}}}]
    ,\ldots
    ,[\BNum{s_{\Intg{s}2}},\ldots,\BNum{s_{\Intg{s}\Intg{r}}}]
    \ran\!\ran
   \end{array}
\\
\label{proposition:Dynamics of the transition function-C2C}
&
\TransFunc_{1+\Intg{n};\Intg{s}}[F,F']\,
\lan\!\lan
 \BNum{r},
 [\BNum{a_{1}},\ldots,\BNum{a_{\Intg{r}}}]
 \!\!\!
 \begin{array}[t]{l}
  ,[\BNum{n_{11}},\ldots,\BNum{n_{1\Intg{r}}}]
  ,\ldots
  ,[\BNum{n_{\Intg{n}1}},\ldots,\BNum{n_{\Intg{n}\Intg{r}}}]\\
  ,[\BNum{s_{11}},\ldots,\BNum{s_{1\Intg{r}}}]
  ,\ldots
  ,[\BNum{s_{\Intg{s}1}},\ldots,\BNum{s_{\Intg{s}\Intg{r}}}]
  \ran\!\ran
  \red^+
 \end{array}
\\
&
\nonumber
    \lan\!\lan
    \BNum{r'}
    ,[\BNum{a'_{2}},\ldots,\BNum{a'_{\Intg{r}}}]
    ,[\BNum{n_{12}},\ldots,\BNum{n_{1\Intg{r}}}]
    ,\ldots
    ,[\BNum{n_{\Intg{n}2}},\ldots,\BNum{n_{\Intg{n}\Intg{r}}}]
    ,[\BNum{s_{12}},\ldots,\BNum{s_{1\Intg{r}}}]
    ,\ldots
    ,[\BNum{s_{\Intg{s}2}},\ldots,\BNum{s_{\Intg{s}\Intg{r}}}]
    \ran\!\ran
\end{align}
\normalsize
\end{proposition}
\textbf{Iterator.}
$\Iter{1+\Intg{n}}{\Intg{s}}{F_0}{F_1}{G}$ realizes a virtual machine that iterates two instances of the transition function starting from an initial configuration. One instance of the transition function depends on the term $F_0$, the other on $F_1$. The initial configuration is built using the term $G$. The choice about which transition function using depends on a copy of the first argument of the iterator, which is a word. A second copy is used by the iterator, through $\BInttoBCconf_{1+\Intg{n};\Intg{s}}$, to generate the initial configuration.
In particular, $\BInttoBCconf_{1+\Intg{n};\Intg{s}}$ exploits the term $\UInttoList$ that \emph{requires} two assumptions: one of them is an elementary partially discharged one, namely it must be $\$$-modal, and will correspond to one of the constant arguments of the safe recursion scheme we shall simulate.
Once the iteration stops, $\BCconftoBInt_{1+\Intg{n};\Intg{s}}$ reads the word, representing the result of the iteration, out of a final configuration.
\small
\begin{align*}
\Iter{1+\Intg{n}}{\Intg{s}}{F_0}{F_1}{G}
\equiv&
 \bs n.\bs n_1\ldots n_{\Intg{n}}.
\\
&
 \EEmbed{1}{0}{1+\Intg{n}+\Intg{s}}{H}
               \,(\LEmbed{1}{1}{\nabla_2}\, n)
               \,(\LEmbed{1}{1}{\Coerc^4}\, n_1)
		 \ldots
		 (\LEmbed{1}{1}{\Coerc^4}\, n_{\Intg{n}})
\\
&
H \text{ being }
\begin{array}[t]{l}
   \bs t
       n_1\ldots n_{\Intg{n}}
       s_1\ldots s_{\Intg{s}}.
   \\
   \phantom{\bs}
   t\, (\bs ab.
        (\bs zy.
	 \BCconftoBInt_{1+\Intg{n};\Intg{s}}\,
         (\BCconftoFConf_{1+\Intg{n};\Intg{s}}
	   \,(z\,(\TransFunc_{1+\Intg{n};\Intg{s}}[\Id,G]\,y)
	     ))
   \\\phantom{\bs t\,(\bs ab.}
	)(a\,\TransFunc_{1+\Intg{n};\Intg{s}}[\BSuccZ,F_0]
	   \,\TransFunc_{1+\Intg{n};\Intg{s}}[\BSuccO,F_1])
   \\\phantom{\bs t\,(\bs ab.)}
	 (\BInttoBCconf_{1+\Intg{n};\Intg{s}}
	               \,n_1\ldots n_{\Intg{n}}
	               \,s_1\ldots s_{\Intg{s}}
		       \,b)
       )
\end{array}
\\
\BInttoBCconf_{1+\Intg{n};\Intg{s}}
\equiv&
\bs n_1\ldots n_{\Intg{n}} s_1\ldots s_{\Intg{s}} w.
\\
&
(\bs t.
 t(\bs k_0\,k_1\ldots k_{\Intg{n}}
          \,h_1\ldots h_{\Intg{s}}.
\\
&\phantom{(\bs t.t(}
 \ListstoConf_{1+\Intg{n};\Intg{s}}
      \, (\UInttoList\, \BNum{0}\, (\USucc\, (\BInttoUInt\, k_0)))
\\
&\phantom{(\bs t.t(\ListstoConf_{1+\Intg{n};\Intg{s}}\,}
    (\UInttoList\, n_{1}\, (\USucc\, (\BInttoUInt\, k_1)))
 	  \ldots
	  (\UInttoList\, n_{\Intg{n}}\, (\USucc\, (\BInttoUInt\, k_{\Intg{n}})))
\\
&\phantom{(\bs t.t(\ListstoConf_{1+\Intg{n};\Intg{s}}\,}
  (\UInttoList\, s_{1}\, (\USucc\, (\BInttoUInt\, h_1)))
  \ldots
  (\UInttoList\, s_{\Intg{s}}\, (\USucc\, (\BInttoUInt\, h_{\Intg{s}})))
                  )
\\
&
)(\nabla_{1+\Intg{n}+\Intg{s}}\, w)
\\
\ListstoConf_{1+\Intg{n};\Intg{s}}
\equiv&
\bs \!
     \begin{array}[t]{l}
	l_0
        l_1           \ldots l_{\Intg{n}}
        l_{\Intg{n}+1}\ldots l_{\Intg{n}+\Intg{s}}.
	\bs d_0
	d_1 \ldots d_{\Intg{n}}
	e_1 \ldots e_{\Intg{s}}.\\
	(\bs b_0
	     b_1\ldots b_{\Intg{n}}
	     c_1\ldots c_{\Intg{s}}.
	  \bs w_0
	      w_1\ldots w_{\Intg{n}}
	      z_1\ldots z_{\Intg{s}}.\\
	  \phantom{(}
	   \bs x.
	   x\,
	   \BNum{0}\,
	   (b_0\, w_0)
	   (b_1\, w_1)
	   \ldots
	   (b_{\Intg{n}}\, w_{\Intg{n}})\,
	   (c_1\, z_{1})
	   \ldots
	   (c_{\Intg{s}}\, z_{\Intg{s}})
        \\
	)\,(l_0\,d_0)
	   (l_1\, d_1) \ldots (l_{\Intg{n}}\, d_{\Intg{n}})
	   (l_{\Intg{n}+1}\, e_{1})
		       \ldots (l_{\Intg{n}+\Intg{s}}\,e_{\Intg{s}})
	\end{array}
\\
\BCconftoFConf_{1+\Intg{n};\Intg{s}}
\equiv&
\bs c.\bs d_0\ldots d_{\Intg{n}}e_1\ldots e_{\Intg{s}}.
\\
&(\bs b w_0\ldots w_{\Intg{n}}z_1\ldots z_{\Intg{s}}.
  b\,w_0\ldots w_{\Intg{n}}\,z_1\ldots z_{\Intg{s}})
  (c\,d_0\ldots d_{\Intg{n}}\,e_1\ldots e_{\Intg{s}})
\\
\BCconftoBInt_{1+\Intg{n};\Intg{s}}
\equiv&
\bs c.
(\bs b. b\,\underbrace{\BNum{0}
                       \cdots
		       \BNum{0}}_{1+\Intg{n}+\Intg{s}}\,
           (\bs r x_0\ldots x_{\Intg{n}+\Intg{s}}. r)
)(c\,\underbrace{\bs xy.x
                 \cdots
		 \bs xy.x}_{1+\Intg{n}+\Intg{s}})
\end{align*}
\normalsize

\begin{proposition}[Typing the iterator]
\label{proposition:Typing the iterator}
Let $\Intg{n}, \Intg{s} \geq 0$, and $m\geq 1$.
Rules derivable in \WALT:
\small
\[
\infer[]
{\emptyset;\emptyset;\emptyset\vdash
 \ta{\ListstoConf_{\Intg{n};\Intg{s}}}{
 (\li^{\Intg{n}}_{i=0} \ListT \$\BIntT)\li
 (\li^{\Intg{s}}_{j=1} \ListT \$^{m} \BIntT) \li
 \bcConfT[1+\Intg{n};\Intg{s};m]}
}
{}
\]
\[
\infer[]
{\emptyset;\emptyset;\emptyset\vdash
 \ta{\BInttoBCconf_{\Intg{n};\Intg{s}}}
    {(\liv^{\Intg{n}}_{i=1}\$^{3}\BIntT)
     \liv
     (\liv^{\Intg{s}}_{j=1}\$^{m+2}\BIntT)\liv
     \BIntT\li
     \$\bcConfT[1+\Intg{n};\Intg{s};m]}
}
{}
\]
\[
\infer[]
{\emptyset;\emptyset;\emptyset\vdash
 \ta{\BCconftoFConf_{1+\Intg{n};\Intg{s}}}
    {\bcConfT[1+\Intg{n};\Intg{s};m]\li
     \FbcConfT[1+\Intg{n};\Intg{s};m]}
}
{}
\]
\[
\infer[]
{\emptyset;\emptyset;\emptyset\vdash
 \ta{\BCconftoBInt_{1+\Intg{n};\Intg{s}}}
    {\FbcConfT[1+\Intg{n};\Intg{s};m]\li\$^{m+1}\BIntT}
}
{}
\]
\[
\infer[]
{
\emptyset;\emptyset;\emptyset
 \vdash
 \ta{\Iter{1+\Intg{n}}{\Intg{s}}{G_0}{G_1}{G_2}}
    {\$\BIntT\liv
     (\liv^{\Intg{n}}_{i=1}\$\BIntT)\liv
     (\liv^{\Intg{s}}_{i=1}\$^{m+4}\BIntT)\liv
     \$^{m+4}\BIntT
    }
}
{
\emptyset;\emptyset;\emptyset
 \vdash
 \ta{G_k}
    {\$\BIntT\liv
     (\liv^{\Intg{n}}_{i=1}\$\BIntT)\liv
     (\liv^{\Intg{s}}_{j=1}\$^{m}\BIntT)\liv
     \$^{m}\BIntT\liv
     \$^{m}\BIntT
    }
 &
 k\in\{0,1,2\}
}
\]
\normalsize
\end{proposition}

\begin{proposition}[Dynamics of the combinators for the iterator.]
\label{proposition:Dynamics of the combinators for the iterator}
Let $\Intg{n}, \Intg{s}, \Intg{r}$, and $m\geq 0$.
\par
For every
$
\BNum{r},
[\BNum{a_{1}},\ldots,\BNum{a_{\Intg{r}}}]
,[\BNum{n_{11}},\ldots,\BNum{n_{1\Intg{r}}}]
,\ldots
,[\BNum{n_{\Intg{n}1}},\ldots,\BNum{n_{\Intg{n}\Intg{r}}}]
,[\BNum{s_{11}},\ldots,\BNum{s_{1\Intg{r}}}],
\ldots,[\BNum{s_{\Intg{s}1}},\ldots,\BNum{s_{\Intg{s}\Intg{r}}}]
$
we have:
\small
\begin{align}
\label{proposition:Dynamics of the transition function-L2C}
&
\ListstoConf_{1+\Intg{n};\Intg{s}}\,
 [\BNum{a_{1}},\ldots,\BNum{a_{\Intg{r}}}]\,
 [\BNum{n_{11}},\ldots,\BNum{n_{1\Intg{r}}}]
 \ldots
 [\BNum{n_{\Intg{n}1}},\ldots,\BNum{n_{\Intg{n}\Intg{r}}}]\,
\\
\nonumber
&\phantom{\ListstoConf_{1+\Intg{n};\Intg{s}}\,[\BNum{a_{1}},\ldots,\BNum{a_{\Intg{r}}}]\,}
 [\BNum{s_{11}},\ldots,\BNum{s_{1\Intg{r}}}]
 \ldots
 [\BNum{s_{\Intg{s}1}},\ldots,\BNum{s_{\Intg{s}\Intg{r}}}]
\red^+
\\
&
\nonumber
\bcConf{
 \BNum{0}
,[\BNum{a_{1}},\ldots,\BNum{a_{\Intg{r}}}]
,[\BNum{n_{11}},\ldots,\BNum{n_{1\Intg{r}}}]
,\ldots
,[\BNum{n_{\Intg{n}1}},\ldots,\BNum{n_{\Intg{n}\Intg{r}}}]
,[\BNum{s_{11}},\ldots,\BNum{s_{1\Intg{r}}}]
,\ldots
,[\BNum{s_{\Intg{s}1}},\ldots,\BNum{s_{\Intg{s}\Intg{r}}}]
}
\end{align}
\normalsize
For every
$\BNum{n_{1}},\ldots,\BNum{n_{\Intg{n}}}
 ,\ldots
 ,\BNum{s_{1}},\ldots,\BNum{s_{\Intg{s}}}$
we have:
\small
\begin{align}
\label{proposition:Dynamics of the transition function-W2C}
&\BInttoBCconf_{1+\Intg{n};\Intg{s}}\,
 \BNum{n_{1}} \ldots \BNum{n_{\Intg{n}}}\,
 \BNum{s_{1}} \ldots \BNum{s_{\Intg{s}}}\,
 \BNum{n}
\red^+
\\
&
\nonumber
\bcConf{
 \BNum{0}
,[\BNum{0},\ldots,\BNum{0}]
,[\BNum{n_{1}},\ldots,\BNum{n_{1}}]
,\ldots
,[\BNum{n_{\Intg{n}}},\ldots,\BNum{n_{\Intg{n}}}]
,[\BNum{s_{1}},\ldots,\BNum{s_{1}}]
,\ldots
,[\BNum{s_{\Intg{s}}},\ldots,\BNum{s_{\Intg{s}}}]
}
\end{align}
\normalsize
where every list of the result has $m+2$ elements whenever $\BNum{n}$ can be written as $\BNum{2^m\nu_m+\ldots+2^0\nu_0}$.
\par
For every
$\BNum{r},
[\BNum{a_{1}},\ldots,\BNum{a_{\Intg{r}}}]
,[\BNum{n_{11}},\ldots,\BNum{n_{1\Intg{r}}}]
,\ldots
,[\BNum{n_{\Intg{n}1}},\ldots,\BNum{n_{\Intg{n}\Intg{r}}}]
,[\BNum{s_{11}},\ldots,\BNum{s_{1\Intg{r}}}],
\ldots,[\BNum{s_{\Intg{s}1}},\ldots,\BNum{s_{\Intg{s}\Intg{r}}}]$
we have:
\small
\begin{align}
\label{proposition:Dynamics of the transition function-C2FC}
&
\BCconftoFConf_{1+\Intg{n};\Intg{s}}\,
\lan\!\lan
 \BNum{r},
 [\BNum{a_{1}},\ldots,\BNum{a_{\Intg{r}}}]
  ,[\BNum{n_{11}},\ldots,\BNum{n_{1\Intg{r}}}]
  ,\ldots
  ,[\BNum{n_{\Intg{n}1}},\ldots,\BNum{n_{\Intg{n}\Intg{r}}}]
\\
\nonumber
&\phantom{\BCconftoFConf_{1+\Intg{n};\Intg{s}}\,\lan\!\lan\BNum{r},
 [\BNum{a_{1}},\ldots,\BNum{a_{\Intg{r}}}]}
  ,[\BNum{s_{11}},\ldots,\BNum{s_{1\Intg{r}}}]
  ,\ldots
  ,[\BNum{s_{\Intg{s}1}},\ldots,\BNum{s_{\Intg{s}\Intg{r}}}]
\ran\!\ran
\\
&
\nonumber
\red^+
\lan\!\lan
 \BNum{r},
 [\BNum{a_{1}},\ldots,\BNum{a_{\Intg{r}}}]
  ,[\BNum{n_{11}},\ldots,\BNum{n_{1\Intg{r}}}]
  ,\ldots
  ,[\BNum{n_{\Intg{n}1}},\ldots,\BNum{n_{\Intg{n}\Intg{r}}}]
  ,[\BNum{s_{11}},\ldots,\BNum{s_{1\Intg{r}}}]
  ,\ldots
  ,[\BNum{s_{\Intg{s}1}},\ldots,\BNum{s_{\Intg{s}\Intg{r}}}]
  \ran\!\ran
\end{align}
\normalsize
\par
For every
$\BNum{r},
[\BNum{a_{1}},\ldots,\BNum{a_{\Intg{r}}}]
,[\BNum{n_{11}},\ldots,\BNum{n_{1\Intg{r}}}]
,\ldots
,[\BNum{n_{\Intg{n}1}},\ldots,\BNum{n_{\Intg{n}\Intg{r}}}]
,[\BNum{s_{11}},\ldots,\BNum{s_{1\Intg{r}}}], \ldots,[\BNum{s_{\Intg{s}1}},\ldots,\BNum{s_{\Intg{s}\Intg{r}}}]$
we have:
\small
\begin{align}
\label{proposition:Dynamics of the transition function-C2W}
&
\BCconftoBInt_{1+\Intg{n};\Intg{s}}\,
\lan\!\lan
 \BNum{r},
 [\BNum{a_{1}},\ldots,\BNum{a_{\Intg{r}}}]
  ,[\BNum{n_{11}},\ldots,\BNum{n_{1\Intg{r}}}]
  ,\ldots
  ,[\BNum{n_{\Intg{n}1}},\ldots,\BNum{n_{\Intg{n}\Intg{r}}}],
\\
\nonumber
&
\phantom{\BCconftoBInt_{1+\Intg{n};\Intg{s}}\,\lan\!\lan\BNum{r},[\BNum{a_{1}},\ldots,\BNum{a_{\Intg{r}}}]\ \ \, }
   [\BNum{s_{11}},\ldots,\BNum{s_{1\Intg{r}}}]
  ,\ldots
  ,[\BNum{s_{\Intg{s}1}},\ldots,\BNum{s_{\Intg{s}\Intg{r}}}]
  \ran\!\ran
  \red^+\BNum{r}
\end{align}
\normalsize
\end{proposition}

\begin{proposition}[Dynamics of the iterator.]
\label{proposition:Dynamics of the iterator}
Let
$\Intg{n}, \Intg{s}\geq 0$,
$\BNum{a},
 \BNum{n},
 \BNum{n_{1}},\ldots,\BNum{n_{\Intg{n}}},
 \BNum{s_{1}},\ldots,\BNum{s_{\Intg{s}}}$ be some words,
$G_0, G_1, G_2$ be three closed typable terms,
$\{\nu_0, \nu_1, \ldots\}$ be a denumerable set of metavariables to range over $\{0, 1\}$,
$[\BNum{x}]^i$ be a notation for a list with $l$ copies of the word $\BNum{x}$, for any $x$ and $i$.
\par
Let
$G_{2}\,
 \BNum{0}\,
 \BNum{n_{1}}\,\ldots\,\BNum{n_{\Intg{n}}}\,
 \BNum{s_{1}}\,\ldots\,\BNum{s_{\Intg{s}}}\,
 \BNum{0}$ rewrite to a word $\BNum{a}$, and
let
$G_{1}\,
 \BNum{0}\,
 \BNum{n_{1}}\,\ldots\,\BNum{n_{\Intg{n}}}\,
 \BNum{s_{1}}\,\ldots\,\BNum{s_{\Intg{s}}}\,
 \BNum{a}$ rewrite to a word
\\
$r[0,a,n_1,\ldots,n_{\Intg{n}},s_1,\ldots,s_{\Intg{s}}]$, and,
for every $m, i$, such that $m\geq0, m-1\geq i\geq 0$:
\small
$$G_{\nu_i}\,
 \BNum{\left(\sum_{j=0}^{m-(i+1)} 2^{m-(i+1)-j} \nu_{m-j}\right)}\,
 \BNum{n_{1}}\,\ldots\,\BNum{n_{\Intg{n}}}\,
 \BNum{s_{1}}\,\ldots\,\BNum{s_{\Intg{s}}}\,
 r[m-(i+1),a,n_1,\ldots,n_{\Intg{n}},s_1,\ldots,s_{\Intg{s}}]
$$
\normalsize
rewrite to a word
$r[m-i,a,n_1,\ldots,n_{\Intg{n}},s_1,\ldots,s_{\Intg{s}}]$.
Then:
\begin{enumerate}
\item
\label{proposition:Dynamics of the iterator-part-a}
For every $m, k, i$, such that $m\geq0, k\geq m+1, m\geq i\geq 0$,
the following iterated application of the transition function:
\small
\begin{align}
\nonumber
&
\TransFunc_{1+\Intg{n};\Intg{s}}[\BSucc{\nu_i},G_{\nu_i}]
\\
\nonumber
&
\qquad
 (\TransFunc_{1+\Intg{n};\Intg{s}}[\BSucc{\nu_{i+1}},G_{\nu_{i+1}}]
(
\ldots
\\
\nonumber
&
\qquad
\qquad
(
\TransFunc_{1+\Intg{n};\Intg{s}}[\BSucc{\nu_m},G_{\nu_m}]
\lan\!\lan
 \BNum{a},
 [\BNum{0}]^k
  ,[\BNum{n_{1}}]^k
  ,\ldots
  ,[\BNum{n_{\Intg{n}}}]^k
  ,[\BNum{s_{1}}]^k
  ,\ldots
  ,[\BNum{s_{\Intg{s}}}]^k
\ran\!\ran
)
\ldots
))
\end{align}
\normalsize
rewrites to the configuration:
\small
$$\lan\!\lan
r[m-i,a,n_1,\ldots,n_{\Intg{n}},s_1,\ldots,s_{\Intg{s}}],
\left[ \,\BNum{\sum_{j=0}^{m-i}2^{m-i-j} \nu_{m-j}}\,\right]^{k-(m-i)-1},
\!\!
\begin{array}[t]{l}
 [\BNum{n_{1}}]^{k-(m-i)-1}
,\ldots
,[\BNum{n_{\Intg{n}}}]^{k-(m-i)-1}
,
\\{[\BNum{s_{1}}]^{k-(m-i)-1}}
,\ldots
,[\BNum{s_{\Intg{s}}}]^{k-(m-i)-1}
\ran\!\ran
\enspace ,
\end{array}
$$
\normalsize
where $\BSucc{\nu_i}$ is $\BSuccZ$ when $\nu_i\equiv 0$,
and $\BSucc{\nu_i}$ is $\BSuccO$ when $\nu_i\equiv 1$.
\item
\label{proposition:Dynamics of the iterator-part-b}
The iterator behaves as follows:
\small
\begin{align}
\nonumber
\Iter{1+\Intg{n}}{\Intg{s}}{G_0}{G_1}{G_2}\,
\BNum{0}\,
\BNum{n_{1}}\,\ldots\,\BNum{n_{\Intg{n}}}\,
\BNum{s_{1}}\,\ldots\,\BNum{s_{\Intg{s}}}
&\red^+ \BNum{a}
\\
\nonumber
\Iter{1+\Intg{n}}{\Intg{s}}{G_0}{G_1}{G_2}\,
\BNum{\left(\sum_{j=0}^{m} 2^{j} \nu_j \right)}\,
\BNum{n_{1}}\,\ldots\,\BNum{n_{\Intg{n}}}\,
\BNum{s_{1}}\,\ldots\,\BNum{s_{\Intg{s}}}
&\red^+
r[m,a,n_1,\ldots,n_{\Intg{n}},s_1,\ldots,s_{\Intg{s}}]
&\left(\BNum{\sum_{j=0}^{m} 2^{j} \nu_j}\neq 0\right)
\end{align}
\normalsize
\end{enumerate}
\end{proposition}
Point~\ref{proposition:Dynamics of the iterator-part-a} of the proposition here above holds by induction on $m-i$. Point~\ref{proposition:Dynamics of the iterator-part-b} holds proceeding by cases on the first argument of the iterator, applying its definition. If it is $\BNum{0}$, then use the assumption on the behavior of $G_2$. Otherwise, it is enough to use the definition of the iterator, and Point~\ref{proposition:Dynamics of the iterator-part-a} just proved.
\subsection{Composition}
\label{subsubsection:Composition}
We shall define combinators that compose terms of \WALT, the goal being the simulation of the composition of \QlSRN.
Intuitively, the composition
$\Wcomp{\Intg{n}}
       {\sum_{i=1}^{\Intg{s}'} \Intg{s}_i}
       {\Intg{n}'}
       {\Intg{s}'}
       {F,G_1,\ldots,G_{\Intg{n}'},H_1,\ldots,H_{\Intg{s}'}}
$ applies the term $F$ to the results of the applications of
$G_1,\ldots,G_{\Intg{n}'},H_1,\ldots,H_{\Intg{s}'}$ to their arguments.
All the terms $H_1,\ldots,H_{\Intg{s}'}$, that we call \textit{safe}, can be thought of as functions with a coincident normal arity $\Intg{n}$, and with a safe arity $\Intg{s}_i$. Analogously, all the terms $H_1,\ldots,H_{\Intg{s}'}$, that we call \textit{normal}, can be thought of as functions only with normal arity $\Intg{n}$.
The composition is meaningful since $F$ is like a function with normal arity $\Intg{n}'$, that equals the number of \textit{normal} terms, and safe arity $\Intg{s}'$, equal to the number of \textit{safe} terms.
Here it is the definition:
\small
\begin{align*}
&
\Wcomp{\Intg{n}}
      {\sum_{i=1}^{\Intg{s}'} \Intg{s}_i}
      {\Intg{n}'}
      {\Intg{s}'}
      {F,G_1,\ldots,G_{\Intg{n}'},H_1,\ldots,H_{\Intg{s}'}}
\equiv
\\
&
\qquad
\qquad
\qquad
\bs n_1\ldots n_{\Intg{n}}.
\EEmbed{2}
       {0}{\Intg{n}+\sum_{i=1}^{\Intg{s}'} \Intg{s}_i}
       {G}
(\LEmbed{1}
        {1}
        {\nabla^1_{\Intg{n}'+\Intg{s}'}}\,n_1)
  \ldots(\LEmbed{1}
                {1}
	        {\nabla^1_{\Intg{n}'+\Intg{s}'}}\,n_{\Intg{n}})
\\
G\equiv&
   \bs \elan x_{11}\ldots x_{\Intg{n}'1}
             y_{11}\ldots y_{\Intg{s}'1}\eran
   \ldots
   \bs \elan x_{1\Intg{n}}\ldots x_{\Intg{n}'\Intg{n}}
             y_{1\Intg{n}}\ldots y_{\Intg{s}'\Intg{n}}\eran.
   \bs w_{1 1}\ldots w_{1 \Intg{s}_{1 }}
   \ldots\ldots
       w_{\Intg{s}'1}\ldots w_{\Intg{s}' \Intg{s}_{\Intg{s}'}}.
\\
&
  \, \EEmbed{m-1}{0}{\Intg{n}'+\Intg{s}'}{F}\,
   (G_1\, x_{11} \ldots x_{1\Intg{n}})
   \ldots
   (G_{\Intg{n}'}\, x_{\Intg{n}'1} \ldots x_{\Intg{n}'\Intg{n}})
\\&\,
   (\EEmbed{m-1}{0}{\Intg{n}+\Intg{s}_1}{H_{1}}\,
     (\LEmbed{1}{1}{\Coerc^{m-1}}\, y_{11})
     \ldots (\LEmbed{1}{1}{\Coerc^{m-1}}\,
                y_{1\Intg{n}})w_{11}\ldots w_{1\Intg{s}_1})
\\&\qquad
    \ldots
    (\EEmbed{m-1}{0}{\Intg{n}+\Intg{s}_{\Intg{s}'}}{H_{\Intg{s}'}}\,
     (\LEmbed{1}{1}{\Coerc^{m-1}}\, y_{\Intg{s}'1})
     \ldots (\LEmbed{1}{1}{\Coerc^{m-1}}\,
               y_{\Intg{s}'\Intg{n}})w_{\Intg{s}'1}\ldots
               w_{\Intg{s}'\Intg{s}_{\Intg{s}'}})
\end{align*}
\normalsize

\begin{proposition}[Typing the composition]
\label{proposition:Typing the composition}
Let $\Intg{n}, \Intg{s'}, \Intg{s_{1}}, \ldots, \Intg{s_{\Intg{s'}}}\geq 0$, and $m\geq 1$. A rule derivable in \WALT:
\small
\[
\infer[]
{
\emptyset;\emptyset;\emptyset
 \vdash
  \begin{array}[t]{l}
     \Wcomp{\Intg{n}}
           {\sum_{i=1}^{\Intg{s}'} \Intg{s}_i}
           {\Intg{n}'}
           {\Intg{s}'}
           {F,G_1,\ldots,G_{\Intg{n}'},H_1,\ldots,H_{\Intg{s}'}}
     \\
     \qquad\!:\!
     (\liv^{\Intg{n}}_{i=1}\$\BIntT)\liv
     (\liv^{\sum_{i=1}^{\Intg{s}'} \Intg{s}_i}_{i=1}\$^{2m+1}\BIntT)\liv
     \$^{2m+1}\BIntT
  \end{array}
}
{
\begin{array}[t]{ll}
\emptyset;\emptyset;\emptyset
 \vdash
 \ta{F}
    {(\liv^{\Intg{n}'}_{i=1}\$\BIntT)\liv
     (\liv^{\Intg{s}'}_{j=1}\$^{m}\BIntT)\liv
     \$^{m}\BIntT
    }
 \\
\emptyset;\emptyset;\emptyset
 \vdash
 \ta{G_i}
    {(\liv^{\Intg{n}}_{i=1}\$\BIntT)\liv\$^{m}\BIntT}
 &
 i\in\{1,\ldots,\Intg{n}'\}
 \\
\emptyset;\emptyset;\emptyset
 \vdash
 \ta{H_j}
    {(\liv^{\Intg{n}}_{i=1}\$\BIntT)\liv
     (\liv^{\Intg{s}_j}_{k=1}\$^{m}\BIntT)\liv
     \$^{m}\BIntT
    }
 &
 j\in\{1,\ldots,\Intg{s}'\}
\end{array}
}
\]
\normalsize
\end{proposition}

\begin{proposition}[Dynamics of the composition]
\label{proposition:Dynamics of the composition}
Let $\Intg{n}, \Intg{s'}, \Intg{s_{1}}, \ldots, \Intg{s_{\Intg{s'}}}\geq 0$, and
$\BNum{n_{1}},\ldots,\BNum{n_{\Intg{n}}},\BNum{s_{11}},\ldots,\BNum{s_{1\Intg{s}_1}}
 \ldots$
\\
$\ldots
 \BNum{s_{\Intg{s}'1}},\ldots,
 \BNum{s_{\Intg{s}'\Intg{s}_{\Intg{s}'}}},
 \BNum{g_1},\ldots,\BNum{g_{n'}},
 \BNum{h_1},\ldots,\BNum{h_{s'}},\BNum{f}
$ be some words. Let us assume:
\small
\begin{align*}
G_i\,\BNum{n_1}\ldots\BNum{n_{\Intg{n}}}
&\red^+\BNum{g_i}
&&(1\leq i\leq n')\\
H_{j}\, \BNum{n_1} \ldots \BNum{n_{\Intg{n}}}\,
         \BNum{s_{j1}} \ldots \BNum{s_{j\Intg{s}_{j}}}
&\red^+ \BNum{h_j}
&&(1\leq j\leq s')
\enspace .
\end{align*}
\normalsize
If
$F\,\BNum{g_1} \ldots \BNum{g_{\Intg{n}'}}\BNum{h_1} \ldots \BNum{h_{\Intg{s}'}}
\red^+ \BNum{f}$, then
$\Wcomp{\Intg{n}}
      {\sum_{i=1}^{\Intg{s}'} \Intg{s}_i}
      {\Intg{n}'}
      {\Intg{s}'}
      {F,G_1,\ldots,G_{\Intg{n}'},H_1,\ldots,H_{\Intg{s}'}}
\,\BNum{n_1}\ldots\BNum{n_{\Intg{n}}},
\BNum{s_{11}}\ldots\BNum{s_{1\Intg{s}_1}}
  \ldots$
\\
$\ldots
  \BNum{s_{\Intg{s}'1}}$ $\ldots\BNum{s_{\Intg{s}'\Intg{s}_{\Intg{s}'}}}
\red^+\BNum{f}$.
\end{proposition}
To prove it, we just apply the definitions.
\section{From \QlSRN\ to \WALL}
\label{section:Algorithmic expressivity of WALL}
\textbf{Functions of \PRNQ\ into \WALL.}
First, we define a map $\etp{\ }$ from the signature $\Sigma_{\PRNQ}$ to \PT:
\begin{enumerate}
\item
\label{embedPRNQWALL:zero}
$\etp{\zero{0}{0}} \equiv \LEmbed{1}{0}{\BNum{0}}$, while
$\etp{\zero{k}{l}}
 \equiv\bs n_{1}\ldots n_{k}\, s_{1}\ldots s_{l}.\etp{\zero{0}{0}}$,
for every $k, l$ such that $k+l\geq 1$.

\item
\label{embedPRNQWALL:sucz}
$\etp{\sucz}\equiv\BEmbed{1}{\BSuccZ}$.

\item
\label{embedPRNQWALL:suco}
$\etp{\suco}\equiv\BEmbed{1}{\BSuccO}$.

\item
\label{embedPRNQWALL:pred}
$\etp{\pred}\equiv\BEmbed{1}{\Pred}$.

\item
\label{embedPRNQWALL:proj}
$\etp{\proj{k}{l}{i}}
 \equiv \bs x_{1}\ldots x_{k+l}.x_{i}$,
   with $1\leq i\leq k+l$.

\item
\label{embedPRNQWALL:bran}
$\etp{\bran}\equiv\bs xyz. \Branch\,x\,y\,z$.

\item
\label{embedPRNQWALL:comp}
Let
$\emptyset;\emptyset;\emptyset\vdash
 \ta{\etp{f}}
    {(\liv^{k'}_{i=1}\$\BIntT)\liv
     (\liv^{l'}_{i=1}\$^m\BIntT)\liv
     \$^{m}\BIntT}$, and
$\emptyset;\emptyset;\emptyset\vdash
 \ta{\etp{g_i}}
    {(\liv^{k}_{i=1}\$\BIntT)\liv
     \$^{m_i}\BIntT}$, with $i\in\{1,\ldots,k'\}$, and
$\emptyset;\emptyset;\emptyset\vdash
 \ta{\etp{h_j}}
    {(\liv^{k}_{i=1}\$\BIntT)\liv
     (\liv^{l_j}_{i=1}\$^{n_j}\BIntT)\liv
     \$^{n_j}\BIntT}$, with $j\in\{1,\ldots,l'\}$.
If $p=\operatorname{max}\{m,m_1,\ldots,m_{k'},n_1,\ldots,n_{l'}\}$, then
\small
\begin{align*}
&
 \etp{
 \comp{k}
      {\sum_{i=1}^{l'} l_i}
      {k'}
      {l'}
      {f,g_1,\ldots,g_{k'},h_1,\ldots,h_{l'}}
     }\equiv
\\
&
\bullet^{k,\sum_{i=1}^{l'} l_i}_{k',l'}
   [\bs x_1\ldots x_{k'}
   .\EEmbed{p-m}{k'}{l'}{\etp{f}}
    (\LEmbed{1}{1}{\Coerc^{p-m-1}}\,x_1)
    \ldots
    (\LEmbed{1}{1}{\Coerc^{p-m-1}}\,x_{k'})
\\
&\phantom{\bullet^{k,\sum_{i=1}^{l'} l_i}_{k',l'}\ }
   ,\bs x_1\ldots x_{k}
   .\EEmbed{p-m_1}{k}{0}{\etp{g_1}}
    (\LEmbed{1}{1}{\Coerc^{p-m_{1}-1}}\,x_1)
    \ldots
    (\LEmbed{1}{1}{\Coerc^{p-m_{1}-1}}\,x_{k})
\\
&\phantom{\bullet^{k,\sum_{i=1}^{l'} l_i}_{k',l'}\ }
   ,\ldots
   ,\bs x_1\ldots x_{k}
   .\EEmbed{p-m_{k'}}{k}{0}{\etp{g_{k'}}}
    (\LEmbed{1}{1}{\Coerc^{p-m_{k'}-1}}\,x_1)
    \ldots
    (\LEmbed{1}{1}{\Coerc^{p-m_{k'}-1}}\,x_{k})
\\
&\phantom{\bullet^{k,\sum_{i=1}^{l'} l_i}_{k',l'}\ }
   ,\bs x_1\ldots x_{k}
   .\EEmbed{p-n_1}{k}{l_1}{\etp{h_1}}
    (\LEmbed{1}{1}{\Coerc^{p-n_{1}-1}}\,x_1)
    \ldots
    (\LEmbed{1}{1}{\Coerc^{p-n_{1}-1}}\,x_{k})
\\
&\phantom{\bullet^{k,\sum_{i=1}^{l'} l_i}_{k',l'}\ }
   ,\ldots
   ,\bs x_1\ldots x_{k}
   .\EEmbed{p-n_{l'}}{k}{l_{l'}}{\etp{h_{l'}}}
    (\LEmbed{1}{1}{\Coerc^{p-n_{l'}-1}}\,x_1)
    \ldots
    (\LEmbed{1}{1}{\Coerc^{p-n_{l'}-1}}\,x_{k})
]
\enspace .
\end{align*}
\normalsize

\item
\label{embedPRNQWALL:rec}
If
$\emptyset;\emptyset;\emptyset\vdash
 \ta{\etp{f_i}}
    {\$\BIntT\liv
     (\liv^{k}_{i=1}\$\BIntT)\liv
     (\liv^{l}_{i=1}\$^{m_i}\BIntT)\liv
     \$^{m_i}\BIntT\liv
     \$^{m_i}\BIntT}$, with $i\in\{0,1\}$, and
$\emptyset;\emptyset;\emptyset\vdash
 \ta{\etp{g}}
    {(\liv^{k}_{i=1}\$\BIntT)\liv
     (\liv^{l}_{i=1}\$^{m}\BIntT)\liv
     \$^{m}\BIntT}$,
then:
\small
$$
\etp{\rec{k+1}{l}{g,f_0,f_1}}\equiv
\Iter{1+\Intg{k}}{l}{F_0}{F_1}{G}
\enspace ,
$$
\normalsize
where
$G\equiv\EEmbed{p-m}
          {k+1}
          {l+1}
          {\bs n_0\,n_1\ldots n_{k}\,s_1\ldots s_{l}\,r.
           \etp{g}\, n_1\ldots n_k\,s_1\ldots s_l}
$,
$F_i\equiv\EEmbed{p-m_i}
            {k+1}
            {l+1}
            {\etp{f_i}}$, with $p=\operatorname{max}\{m_0,
\\
m_1,m\}$,
and $i\in\{0,1\}$.
\end{enumerate}
\textbf{Interpreting \QlSRN\ to \WALT.}
\label{definition:Interpreting QlSRN to WALT}
Let $\mathcal R$ be the set of environments, such that, every $\rho\in{\mathcal R}$ is a map from $\QlSRNVnames$ to $\Nat$. Then, $\srtw{\ }{}$ is a map from a pair in
$(\PRNQ\cup\Sigma_{\PRNQ})\times{\mathcal R}$,
to \PT, inductively defined on its first argument:
\small
\begin{align*}
\srtw{x}{\rho}&=\srtw{\rho(x)}{\rho}\qquad\qquad\qquad\qquad\qquad\qquad\qquad
(x\in\QlSRNVnames)\\
\srtw{0}{\rho}&=\etp{0}\\
\srtw{f}{\rho}&=\etp{f}\quad\qquad\qquad\qquad\qquad\qquad\qquad\qquad
(f\in\Sigma_{\PRNQ})\\
\srtw{f(t_1,\ldots, t_k, u_1,\ldots, u_l)}{\rho}
&=
\\
&
\vspace{-1cm}
\EEmbed{v-u+1-m}
       {0}
       {l}
       {\EEmbed{u-1}{0}{k+l}{\srtw{f}{\rho}}
        (\LEmbed{u-p_{1}}{0}{\srtw{t_1}{\rho}})
        \ldots
        (\LEmbed{u-p_{k}}{0}{\srtw{t_k}{\rho}})
       }
\\
&
\qquad\ \ \
(\LEmbed{v-q_{1}}{0}{\srtw{u_1}{\rho}})
\ldots
(\LEmbed{v-q_{l}}{0}{\srtw{u_l}{\rho}})
\qquad
(f\in\Sigma^{k,l}_{\QlSRN})
\end{align*}
\normalsize
when $u=\operatorname{max}\{m,p_1,\ldots,p_k\}$, $v=\operatorname{max}\{u-1+m,q_1\ldots,q_l\}$, and:
\small
\begin{align*}
&\emptyset;\emptyset;\emptyset\vdash
 \ta{\srtw{f}{\rho}}
    {(\liv^{k}_{i=1}\$\BIntT)\liv
     (\liv^{l}_{j=1}\$^m\BIntT)\liv
     \$^m\BIntT}\\
&\emptyset;\emptyset;\emptyset\vdash
 \ta{\srtw{t_i}{\rho}}
    {\$^{p_i}\BIntT}& i\in\{1,\ldots,k\}
\phantom{\enspace .}
\\
&\emptyset;\emptyset;\emptyset\vdash
 \ta{\srtw{u_j}{\rho}}
    {\$^{q_j}\BIntT}& j\in\{1,\ldots,l\}
\enspace .
\end{align*}
\normalsize
Otherwise, $\srtw{\ }{}$ is undefined.
\par
\textbf{Weight of a term in \QlSRN.}
\label{definition:Weight of a term in QlSRN}
For proving the statement that formalizes how we can embed \QlSRN\ into \WALT\ (Theorem~\ref{theorem:QlSRN is a subsystems of WALT} below)
we need a notion of weight of a \textit{closed} term in \QlSRN, which, essentially, gives a measure of its impredicativity.
For every \textit{closed} term $t\in\QlSRN\cup\Sigma_{\PRNQ}$, $\wght{t}{}$ is the \textit{weight of $t$}, defined by induction on $t$. If $t$ is one among zero, predecessor, successor, projection, and branching, then $\wght{t}{}=0$.
Otherwise:
\small
\begin{align*}
\wght{
\comp{k}
     {\sum_{i=1}^{l'} l_i}
     {k'}
     {l'}
     {f,g_1,\ldots,g_{k'},h_1,\ldots,h_{l'}}
}{}
&=
3\operatorname{max}\{
\wght{f}{},
\wght{g_1}{}
,\ldots,
\wght{g_k}{},
\wght{h_1}{}
,\ldots,
\wght{h_l}{}
,\frac{1}{3}\}
\\
\wght{
\rec{k+1}{l}{g,h_0,h_1}
}{}
&=
2\operatorname{max}\{
\wght{g}{},
\wght{h_0}{},
\wght{h_1}{}
,\frac{1}{2}
\}
\\
\wght{f(t_1, \ldots, t_k, u_1, \ldots, u_l)}{}
&=
2\operatorname{max}\{
\wght{f}{},
\wght{t_1}{}
,\ldots,
\wght{t_k}{},
\wght{u_1}{}
,\ldots,
\wght{u_l}{}
,\frac{1}{2}\}
\end{align*}
\normalsize
\vspace{-.5cm}
\begin{theorem}[\QlSRN\ is a subsystem of \WALT.]
\label{theorem:QlSRN is a subsystems of WALT}
Let $k, l\in\Nat$, $f\in\Sigma^{k,l}_{\PRNQ}$, and $t, t_1, \ldots, t_k, u_1, \ldots,$ $ u_l$ be terms of \QlSRN.
\begin{enumerate}
\item
\label{theo:realizPRNQ1}
There is an $m\geq 1$ such that
$\emptyset;\emptyset;\emptyset
 \vdash
 \ta{\etp{f}}
    {(\liv_{i=1}^{k}\$\BIntT)\liv
     (\liv_{j=1}^{l}\$^{m}\BIntT)\liv
     \$^{m}\BIntT}
$.

\item
\label{theo:realizPRNQ2}
$\srtw{f(t_1, \ldots, t_k, u_1, \ldots, u_l)}{\rho}$ is defined,
for every $\rho$.

\item
\label{theo:realizPRNQ3}
$\emptyset;\emptyset;\emptyset\vdash
 \ta{\srtw{t}{}}
    {\$^{m}\BIntT}$ with $m\leq \wght{t}{}$.

\item
\label{theo:realizPRNQ4}
$\srtw{n}{}\red^+ \BNum{n}$, for every $n\geq 0$.

\item
\label{theo:realizPRNQ5}
If $f(n_1,\ldots,n_k,s_1,\ldots,s_l)=n$, then
$\srtw{f(n_1,\ldots,n_k,s_1,\ldots,s_l)}{}\red^*\BNum{n}$,
for every $n_1,\ldots,n_k,s_1,\ldots,s_l\in\Nat$.
\end{enumerate}
\end{theorem}
Point~\ref{theo:realizPRNQ1} is a direct consequence of the typing of the combinators of \WALT\ that we use in the definition of $\etp{f}$.
Point~\ref{theo:realizPRNQ2} follows from point~\ref{theo:realizPRNQ1} here above and from the definition of $\srtw{\ }{}$.
Point~\ref{theo:realizPRNQ3} holds by induction on $t$.
Point~\ref{theo:realizPRNQ4} holds by induction on $n$.
Point~\ref{theo:realizPRNQ5} holds by induction on $f$.
Finally, by structural induction on $t$, we have:
\begin{corollary}[The embedding of \QlSRN\ into \WALT\ is sound.]
\label{corollary:The embedding of QlSRN into WALT is sound}
Let $t\in\QlSRN$, and $n\in\Nat$.
If $t=n$, then $\srtw{t}{\rho}\red^{+}\BNum{n}$, for every environment $\rho$.
\end{corollary}
\section{Conclusions and further work}
\label{section:Conclusions and further work}
\WALL\ is a type assignment for pure $\lambda$-terms, typable in \SF, that characterizes the class of poly-time computable functions. Its design principles relax the stratification of deductions of \LLL, \LAL, \DLAL.
The subject reduction holds for a suitable restriction of both call-by-name and call-by-value $\beta$-reduction coherently with the idea that we can simulate a call-by-value poly-time sound rewriting system, like \SRN\ is.
A call-by-value behavior implies that \WALL\ presents the typical aspects of the standard call-by-value $\lambda$-calculus, that we can outline by a simple example.
Let us assume to use the call-by-value rewriting step $\rightarrow_{v}$ on \SF\ --- recall that $(\bs x.M)N\rightarrow_{v} M\subs{N}{x}$ if $N$ is either a variable or a $\lambda$-abstraction ---. 
Then, we can write $\vdash_{\SF}\ta{\bs fx.(\bs w.f\ w)(f x)}{\forall \alpha.(\alpha\rightarrow\alpha)\rightarrow\alpha\rightarrow\alpha}$. Namely, we can give the type of the Church numerals to $\rightarrow_{v}$-normal forms not having the canonical form $\bs fx.f(\ldots(f x)\ldots)$. But this is an intrinsic aspect of the call-by-value and \WALL\ cannot escape it.
The call-by-value nature of \WALL\ is a further example that the call-by-value operational semantics has a role in the domain of Linear logic 
\cite{Pravato+Ronchi+Roversi:1999-MSCS,Cop-DLag-Ron:EALCBV-04}.
\par
Finally, future work might address, at least, the following subjects.
\par
\textbf{Completeness of \WALT\ with respect to \SRN, or other systems like the tiered ones.}
The conjecture is that \WALT\ is, in fact, \SRN\ complete.  We know from 
\cite{DalLago+Martini+Roversi:2004-TYPES} how using an iterator to supply the same value to two distinct arguments of a given function. So, composing enough iterators in \WALT\ to duplicate a safe argument, and ``dispatching'' the copies as needed, using a linear exchange combinator, looks a promising strategy to prove the \SRN\ completeness of \WALT.
\par
\textbf{Generalizing the design principles of \WALL.} 
\WALT\ is a generalization of some basic structural proof-theoretical principles.
It is natural to ask if it is the larger system that extends such principles. In fact, preliminary investigations, say that there is room for further poly-time sound generalizations of \WALT.
\par
\textbf{Polynomial $\lambda$-calculi.}
We think that the rewriting relation $\red$ that \WALT\ induces on the $\lambda$-terms deserves further study. The emphasis should be put on the $\lambda$-terms $M$, typable in \SF, that can be reduced to their normal form  $M'$ by some poly-time sound normalization strategy, based on the call-by-value, or call-by-name, $\beta$-reduction.
If $M''$ is the normal form that we can get from $M$ by using $\red^*$, how much does it cost to rewrite $M''$ to $M'$, by iterating the the call-by-name, or call-by-value, $\beta$-reduction?
\bibliographystyle{alpha}
\bibliography{walt}
\appendix
\section{Completeness}
\label{section:Completeness}
We explicitly show that the poly-time completeness of \WALL\ holds. 
Notice that this property must be explicitly proved because \WALT\ does not exactly contain \LAL\ as its subsystem.
We shall represent and simulate poly-time Turing machines inside \WALT: the depth of the derivation encoding the given machine will be constant, independently of the length of the representation of the tapes. We borrow and integrate ideas in \cite{Roversi:1999-CSL,Asperti02TOCL,MollerMairson:2002-ICC} showing:
\begin{theorem}[Poly-time completeness.]
\label{theorem:poly-time-completeness}
For every poly-time Turing machine $\TM$ with a set of states ${\mathcal S}$, a tape alphabet $\Sigma$, a transition function $\TrF$, and a polynomial $\TMPoly{k}{x}$ of degree $k$, we can write a closed $\lambda$-term $\ol{\TM}$ with type 
$\ListT\AlphT \li \$^{4e+1}\ConfT$ such that $\ol{\TM}$ takes a list of type $\ListT\AlphT$, that represents the input tape, and evaluates it to a configuration of type $\$^{4e+1}\ConfT$. 
\end{theorem}
The configuration contains the encoding of the output tape, and $e$ is the least value such that $k\leq 2^e$. $\ol{\TM}$ is the composition of two parts. The quantitative one is a term that represents $\TMPoly{k}{x}$. It calculates how long the simulation of $\TM$ lasts. The qualitative part implements the transition function $\TrF$ of $\TM$ as an iterable $\lambda$-term. The quantitative and qualitative parts are put together so that the Church numeral, result of the quantitative part, iterates the qualitative one, starting from an initial configuration, that contains the representation of the input tape.
Before proceeding, without loss of generality, we fix some simplifying assumption and convention on the poly-time Turing machines we shall represent. 
$\mathcal S$ contains the \textit{initial} and the \textit{accepting} states $s_0$, and $s_a$, respectively. 
$\Sigma$ contains at least the symbols \textit{false}, \textit{true}, and \textit{blank}, identified as $\sigma_1,\sigma_2, \sigma_3$, respectively.
For any given $\TM$, $\TrF$ cannot leave $s_a$, once entered it. Namely, $\TMPoly{k}{x}$ may overestimate the time required to yield the result. So, $\TrF$, besides the \textit{leftward move} $\MoveL$ and the \textit{rightward move} $\MoveR$, can issue the \textit{stay-there} ``move'' $\DoNotMove$ to the read head.
Finally, we assume that $\TM$ enters $s_a$ with its head on the leftmost symbol of the \textit{output portion} of the tape, where, by ``output portion'', we mean the symbol under the head and all the symbols to its right hand side, up to the first blank.
\subsection{Preliminaries}
\textbf{Notations and definitions.} 
If $X$ is any finite set, $\size{X}$ is its cardinality.
$\$^nA$ denotes $\$\cdots\$A$ with $n\geq 0$ occurrences of $\$$. An analogous meaning holds for $!^nA$. $(\li^{n}_{i=1} A_i)$ abbreviates $A_1\li\cdots\li A_n$, while $(\mliv^{n}_{i=1} A_i)$ abbreviates $A_1\liv\cdots\liv A_n$. If useful, $B^A$ shortens $A\li B$, for any $A, B$. Finally, the $\lambda$-term identity  $\bs x.x$ is $\Id$.
\par
Table~\ref{table:Basic data-types} introduces and redefines some basic data-types we shall use to define the encoding of a given $\TM$.
\begin{Table}[ht]
\begin{center}
\small
{\renewcommand{\arraystretch}{1.2}
\begin{tabular}{|r|l|}
\hline
\textbf{Type name} & \textbf{Type definition and canonical terms} \\
\hline\hline
List of type $A$ &
\begin{minipage}{.7\textwidth}
$
\begin{array}{r@{\hspace{.4em}}c@{\hspace{.4em}}ll}
\ListT\ A &\equiv& \forall \alpha\beta.
	      !(A\li(\beta\li\alpha)\li\alpha)\li
	      \$(\alpha\li\beta\li\alpha)\\
\Nil &\equiv& \bs c.\bs x.\bs y .x \\
{[M_1,\ldots,M_m]}
     &\equiv&\bs c.\bs x.
        \bs y.c M_1( \cdots (\bs y. c M_m(\bs y.x)) \cdots)
       &m\geq 1
\end{array}
$
\end{minipage}
\\
\hline
Alphabet &
\begin{minipage}{.7\textwidth}
$
\begin{array}{r@{\hspace{.4em}}c@{\hspace{.4em}}ll}
\AlphT &\equiv& \BoolT_{\size{\Sigma}+1}\\
\CAlph{\sigma_i}
	&\equiv&\pi^{\size{\Sigma}+1}_{i}
        &1\leq i\leq \size{\Sigma} \\
\LTape  &\equiv&\pi^{\size{\Sigma}+1}_{0} \\
\RTape  &\equiv&\pi^{\size{\Sigma}+1}_{\size{\Sigma}+1} \\
\end{array}
$
\end{minipage}
\\
\hline State &
\begin{minipage}{.7\textwidth}
$
\begin{array}{r@{\hspace{.4em}}c@{\hspace{.4em}}ll}
\StatT &\equiv& \BoolT_{\size{{\mathcal S}}}\\
\CStat{s_i}
	&\equiv&\pi^{\size{\mathcal S}}_{i}
        &0\leq i\leq \size{\mathcal S}-1
\end{array}
$
\end{minipage}
\\
\hline
\end{tabular}}
\normalsize
\end{center}
\caption{Basic data-types}
\label{table:Basic data-types}
\end{Table}
We redefine the structure of the lists to simplify their inductive manipulation in a call-by-value setting.
The alphabet representation has the term $\CAlph{\sigma_i}$ for every symbol $\sigma_i\in\Sigma$, plus $\LTape$, that marks the \textit{left hand border} of the represented tape, and $\RTape$ which marks the \textit{right hand border}. This allows us to give a finite representation of the tape, and to extend it on-demand to the left, or to the right, when the read head reaches one of the two borders. Namely, our choice is different from \cite{MollerMairson:2002-ICC}, where no border is explicitly used in the representation of the tape. There, the side effect is that every application of the representation of $\TrF$ to the representation of the tape extends this latter by one symbol per part.

\begin{lemma}[Typing Alphabet and State.]
\label{lemma:Typing the canonical terms}
Recall that, for every $0\leq i\leq m-1$,
$\emptyset;\emptyset;\emptyset\vdash \ta{\pi^m_i}{\BoolT_m}$.
The type of $\CAlph{\sigma_i}$ and $\CStat{s_i}$ are obvious, once instantiated  $m$ with $\size{\AlphT}$ and $\size{\StatT}$, respectively.
\end{lemma}

\begin{lemma}[Dynamics of Alphabet and State.]
\label{lemma:Dynamics of some canonical term}
Recall that $\pi^{m}_{i}\, \lan M_0,\ldots,M_{m-1}\ran \red^+ M_i$. An analogous behavior exists for $\CAlph{\sigma_i}$ and $\CStat{s_i}$, once replaced 
$m$ by $\size{\Sigma}+1$ and $\size{\mathcal S}$, respectively.
\end{lemma}
Table~\ref{table:Basic combinators} introduces the \textit{Basic combinators} useful to build the more complex terms of the qualitative and quantitative parts of $\ol{\TM}$.
Some of them have already been defined in previous sections. Any explicit reintroduction is justified by the attempt to improve the readability by means of a uniform naming of the combinators with similar behavior.
\begin{Table}[ht]
\begin{center}
\small
{\renewcommand{\arraystretch}{1.2}
\begin{tabular}{|r|l|}
\hline
\textbf{Class} & \textbf{Definition} \\
\hline\hline
Successor &
\begin{minipage}{.7\textwidth}\vspace{2pt}
$
\begin{array}{r@{\hspace{.4em}}c@{\hspace{.4em}}ll}
\NSucc &\equiv& \bs mf.(\bs zx.f(z\,x))(m\,f)
\\
\NSuccE &\equiv&
\bs\elan x\, y\eran.
\elan \BEmbed{1}{\NSucc}\,x
    , \BEmbed{1}{\NSucc}\,y \eran
\\
{\LSucc[M]} 
       &\equiv& \bs xy. \LPush[M]\, x\, (y\, \Id)\\
       & & \text{ \textbf{where} }
           \LPush[M] \equiv \bs xl.
                       \bs c.(\bs yzw.c\,z(y\,w))(l\,c)(M\,x)
\\
{\LSuccE[M,N]} 
       &\equiv& \bs xy. 
           \LPushE[M] 
	   (N\,x)
           ((\bs \elan z_1\, z_2\eran.
	       \elan \bs y.z_1,\bs y.z_2\eran)
	    (y\, \Id))\\
       & & \text{ \textbf{where} }
           \LPushE[M] \equiv 
	   \bs \elan x_1\, x_2\eran.
	   \bs \elan y_1\, y_2\eran.
	   \\
       & & \phantom{\text{ \textbf{where} } \LPushE[M] \equiv \bs }
	   \elan
	    \LEmbed{1}{2}{\LSucc[M]}\,x_1\,y_1
	    ,\LEmbed{1}{2}{\LSucc[M]}\,x_2\,y_2
	   \eran
\end{array}
$
\end{minipage}
\\
\hline
Coerce &
\begin{minipage}{.7\textwidth}\vspace{2pt}
$
\begin{array}{r@{\hspace{.4em}}c@{\hspace{.4em}}ll}
\ACoerce&\equiv& \bs x. x\,
   \lan
    \CAlph{\sigma_0},\ldots,\CAlph{\sigma_{\size{\Sigma}+1}}
   \ran
\\
\NCoerce &\equiv&
\bs m.(\bs z.z\,\UNum{0})(m(\bs y.\NSucc\, y))
\\
{\LCoerce[M]} &\equiv& \bs l.(\bs z.z\, \Nil\, \Id)(l\,\LSucc[M])
\\
\MLCoerce^1[M] &\equiv& \LCoerce[M]\\
\MLCoerce^n[M] &\equiv& \bs x. 
                   \LEmbed{1}{1}{\MLCoerce^{n-1}[M]}
		   (\LCoerce[M]\, x)
	       & n> 1
\end{array}
$
\end{minipage}
\\
\hline
Diagonal &
\begin{minipage}{.8\textwidth}\vspace{2pt}
$
\begin{array}{r@{\hspace{.4em}}c@{\hspace{.4em}}ll}
\ADiag&\equiv&
\bs x.x\,
      \lan 
      \lan\CAlph{\sigma_0},\CAlph{\sigma_0}\ran
          ,\ldots,
      \lan\CAlph{\sigma_{\size{\Sigma}+1}},
          \CAlph{\sigma_{\size{\Sigma}+1}}\ran 
      \ran
\\
\ADiagE&\equiv&
\bs x.x\,
      \lan 
      \elan\CAlph{\sigma_0},\CAlph{\sigma_0}\eran
           ,\ldots,
      \elan\CAlph{\sigma_{\size{\Sigma}+1}},
           \CAlph{\sigma_{\size{\Sigma}+1}}\eran
      \ran
\\
\NDiagE&\equiv&
\bs m.(\bs z.z\,\elan\UNum{0},\UNum{0}\eran)
       (m\,\NSuccE)
\\
\LDiagE&\equiv& 
\bs l.(\bs z.z\,\elan\Nil,\Nil\eran\, I)(l\,\LSuccE[M,N])
\end{array}
$
\end{minipage}
\\\hline
\end{tabular}}
\normalsize
\end{center}
\caption{Basic combinators}
\label{table:Basic combinators}
\end{Table}
Every \textit{successor} takes an instance of some Basic data-types and yields its successor, whatever this means.
Every \textit{coerce} takes an instance of some Basic data-types and gives back the same instance inside some boxes.
Every \textit{diagonal} replicates the instance of some Basic data-types inside some boxes.
\begin{lemma}[Typing the Basic combinators.]
\label{lemma:Typing the Basic combinators}
\begin{description}
\item [Successor.]
\begin{enumerate}
\item 
$\emptyset;\emptyset;\emptyset\vdash\ta{\NSucc}{\UIntT\li\UIntT}$.

\item 
$\emptyset;\emptyset;\emptyset\vdash
\ta{\NSuccE}
   {(\$\UIntT\odot\$\UIntT)\li(\$\UIntT\odot\$\UIntT)}$.

\item 
Let $\emptyset;\emptyset;\emptyset\vdash\ta{M}{A\li\$A}$. Then,
both 
$\emptyset;\emptyset;\emptyset
\vdash\ta{\LPush[M]}{A\li\ListT A\li\ListT A}$, and
$\emptyset;\emptyset;\emptyset
\vdash\ta{\LSucc[M]}{A\li(\GG\li\ListT A)\li\ListT A}$.

\item 
Let $\emptyset;\emptyset;\emptyset\vdash\ta{M}{A\li\$A}$
and $\emptyset;\emptyset;\emptyset\vdash\ta{N}{A\li(\$A\odot\$A)}$.
Then, both
$\emptyset;\emptyset;\emptyset
\vdash
\ta{\LPushE[M,N]}
   {(\$A\odot\$A)\li
    (\$(\GG\li\ListT A)\odot\$(\GG\li\ListT A))\li
    (\$\ListT A\odot\$\ListT A)}$, and
$\emptyset;\emptyset;\emptyset
\vdash
\ta{\LSuccE[M]}
   {A\li(\GG\li(\$\ListT A\odot\$\ListT A))
     \li(\$\ListT A\odot\$\ListT A)}$.
\end{enumerate}

\item [Coerce.]
\begin{enumerate}
\item 
$\emptyset;\emptyset;\emptyset\vdash\ta{\ACoerce}{\AlphT\li\$\AlphT}$.

\item 
$\emptyset;\emptyset;\emptyset\vdash\ta{\NCoerce}{\UIntT\li\$\UIntT}$.

\item 
Let 
$\emptyset;\emptyset;\emptyset
\vdash\ta{M}{A\li\$A}$. Then
$\emptyset;\emptyset;\emptyset
\vdash\ta{\LCoerce[M]}{\ListT A\li\$\ListT A}$.

\item 
Let 
$\emptyset;\emptyset;\emptyset
\vdash\ta{M}{A\li\$A}$. Then
$\emptyset;\emptyset;\emptyset
\vdash\ta{\MLCoerce^n[M]}{\ListT A\li\$^n\ListT A}$, for every $n\geq 1$.
\end{enumerate}

\item [Diagonal.]
\begin{enumerate}
\item 
$\emptyset;\emptyset;\emptyset
\vdash\ta{\ADiag}{\AlphT\li(\AlphT\otimes\AlphT)}$.

\item 
$\emptyset;\emptyset;\emptyset
\vdash\ta{\ADiagE}{\AlphT\li(\$\AlphT\odot\$\AlphT)}$.

\item 
$\emptyset;\emptyset;\emptyset
\vdash\ta{\NDiagE}{\UIntT\li\$(\$\UIntT\odot\$\UIntT)}$.

\item 
Let 
$\emptyset;\emptyset;\emptyset
\vdash\ta{M}{A\li\$A}$ and $\emptyset;\emptyset;\emptyset
\vdash\ta{N}{A\li(\$A\odot\$A)}$. Then
$\emptyset;\emptyset;\emptyset
\vdash\ta{\LDiagE[M,N]}{\ListT A\li\$(\$\ListT A\odot\$\ListT A)}$.
\end{enumerate}
\end{description}
\end{lemma}

\begin{lemma}[Dynamics of the Basic combinators.]
\label{lemma:Dynamics of the Basic combinators}
\begin{description}
\item [Successor.]
\begin{enumerate}
\item 
$\NSucc\,\UNum{n}\red^+\UNum{n+1}$.

\item 
$\NSuccE\,\elan\UNum{m},\UNum{n}\eran\red^+
          \elan\UNum{m+1},\UNum{n+1}\eran$.

\item 
For every $n\geq 0$, let
$M\,P$, $\LPush[M]\,P\,[M_1,\ldots,M_n]$, and
$\LSucc[M]\,P\,(\bs y. [M_1,\ldots,M_n])$ be typeable.
For every
$y\not\in\bigcup^{n}_{i=1}\FV{M_i}$,
$M\,P\red^+ P$ implies
$\LPush[M]\,P\,[M_1,\ldots,M_n]
   \red^+ [P,M_1,\ldots,M_n]$
and
$\LSucc[M]\,P\,(\bs y. [M_1,\ldots,M_n])
   \red^+ [P,M_1,\ldots,M_n]$.

\item 
For every $m, n\geq 0$, and $i,j\in\{0,1,2\}$, let us assume that
$M\,P_i$, 
$N\,P_i$,
\begin{align*}
&\LPushE[M]\,
  \elan P_1,P_2\eran\,
  \elan \bs y.[M_1,\ldots,M_m], \bs y.[N_1,\ldots,N_n]\eran
  \enspace, \text{ and}
\\
&\LSuccE[M,N]\,
  P_0\,
  \elan [M_1,\ldots,M_m], [N_1,\ldots,N_n]\eran
\end{align*}
be typeable.
For every
$y\not\in(\bigcup^{m}_{i=1}\FV{M_i})\cup
         (\bigcup^{n}_{i=1}\FV{N_i})$,
if $M\,P_i\red^+ P_i$ and $N\,P_i\red^+ \elan P_i,P_i\eran$, then:
\begin{eqnarray*}
\lefteqn{
\LPushE[M]\,
  \elan P_1,P_2\eran\,
  \elan \bs y.[M_1,\ldots,M_m], \bs y.[N_1,\ldots,N_n]\eran
  \red^+
}\\
&&\qquad\qquad\qquad\qquad\qquad\qquad\qquad\qquad
  \elan [P_1,M_1,\ldots,M_m], [P_2,N_1,\ldots,N_n]\eran
\\
\lefteqn{
\LSuccE[M,N]\,
  P_0\,
  \bs y.\elan [M_1,\ldots,M_m], [N_1,\ldots,N_n]\eran
  \red^+
}\\
&&\qquad\qquad\qquad\qquad\qquad\qquad\qquad\qquad
 \elan [P_0,M_1,\ldots,M_m], [P_0,N_1,\ldots,N_n]\eran
\end{eqnarray*}
\end{enumerate}

\item [Coerce.]
\begin{enumerate}
\item 
$\ACoerce\,M\red^+M$, for every $M\in\{\LTape,\CAlph{\sigma_1},\ldots,\CAlph{\sigma_{\size{\Sigma}}},\RTape\}$.

\item 
$\NCoerce\,\UNum{m}\red^+\UNum{m}$, for every $m\geq 0$.

\item 
For every $m\geq 0$, let $\LCoerce[M]\,[M_1,\ldots,M_m]$ be typeable.
If $M\,M_i\red^+M_i$, for every $0\leq i \leq m$, then
$\LCoerce[M]\,[M_1,\ldots,M_m]\red^+[M_1,\ldots,M_m]$.

\item 
For every $m\geq 0, n\geq 1$, let $\MLCoerce^n[M]\,[M_1,\ldots,M_m]$ be typeable.
If $M\,M_i\red^+M_i$, for every $0\leq i \leq m$, then
$\MLCoerce^n[M]\,[M_1,\ldots,M_m]\red^+[M_1,\ldots,M_m]$.
\end{enumerate}

\item [Diagonal.]
\begin{enumerate}
\item 
$\ADiag\,M\red^+\lan M,M \ran$, for every $M\in\{\LTape,\CAlph{\sigma_1},\ldots,\CAlph{\sigma_{\size{\Sigma}}},\RTape\}$.

\item 
$\ADiagE\,M\red^+\elan M,M \eran$, for every $M\in\{\LTape,\CAlph{\sigma_1},\ldots,\CAlph{\sigma_{\size{\Sigma}}},\RTape\}$.

\item 
$\NDiagE\,\UNum{m}\red^+\elan \UNum{m},\UNum{m} \eran$, for every $m\geq 0$.

\item 
For every $0\leq i\leq m$, let us assume $M, M_i, N, M_i$, and
$\LDiagE[M,N]\,[M_1,\ldots,M_m]$ be typeable.
If $M\,M_i\red^+ M_i$, and 
$N\,M_i\red^+ \elan M_i,M_i\eran$, then
$\LDiagE[M,N]\,[M_1,\ldots,M_m]\red^+
\elan[M_1,\ldots,M_m],$ \\ $[M_1,\ldots,M_m]\eran$.
\end{enumerate}
\end{description}
\end{lemma}
For proving the statement, we need to consider four cases for the class Successor. With $\NSucc$ proceed by induction on $n$. With $\NSuccE$ prove that $\BEmbed{1}{\NSucc}\,\UNum{n}\red^+\UNum{n+1}$ by induction on $n$.
With $\LPush[M]$, and $\LSucc[M]$ proceed by induction on $n$.
With $\LPushE[M]$, and $\LSuccE[M,N]$ apply the definition and use the result on the dynamics of $\LSucc[M]$.
\par
For the class Coerce we have four cases. 
With $\ACoerce$ just apply the definition.
With $\NCoerce$ and $\LCoerce[M]$ proceed by induction on $m$.
With $\LCoerce^n[M]$ proceed by induction on $n$, using the previous case.
\par
For the class Diagonal we have three cases. 
With $\ADiag$ just apply the definition.
With $\NDiagE$ proceed by induction on $m$.
With $\LDiagE[M,N]$ proceed by induction on $m$, using the previous points.
\subsection{Quantitative part}
The quantitative part requires to represent polynomials. Recall that, by assumption, if our poly-time Turing machine enters the accepting state $\sigma_a$, it never leaves it, even if the clock,
bounded by a polynomial $\TMPoly{k}{x}=\sum^{k}_{i=0}k_i x^i$, of degree $k$, keeps ticking.
This observation, together with the assumptions $x, k>0$, allows to simplify the form of the polynomials to represent.
First, we observe that $\sum^{k}_{i=0}k_i x^i\leq x^k\sum^{k}_{i=0}k_i\leq K x^k\leq K x^{2^e}$, for the least $e>0$ such that $k\leq 2^e$, and for every big enough $K$.
Second, we write a closed term $\NPoly{e}{K}$, with type $\$\UIntT\liv\$^{4e}\UIntT$, that applied to $\UNum{n}$, yields $\UNum{Kn^{2^e}}$. So, the result of $\NPoly{e}{K}\,\UNum{n}$ will be the clock in the representation of the given Turing machine.
Table~\ref{table:Representing a polynomial with exponential exponent} introduces $\NPoly{e}{K}$.
\begin{Table}
\begin{minipage}{\textwidth}
\small
\begin{center}
{\renewcommand{\arraystretch}{1.2}
\begin{tabular}{|c|}
\hline
$\NPoly{e}{K}$ \textbf{and the combinators that define it}
\\\hline\hline
$
\begin{array}{r@{\hspace{.4em}}c@{\hspace{.4em}}ll}
\NSum &\equiv& \bs mnf. (\bs wzx.w(z\,x))(m\,f)(n\,f)
\\
\NMult &\equiv& 
  \bs mn.(\bs z.z\,\UNum{0})(m(\bs y.\NSum n\,y)))
\\
\NMultE &\equiv& 
  \bs \elan m\, n\eran.
  (\bs z.z\,\UNum{0})(m(\bs y.\NSum n\,y))
\\
\NSquare^0 &\equiv&
  \bs n. \LEmbed{1}{1}{\NMultE}
         (\LEmbed{1}{1}
	   {
	    \bs \elan y\,z\eran.
	    \elan
	       y,
	       \BEmbed{1}{\NCoerce}\,z
	     \eran
	   }
          (\NDiagE\, n))
\\
\NSquare^1 &\equiv&
\BEmbed{1}{\NSquare^0}
\\
\NSquare^e &\equiv&
\bs x.
\BEmbed{4}{\NSquare^{e-1}}(\BEmbed{1}{\Coerc^4}\, x)
&(e>1)
\\
\NPoly{e}{K} &\equiv&
\bs x.
(\bs xy. \NMult\,x\,y)\,\UNum{K}\,(\NSquare^e\, x)
& (e\geq 1)
\end{array}
$
\\\hline
\end{tabular}}
\end{center}
\normalsize
\end{minipage}
\caption{Representing a polynomial with exponential exponent}
\label{table:Representing a polynomial with exponential exponent}
\end{Table}

\begin{lemma}[Typing and dynamics.]
\begin{enumerate}
\item 
$\emptyset;\emptyset;\emptyset\vdash
 \ta{\NSum}
    {\UIntT\li\UIntT\li\UIntT}$
and
$\NSum\,\UNum{m}\,\UNum{n}\red^+\UNum{m+n}$.

\item 
$\emptyset;\emptyset;\emptyset\vdash
 \ta{\NMult}
    {\UIntT\li\$\UIntT\liv\$\UIntT}$
and
$\NMult\,\UNum{m}\,\UNum{n}\red^+\UNum{mn}$.

\item 
$\emptyset;\emptyset;\emptyset\vdash
 \ta{\NMultE}
    {(\$\UIntT\odot\$^2\UIntT)\li\$^2\UIntT}$
and
$\NMultE\,\UNum{m}\,\UNum{n}\red^+\UNum{m n}$.

\item 
$\emptyset;\emptyset;\emptyset\vdash
 \ta{\NSquare^0}
    {\UIntT\li\$^3\UIntT}$
and
$\NSquare^0\,\UNum{m}\red^+\UNum{m^2}$.

\item 
$\emptyset;\emptyset;\emptyset\vdash
 \ta{\NSquare^e}
    {\$\UIntT\liv\$^{4e}\UIntT}$
and
$\NSquare^e\,\UNum{m}\red^+\UNum{m^{2^e}}$, for every $e\geq 1$.

\item 
$\emptyset;\emptyset;\emptyset\vdash
 \ta{\NPoly{e}{K}}
    {\$\UIntT\liv\$^{4e}\UIntT}$
and
$\NPoly{e}{K}\,\UNum{m}\red^+\UNum{Km^{2^e}}$, for every $e\geq 1$.
\end{enumerate}
\end{lemma}
To prove it we apply the definitions in Table~\ref{table:Representing a polynomial with exponential exponent}, and use Lemma~\ref{lemma:Typing the canonical terms}, 
\ref{lemma:Dynamics of some canonical term},
\ref{lemma:Typing the Basic combinators}, 
and~\ref{lemma:Dynamics of the Basic combinators}. In particular, for typing $\NSquare^{e}$, we inductively assume 
$\emptyset;\emptyset;\emptyset\vdash
 \ta{\NSquare^{e-1}}
    {\$\UIntT\liv\$^{4(e-1)}\UIntT}$.
Finally, it is necessary to show 
$\emptyset;\emptyset;\emptyset\vdash
\ta{\bs xy. \NMult\,x\,y}
{\$^{4e-1}\UIntT\li\$^{4e}\UIntT\liv\$^{4e}\UIntT}$.
\subsection{Qualitative part}
\textbf{Notations.}
$\seq{M}{m}{n}$ denotes a \textit{sequence} $M_m,\ldots,M_n$ of terms.
If $m>n$, then $\seq{M}{m}{n}$ is empty, and we denote it by $\eseq$. The element of position $m\leq i\leq n$ in $\seq{M}{m}{n}$ is $\seqel{M}{i}$.
The \textit{quasi-Tapes} are almost lists and are used as a convenient shortening of both the configurations and the quasi-configurations of Table~\ref{table:Quasi-Tapes, configurations and quasi-Configurations}. The definition of quasi-Tapes is:
\small
$$
\begin{array}{r@{\hspace{.4em}}c@{\hspace{.4em}}ll}
\PreTape{\eseq}{\eseq}{x}&\equiv& x
\\
\PreTape{\seq{c}{1}{m}}{\seq{M}{1}{m}}{x}&\equiv&
\seqel{c}{1}\,\seqel{M}{1}
(\bs y.\PreTape{\seq{c}{2}{m}}{\seq{M}{2}{m}}{x})
& (m\geq 1)
\end{array}
$$
\normalsize
\par
\textit{Configurations} and \textit{quasi-Configurations} of Table~\ref{table:Quasi-Tapes, configurations and quasi-Configurations} are the data-types that we use to define $\WALLTrF$, our encoding of $\TrF$.
Every instance of Configuration represents the left hand side of the tape, the current state, and the right hand side of the tape, under the conventions in Table~\ref{table:Basic data-types} on the representations of the states of $\mathcal S$, by the terms of type $\StatT$, and the symbols of $\Sigma$, by the terms of type $\AlphT$.
For example, an \textit{initial configuration} is:
\small
\begin{eqnarray}
\label{eqn:initial-conf}
\Conf{\LTape}{\CStat{s_0}}{\seq{R}{1}{m},\RTape}
\end{eqnarray}
\normalsize
where $\seqel{R}{i}=\CAlph{\sigma_{k}}$, for every $1\leq i\leq m$, $1\leq k\leq \size{\Sigma}$, and $m\geq 0$. 
The left hand side of the tape is empty, so it contains only the symbol $\LTape$ that marks its border.
The right hand side, besides its right border $\RTape$, is assumed to contain $\seq{R}{1}{m}$ the input tape of $\TM$.
Observe that, if $m=0$, then the represented tape is empty. In a few, we shall see that the condition ``$\seqel{R}{i}$ different from $\LTape$ and $\RTape$'' has consequences in the definition of the look-up function that determines the moves and the tape symbols, written by the head in a configuration and which must determine when the left or the right hand side of the tapes in a configuration must be extended.
Every instance of pre-Configuration is an intermediate step between two consecutive configurations, the second being obtained by applying the transition function to the first one.
For example, let us assume that $\seqel{R}{1}=\CAlph{\sigma'}$ and that we need to simulate 
$\TrF(\sigma',s_0)= (\MoveR,s,\sigma)$. 
We must move rightward in the state $\CStat{s}$, writing $\CAlph{\sigma}$, once read $\seqel{R}{1}$ in \eqref{eqn:initial-conf}. The resulting configuration would be:
\small
\begin{eqnarray}
\label{eqn:next-conf}
\Conf{\CAlph{\sigma},\LTape}{\CStat{s}}{\seq{R}{2}{m},\RTape}
\end{eqnarray}
\normalsize
Then, the quasi-configuration generated between \eqref{eqn:initial-conf} and \eqref{eqn:next-conf} is
\small
\begin{eqnarray}
\label{eqn:intial-next-quasi-conf}
\PreConf{\LTape}{\eseq}
        {\CStat{s_0}}
	{\seqel{R}{1}}{[\seq{R}{2}{m},\RTape]}
\end{eqnarray}
\normalsize
Namely, a quasi-configuration, besides the state, makes available the symbol under the head and the symbol to its immediate left.
These two symbols and the state find in a \textit{look-up table} the next move, state and symbol to be written. The definition of the look-up table must also define what to do when the read symbol, or the one immediately to its left are $\RTape$ or $\LTape$, respectively.

\begin{Table}[ht]
\begin{center}
\small
{\renewcommand{\arraystretch}{1.2}
\begin{tabular}{|r|l|}
\hline
\textbf{Type name} & \textbf{Type definition and canonical terms} \\
\hline\hline
Configuration &
\begin{minipage}{.7\textwidth}
\scriptsize
$
\begin{array}{r@{\hspace{.4em}}c@{\hspace{.4em}}ll}
\ConfT &\equiv& \forall \alpha\beta.
	   !(\AlphT\li(\BB\li\alpha)\li\alpha)\li
       \\
       & & \qquad
           \$(\alpha\li\alpha\li
	         ((\BB\li\alpha)
	          \otimes
		  \StatT
		  \otimes
		  (\BB\li\alpha)))
\\
\Conf{\seq{L}{1}{m},\LTape}
     {\CStat{s}}
     {\seq{R}{1}{n},\RTape}
   &\equiv& \bs c.\bs lr.
       \lan
        \bs y.\PreTape{\seq{c}{1}{m+1}}{\seq{L}{1}{m},\LTape}{l},
	\CStat{s},
\\&&\phantom{\bs c.\bs lr.\lan\bs y.}
        \bs y.\PreTape{\seq{c}{1}{n+1}}{\seq{R}{1}{n},\RTape}{r}
       \ran\\
   & & m,n\geq 0 \textbf{ and } \seqel{c}{i}= c 
       \textbf{ for every } 1\leq i\leq \max\{m,n\}+1
\end{array}
$
\normalsize
\end{minipage}
\\
\hline
quasi-Configuration &
\begin{minipage}{.7\textwidth}
\scriptsize
$
\begin{array}{r@{\hspace{.4em}}c@{\hspace{.4em}}ll}
\TypeT{\alpha}{\beta} &\equiv&
(\AlphT\li\BB\li\alpha)\li\AlphT\li\alpha\li\alpha
\\
\TypeU{\alpha}{\beta} &\equiv&
\AlphT\otimes\TypeT{\alpha}{\beta}\otimes\alpha
\\
\PreConfT 
  &\equiv& \forall \alpha\beta.
      !(\AlphT\li(\BB\li\alpha)\li\alpha)\li
  \\&&\$(\alpha\li\alpha\li
         (
	  (\BB\li\TypeU{\alpha}{\beta})\otimes
	  \StatT\otimes
	  (\BB\li\TypeU{\alpha}{\beta})
	 )
        )
\\
\PreConf{L}{[\seq{L}{1}{m}]}
        {\CStat{s}}
	{R}{[\seq{R}{1}{n}]}
     &\equiv&\bs c.\bs lr.
     \\&&
        \lan
	 \bs y.\lan
	         L,F,
		 \PreTape{\seq{c}{1}{m}}{\seq{L}{1}{m}}{l}
	       \ran
     \\&&
	 ,\CStat{s}
     \\&&,\bs y.\lan
	          R,F,
 		  \PreTape{\seq{L}{1}{n}}{\seq{R}{1}{n}}{r}
	        \ran
	\ran
\\
&& m,n\geq 0, F\in\{\bs cht.t, \bs cht.cht\}
\\
&& \textbf{ and } \seqel{c}{i}= c 
   \textbf{ for every } 1\leq i\leq \max\{m,n\}
\end{array}
$
\normalsize
\end{minipage}
\\ \hline
\end{tabular}}
\normalsize
\end{center}
\caption{Configurations and quasi-Configurations}
\label{table:Quasi-Tapes, configurations and quasi-Configurations}
\end{Table}

\subsubsection{The look-up table}
The look-up table $\Lookup{\gseq{c}}$, where $\gseq{c}$ is the sequence of the free variables of the look-up table itself, must allow to define the transition function $\WALLTrF$ as a coherent extension of $\TrF$ on $\RTape, \LTape$.
We have to think of using $\Lookup{\gseq{c}}$ by applying to it, at time $t$, the symbol $\CAlph{\sigma_h^t}$ under the head, the state $\CStat{s^t}$, and the symbol $\CAlph{\sigma_l^t}$ to the immediate left of the head. Namely, 
$\CAlph{\sigma_l^t}
 (\CStat{s^t}
  (\CAlph{\sigma_h^t}\,\Lookup{\gseq{c}}))$ 
will occur in the definition of $\WALLTrF$.
The types of 
$\CAlph{\sigma_h^t}, \CStat{s^t}, \CAlph{\sigma_l^t}$ (Table~\ref{table:Basic data-types})
suggest that $\Lookup{\gseq{c}}$ be a tuple with $\size{\Sigma}+2$ tuples, each containing $\size{\mathcal S}$ tuples with
$\size{\Sigma}+2$ \textit{triples}
$\LookupEl{\CAlph{\sigma_h^t}}{\CStat{s^t}}{\CAlph{\sigma_l^t}}$. Table~\ref{table:triples} defines a $\LookupEl{\CAlph{\sigma_h^t}}{\CStat{s^t}}{\CAlph{\sigma_l^t}}$, for every combination of $\CAlph{\sigma_h^t}, \CStat{s^t}$, and $\CAlph{\sigma_l^t}$, depending on the value of $\TrF(\sigma_h^t,s^t)$, which yields the move direction, the new character $\sigma^{t+1}$ and the new state $s^{t+1}$. The first three rows of Table~\ref{table:triples} define the triples on the tape symbols of $\Sigma$.
The last row covers the case where the head of $\ol{\TM}$, we are going to define, has passed the right hand border of the represented tape. This is meaningless from the point of view of both $\TrF$, which is undefined, and of the definition of the triples. So, we yield the conventional dummy value 
$\lan \Id, \CStat{s_{a}}, \Id \ran$. The three remaining clauses manage the situations where the head of $\ol{\TM}$ reaches one of the borders of the tape. These cases must be treated coherently with the definition of $\TrF$, suitably extending with the correct symbols the content of the represented tape.
\begin{Table}[ht]
\begin{center}
\small
{\renewcommand{\arraystretch}{1.5}
\begin{tabular}{|c|c|c||c|}
\hline
$\CAlph{\sigma_l^{t}}$ & 
$\CAlph{\sigma_h^{t}}$ & 
$\TrF$ &
$\LookupEl{\CAlph{\sigma_h^t}}
          {\CStat{s^t}}
          {\CAlph{\sigma_l^t}}$
\\
\hline\hline
$\neq\LTape$ & 
$\neq\RTape$ & 
$\TrF(\sigma_h^{t},s^{t}\neq s_{a})
 \equiv(\MoveL,\sigma^{t+1},s^{t+1})$ &
$\lan
  \Id
 ,\CStat{s^{t+1}}
 ,\bs ty.c\,\CAlph{\sigma_l^{t}}
    (\bs y.c'\,\CAlph{\sigma^{t+1}}\,t)
 \ran$
\\\hline
$\neq\LTape$ & 
$\neq\RTape$ & 
$\TrF(\sigma_h^{t},s^{t}\neq s_{a})
 \equiv(\MoveR,\sigma^{t+1},s^{t+1})$ &
$\lan
  \bs ty.c\,\CAlph{\sigma^{t+1}}
    (\bs y.c'\,\CAlph{\sigma_l^{t}}\,t)
 ,\CStat{s^{t+1}}
 ,\Id
 \ran$
\\\hline
\textit{any} & 
\textit{any} & 
$\TrF(\sigma_h^{t},s_{a})
 \equiv(\DoNotMove,\sigma_h^{t},s_{a})$ &
$\lan
  \bs ty.c\,\CAlph{\sigma_l^{t}}\,t
 ,\CStat{s_{a}}
 ,\bs ty.c'\,\CAlph{\sigma_h^{t}}\,t
 \ran$
\\\hline
$\LTape$ & 
\textit{any}&
$\TrF(\sigma_h^{t},s^{t}\neq s_{a})
 \equiv(\MoveL,\sigma^{t+1},s^{t+1})$ &
\begin{tabular}{l}
 $\lan
  \bs ty.c\,\LTape\,t
 ,\CStat{s^{t+1}}
 ,\bs ty.c'\,\CAlph{\sigma_{3}}
    (\bs y.c''\,\CAlph{\sigma^{t+1}}\,t)
 \ran$ 
 \\
 \textbf{where $\sigma_{3}$ is the separator} \textit{blank}
\end{tabular}
\\\hline
\textit{any} & 
$\RTape$     &
$\TrF(\sigma_h^{t},s^{t}\neq s_{a})
 \equiv(\MoveR,\sigma^{t+1},s^{t+1})$ &
$\lan
  \bs ty.c\,\CAlph{\sigma^{t+1}}
    (\bs y.c'\,\CAlph{\sigma_l^{t}}\,t)
 ,\CStat{s^{t+1}}
 ,\bs ty.c''\,\RTape\,t
 \ran$
\\\hline
\textit{any} & 
$\RTape$     &
$\TrF(\sigma_h^{t},s^{t}\neq s_{a})
 \equiv(\MoveL,\sigma^{t+1},s^{t+1})$ &
$\lan
  \Id
 ,\CStat{s^{t+1}}
 ,\bs ty.c\,\CAlph{\sigma_h^{t}}
    (\bs y.c'\,\CAlph{\sigma^{t+1}}(\bs y.c''\,\RTape\,t))
 \ran$
\\\hline
$\RTape$ & 
\textit{any} & 
\textit{undefined} &
$\lan
  \Id
 ,\CStat{s_{a}}
 ,\Id
 \ran$
\\\hline
\end{tabular}
}
\normalsize
\end{center}
\caption{The triples elements of the look-up table}
\label{table:triples}
\end{Table}
Given the triples, we can define $\Lookup{\gseq{c}}$ as in Table~\ref{table:Representing the look-up table}, whose \textbf{proviso (*)} reads as follows: $\gseq{c}$ is the sequence containing the free variables of every $\LookupEl{\sigma}{s}{\sigma'}$, for every combination of $\sigma,s$, and $\sigma'$, given that, 
$\nocc{c}{\LookupEl{\sigma}{s}{\sigma'}}=1$, for
every $c\in\FV{\LookupEl{\sigma}{s}{\sigma'}}$, and given that
the set of the free variables of any two triples be disjoint.
The idea is that
$\CAlph{\sigma_h^t}\,\Lookup{\gseq{c}}$ extracts a ``row''
$\LookupRow{\CAlph{\sigma_h^t}}$, from which 
$\CStat{s^t}\,\LookupRow{\CAlph{\sigma_h^t}}$ gives a ``column''
$\LookupCol{\CAlph{\sigma_h^t}}{\CStat{s^t}}$.
Finally, $\CAlph{\sigma_l^t}\,\LookupCol{\CAlph{\sigma_h^t}}{\CStat{s^t}}$ yields the triple $\LookupEl{\CAlph{\sigma_h^t}}{\CStat{s^t}}{\CAlph{\sigma_l^t}}$.

\begin{Table}[ht]
\begin{center}
\small
{\renewcommand{\arraystretch}{1.3}
\begin{tabular}{|c|}
\hline
$\Lookup{\gseq{c}}$ \textbf{and the combinators that define it}\\
\hline\hline
\begin{minipage}{.85\textwidth}
$
\begin{array}{r@{\hspace{.4em}}c@{\hspace{.4em}}ll}
\LookupCol{\sigma}{s} &\equiv&
\lan
 \LookupEl{\sigma}{s}{\LTape},
 \LookupEl{\sigma}{s}{\sigma_0},
 \ldots,
 \LookupEl{\sigma}{s}{\sigma_{\size{\Sigma}}}
 \LookupEl{\sigma}{s}{\RTape}
\ran
&  \textbf{ with }
   \sigma\in\{\LTape,
             \CAlph{\sigma_1},\ldots,
	        \CAlph{\sigma_{\size{\Sigma}}},
	     \RTape\}
\\&&& \textbf{ and }
   s\in\{\CStat{s_0},\ldots,
	        \CStat{s_{\size{\mathcal S}}}\}
\\
\LookupRow{\sigma} &\equiv&
\lan
 \LookupCol{\sigma}{\CStat{s_0}},
 \ldots,
 \LookupCol{\sigma}{\CStat{s_{\size{S}}}}
\ran
& \textbf{ with }
  \sigma\in\{\LTape,
             \CAlph{\sigma_1},\ldots,
	        \CAlph{\sigma_{\size{\Sigma}}},
	     \RTape\}
\\
\Lookup{\gseq{c}} &\equiv&
\lan
 \LookupRow{\LTape},
 \LookupRow{\CAlph{\sigma_0}},
 \ldots,
 \LookupRow{\CAlph{\sigma_{\size{\Sigma}}}},
 \LookupRow{\RTape}
\ran
& \textbf{ with the proviso (*)}
\end{array}
$
\end{minipage}
\\\hline
\end{tabular}}
\normalsize
\end{center}
\caption{Representing the look-up table}
\label{table:Representing the look-up table}
\end{Table}
\subsubsection{The transition map}
Table~\ref{table:transition map} defines the transition map as the composition of two terms $\CtoP$ and $\PtoC$. The former maps a configuration, at time $t$, to a quasi-configuration. The latter goes in the opposite direction, yielding a new configuration, at time $t+1$, starting from the quasi-configuration. $\CtoP$ iterates the step function $\SCtoP{c}$ on the two sides of the represented tape, starting from the base function $\BCtoP{M}{c}$.
Finally, $\Lookup{\gseq{x}\subs{c}{\gseq{x}}}$ means that, in $\CtoP$, $\Lookup{\gseq{x}}$ is used substituting $c$ for every of the free variables in $\gseq{x}$.

\begin{Table}[h]
\begin{center}
\small
{\renewcommand{\arraystretch}{1.3}
\begin{tabular}{|c|}
\hline
$\WALLTrF$ \textbf{and the combinators that define it}\\
\hline\hline
\begin{minipage}{.85\textwidth}
$
\begin{array}{r@{\hspace{.4em}}c@{\hspace{.4em}}ll}
\BCtoP{M}{x} &\equiv& \lan M,\bs cht.t,x\ran
\\
\SCtoP{c} &\equiv&
\bs eg.(\bs \lan h\,f\,t\ran.
        \lan e, \bs cht.cht, f\,c\,h\,t\ran)(g\,\Id)
\\
\CtoP &\equiv&
\bs nc.
(\bs zlr.z\,\BCtoP{\LTape}{l}\,\BCtoP{\RTape}{r})
(n\,\SCtoP{c})
\\
\PtoC &\equiv&
\bs nc.
\\&&
(\bs zlr.
\\&&\phantom{(}
 (\bs \lan x\,s\,y\ran.
\\&&\phantom{((}
  (\bs \lan e^l\,f^l\,t^l\ran.\bs \lan e^r\,f^r\,t^r\ran.
\\&&\phantom{(((}  
   (\bs \lan h^l\,s\,h^r\ran.\lan
                              h^l(\bs y.t^l),s,h^r(\bs y.t^r)
		   	     \ran
   )(e^l(s(e^h\,\Lookup{\gseq{x}\subs{c}{\gseq{x}}})))
\\&&\phantom{(((}  
  )(x\,\Id)(y\,\Id)
\\&&\phantom{((}  
 )(z\,l\,r)
\\&&\phantom{(}
)(n\,c)
\\
\WALLTrF &\equiv& \bs c.\PtoC(\CtoP\,c)
\\
\end{array}
$
\end{minipage}
\\\hline
\end{tabular}}
\normalsize
\end{center}
\caption{Representing the transition map}
\label{table:transition map}
\end{Table}

\begin{lemma}[Typing and dynamics.]
\begin{enumerate}
\item
$\ta{w}{\AlphT},\ta{x}{\alpha};\emptyset;\emptyset\vdash
\ta{\BCtoP{w}{x}}
   {\TypeU{\alpha}{\beta}}
   $.

\item 
$\ta{c}{\AlphT\li(\BB\li\alpha)\li\alpha};\emptyset;\emptyset
\vdash\ta{\SCtoP{c}}
         {\AlphT\li(\BB\li\TypeU{\alpha}{\beta})
	  \li\TypeU{\alpha}{\beta}}
   $. 
   
\item 
$\emptyset;\emptyset;\emptyset\vdash
\ta{\CtoP}
   {\ConfT\li\PreConfT}
   $, and
$\CtoP\,
\Conf{\seq{L}{1}{m}}
     {\CStat{s}}
     {\seq{R}{1}{n}}
\red^+
\PreConf{\seqel{L}{1}}{[\seq{L}{2}{m}]}
        {\CStat{s}}
	{\seqel{R}{1}}{[\seq{R}{2}{n}]}
$, for every $m,n\geq 1$.

\item 
$\emptyset;\emptyset;\emptyset\vdash
\ta{\PtoC}
   {\PreConfT\li\ConfT}
   $, and
$\PtoC\,
\PreConf{L}{[\seq{L}{1}{m}]}
        {\CStat{s}}
	{R}{[\seq{R}{1}{n}]}
\red^+\\
\Conf{\seq{L'}{1}{m'},\seq{L}{2}{m}}
     {\CStat{s'}}
     {\seq{R'}{1}{n'},\seq{R}{2}{n}}
$, for every $m,n\geq 1$, some $m',n'\geq0$ and $s'$.

\item 
$\emptyset;\emptyset;\emptyset\vdash
\ta{\WALLTrF}
   {\ConfT\li\ConfT}$, and
$\WALLTrF\,
\Conf{\seq{L}{1}{m},\LTape}
     {\CStat{s}}
     {\seq{R}{1}{n},\RTape}
\red^+
\Conf{\seq{L'}{1}{m'},\seq{L}{2}{m}}
     {\CStat{s'}}
     {\seq{R'}{1}{n'},\seq{R}{2}{n}}
$, for every $m,n\geq 1$, some $m',n'\geq0$ and $s'$.
\end{enumerate}
\end{lemma}
\begin{proof}
The typing of $\BCtoP{w}{x}$ is simple.
The key judgments to type $\SCtoP{c}$ are:
\small
\begin{eqnarray*}
&&
\ta{g}{\BB\li\TypeU{\alpha}{\beta}};\emptyset;\emptyset
\vdash\ta{g\,\Id}
         {\TypeU{\alpha}{\beta}}
\\
&&
\ta{c}{\AlphT\li(\BB\li\alpha)\li\alpha},
\ta{e}{\AlphT};\emptyset;\emptyset
\vdash\ta{\bs \lan h\,f\,t\ran.
          \lan e,\bs cht.cht,fcht\ran}
         {\TypeU{\alpha}{\beta}\li\TypeU{\alpha}{\beta}}
\\
&&
\ta{f}{\TypeT{\alpha}{\beta}},
\ta{c}{\AlphT\li(\BB\li\alpha)\li\alpha},
\ta{h}{\AlphT},
\ta{t}{\alpha}
;\emptyset;\emptyset
\vdash\ta{fcht}
         {\alpha}
\end{eqnarray*}
\normalsize

\par
The key judgments to type $\CtoP$ are:
\small
\begin{eqnarray*}
&&
\ta{n}{\ConfT};\emptyset;
\{(\emptyset;\ta{c}
                {\AlphT\li(\BB\li\alpha)\li\alpha})\}
\vdash
\ta{n\,\SCtoP{c}}
   {\$(\TypeU{\alpha}{\beta}\li\TypeU{\alpha}{\beta}\li
       ((\BB\li\TypeU{\alpha}{\beta})\otimes
        \StatT\otimes
	(\BB\li\TypeU{\alpha}{\beta})))}
\\
&&
\emptyset;\emptyset;\emptyset\vdash
\ta{\bs zlr.z\,\BCtoP{\LTape}{l}\,\BCtoP{\RTape}{r}}
   {\begin{array}[t]{l}
    \$(\TypeU{\alpha}{\beta}\li\TypeU{\alpha}{\beta}\li
       ((\BB\li\TypeU{\alpha}{\beta})\otimes
        \StatT\otimes
	(\BB\li\TypeU{\alpha}{\beta})))
    \li\\
    \qquad\qquad\qquad
    \$(\alpha\li\alpha\li
        ((\BB\li\TypeU{\alpha}{\beta})\otimes
         \StatT\otimes
 	 (\BB\li\TypeU{\alpha}{\beta})))
    \end{array}
   }
\end{eqnarray*}
\normalsize
For the dynamics, just apply the definitions.   

\par
The key judgments to type $\PtoC$ are:
\small
\begin{eqnarray*}
&&
\Gamma,
\ta{e^l}{\AlphT},
\ta{s}{\StatT},
\ta{e^h}{\AlphT};
\emptyset;\emptyset
\vdash
\ta{e^l(s(e^h\,\Lookup{\gseq{x}}))}
   {((\BB\li\alpha)\li\BB\li\alpha)\otimes
    \StatT\otimes
    ((\BB\li\alpha)\li\BB\li\alpha)}
\\
&&
\ta{t^l}{\alpha},
\ta{t^h}{\alpha};
\emptyset;\emptyset
\vdash\!\!
\begin{array}[t]{l}
\bs \lan h^l\,s\,h^r\ran.
    \lan
     h^l(\bs y.t^l),s,h^r(\bs y.t^r)
    \ran
\!:\!
\\\qquad\qquad
     (((\BB\li\alpha)\li\BB\li\alpha)\otimes
        \StatT\otimes
       ((\BB\li\alpha)\li\BB\li\alpha))
     \li
\\\qquad\qquad\qquad\qquad\qquad\qquad\qquad\qquad\qquad\qquad
     ((\BB\li\alpha)\otimes
      \StatT\otimes
      (\BB\li\alpha))
\end{array}
\\
&&
\ta{n}{\PreConfT};\emptyset;
\{(\emptyset;\ta{c}
                {\AlphT\li(\BB\li\alpha)\li\alpha})\}
\vdash
\ta{n\,c}
   {\$(\alpha\li\alpha\li
       ((\BB\li\TypeU{\alpha}{\beta})\otimes
        \StatT\otimes
	(\BB\li\TypeU{\alpha}{\beta})))}
\\
&&
\ta{z}{\alpha\li\alpha\li
        ((\BB\li\TypeU{\alpha}{\beta})\otimes
         \StatT\otimes
 	 (\BB\li\TypeU{\alpha}{\beta}))},
\ta{l}{\alpha},
\ta{r}{\alpha};\emptyset;\emptyset
\vdash
\ta{z\,l\,r}
   {(\BB\li\TypeU{\alpha}{\beta})\otimes
    \StatT\otimes
    (\BB\li\TypeU{\alpha}{\beta})
   }
\end{eqnarray*}
\normalsize
where $\ta{x_i}{\AlphT\li(\BB\li\alpha)\li\alpha}\in\Gamma$, for every $x_i\in\FV{\Lookup{\gseq{x}}}$.
For the dynamics just apply the definitions, and observe that 
$m',n'$ may vary, depending on the direction move of the head.
Finally, the typing of $\WALLTrF$. It is a trivial composition of the typing above, and for its dynamics just apply the definitions.
\end{proof}
\subsection{Encoding of a poly-time Turing machine}
Table~\ref{table:Turing machine} defines the term $\WALLTM$ that represents the poly-time Turing machine $\TM$.
\begin{Table}[h]
\begin{center}
\small
{\renewcommand{\arraystretch}{1.3}
\begin{tabular}{|c|}
\hline
$\WALLTM$ \textbf{and the combinators that define it}\\
\hline\hline
\begin{minipage}{.85\textwidth}
$
\begin{array}{r@{\hspace{.4em}}c@{\hspace{.4em}}ll}
\LtoC &\equiv&
\bs lc.
(\bs zlr.
 \lan
  \bs y.c\,\LTape(\bs y.l),
  \CStat{s_0},
  z(c\,\RTape(\bs y.r))
 \ran
)(l\,c)
\\
\LtoN &\equiv&
\bs lf.
(\bs zx.z\,x\,\Id)(l(\bs et.f(t\,\Id)))
\\
\WALLTM &\equiv&
\bs x.
\LEmbed{1}{1}
      {
       \bs\elan z_p,z_c\eran.M_{z_p,z_c}
      }
      (\LDiagE[\ACoerce,\ADiagE]\,x)
\\
M_{z_p,z_c}&\equiv&
(\bs xy. (\bs i\,l.(\bs z.z\,(\LtoC\,l))(i\,\WALLTrF))\,x\,y)
       (\NPoly{e}{K}(\BEmbed{1}
                            {\LtoN}\,z_p
		     )
       )
       (\BEmbed{1}{1}
              {\MLCoerce^{4e}[\ACoerce]}\,z_c
       )
\end{array}
$
\end{minipage}
\\\hline
\end{tabular}}
\normalsize
\end{center}
\caption{Representing the Turing machine}
\label{table:Turing machine}
\end{Table}
$\WALLTM$ requires to duplicate its input $x$, that represents a given input tape of $\TM$. $\LDiagE[\ACoerce,\ADiagE]$ duplicates the instance of $x$. $\NPoly{e}{K}(\BEmbed{1}{1}{\LtoN}\,z_p)$ is in charge of using the copy $z_p$ of $x$ to obtain the length of the computation, represented as a Church numeral $\UNum{n}$, for some $n$.
The second copy $z_c$ of $x$, once embedded into a suitable number of $\$$-boxes, is transformed into an initial configuration, by 
$\LtoC$. Finally, $\UNum{n}$ iterates the transition function $\WALLTrF$, and the result is applied to the initial configuration.

\begin{lemma}[Typing and Dynamics.]
\label{lemma:TM Typing and Dynamics}
\begin{enumerate}
\item 
$\emptyset;\emptyset;\emptyset\vdash
\ta{\LtoC}{\ListT\AlphT\li\ConfT}$ and
$\LtoC\,[\seq{\CAlph{\sigma}}{1}{n}]\red^+\\
\Conf{\LTape}{\CStat{s_0}}{\seq{\CAlph{\sigma}}{1}{n},\RTape}$, for every $n\geq 0$.

\item 
$\emptyset;\emptyset;\emptyset\vdash
\ta{\LtoN}{\ListT\AlphT\li\UIntT}$ and
$\LtoN\,[\seq{\CAlph{\sigma}}{1}{n}]\red^+
\UNum{n}$, for every $n\geq 0$.

\item 
$\emptyset;\emptyset;\emptyset\vdash
\ta{\WALLTM}{\ListT\AlphT \li \$^{4e+1}\ConfT}$ and
$\WALLTM\,[\seq{\CAlph{\sigma}}{1}{n}]\red^+
\Conf{\seq{\CAlph{\sigma^l}}{1}{p},\LTape}
     {\CStat{s_a}}
     {\seq{\CAlph{\sigma^r}}{1}{q},\RTape}$, 
for every $n\geq 0$, and some sequences $\seq{\CAlph{\sigma^l}}{1}{p}, \seq{\CAlph{\sigma^r}}{1}{q}$, with $p, q\geq 0$.
\end{enumerate}
\end{lemma}
For the proof, the key judgments to type $\LtoC$ are:
\small
\begin{align*}
&
\emptyset;
\emptyset;
\{(\emptyset;\{\ta{c_l}{\AlphT\li(\BB\li\alpha)\li\alpha},
               \ta{c_r}{\AlphT\li(\BB\li\alpha)\li\alpha}
	       \})\}
\vdash
\\
&
\qquad\qquad\qquad
\bs zlr.
 \lan
  \bs y.c_l\,\LTape(\bs y.l),
  \CStat{s_0},
  z(c_r\,\RTape(\bs y.r))
 \ran
\!:\!
\\
&
\qquad\qquad\qquad\qquad\qquad\qquad
\$(\alpha\li\BB\li\alpha
  )\li
\$(\alpha\li\alpha\li
   ((\BB\li\alpha)\otimes
    \StatT\otimes
    (\BB\li\alpha)
   )
  )
\\
&
\ta{l}{\ListT\AlphT};
\emptyset;
\{(\emptyset;\ta{c'}{\AlphT\li(\BB\li\alpha)\li\alpha})\}
\vdash
\ta{l\,c'}{\$(\alpha\li\BB\li\alpha)}
\end{align*}
\normalsize
\par
The key judgments to type $\LtoN$ are:
\small
\begin{align*}
&
\emptyset;
\emptyset;
\emptyset\vdash
\ta{\bs zx.z\,x\,\Id}{\$(\alpha\li\BB\li\alpha)\li
                      \$(\alpha\li\alpha)}
\\
&
\ta{l}{\ListT\AlphT};
\emptyset;
\{(\emptyset;\ta{f}{\alpha\li\alpha})\}
\vdash
\ta{l(\bs et.f(t\,\Id))}
   {\$(\alpha\li\BB\li\alpha)}
\end{align*}
\normalsize
The key judgments to type $\WALLTM$ are:
\small
\begin{align*}
&
\emptyset;\emptyset;\emptyset\vdash
\ta{
\bs xy.(\bs i\,l.(\bs z.z\,(\LtoC\,l))(i\,\WALLTrF))\,x\,y
}{\$^{4e}\UIntT\liv\$^{4e+1}\ListT\AlphT\li\$^{4e+1}\ConfT}
\\
&
\emptyset;\emptyset;\{(\ta{z_c}{\ListT\AlphT};\emptyset)\}
\vdash
\ta{\BEmbed{1}
          {\MLCoerce^{4e}[\ACoerce]}\,z_c}
   {\$^{4e+1}\ListT\AlphT}
\\
&
\emptyset;\emptyset;\{(\ta{z_p}{\ListT\AlphT};\emptyset)\}
\vdash
\ta{\NPoly{e}{K}(\BEmbed{1}{\LtoN}\,z_p)}
   {\$^{4e}\UIntT}
\\
&
\ta{x}{\ListT\AlphT};\emptyset;\emptyset
\vdash
\ta{\LDiagE[\ACoerce,\ADiagE]\,x}
   {\$(\$\ListT\AlphT\odot\$\ListT\AlphT)}
\end{align*}
\normalsize
For the dynamics, just apply the definitions.
\par
So, using the initial assumptions on the poly-time Turing machines we want to simulate, and consistently extending $\TrF$ to both $\LTape$, and $\RTape$ by $\WALLTrF$, and applying Lemma~\ref{lemma:TM Typing and Dynamics}, we can state:
\begin{quote}
\textit{If $\TM$, applied to the input tape $\seq{\sigma}{1}{n}$, for some $n\geq 0$, produces the output portion $\seq{\sigma'}{1}{p}$ of the tape, then $\WALLTM\,[\seq{\CAlph{\sigma}}{1}{n}]$ simulates it.
Namely, $\WALLTM\,[\seq{\CAlph{\sigma}}{1}{n}]$ evaluates to
$\Conf{\seq{\CAlph{\sigma^l}}{1}{m},\LTape}
     {\CStat{s_a}}
     {\seq{\CAlph{\sigma'}}{1}{p},
      \seq{\CAlph{\sigma^r}}{1}{q},\RTape}$,
for some sequences $\seq{\CAlph{\sigma^l}}{1}{m}, \seq{\CAlph{\sigma^r}}{1}{q}$, with $m, p, q\geq 0$.}
\end{quote}
which implies Theorem~\ref{theorem:poly-time-completeness}.
\section{Details about the proofs}
\label{section:Details about the proofs}
\paragraph*{Proof of Lemma~\ref{lemma:structural-properties-WALL} (Structural properties)}
\begin{description}
\item [Point~\ref{lemma:structural-properties-WALL-2}  of Lemma~\ref{lemma:structural-properties-WALL}.]
We proceed by structural induction on $\Pi$, which can be written as $\Pi_M(R)$. 
The base cases occur with $R\in\{A, \$, !\}$, $n=1$, and $q_1=0$.
We focus on the inductive steps, starting with $\Pi_M$ being
$\Pi_{M\{^{x}/_{z}\,^{x}/_{y}\}}(C,\Pi'_{M})\rhd
\Gamma;\Delta;
{\mathcal E}\sqcup\{(\Theta_z,\Theta_y;\ta{x}{A})\}
\vdash\ta{M\{^{x}/_{z}\,^{x}/_{y}\}}{B}$.
\begin{enumerate}
\item 
Let us call ${\mathcal E}_y$ the set ${\mathcal E}\sqcup\{(\Theta_y;\ta{y}{A})\}$.
By induction on 
$\Pi'_{M}\rhd
\Gamma;\Delta;
{\mathcal E}_y,
(\Theta_z;\ta{z}{A})\vdash\ta{M}{B}$, there are $n\geq 1$ and $q_1,\ldots,q_n\geq 0$ such that $\wdth{1}{\Pi'_{M}}\geq q_1+\ldots+q_n$ and:
(i) $M$ can be written as 
$M'\subs{z}{z^1_1\ldots z^1_{q_1}\ldots\ldots z^n_1\ldots z^n_{q_n}}$ for some $M'$;
(ii) for every $1\leq i\leq n$, there is
$\Pi^z_i(R_i)\rhd
 \Gamma^z_i;
 \Delta^z_i;
 {\mathcal E}^z_i,
 \{(\Theta^i_1;\ta{z^i_1}{A})\},\ldots,
 \{(\Theta^i_{q_i};\ta{z^i_{q_i}}{A})\}
 \vdash\ta{P^z_i}{C^z_i}$, subdeduction of $\Pi'_{M}$, with $R_i\in\{A, \$, !\}$, that introduce
$\ta{z^i_1}{A},\ldots,\ta{z^i_{q_i}}{A}$;
(iii) $q_1+\ldots+q_n-1$ instances of $C, \li E, \li E_!, \liv E$ in the tree with the conclusion of $\Pi'_{M}$ as root and the conclusions of every $\Pi^z_i$ as leaves, are required to contract $z^1_1\ldots z^1_{q_1}\ldots\ldots z^n_1\ldots z^n_{q_n}$ to $z$.

\item 
Now, let us call ${\mathcal E}_z$ the set ${\mathcal E}\sqcup\{(\Theta_z;\ta{z}{A})\}$,
and, proceed again by induction on $\Pi'_M$ that we see as
$\Pi'_{M'\subs{z}{z^1_1\ldots z^1_{q_1}\ldots\ldots z^n_1\ldots z^n_{q_n}}}
\rhd
\Gamma;\Delta;
{\mathcal E}_z,
(\Theta_y;\ta{y}{A})
\vdash\ta{M'\subs{z}{z^1_1\ldots z^1_{q_1}
                     \ldots\ldots 
                     z^n_1\ldots z^n_{q_n}}}{B}$.
There are $m\geq 1$, and $p_1,\ldots,p_m\geq 0$ such that
$\wdth{1}{\Pi_{M'\subs{z}{z_1\ldots z_{q_z}}}}\geq p_1+\ldots+p_m$ and:
(i) $M'\subs{z}{z^1_1\ldots z^1_{q_1}\ldots\ldots z^n_1\ldots z^n_{q_n}}$
can be written as
$M''\subs{z}{z^1_1\ldots z^1_{q_1}\ldots\ldots z^n_1\ldots z^n_{q_n}}
    \subs{y}{y^1_1\ldots y^1_{p_1}\ldots\ldots y^m_1\ldots y^m_{p_m}}
$,
for some $M''$;
(ii) for every $1\leq j\leq m$, there is
$\Pi^y_j(R_j)
\rhd
\Gamma^y_j;
\Delta^y_j;
{\mathcal E}^y_j,
\{(\Xi^j_1;\ta{y^j_1}{A})\},
\ldots,
\{(\Xi^j_{p_j};$ $\ta{y^j_{p_j}}{A})\}
\vdash\ta{P^y_j}{C^y_j}$, subdeduction of $\Pi_{M}$, with $R_j\in\{A, \$, !\}$,
that introduce $\ta{y^j_1}{A},\ldots,\ta{y^j_{p_j}}{A}$;
(iii) $p_1+\ldots+p_m-1$ instances of $C, \li E, \li E_!, \liv E$ in the tree with the conclusion of 
$\Pi'_{M'\subs{z}{z^1_1\ldots z^1_{q_1}\ldots\ldots z^n_1\ldots z^n_{q_n}}}$ 
as root and the conclusions of every $\Pi^y_j$ as leaves, are required to contract 
$y^1_1\ldots y^1_{p_1}\ldots\ldots y^m_1\ldots y^m_{p_m}$ to $y$.
\end{enumerate}
So, we have $n+m$ subdeductions of $\Pi_{M\{^{x}/_{z}\,^{x}/_{y}\}}$ in which we can count
$q_1+\ldots+q_n-1+p_1+\ldots+p_m-1+1=q_1+\ldots+q_n+p_1+\ldots+p_m-1$ instances of $C, \li E, \li E_!, \liv E$
in the tree, say $\tau$, with the conclusion of $\Pi_{M\{^{x}/_{z}\,^{x}/_{y}\}}(C)$ as root and the conclusions of all
$\Pi^z_j$s, and $\Pi^y_j$s as leaves that contract
$z^1_1\ldots z^1_{q_1}\ldots\ldots z^n_1\ldots z^n_{q_n}$ to $z$,
$y^1_1\ldots y^1_{p_1}\ldots\ldots y^m_1\ldots y^m_{p_m}$ to $y$, and
$z, y$ to $x$.
The reason why $q_1+\ldots+q_n+p_1+\ldots+p_m\leq \wdth{1}{\Pi_M}$ follows from the definition of width at depth 1 that, applied to $\Pi_M$ counts at least the occurrences of $C, \li E, \li E_!, \liv E$ in $\tau$ here above. The resulting value cannot exceed the total number of instances of $C, \li E, \li E_!, \liv E$ at depth 1.
\par
With $\Pi_{PQ}(\li E)$ and $\Pi_{PQ}(\li E_!)$ we can proceed analogously. All the other cases, $\Pi_{PQ}(\liv E)$ included, routinely apply the induction.

\item [Point~\ref{lemma:structural-properties-WALL-3.1} of Lemma~\ref{lemma:structural-properties-WALL}.]
As a \textbf{first case}, let $M$ be $x$. 
Assuming that $R\in\{\li E, \li E_!, \liv E, \li I, \li I_{\$}, \li I_{!}, \liv I\}$,
contradicts the hypothesis that the subject of the conclusion of $\Pi$ be $x$.
Assuming that
$R\in\{\forall I, \$\}$
contradicts the hypothesis that the type of $x$ be $!A$.
So, $R$ can be $A$. However, it cannot be of the form 
$\Gamma,\ta{x}{\ !A};\Delta;{\mathcal E}\vdash\ta{x}{\ !A}$ because, by definition, illegal.
So, $R$ can only be one among $!,C$, and $\forall E$.
If $R$ is $!$, then $\Pi(!)$ is:
\small
\[
\infer[!]
{\Gamma';\Delta;\{(\$\Theta';\emptyset)\}\sqcup\{(\Theta;\ta{x}{A})\}
 \vdash\ta{x}{\,!A}
}
{\begin{array}{l}
 \Gamma,\ta{x}{A};\emptyset;\{(\Theta';\emptyset)\}
 \vdash\ta{x}{A}
 \\
 \{\ta{x}{A}\}\cup\Gamma\subseteq \Theta\cup\Phi
 \\
 \Theta\neq\emptyset \Rightarrow \dom{\Phi}\cap\FV{x}=\{x\}
 \end{array}
}
\]
\normalsize
Namely, $x$ forcefully belongs to $\operatorname{Dom}(\Phi)$. It follows that the presence of a single polynomial assumption excludes that we can apply $C$ below the given instance of the rule $!$

\par
Only the case $R\equiv\forall E$, is left, but we shall see that this is impossible. Indeed, by definition, the modality in front of $A$ could not be introduced by the rule $\forall E$ and it had to be present before the substitution on types takes place. Namely, we should be starting with:
\small
\[
\infer[A]
{\ta{x}{\forall \alpha.!B},\Gamma;\Delta;{\mathcal E}
 \vdash\ta{x}{\forall \alpha.!B}
}
{}
\]
\normalsize
 which is an illegal instance of the axiom.
\par
As a \textbf{second case}, let $M$ be $\bs x.N$. We observe that $R\in\{A, \li E, \li E_!, \liv E\}$ contradicts the hypothesis that the subject of the conclusion of $\Pi$ be $\bs x.N$. Moreover, assuming that
$R\in\{\li I, \li I_{\$}, \li I_{!}, \liv I, \forall I, \$\}$
contradicts the hypothesis that the type of $\bs x.N$ be $!A$.
So, $R$ can only be one between $!, C$, and $\forall E$.
If $R$ is $!$, then $\Pi(!)$ is:
\small
\[
\infer[!]
{\Gamma;\Delta;\{(\$\Theta';\emptyset)\}\sqcup\{(\Theta;\Phi)\}
 \vdash\ta{\bs x.N}{\,!A}
}
{\Gamma';\emptyset;\{(\Theta';\emptyset)\}
 \vdash\ta{\bs x.N}{A}
 & \Gamma'\subseteq \Theta\cup\Phi
 & \Theta\neq\emptyset \Rightarrow \dom{\Phi}\cap\FV{M}\neq\emptyset
}
\]
\normalsize
with
$
\FV{\bs x.N}
\subseteq
\operatorname{Dom}(\Gamma')
\subseteq
\operatorname{Dom}(\Theta)\cup\operatorname{Dom}(\Phi)
$.
\par
In the case $R$ be equal to $C$, since the subject is a $\lambda$-abstraction with type $!A$, the modality in front of $A$ must be introduced by an instance of the rule $!$ that, given our assumptions, cannot be followed by any rule but a sequence of instances of the rule $C$. The situation can be summarized as follows:
\small
\[
\infer[C]
{
 \Gamma;\Delta;
 {\mathcal E}
 \vdash\ta{\bs x.N}{\,!A}
}
{
 \infer[C]
 {\vdots
 }
 {
  \infer[!]
  {
   \Gamma;\Delta;
   \{(\$\Theta';\emptyset)\}\sqcup
   \{(\Theta;\Phi)\}
   \vdash\ta{\bs x.N}{\,!A}
  }
  {\Gamma';\emptyset;\{(\Theta';\emptyset)\}
   \vdash\ta{\bs x.N}{A}
   &
   \Gamma'
   \subseteq
   \Theta\cup\Phi
   &
   \Theta\neq\emptyset \Rightarrow \dom{\Phi}\cap\FV{M}\neq\emptyset
  }
 }
}
\]
\normalsize
But $\Phi$ is either empty or a singleton, hence, the sequences of instances of $C$ is empty.
\par
Only the case $R\equiv\forall E$, is left, but we shall see that this is impossible. Indeed, by definition, the modality in front of $A$ could not be introduced by the rule $\forall E$ and it had to be present before the substitution on types takes place. Namely, the type of $\bs x.N$ prior to the substitution would have form $\forall \alpha.!B$, for some $B$. But this would contradict the definition of the rule $\forall E$.
\par
To sum up, we can conclude that the only admissible conclusion is the rule $!$.

\item [Point~\ref{lemma:structural-properties-WALL-3.3} of Lemma~\ref{lemma:structural-properties-WALL}.]
Proceed in analogy to the proof of point~\ref{lemma:structural-properties-WALL-3.1} here above.
\end{description}
\paragraph*{Proof of Lemma~\ref{lemma:subst-vs-wght}
(Substitution property)}
\begin{description}
\item 
[Point~\ref{lemma:subst-vs-wght-01} of Lemma~\ref{lemma:subst-vs-wght}.] 
Notice that the hypothesis $x\in\FV{M}$ excludes $\Pi_{M}(R)$, with $R\in\{!,\$\}$, since, in those cases, $x\not\in\FV{M}$.
We detail out a couple of points, proceeding by induction on $\Pi_M$.
\par
As a \textbf{first} case, let $\Pi_M$ be such that
$\Pi_M(A)
 \rhd\Gamma_x,\ta{x}{L};\Delta_x;{\mathcal E}_x
 \vdash\ta{x}{L}$.
Then, $\Pi_{x\subs{N}{x}}$ coincides to 
$\Pi'_N
 \rhd\Gamma_x,\Gamma_N;
     \Delta_x,\Delta_N;
     {\mathcal E}_x\sqcup{\mathcal E}_N
 \vdash\ta{N}{L}$, which is obtained from $\Pi_N$, by 
Lemma~\ref{lemma:structural-properties-WALL}, point~\ref{lemma:structural-properties-WALL-1}.
The same point of the same lemma implies the following statements.
\par
Subpoint~\textbf{\ref{lemma:subst-vs-wght-01-b}} of 
\textbf{Point~\ref{lemma:subst-vs-wght-01}} is
$\dpth{\Pi_{x\subs{N}{x}}}=\dpth{\Pi'_{N}}=
\max\{0,\dpth{\Pi_{N}}\}=$ 
\\$\max\{\dpth{\Pi_x},$ $\dpth{\Pi_{N}}\}$.
\par
Subpoint~\textbf{\ref{lemma:subst-vs-wght-01-a}} of 
\textbf{Point~\ref{lemma:subst-vs-wght-01}} is
$\wdth{d}{\Pi_{x\subs{N}{x}}}=\wdth{d}{\Pi'_{N}}=\wdth{d}{\Pi_{N}}=0+\wdth{d}{\Pi_{N}}$
\\$=\wdth{d}{\Pi_{x}}+\wdth{d}{\Pi_{N}}$, 
for every $d\geq 0$.
\par
Subpoint~\textbf{\ref{lemma:subst-vs-wght-01-c}} of 
\textbf{Point~\ref{lemma:subst-vs-wght-01}} is
$\psz{0}{\Pi_{x\subs{N}{x}}}=\psz{0}{\Pi'_{N}}<
1+\psz{0}{\Pi_{N}}=\psz{0}{\Pi_{x}}+\psz{0}{\Pi_{N}}$.
\par
Subpoint~\textbf{\ref{lemma:subst-vs-wght-01-d}} of 
\textbf{Point~\ref{lemma:subst-vs-wght-01}} is
$\psz{d}{\Pi_{x\subs{N}{x}}}=\psz{d}{\Pi'_{N}}=
0+\psz{d}{\Pi_{N}}=\psz{d}{\Pi_{x}}+\psz{d}{\Pi_{N}}$,
for every $1\leq d$.
\par
As a \textbf{second} case 
let $\Pi_M$ be such that
$\Pi_M(\li E,\Pi_P,\Pi_Q)
 \rhd\Gamma_P,\ta{x}{L},\Gamma_Q;
     \Delta_P,$ $\Delta_Q;
     {\mathcal E}_P\sqcup{\mathcal E}_Q
 \vdash\ta{PQ}{B}$.
Either $x\in\FV{Q}$, or $x\in\FV{P}$. Let us assume $x\in\FV{P}$, the other case being symmetric.
By induction, there exists
$\Pi_{P\subs{N}{x}}
 \rhd\Gamma_P,\Gamma_N;
     \Delta_P,\Delta_N;
     {\mathcal E}_P\sqcup{\mathcal E}_N
\vdash\ta{P\subs{N}{x}}{C\li B}$, that we can use as a principal premise of an instance of $\li E$, whose secondary premise be
$\Pi_{Q}
 \rhd\Gamma_Q;
     \Delta_Q;
     {\mathcal E}_Q
\vdash\ta{Q}{C}$. So,
$\Pi_{(PQ)\subs{N}{x}}
 \rhd\Gamma_P,\Gamma_N,\Gamma_Q;
     \Delta_P,\Delta_N,\Delta_Q;
     {\mathcal E}_P\sqcup{\mathcal E}_N\sqcup{\mathcal E}_Q
\vdash\ta{(PQ)\subs{N}{x}}{B}$.
\par
Subpoint~\textbf{\ref{lemma:subst-vs-wght-01-b}} of  \textbf{Point~\ref{lemma:subst-vs-wght-01}} is
$\dpth{\Pi_{(PQ)\subs{N}{x}}}
=\max\{\dpth{\Pi_{P\subs{N}{x}}},\wdth{}{\Pi_{Q}}\}
=\max\{\dpth{\Pi_{P}},
       \dpth{\Pi_{N}},$ $
       \dpth{\Pi_{Q}}\}
=\max\{\dpth{\Pi_{PQ}}
      ,\dpth{\Pi_{N}}\}$,
using the inductive hypothesis.
\par
Subpoint~\textbf{\ref{lemma:subst-vs-wght-01-a}} of \textbf{Point~\ref{lemma:subst-vs-wght-01}} requires two cases.
The statement holds for
\\ 
$\wdth{0}{\Pi_{(PQ)\subs{N}{x}}}$ because $\wdth{0}{\Pi}=0$, for every $\Pi$.
It also holds with $d=1$, since 
$\wdth{1}{\Pi_{(PQ)\subs{N}{x}}}
=\wdth{1}{\Pi_{P\subs{N}{x}}}+\wdth{1}{\Pi_{Q}}+1
= \wdth{1}{\Pi_{P}}+\wdth{1}{\Pi_{Q}}+\wdth{1}{\Pi_{N}}+1
=\wdth{1}{\Pi_{PQ}}+\wdth{1}{\Pi_{N}}$, using the inductive hypothesis. For every $d>1$, it holds not counting the application.
\par
Subpoint~\textbf{\ref{lemma:subst-vs-wght-01-c}} of \textbf{Point~\ref{lemma:subst-vs-wght-01}} is
$\psz{0}{\Pi_{(PQ)\subs{N}{x}}}
=\psz{0}{\Pi_{P\subs{N}{x}}}
+\psz{0}{\Pi_{Q}}+1
<\psz{0}{\Pi_{P}}
+\psz{0}{\Pi_{N}}
+\psz{0}{\Pi_{Q}}+1
=\psz{0}{\Pi_{PQ}}
+\psz{0}{\Pi_{N}}$,
using the inductive hypothesis.
\par
Subpoint~\textbf{\ref{lemma:subst-vs-wght-01-d}} of \textbf{Point~\ref{lemma:subst-vs-wght-01}} is
$\psz{d}{\Pi_{(PQ)\subs{N}{x}}}
=\psz{d}{\Pi_{P\subs{N}{x}}}
+\psz{d}{\Pi_{Q}}+1
=\psz{d}{\Pi_{P}}
+\psz{d}{\Pi_{N}}
+\psz{d}{\Pi_{Q}}+1
=\psz{d}{\Pi_{PQ}}
+\psz{d}{\Pi_{N}}$, for every $d\geq 1$,
using the inductive hypothesis.
It is enough to proceed analogously in the cases $\Pi_M(\li E_!), \Pi_M(\liv E)$, where, in particular, $x$ can be free only in its principal premise;
the proof is simpler with $\Pi_M(R)$ whose $R$ has a single premise.

\item [Point~\ref{lemma:subst-vs-wght-03} of Lemma~\ref{lemma:subst-vs-wght}.]

Notice that the hypothesis $x\in\FV{M}$ excludes $\Pi_{M}(R)$, with $R\in\{A,!\}$, since, in those cases, $x\not\in\FV{M}$.
Then we proceed by structural induction on $\Pi_M$.

\par
The \textbf{first} case is with $\Pi_M(\$)$ concluding by:
\small
$$\infer[\$]
{\Gamma';
 \$\Delta'_{M},\Delta_M,\ta{x}{L};
 \{(\$\Theta'_M;\emptyset)\}\sqcup
 \{(\Theta_1;\Phi_1)\}\sqcup\ldots\sqcup\{(\Theta_m;\Phi_m)\}
 \vdash\ta{M}{\$B}
}
{\begin{array}[t]{l}
\Pi'_M\rhd
\Gamma_M,\ta{x}{L};\Delta'_{M};\{(\Theta'_M;\emptyset)\}\vdash\ta{M}{B}
\\
\Gamma_M,\ta{x}{L}\subseteq
\Delta_{M}\cup\{\ta{x}{L}\}\cup
\bigcup^{m}_{i=1}\Theta_i\cup
\bigcup^{m}_{i=1}\Phi_i
\\
\Theta_i\neq\emptyset\Leftrightarrow\Phi_i=\emptyset
\end{array}
}$$
\normalsize
\par
Point~\ref{lemma:structural-properties-WALL-3.3} of
Lemma~\ref{lemma:structural-properties-WALL}, applied to
$\Pi_N\rhd
\Gamma_N;\Delta_N;{\mathcal E}_N\vdash\ta{N}{\$L}$, requires to focus on \textbf{two cases}:
\begin{itemize}
\item 
Let $\Pi_N(\$)$ be:
\small
$$
\infer[\$]
{
\Gamma'_N;\$\Delta'_N,\Delta_N;
\{(\$\Theta'_N;\emptyset)\}\sqcup
\{(\Theta'_1;\Phi'_1)\}\sqcup\ldots\sqcup\{(\Theta'_n;\Phi'_n)\}
\vdash\ta{N}{\$L}
}
{\begin{array}{l}
 \Pi'_N\rhd
 \Gamma_N;
 \Delta'_N;
 \{(\Theta'_N;\emptyset)\}
 \vdash\ta{N}{L}
 \\
 \Gamma_N\subseteq
 \Delta_N\cup\bigcup_{i=1}^{n}\Theta'_i\cup\bigcup_{i=1}^{n}\Phi'_i
 \\
 \Theta'_i\neq\emptyset \Leftrightarrow \Phi'_i=\emptyset
 \end{array}
}
$$ 
\normalsize
Using point~\ref{lemma:subst-vs-wght-01} of
Lemma~\ref{lemma:subst-vs-wght} on
$\Pi'_M$, and $\Pi'_N$, we get
$\Pi'_{M\subs{N}{x}}
\rhd
\Gamma_M,\Gamma_N;
\Delta'_{M},$
\\$\Delta'_{N};
\{(\Theta'_M;\emptyset)\}\sqcup
\{(\Theta'_N;\emptyset)\}
\vdash\ta{M\subs{N}{x}}{B}$, that allows to conclude by an instance of $\$$, yielding $\bar{\Pi}_{M\subs{N}{x}}$ with the right conclusion.

\item 
Let $\Pi_N(C)\rhd\Gamma''_N;\Delta''_N;{\mathcal E}''_N\vdash\ta{N}{\$L}$. Iteratively applying Lemma~\ref{lemma:structural-properties-WALL}, point~\ref{lemma:structural-properties-WALL-3.3}, we eventually prove the existence of 
$\Pi''_N(\$)\rhd
\Gamma'_N;\$\Delta'_N,\Delta_N;
\{(\$\Theta'_N;\emptyset)\}\sqcup
\{(\Theta'_1;\Phi'_1)\}\sqcup\ldots\sqcup\{(\Theta'_n;\Phi'_n)\}
\vdash\ta{N}{\$L}$, which, as in the previous case, we assume to have a premise derived from $\Pi'_N$. So, we use point~\ref{lemma:subst-vs-wght-01} of
Lemma~\ref{lemma:subst-vs-wght} on
on $\Pi'_M$ and $\Pi'_N$ as before. After an instance of the rule $\$$ we still get $\bar{\Pi}_{M\subs{N}{x}}$ with the right conclusion. Now, we proceed by the application of as many instances of $C$ as those we can count in $\Pi_N$ below $\Pi''_N(\$)$. Let us say they be $r$ and let us call $\bar{\bar{\Pi}}_{M\subs{N}{x}}(C)$ the final deduction.
\end{itemize}
\par
Subpoint~\textbf{\ref{lemma:subst-vs-wght-03-b}} of
\textbf{Point~\ref{lemma:subst-vs-wght-03}}
has two cases. The first case is:
\small
\begin{align}
\nonumber
\dpth{\bar{\Pi}_{M\subs{N}{x}}(\$,\Pi'_{M\subs{N}{x}})}
&=\dpth{\Pi'_{M\subs{N}{x}}}+1
\\
\nonumber
&=\max\{\dpth{\Pi'_{M}}
       ,\dpth{\Pi'_{N}}\}+1
&(\text{by induction})
\\
\nonumber
&=\max\{\dpth{\Pi'_{M}}+1
       ,\dpth{\Pi'_{N}}+1\}
\\
\label{align:lemma:subst-vs-wght-03-a-1}
&=\max\{\dpth{\Pi_{M}(\$,\Pi'_M)}
       ,\dpth{\Pi_{N}(\$,\Pi'_N)}\}
\end{align}
\normalsize
Step~\eqref{align:lemma:subst-vs-wght-03-a-1}
holds because both $\Pi'_M$ and $\Pi'_N$ are followed by the rule $\$$.
The second case is based on the first here above, observing that
$\dpth{\Pi(C,\Pi')}=\dpth{\Pi'}$, for every deduction $\Pi$.
\par
Subpoint~\textbf{\ref{lemma:subst-vs-wght-03-a}} of
\textbf{Point~\ref{lemma:subst-vs-wght-03}}
has two cases.
The first case has $\Pi_N(\$)$. If $d=0$, the statement holds because $\wdth{0}{\Pi}=0$ for any $\Pi$. If $d\geq 1$:
\small
\begin{align}
\wdth{d}{\bar{\Pi}_{M\subs{N}{x}}(\$,\Pi'_{M\subs{N}{x}})}
&=\wdth{d-1}{\Pi'_{M\subs{N}{x}}}
\nonumber\\
\nonumber
&=\wdth{d-1}{\Pi'_{M}}+\wdth{d-1}{\Pi'_{N}}
&(\text{by induction})
\\
\label{align:lemma:subst-vs-wght-03-b-1}
&=\wdth{d}{\Pi_{M}(\$,\Pi'_M)}+\wdth{d}{\Pi_{N}(\$,\Pi'_N)}
\end{align}
\normalsize
The second case is with $\Pi_{N}(C)$. 
If $d=0$, we have $\wdth{0}{\Pi}=0$ for any $\Pi$.
With $d\geq 1$ we can still write:
\small
\begin{align}
\wdth{d}{\bar{\Pi}_{M\subs{N}{x}}(\$,\Pi'_{M\subs{N}{x}})}
&=\wdth{d-1}{\Pi'_{M\subs{N}{x}}}
\nonumber\\
\nonumber
&\leq\wdth{d-1}{\Pi'_{M}}+\wdth{d-1}{\Pi'_{N}}
&(\text{by induction})
\\
\nonumber
&=\wdth{d}{\Pi_{M}(\$,\Pi'_M)}+\wdth{d}{\Pi''_{N}(\$,\Pi'_N)}
\end{align}
\normalsize
where $\Pi''_{N}(\$)$ replaces $\Pi_{N}(\$)$ of~\eqref{align:lemma:subst-vs-wght-03-b-1}.
We can conclude by observing the two following facts:
\begin{itemize}
\item 
if $d=1$, then we count $r$ instances of $C$ below
$\wdth{d}{\bar{\Pi}_{M\subs{N}{x}}(\$)}$ to obtain
$\wdth{d}{\bar{\bar{\Pi}}_{M\subs{N}{x}}(C)}$, and $r$ below
$\wdth{d}{\Pi''_{N}(\$)}$ to obtain $\wdth{d}{\Pi_{N}(C)}$, so getting:
\small
\begin{align*}
\wdth{d}{\bar{\bar{\Pi}}_{M\subs{N}{x}}(C)}
&\leq\wdth{d}{\Pi_{M}(\$)}+\wdth{d}{\Pi_{N}(C)}\enspace .
\end{align*}
\normalsize

\item 
if $d>1$, by definition of width, we do not count any instance 
of $C$ neither below $\wdth{d}{\bar{\Pi}_{M\subs{N}{x}}(\$)}$ nor below
$\wdth{d}{\Pi''_{N}(\$)}$, even if they exist. So, we get again:
\small
\begin{align*}
\wdth{d}{\bar{\bar{\Pi}}_{M\subs{N}{x}}(C)}
&\leq\wdth{d}{\Pi_{M}(\$)}+\wdth{d}{\Pi_{N}(C)}\enspace .
\end{align*}
\normalsize
\end{itemize}
\par
Subpoint~\textbf{\ref{lemma:subst-vs-wght-03-c}} of
\textbf{Point~\ref{lemma:subst-vs-wght-03}}
requires to observe that we can only have
$\Pi_M(\$),$
\\
$\Pi_{N}(R),$ $\bar{\Pi}_{M\subs{N}{x}}(\$),$
$\bar{\bar{\Pi}}_{M\subs{N}{x}}(R)$ with $R\in\{C, \$\}$.
By definition, the partial size at level $0$ is $0$ on any deduction terminating by the rules $\$$, and $C$. So, the point trivially holds.
\par
Subpoint~\textbf{\ref{lemma:subst-vs-wght-03-d}} of
\textbf{Point~\ref{lemma:subst-vs-wght-03}} has $d\geq 1$.
The union of the subpoints~\textbf{\ref{lemma:subst-vs-wght-01-c}} and~\textbf{\ref{lemma:subst-vs-wght-01-d}} of the lemma we are proving, applied to $\Pi'_M$, and $\Pi'_N$, imply 
$\psz{d-1}{\Pi'_{M\subs{N}{x}}}
\leq\psz{d-1}{\Pi'_{M}}+\psz{d-1}{\Pi'_{N}}$ since $x$ is linear in $\Pi'_M$.
So we get
$\psz{d}{\bar{\Pi}_{M\subs{N}{x}}(\$)}\leq\psz{d}{\Pi_{M}(\$)}+\psz{d}{\Pi_{N}(\$)}$, or
$\psz{d}{\bar{\bar{\Pi}}_{M\subs{N}{x}}(C)}\leq\psz{d}{\Pi_{M}(\$)}+\psz{d}{\Pi_{N}(C)}$, depending on the last rule of $\Pi_N$.
\par
The \textbf{second} case is with $\Pi_M(\$)$ concluding by:
\small
$$\infer[\$]
{
 \Gamma';
 \$\Delta'_{M},\ta{x}{\$A},\Delta_M;
 \{(\$\Theta'_M;\emptyset)\}\sqcup
 \{(\Theta_1;\Phi_1)\}\sqcup\ldots\sqcup\{(\Theta_m;\Phi_m)\}
 \vdash\ta{M}{\$B}
}
{\begin{array}{ll}
 \Pi'_M\rhd
 \Gamma_M;\Delta'_{M},\ta{x}{A};\{(\Theta'_M;\emptyset)\}\vdash\ta{M}{B}
 \\
 \Gamma_M\subseteq
 \Delta_{M}\cup
 \bigcup^{m}_{i=1}\Theta_i\cup
 \bigcup^{m}_{i=1}\Phi_i
 &
 \Theta_i\neq\emptyset\Leftrightarrow\Phi_i=\emptyset
 \end{array}
}$$
\normalsize
To match the above assumption about $\Pi_M(\$)$ we must assume $\Pi_N(R)\rhd\Gamma''_N;$
\\$\Delta''_N;{\mathcal E}''_N\vdash\ta{N}{\$^2A}$.
Point~\ref{lemma:structural-properties-WALL-3.3} of
Lemma~\ref{lemma:structural-properties-WALL}, applied to
$\Pi_N(R)\rhd
\Gamma''_N;\Delta''_N;{\mathcal E}''_N\vdash\ta{N}{\$^2A}$, requires to focus on \textbf{two cases}:
\begin{itemize}
\item 
Let $\Pi_N(\$)$ be:
\small
\[
\infer[\$]
{
\Gamma'_N;
\$\Delta'_N,\Delta_N;
\{(\$\Theta'_N;\emptyset)\}\sqcup
\{(\Theta'_1;\Phi'_1)\}\sqcup\ldots\sqcup\{(\Theta'_n;\Phi'_n)\}
\vdash\ta{N}{\$^2A}
}
{\begin{array}{ll}
 \Pi'_N\rhd
 \Gamma_N;
 \Delta'_N;
 \{(\Theta'_N;\emptyset)\}
 \vdash\ta{N}{\$A}
 \\
 \Gamma_N\subseteq
 \Delta_N\cup\bigcup_{i=1}^{n}\Theta'_i\cup\bigcup_{i=1}^{n}\Phi'_i
 &
 \Theta'_i\neq\emptyset \Leftrightarrow \Phi'_i=\emptyset
 \end{array}
}
\]
\normalsize
So, applying the inductive hypothesis on $\Pi'_{M}$, and $\Pi'_{N}$ we get
$\Pi'_{M\subs{N}{x}}
\rhd
\Gamma_M,\Gamma_N;
\Delta'_{M},$ $\Delta'_{N};
\{(\Theta'_M;\emptyset)\}\sqcup
\{(\Theta'_N;\emptyset)\}
\vdash\ta{M\subs{N}{x}}{B}$ that allows to conclude by an application of $\$$.

\item 
Let $\Pi_N(C)\rhd\Gamma''_N;\Delta''_N;{\mathcal E}''_N\vdash\ta{N}{\$^2A}$. Iteratively applying Lemma~\ref{lemma:structural-properties-WALL}, point~\ref{lemma:structural-properties-WALL-3.3}, we eventually prove the existence of 
$\Pi''_N(\$)\rhd
\Gamma'_N;\$\Delta'_N,\Delta_N;
\{(\$\Theta'_N;\emptyset)\}\sqcup
\{(\Theta'_1;\Phi'_1)\}\sqcup\ldots\sqcup\{(\Theta'_n;\Phi'_n)\}
\vdash\ta{N}{\$^2A}$ which, as in the previous case, we assume  to have a premise derived from $\Pi'_N$. So, we can proceed on $\Pi'_M$, and $\Pi'_N$ as before, to conclude by a applying the same number of instances of $C$, that we can count in $\Pi_N$, below $\Pi''_N(\$)$.
\end{itemize}
To prove Subpoints~\textbf{\ref{lemma:subst-vs-wght-03-b}},
\textbf{\ref{lemma:subst-vs-wght-03-a}},
\textbf{\ref{lemma:subst-vs-wght-03-c}}, and
\textbf{\ref{lemma:subst-vs-wght-03-d}} of
\textbf{Point~\ref{lemma:subst-vs-wght-03}}
proceed in analogy to what we did above. All the other cases of $\Pi_M(R)$ routinely apply the inductive hypothesis.

\item [Point~\ref{lemma:subst-vs-wght-05} of Lemma~\ref{lemma:subst-vs-wght}.]
Notice that the hypothesis $x\in\FV{M}$ excludes $\Pi_{M}(A)$, otherwise $x\not\in\FV{M}$. Then we proceed by structural induction on $\Pi_M$, by detailing out \textbf{four} base \textbf{cases}.
\par
The \textbf{first} case has
$\Pi_M(\$)$ that concludes by:
\small
$$\infer[\$]
{
 \Gamma';
 \$\Delta'_{M},\Delta_{M};
 \{(\$\Theta'_M,\ta{x}{\$A};\emptyset)\}\sqcup
 \{(\Theta_1;\Phi_1)\}\sqcup\ldots\sqcup\{(\Theta_m;\Phi_m)\}
 \vdash\ta{M}{\$B}
}
{
\begin{array}{ll}
\Pi'_M\rhd
\Gamma_M;\Delta'_{M};\{(\Theta'_M,\ta{x}{A};\emptyset)\}\vdash\ta{M}{B}
\\ 
\Gamma_M\subseteq
\Delta_{M}\cup
\bigcup^{m}_{i=1}\Theta_i\cup
\bigcup^{m}_{i=1}\Phi_i
&
\Theta_i\neq\emptyset\Leftrightarrow\Phi_i=\emptyset
\end{array}
}$$
\normalsize
To match the above assumption about $\Pi_M(\$)$ we must assume $\Pi_N(R)\rhd\emptyset;\emptyset;{\mathcal E}_N$
$\vdash\ta{N}{\$^{2}A}$,
with ${\mathcal E}_N\subseteq\{(\Theta_N;\emptyset)\}$, in accordance with the statement we have to prove. Then, point~\ref{lemma:structural-properties-WALL-3.3} of
Lemma~\ref{lemma:structural-properties-WALL}, applied to
$\Pi_N(R)$, requires to focus on the following \textbf{single case},
since
$\Pi_N(C)\rhd\emptyset;\emptyset;{\mathcal E}_N\vdash\ta{N}{\$^2A}$ is excluded by the assumption on ${\mathcal E}\subseteq\{(\Theta_{N};\emptyset)\}$.
Let $\Pi_N(\$)$ be:
\small
\[
\infer[\$]
{
\emptyset;\emptyset;
\{(\$\Theta'_N,\Theta';\emptyset)\}
\vdash\ta{N}{\$^{2}A}
}
{\Pi'_N\rhd
\Gamma_N;\emptyset;\{(\Theta'_N;\emptyset)\}\vdash\ta{N}{\$A}
&\Gamma_N\subseteq\Theta'
}
\]
\normalsize
Point~\ref{lemma:structural-properties-WALL-3.3} of
Lemma~\ref{lemma:structural-properties-WALL} applied to $\Gamma_N;\emptyset;\{(\Theta'_N;\emptyset)\}\vdash\ta{N}{\$A}$ implies that $\FV{N}\subseteq\dom{\Theta'_N}$. Namely, $\Gamma_N=\emptyset$. So, we can apply the inductive hypothesis on both $\Pi'_M$ and $\Pi'_N$. We get
$\Pi'_{M\subs{N}{x}}
\rhd
\Gamma_M;
\Delta'_{M};
\{(\Theta'_M,\Theta'_N;\emptyset)\}
\vdash\ta{M\subs{N}{x}}{B}$ that allows to conclude by an application of $\$$.
\par
The \textbf{second} case has $\Pi_M(!)$ that concludes by:
\small
$$\infer[!]
{\Gamma';\Delta;
 \{(\$\Theta'_M,\ta{x}{\$A};\emptyset)\}\sqcup\{(\Theta_M;\Phi_M)\}
 \vdash\ta{M}{\ !B}
}
{\begin{array}{ll}
 \Gamma_M;\emptyset;\{(\Theta'_M,\ta{x}{A};\emptyset)\}\vdash\ta{M}{B}
 \\
 \Gamma_M\subseteq\Theta_M\cup\Phi_M
 & 
 \Theta\neq\emptyset \Rightarrow \dom{\Phi}\cap\FV{M}\neq\emptyset
\end{array}
}$$
\normalsize
We can proceed with the same assumptions and arguments used in the proof of the first case here above. 
Applying point~\ref{lemma:structural-properties-WALL-3.3} of
Lemma~\ref{lemma:structural-properties-WALL} to $\Pi'_N$, and the inductive hypothesis to both $\Pi'_M$, and $\Pi'_N$, we get
$\Pi'_{M\subs{N}{x}}\rhd
 \Gamma_M;\emptyset;
 \{(\Theta'_M,\Theta'_N;\emptyset)\}
\vdash \ta{M\subs{N}{x}}{B}$, which allows to conclude by using the rule $!$.
\par
The \textbf{third} case has
$\Pi_M(\$)$ that concludes by:
\small
$$\infer[\$]
{\Gamma';
 \$\Delta'_{M},\Delta_{M};
 \{(\$\Theta'_M;\emptyset)\}\sqcup
 \{(\Theta_j,\ta{x}{L};\Phi_j)\}\sqcup
   \bigsqcup^{m}_{i=1, i\neq j}\{(\Theta_i;\Phi_i)\}
 \vdash\ta{M}{\$B}
}
{\begin{array}{ll}
\Gamma_M,\ta{x}{L};\Delta'_{M};\{(\Theta'_M;\emptyset)\}\vdash\ta{M}{B}
\\
\Gamma_M,\ta{x}{L}\subseteq
\Delta_{M}\cup\{\ta{x}{L}\}\cup
\bigcup^{m}_{i=1}\Theta_i\cup
\bigcup^{m}_{i=1}\Phi_i
&
\Theta_i\neq\emptyset\Leftrightarrow\Phi_i=\emptyset
\end{array}
}$$
\normalsize
The above assumption about $\Pi_M(\$)$ is matched by assuming $\Pi_N(R)\rhd\emptyset;\emptyset;{\mathcal E}_N\vdash\ta{N}{\$L}$,
with ${\mathcal E}_N\subseteq\{(\Theta_N;\emptyset)\}$. Then, Point~\ref{lemma:structural-properties-WALL-3.3} of
Lemma~\ref{lemma:structural-properties-WALL}, applied to
$\Pi_N(R)$, requires to focus on the following \textbf{single case},
since the case with 
$\Pi_N(C)\rhd\emptyset;\emptyset;{\mathcal E}_N\vdash\ta{N}{\$L}$ is excluded by the assumption on ${\mathcal E}_N\subseteq\{(\Theta_N;\emptyset)\}$.
Let $\Pi_N(\$)$ be:
\small
\[
\infer[\$]
{
\emptyset;\emptyset;
\{(\$\Theta'_N,\Theta';\emptyset)\}
\vdash\ta{N}{\$L}
}
{
\Pi'_N\rhd
\Gamma_N;\emptyset;\{(\Theta'_N;\emptyset)\}\vdash\ta{N}{L}
&\Gamma_N\subseteq\Theta'
}
\]
\normalsize
So, using point~\ref{lemma:subst-vs-wght-01} of
Lemma~\ref{lemma:subst-vs-wght} on both $\Pi'_M$, and $\Pi'_N$ we get
$\Pi'_{M\subs{N}{x}}
\rhd
\Gamma_M,\Gamma_N;$
\\$
\Delta'_{M};
\{(\Theta'_M;\emptyset)\}\sqcup\{(\Theta'_N;\emptyset)\}
\vdash\ta{M\subs{N}{x}}{B}$ that allows to conclude by an application of $\$$.
\par
The \textbf{fourth} case has $\Pi_M(!)$ that concludes by:
\small
$$\infer[!]
{\Gamma';\Delta;
 \{(\$\Theta'_M;\emptyset)\}\sqcup\{(\Theta_M,\ta{x}{L};\Phi_M)\}
 \vdash\ta{M}{\ !B}
}
{\begin{array}{ll}
\Pi'_M\rhd
\Gamma_M,\ta{x}{L};\emptyset;\{(\Theta'_M;\emptyset)\}\vdash\ta{M}{B}
\\
\Gamma_M,\ta{x}{L}\subseteq\Theta_M\cup\{\ta{x}{L}\}\cup\Phi_M
& 
\Theta\neq\emptyset \Rightarrow \dom{\Phi}\cap\FV{M}\neq\emptyset
\end{array}
}$$
\normalsize
We can proceed with the same assumptions and arguments used in the proof of the third case here above. Applying point~\ref{lemma:subst-vs-wght-01} of Lemma~\ref{lemma:subst-vs-wght} to both $\Pi'_M$, and $\Pi'_N$ we get
$\Pi'_{M\subs{N}{x}}\rhd
 \Gamma_M,\Gamma_N;\emptyset;
 \{(\Theta'_M;\emptyset)\}\sqcup\{(\Theta'_N;\emptyset)\}
\vdash \ta{M\subs{N}{x}}{B}$, which allows to conclude by using rule $!$.
\par
The remaining cases of $\Pi_M(R)$ routinely apply the inductive hypothesis.
\par
To prove the subpoints \textbf{\ref{lemma:subst-vs-wght-05-b}},
\textbf{\ref{lemma:subst-vs-wght-05-a}}, \textbf{\ref{lemma:subst-vs-wght-05-c}},
and~\textbf{\ref{lemma:subst-vs-wght-05-d}}
of point \textbf{\ref{lemma:subst-vs-wght-05}}
we can proceed as we did for the subpoints
\textbf{\ref{lemma:subst-vs-wght-03-b}},
\textbf{\ref{lemma:subst-vs-wght-03-a}}, \textbf{\ref{lemma:subst-vs-wght-03-c}},
and~\textbf{\ref{lemma:subst-vs-wght-03-d}} above.

\item [Point~\ref{lemma:subst-vs-wght-02} of Lemma~\ref{lemma:subst-vs-wght}.]
Notice that $x\in\FV{M}$ excludes $\Pi_{M}(A)$. Otherwise $x\not\in\FV{M}$. We proceed by induction on $\Pi_M$.
\par
Lemma~\ref{lemma:structural-properties-WALL}, point~\ref{lemma:structural-properties-WALL-2}, applied to
$\Pi_M\rhd \Gamma_M;\Delta_M;{\mathcal E}_M,(\emptyset;\ta{x}{A})\vdash \ta{M}{B}$ implies the existence of $n\geq 1$ and $q_1,\ldots,q_n\geq 0$ such that $\wdth{1}{\Pi_M}\geq q_1+\ldots+q_n$ and:
\begin{itemize}
\item $M$ can be written as 
$M'
\subs{x}
{{x^{1}_{1}\ldots x^{1}_{q_1}
\ldots\ldots
x^{n}_{1}\ldots x^{n}_{q_n}}}$,
for some $M'$;
\item there are $n$ subdeductions 
$\Pi'_{P_i}(R_i)
\rhd
\Gamma_i;
\Delta_i;
{\mathcal E}_i,
(\Theta^{i}_1;\ta{x^{i}_1}{A}),
\ldots,
(\Theta^{i}_{q_i};\ta{x^{i}_{q_i}}{A})
\vdash \ta{P_i}{C_i}$, with $R_i\in\{A,\$,!\}$, that introduce $\ta{x^{i}_{1}}{A},\ldots,\ta{x^{i}_{q_i}}{A}$;
\item 
$q_1+\ldots+q_n-1$ instances of $C, \li E, \li E_!, \liv E$ are required in the tree with the conclusion of $\Pi_M$ as root and the conclusions of $\Pi'_{P_i}$s as leaves to contract
${x^{1}_{1},\ldots,x^{1}_{q_1},
\ldots\ldots,
x^{n}_{1},\ldots,x^{n}_{q_n}}$ to $x$.
\end{itemize}
\par
Using the assumption on ${\mathcal E}_N$ we focus on the case 
${\mathcal E}_N=\{(\Theta_N;\ta{y}{C})\}$, 
the other being simpler.
\par
Lemma~\ref{lemma:structural-properties-WALL}, point~\ref{lemma:structural-properties-WALL-3.1} can be applied to $\Pi_N$ implying that its conclusion be:
\small
$$
\infer[!]
{\Gamma_N;
 \Delta_N;
 \{(\Theta_N;\ta{y}{C})\}
 \vdash\ta{N}{\ !A}
}
{
\Pi'_N
\rhd
\ta{y}{C},
\Gamma'_N;
\emptyset;\emptyset
\vdash\ta{N}{A}
&
\{\ta{y}{C}\}\cup\Gamma'_N
\subseteq
\Theta_N\cup\{\ta{y}{C}\}
}$$
\normalsize
where $C$ is linear.
Now we can split the set of all $\Pi'_{P_i}$s into two complementary sets.
The first set $\mathcal G$ contains all the deductions $\Pi'_{P_i}(R_i)$ such that both $R_i\in\{\$, !\}$ and at least one among 
$\ta{x^{i}_{1}}{A},\ldots,\ta{x^{i}_{q_i}}{A}$
is a linear type assignment in the premise of $R_i$.
The other set $\mathcal B$ is equal to 
$\{ \Pi'_{P_i}(R_i) \mid 1\leq i\leq n\}\setminus{\mathcal G}$, 
namely the set of all $\Pi'_{P_i}(R_i)$ whose conclusion is either an axiom, or a modal rule that introduces every of $\ta{x^{i}_{1}}{A},\ldots,\ta{x^{i}_{q_i}}{A}$ as a fake polynomially partially discharged assumption in the conclusion of $R_i$.
The assumption $\nocc{x}{M}=1$ says that ${\mathcal G}$ contains a single $\Pi'_{P_i}$ in which only one among
$\ta{x^{i}_{1}}{A},\ldots,\ta{x^{i}_{q_i}}{A}$
is a linear type assignment. We can assume it be $\ta{x^{i}_{1}}{A}$ with $i=1$.
\par
We can apply Lemma~\ref{lemma:subst-vs-wght}, point~\ref{lemma:subst-vs-wght-01}, to the premise 
$\Pi'_{N\subs{w^1_1}{y}}\rhd
\ta{w^1_1}{C},\Gamma'_N;
\emptyset;
\emptyset
\vdash
\ta{N\subs{w^1_1}{y}}{A}$ of $\Pi_{N}$,
and to the premise of $\Pi'_{P_1}$ in $\mathcal G$ that, by definition, has form:
\small
$$\infer[R_1]
{\Gamma_1;\Delta_1;
 {\mathcal E}_1,
 (\Theta^1_1;\ta{x^1_1}{A}),
 \ldots,
 (\Theta^1_{q_1};\ta{x^1_{q_1}}{A})
 \vdash\ta{P_1}{R_1 C'_1}
}
{
\Pi''_{P_1}
\rhd
\Gamma''_1,\ta{x^1_1}{A};\Delta''_1;
{\mathcal E}''_1\vdash\ta{P_1}{C'_1}
}$$
\normalsize
We can get a set ${\mathcal G}'$ with a single deduction
$\Pi'_{P_1\subs{N\subs{w^1_1}{y}}{x^1_1}}$ with form:
\scriptsize
$$\infer[R_1]
{\Gamma_1;
 \Delta_1;
 {\mathcal E}_1\sqcup
 \{
 (\Theta_N;\ta{w^1_{1}}{C}),
 (\Theta^1_2;\ta{w^1_{2}}{C}),
 \ldots,
 (\Theta^1_{q_1};\ta{w^1_{q_1}}{C})
 \}
 \vdash\ta{P_1\subs{N\subs{w^1_1}{y}}{x^1_1}}{R_1 C'_1}
}
{
\Pi''_{P_1\subs{N\subs{w^1_1}{y}}{x^1_1}}
\rhd
\Gamma''_1,\ta{w^1_1}{C},\Gamma'_N;
\Delta''_1;
{\mathcal E}''_1
\vdash\ta{P_1\subs{N\subs{w^1_1}{y}}{x^1_1}}{C'_1}
}$$
\normalsize
Moreover, we can build a set ${\mathcal B}'$, from ${\mathcal B}$, so that every deduction of ${\mathcal B}'$ is identical to one of ${\mathcal B}$ up to the introduction of the fake assumptions 
$\ta{w^{j}_{1}}{C},\ldots,\ta{w^{j}_{q_j}}{C}$
in place of
$\ta{x^{j}_{1}}{A},\ldots,\ta{x^{j}_{q_j}}{A}$, for every $1\leq j\leq n$ and $j\neq i$,
using the rules $A, \$, !$.
Finally, we choose, arbitrarily in ${\mathcal G}'\cup{\mathcal B}'$ a deduction to introduce $\Gamma_N$ and $\Delta_N$ as fake assumptions in its conclusion.
We observe that the set of assumptions of the deductions in ${\mathcal G}'\cup{\mathcal B}'$ are those of ${\mathcal G}\cup{\mathcal B}$ up to the changes due to the substitution of $N\subs{w^1_1}{y}$ for $x^1_1$ in $P_1$.
So, we can apply to ${\mathcal G}'\cup{\mathcal B}'$, and to the subdeductions of $\Pi_M$ not in ${\mathcal G}\cup{\mathcal B}$, required to build $\Pi_M$, the same sequence of rules that, from ${\mathcal G}\cup{\mathcal B}$, lead to $\Pi_M$ itself.
This implies to apply at least the $q_1+\ldots+q_n-1$ instances of $C, \li E, \li E_!, \liv E$ that contract $w^1_1,\ldots,w^1_{q_1},\ldots\ldots,w^n_1,\ldots,w^n_{q_n}$ to $y$. We end up with a deduction with conclusion:
\small
\[
\Pi_{\overline{M}}
\rhd
\Gamma_M,\Gamma_N;
\Delta_M,\Delta_N;
{\mathcal E}_M\sqcup
\{(\Theta_N;\ta{y}{C})\}
\vdash\ta{\overline{M}}{B}
\]
\normalsize
where ${\mathcal E}_M\sqcup\{(\Theta_N;\ta{y}{C})\}$ is exactly 
${\mathcal E}_M\sqcup{\mathcal E}_N$ and:
\small
\begin{align*}
\overline{M}
&=(M'[P_1
      \subs{N\subs{w^1_1}{y}}
           {x^1_1}
      \ldots
      \subs{N\subs{w^1_{q_1}}{y}}
           {x^1_{q_1}}
      \ldots
\\
&\qquad\qquad\qquad
      \ldots
      P_n
      \subs{N\subs{w^{n}_{1}}{y}}
           {x^n_{1}}
      \ldots
      \subs{N\subs{w^{n}_{q_n}}{y}}
           {x^n_{q_n}}
     ])
  \subs{y}{w^1_1\ldots w^1_{q_1}\ldots\ldots w^n_1\ldots w^n_{q_n}}
\\
&=(M'[P_1\subs{N\subs{w^1_1}{y}}{x^1_1}
      P_2\ldots P_n])
  \subs{y}{w^1_1}
\\
&=(M'[P_1
     \ldots
     P_n]
    \{{}^{N\subs{w^1_1}{y}}/_{x^1_1}\})
  \subs{y}{w^1_1}
\\
&=(M'[P_1
     \ldots
     P_n])
   \subs{N}{x^1_1}
=((M'[P_1
      \ldots
      P_n])
    \subs{x}{x^1_1})
    \subs{N}{x} = M\subs{N}{x}
\enspace .
\end{align*}
\normalsize
$M[M_1\ldots M_n]$ highlights that the terms $M_1,\ldots,M_n$ occur in $M$.
\par
Subpoint~\textbf{\ref{lemma:subst-vs-wght-02-b}} holds because we have not introduced any new instance of modal rules in the course of the reconstruction of $\Pi_{\overline{M}}$.
\par
Subpoint~\textbf{\ref{lemma:subst-vs-wght-02-a}} holds for $d=0$ by definition of width.
If $d=1$, then $\wdth{1}{\Pi_{M\subs{N}{x}}}$\\
$=q_1+\ldots+q_n-1+k$, where $q_1+\ldots+q_n-1$ counts the number of rules required to contract
$w^1_1,\ldots,w^1_{q_1},\ldots\ldots, \\ 
 w^n_1,\ldots,$ $w^n_{q_n}$ to $y$, and $k$
counts the contribution to the width of $\Pi_{M\subs{N}{x}}$ by the instances of $\li E, \li E_!, \liv E, C$ that do not contract the polynomial assumptions of the deductions in ${\mathcal G}\cup{\mathcal B}$, but which may exist to produce the whole $\Pi_M$. However, all the $q_1+\ldots+q_n-1+k$ rules exist in $\Pi_{M\subs{N}{x}}$ exactly because they exist in $\Pi_{M}$. So, $\wdth{d}{\Pi_{M\subs{N}{x}}}=\wdth{d}{\Pi_M}$, for $0\leq d\leq 1$. Notice that ${\mathcal E}=\emptyset$ may imply that we can avoid the use of some of the instances of $\li E, \li E_!, \liv E, C$ to build $\Pi_{M\subs{N}{x}}$. So,
$\wdth{d}{\Pi_{M\subs{N}{x}}}<\wdth{d}{\Pi_M}$, for $0\leq d\leq 1$.
\par
Point~\textbf{\ref{lemma:subst-vs-wght-02-c}} holds because all the substitutions in $M$ occur at level $1$, and we do not change the structure at level $0$ thanks to the way we build ${\mathcal G}',{\mathcal B}'$, and $\Pi_{\bar{M}}$.
\par
Point~\textbf{\ref{lemma:subst-vs-wght-02-d}} holds because:
\small
\begin{align}
\psz{d}{\Pi_{M\subs{N}{x}}}
&=(q_1+\ldots+q_n-1)
 +\psz{d}{\Pi'_{P_1\subs{N\subs{w^1_1}{y}}{x^1_1}}}
 +\sum^{n}_{i=2}\psz{d}{\Pi'_{P_i}}
 +k
\nonumber
\\
\label{align:lemma:subst-vs-wght-02-d-10}
&\leq(q_1+\ldots+q_n-1)
 +\psz{d}{\Pi'_{P_1}}
 +\sum^{n}_{i=2}\psz{d}{\Pi'_{P_i}}
 +k
 +\psz{d}{\Pi_{N}}
\\
\nonumber
&=
\psz{d}{\Pi_{M}}+\psz{d}{\Pi_{N}}
\end{align}
\normalsize
where step~\eqref{align:lemma:subst-vs-wght-02-d-10} holds by induction on $\Pi'_{P_1\subs{N\subs{w^1_1}{y}}{x^1_1}}$, which has the same structure as $\Pi'_{P_1\subs{N}{x^1_1}}$, and $k$ takes into account the contribution to the width of $\Pi_{M\subs{N}{x}}$ by the instances of $\li E, \li E_!, \liv E$, and $C$ that do not contract the polynomial assumptions of the deductions in ${\mathcal G}\cup{\mathcal B}$, but which may exist to produce the whole $\Pi_M$.

\item [Point~\ref{lemma:subst-vs-wght-04} of Lemma~\ref{lemma:subst-vs-wght}.]
Notice that $x\in\FV{M}$ excludes $\Pi_{M}(A)$. Otherwise $x\not\in\FV{M}$. We proceed by induction on $\Pi_M$.
\par
Lemma~\ref{lemma:structural-properties-WALL}, point~\ref{lemma:structural-properties-WALL-2}, applied to
$\Pi_M\rhd \Gamma_M;\Delta_M;{\mathcal E}_M,(\emptyset;\ta{x}{A})\vdash \ta{M}{B}$ implies the existence of $n\geq 1$ and $q_1,\ldots,q_n\geq 0$ such that $\wdth{1}{\Pi_M}\geq q_1+\ldots+q_n$ and:
\begin{itemize}
\item 
$M$ can be written as
$M'
\subs{x}
{x^1_1\ldots x^1_{q_1}
 \ldots\ldots
 x^n_1\ldots x^n_{q_n}}
$, for some $M'$;
\item 
there are $n$ subdeductions
$\Pi'_{P_i}(R_i)\rhd
\Gamma_i;
\Delta_i;
{\mathcal E}_i,
(\Theta^i_{1};\ta{x^i_{1}}{A}),
\ldots,
(\Theta^i_{q_i};\ta{x^i_{q_i}}{A})
\vdash \ta{P_i}{C_i}$
that introduce the polynomially partially discharged assumptions
$\ta{x^i_{1}}{A},\ldots,\ta{x^i_{q_i}}{A}$,
and such that $R_i\in\{A,\$,!\}$;
\item 
$q_1+\ldots+q_n-1$ instances of $C, \li E, \li E_!, \liv E$ are required in the tree with the conclusion of $\Pi_M$ as root and the conclusions of $\Pi'_{P_i}$s as leaves to contract
$x^1_1\ldots x^1_{q_1}
 \ldots\ldots
 x^n_1\ldots x^n_{q_n}$ to $x$.
\end{itemize}
\par
We observe that every $\Theta^i_{1},\ldots,\Theta^i_{q_i}$, with $1\leq i\leq n$, must be equal to $\emptyset$. If not, the only way to get rid of them, to obtain the pair $(\emptyset;\ta{x}{A})$ in the conclusion of $\Pi_M$, would be the use of the rule
$\liv I$ in some positions between the conclusion of $\Pi'_{P_i}(R_i)$ and $\Gamma_M;\Delta_M;{\mathcal E}_M,(\emptyset;\ta{x}{A})\vdash \ta{M}{B}$. However, $\liv I$ can only be applied in absence of the assumptions 
$
\ta{x^1_1}{A},\ldots,\ta{x^1_{q_1}}{A},
\ldots\ldots,
\ta{x^n_1}{A},\ldots,\ta{x^n_{q_n}}{A}$, some of which, instead, we know to exist, thanks to $\nocc{x}{M}>1$.
We focus on the case ${\mathcal E}_N=\{(\emptyset;\ta{y}{C})\}$, the other being simpler.
\par
Lemma~\ref{lemma:structural-properties-WALL}, point~\ref{lemma:structural-properties-WALL-3.1}, applied to $\Pi_N$, implies that its conclusion be:
\small
\[
\infer[!]
{\Gamma_N;
 \Delta_N;
 \{(\emptyset;\ta{y}{C})\}
 \vdash\ta{N}{\ !A}
}
{
\Pi'_N
\rhd
\ta{y}{C};
\emptyset;
\emptyset
\vdash\ta{N}{A}
}
\]
\normalsize
where $C$ is linear.
Now we can split the set of all $\Pi'_{P_i}$s into two complementary sets.
\par
The first set $\mathcal G$ contains all the deductions $\Pi'_{P_i}(R_i)$ such that both $R_i\in\{\$, !\}$ and at least one among $\ta{x^i_{1}}{A},\ldots,\ta{x^i_{q_i}}{A}$ is a linear type assignment in the premise of $R_i$.
\par
The other set $\mathcal B$ is equal to 
$\{ \Pi'_{P_i}(R_i) \mid 1\leq i\leq n\}\setminus{\mathcal G}$,
namely the set of all $\Pi'_{P_i}(R_i)$ whose conclusion is either an axiom, or a modal rule that introduces every $\ta{x^i_{1}}{A},\ldots,\ta{x^i_{q_i}}{A}$ as a fake polynomially partially discharged assumption in the conclusion of $R_i$.
\par
We can apply Lemma~\ref{lemma:subst-vs-wght}, point~\ref{lemma:subst-vs-wght-01}, to the premise
$\Pi'_{N\subs{w^i_j}{y}}\rhd
\ta{w^i_j}{C};
\emptyset;
\emptyset
\vdash\ta{N\subs{w^i_j}{y}}{A}$ of $\Pi_{N}$, for every $1\leq j\leq p_i$, and to the premise of every $\Pi'_{P_i}$ of $\mathcal G$ that, by definition, has form:
\small
$$
\infer[R_i]
{\Gamma_i;\Delta_i;
 {\mathcal E}_i,
 (\emptyset;\{\ta{x^i_1}{A},\ldots,\ta{x^i_{q_i}}{A}\})
 \vdash\ta{P_i}{R_i C'_i}
}
{
\Pi''_{P_i}
\rhd
\Gamma''_i,
\ta{z^i_1}{A},\ldots,\ta{z^i_{p_i}}{A}
;\Delta''_i;
{\mathcal E}''_i
\vdash\ta{P_i}{C'_i}
}
$$
\normalsize
with $p_i\leq q_i$, and 
$\{z^i_1,\ldots,z^i_{p_i}\}\subseteq\{x^i_1,\ldots,x^i_{p_i}\}$.
We get a set ${\mathcal G}'$ of deductions 
$\Pi'_{\bar{P_i}}$ with form:
\small
$$\infer[R_i]
{\Gamma_i;
 \Delta_i;
 {\mathcal E}_i,
 (\emptyset;\ta{w^i_1}{C},\ldots,\ta{w^i_{q_i}}{C})
 \vdash\ta{\bar{P_i}}{R_i C'_i}
}
{
\Pi''_{\bar{P_i}}
\rhd
\Gamma''_i,\ta{\bar{w}^i_1}{C},\ldots,\ta{\bar{w}^i_{p_i}}{C};
\Delta''_i;
{\mathcal E}''_i\vdash
\ta{\bar{P_i}}{C'_i}
}$$
\normalsize
where $\bar{P_i}$ is
$P_i
\{
^{N\subs{\bar{w}^i_1}{y}}/_{z^{i}_{1}}
\ldots
^{N\subs{\bar{w}^i_{p_i}}{y}}/_{z^{i}_{p_i}}
\}$,
and 
$\{\bar{w}^i_1,\ldots,\bar{w}^i_{p_i}\}
\subseteq
\{\bar{w}^i_1,\ldots,\bar{w}^i_{q_i}\}$.
\par
Moreover, we can build a set ${\mathcal B}'$, from ${\mathcal B}$, so that every deduction of ${\mathcal B}'$ is identical to one of ${\mathcal B}$ up to the introduction of the fake assumptions $\ta{w^i_1}{C},\ldots,\ta{w^i_{q_i}}{C}$ in place of $\ta{x^i_1}{A},\ldots,\ta{x^i_{q_i}}{A}$, by using the rules $A, \$, !$.
Finally, we choose, arbitrarily in ${\mathcal G}'\cup{\mathcal B}'$ a deduction to introduce $\Gamma_N$ and $\Delta_N$ as fake assumptions in its conclusion.
We observe that the set of assumptions of the deductions in 
${\mathcal G}'\cup{\mathcal B}'$ are those of ${\mathcal G}\cup{\mathcal B}$ up to the changes due to the substitutions of the terms
$N\subs{\bar{w}^i_j}{y}$s for $z^i_j$ which is one among $x^i_1,\ldots,x^i_{q_i}$.
So, we can apply to ${\mathcal G}'\cup{\mathcal B}'$ the same sequence of rules that, from ${\mathcal G}\cup{\mathcal B}$, lead to $\Pi_M(R)$.
This implies to apply at least the $q_1+\ldots+q_n-1$ instances of $C, \li E, \li E_!, \liv E$ that contract $w^1_1,\ldots,w^1_{q_1},\ldots\ldots,w^n_1\ldots w^n_{q_n}$ to $y$.
We end up with a deduction whose conclusion is:
$\Pi_{\overline{M}}
\rhd
\Gamma_M,\Gamma_N;
\Delta_M,\Delta_N;
{\mathcal E}_M,(\emptyset;\ta{y}{C})\vdash\ta{\overline{M}}{B}
$,
where ${\mathcal E}_M,(\emptyset;\ta{y}{C})$ is ${\mathcal E}_M\sqcup{\mathcal E}_N$ and:
\scriptsize
\begin{align*}
\overline{M}
&=M'[P_1
     \{
      {}^{N\subs{\bar{w}^1_1}{y}}/_{x^1_1}
      \ldots
      {}^{N\subs{\bar{w}^1_{p_1}}{y}}/_{x^1_{p_1}}
     \}
      \ldots
\\
&\qquad\qquad\qquad
      \ldots
     P_n
     \{
      {}^{N\subs{\bar{w}^{n}_1}{y}}/_{x^{n}_1}
      \ldots
      {}^{N\subs{\bar{w}^{n}_{p_n}}{y}}/_{x^{n}_{p_n}}
     \}
     ]
\subs{y}{\bar{w}^1_{1}\ldots\bar{w}^1_{p_1}
         \ldots\ldots
         \bar{w}^n_{1}\ldots\bar{w}^n_{p_n}
         }
\\
&=M'[P_1
     \{
      {}^{N\subs{w^1_1}{y}}/_{x^1_1}
      \ldots
      {}^{N\subs{w^1_{q_1}}{y}}/_{x^1_{p_1}}
     \}
      \ldots
\\
&\qquad\qquad\qquad
      \ldots
     P_n
     \{
      {}^{N\subs{w^{n}_1}{y}}/_{x^{n}_1}
      \ldots
      {}^{N\subs{w^{n}_{q_n}}{y}}/_{x^{n}_{p_n}}
     \}
     ]
\subs{y}{w^1_{1}\ldots w^1_{q_1}
         \ldots\ldots
         w^n_{1}\ldots w^n_{q_n}
         }
\\
&=(M'[P_1\ldots P_n])
     \subs{N}{x^1_{1}\ldots x^1_{q_1}
              \ldots\ldots
              x^n_{1}\ldots x^n_{q_n}}
\\
&=((M'[P_1\ldots P_n])
    \subs{x}{x^1_{1}\ldots x^1_{q_1}
             \ldots\ldots
             x^n_{1}\ldots x^n_{q_n}})
    \subs{N}{x}
= M\subs{N}{x}
\enspace .
\end{align*}
\normalsize
$M[M_1\ldots M_n]$ highlights that $M_1\ldots M_n$ are subterms of $M$.
\par
Subpoints~\textbf{\ref{lemma:subst-vs-wght-04-b}}, \textbf{\ref{lemma:subst-vs-wght-04-a}}, and~\textbf{\ref{lemma:subst-vs-wght-04-c}}
of point~\textbf{\ref{lemma:subst-vs-wght-04}} holds for reasons analogous to the ones that justify the subpoints~\textbf{\ref{lemma:subst-vs-wght-02-b}}, \textbf{\ref{lemma:subst-vs-wght-02-a}}, and~\textbf{\ref{lemma:subst-vs-wght-02-c}} above, respectively.
\par
Subpoint~\textbf{\ref{lemma:subst-vs-wght-04-d}} of point~\textbf{\ref{lemma:subst-vs-wght-04}} holds because:
\scriptsize
\begin{align}
\psz{d}{\Pi_{M\subs{N}{x}}}
&=(q_1+\ldots+q_n-1)+\sum^{n}_{i=1}\psz{d}{\Pi'_{\bar{P_i}}}+k
\nonumber
\\
\nonumber
&=
(q_1+\ldots+q_n-1)+\sum^{n}_{i=1}\psz{d-1}{\Pi''_{\bar{P_i}}}+k
\\
\label{align:lemma:subst-vs-wght-04-d-10}
&\leq
(q_1+\ldots+q_n-1)+
\sum^{n}_{i=1}
(\psz{d-1}{\Pi'_{P_i}}
+\psz{d-1}{\Pi'_{N\subs{w^i_1}{y}}}
+\ldots
+\psz{d-1}{\Pi'_{N\subs{w^i_{p_i}}{y}}}
)+k
\\
\nonumber
&=
(q_1+\ldots+q_n-1)
+\sum^{n}_{i=1}(\psz{d-1}{\Pi'_{P_i}}+p_i\psz{d-1}{\Pi'_N})+k
\\
&=
(q_1+\ldots+q_n-1)
+\sum^{n}_{i=1}\psz{d-1}{\Pi'_{P_i}}
+k
+\psz{d-1}{\Pi'_N} \sum^{n}_{i=1}p_i
\nonumber
\\
&=
(q_1+\ldots+q_n-1)
+\sum^{n}_{i=1}\psz{d}{\Pi_{P_i}}
+k
+\psz{d}{\Pi_N} \sum^{n}_{i=1}p_i
\nonumber
=\psz{d}{\Pi_{M}}+\nocc{x}{M}\psz{d}{\Pi_N}
\enspace .
\end{align}
\normalsize
Step~\eqref{align:lemma:subst-vs-wght-04-d-10} holds by iteratively applying the points~\textbf{\ref{lemma:subst-vs-wght-01-c}}, and~\textbf{\ref{lemma:subst-vs-wght-01-c}} of Lemma~\ref{lemma:subst-vs-wght}, and using the observation that every $p_1,\ldots,p_n$ is the effective number of linear type assignments in $\Pi'_{P_i}$ which are replaced by the linear type assignment of $\Pi_N$, if any, up to a renaming of, at most, the single free variable of $N$.
$k$ counts the contribution to the width of $\Pi_{M\subs{N}{x}}$ by the instances of $\li E, \li E_!, \liv E$, and $C$ that do not contract the polynomial assumptions of the deductions in ${\mathcal G}\cup{\mathcal B}$, but which may exist to produce the whole $\Pi_M$.
\end{description}
\paragraph*{Proof of Lemma~\ref{lemma:substitution-property}
(Subject reduction at depth $0$.)}
As a first step, we inspect the structure of
$\Pi_{(\bs x.M)N}\rhd\Gamma;\Delta;{\mathcal E}\vdash\ta{(\bs x.M)N}{B}$.
In general, it contains an instance of one of the arrow eliminations that assume the generic form:
\small
\begin{align}
\label{eqn:genericappl} 
\begin{minipage}{5cm}
\infer[R_E]
  {\Gamma_M, \Gamma_N;
   \Delta_M, \Delta_N;
   {\mathcal E}_M
   \sqcup
   {\mathcal E}_N
   \vdash \ta{(\bs x.M)N}{C}
  }
  {\begin{array}{ll}
    \Pi_{\bs x.M}\rhd
    \Gamma_M; 
    \Delta_M;
    {\mathcal E}_M
    \vdash \ta{\bs x.M}{A\supset C}
    \\
    \Pi_{N}\rhd
    \Gamma_N; 
    \Delta_N;
    {\mathcal E}_N
    \vdash \ta{N}{A}
    &
    \supset\in\{\li,\liv\}
   \end{array}
  }
\end{minipage}
\end{align}
\normalsize
followed by a, possibly empty, sequence $\sigma$ of instances of the rules $C,\forall I$, and $\forall E$, with $r\geq 0$ instances of $C$. No other rules can belong to $\sigma$, since we are at depth $0$.
Moreover,
$\Pi_{\bs x.M}\rhd
\Gamma_M;\Delta_M;{\mathcal E}_M\vdash \ta{\bs x.M}{A\supset C}$
must be obtained by an instance of:
\small
\begin{align}
\label{eqn:genericabstr}
\begin{minipage}{5cm}
\infer[R_I]
  {\Gamma_M; \Delta_M; 
   {\mathcal E}'_M
   \vdash \ta{\bs x.M}{A\supset C}}
  {\Pi_{M}\rhd\Gamma'_M; \Delta'_M; {\mathcal E}''_M\vdash \ta{M}{C}}
\end{minipage}
\end{align}
\normalsize
where $R_I$ is some arrow introduction, followed by a, possibly empty, sequence $\rho$ of instances of the rule $C, \forall E, \forall I$, with $s\geq 0$ instances of $C$,
and $x\in\dom{\Gamma'_M}\cup\dom{\Delta'_M}\cup\dom{{\mathcal E}''_M}$.
The possible combinations of pairs $(R_I,R_E)$ are:
$(\li I,\li E),(\li I_{\$},\li E),(\li I_!,\li E_!),(\liv I,\liv E)$.
If some instances of the rules $\forall E, \forall I$ exist in $\rho$, the types of the whole deduction can be rearranged to eliminate them, using Lemma~\ref{lemma:structural-properties-WALL}, point~\ref{lemma:structural-properties-WALL-(-1)}. So, we can assume that $\rho$ contains only $s$ occurrences of $C$.

\par{\textbf{First case.}}
We assume $\nocc{x}{M}=0$ and $N\in\text{\PTT}$.
\par
As \textbf{\textit{first hypothesis}} we let \eqref{eqn:genericabstr} be:
\small
\begin{align}
\label{eqn:!abstr}
\begin{minipage}{5cm}
\infer[\li I_{!}]
  {\Gamma_M; \Delta_M; 
   {\mathcal E}'''_M\sqcup\{(\Theta_M;\emptyset)\}
   \vdash \ta{\bs x.M}{\,!A\li C}}
  {\Pi_{M}\rhd\Gamma_M; \Delta_M; {\mathcal E}'''_M,(\Theta_M;\ta{x}{A})
  \vdash \ta{M}{C}}
\end{minipage}
\end{align}
\normalsize
Lemma~\ref{lemma:structural-properties-WALL}, point~\ref{lemma:structural-properties-WALL-2},
implies the existence of $n\geq 1$ and $q_1,\ldots,q_n\geq0$ such that $\wdth{1}{\Pi_M}\geq q_1+\ldots+q_n$, and:
(i) $M$ can be written as 
$M'\subs{x}
        {x^1_1\ldots x^1_{q_1}
         \ldots\ldots
         x^n_1\ldots x^n_{q_n}}$, for some $M'$;
(ii) there are $n$ deductions
$\Pi'_i(R_i)\rhd
\Gamma_i;
\Delta_i;
{\mathcal E}_i,
(\Theta^i_1;\ta{x^i_{1}}{A}),
\ldots,
(\Theta^i_1;\ta{x^i_{q_i}}{A})
\vdash\ta{P_i}{C_i}$
such that $\Pi'_i\preceq\Pi_{M}$ and $R_i\in\{A,\$,!\}$;
(iii) $q_1+\ldots+q_n-1$ instances of $C, \li E, \li E_!, \liv E$, in the tree with the conclusion of $\Pi_M$ as root and the conclusions of $\Pi'_i$s as leaves, are required to contract $x^1_1,\ldots,x^1_{q_1},\ldots\ldots,x^n_1,\ldots,x^n_{q_n}$ to $x$. 
Call $\tau$ such a tree.
We observe that $M$ is $M'$ since none of the $x^i_j$s occurs in the corresponding $P_i$, otherwise we could not have $\nocc{x}{M}=0$. However, this does not prevent to have instances of $C$ in $\tau$ (uselessly) contracting some $x_i$.
Lemma~\ref{lemma:structural-properties-WALL}, point~\ref{lemma:structural-properties-WALL-4}, implies the existence
of $\Pi''_i(R_i)\rhd\Gamma_i;\Delta_i;{\mathcal E}_i\sqcup\{(\Theta_i;\emptyset)\}
\vdash\ta{P_i}{C_i}$, with $1\leq i\leq n$, to which we can apply the instances of the rules in the tree $\tau$, but the contractions $C$ occurring in it, obtaining:
\small
\begin{align}
\Pi_{M}\rhd\Gamma_M; \Delta_M; {\mathcal E}'''_M\sqcup\{(\Theta_M;\emptyset)\}\vdash \ta{M}{C} \enspace .
\label{eqn:!!abstr1}
\end{align}
\normalsize
\eqref{eqn:!!abstr1} can be followed by the sequence $\rho$, which yields:
\small
\begin{align}
\Gamma_M; \Delta_M; {\mathcal E}_M\vdash \ta{M}{C} \enspace .
\label{eqn:!!abstr2}
\end{align}
\normalsize
Lemma~\ref{lemma:structural-properties-WALL}, point~\ref{lemma:structural-properties-WALL-1}, applied to \eqref{eqn:!!abstr2}, implies the existence of $\Gamma_M,\Gamma_N; \Delta_M,\Delta_N; {\mathcal E}_M$
\\$\sqcup{\mathcal E}_N\vdash \ta{M}{C}$, which, followed by $\sigma$, becomes $\Gamma; \Delta; {\mathcal E}\vdash \ta{M}{C}$.
\par
As a \textbf{\textit{second hypothesis}} we let \eqref{eqn:genericabstr} be:
\small
\begin{align}
\label{eqn:livabstr}
\begin{minipage}{5cm}
\infer[\liv I]
  {\Gamma_M; \Delta_M; 
  {\mathcal E}'''_M\sqcup
  \{(\Theta_M;\emptyset)\}
  \vdash \ta{\bs x.M}{\$A\liv C}}
  {\Pi_{M}\rhd\Gamma_M; \Delta_M; 
  {\mathcal E}'''_M,(\Theta_M,\ta{x}{A};\emptyset)
  \vdash \ta{M}{C}}
\end{minipage}
\end{align}
\normalsize
Lemma~\ref{lemma:structural-properties-WALL}, point~\ref{lemma:structural-properties-WALL-6}
implies
$\Pi'(R)\rhd
\Gamma';\Delta';{\mathcal E}',(\Theta',\ta{x}{A};\emptyset)
\vdash\ta{P}{C}$, that introduces $\ta{x}{A}$, and
such that $\Pi'\preceq\Pi_M$,
$R\in\{A, \$, !\}$, for some ${\mathcal E}'',\Gamma', \Delta', {\mathcal E}', \Theta'$. Since $\nocc{x}{M}=0$ implies $\nocc{x}{P}=0$, we have that
Lemma~\ref{lemma:structural-properties-WALL}, point~\ref{lemma:structural-properties-WALL-4}, allows to deduce:
\small
\begin{align}
\Pi''(R)\rhd\Gamma';\Delta';{\mathcal E}',\{(\Theta';\emptyset)\}
\vdash\ta{P}{C}
\label{eqn:livabstr0}
\end{align}
\normalsize
that can be used to yield:
\small
\begin{align}
\Pi_{M}\rhd\Gamma_M; \Delta_M; 
{\mathcal E}'''_M\sqcup\{(\Theta_M;\emptyset)\}\vdash \ta{P}{C}
\label{eqn:livabstr1}
\end{align}
\normalsize
exactly like
$\Pi_{M}\rhd
\Gamma_M; \Delta_M; {\mathcal E}'''_M,(\Theta_M,\ta{x}{A};\emptyset)
\vdash \ta{M}{C}$ can be deduced from
$\Pi'(R)\rhd
\Gamma';\Delta';{\mathcal E}',(\Theta',\ta{x}{A};\emptyset)
\vdash\ta{P}{C}$ because the presence of $\ta{x}{A}$ is not essential.
\eqref{eqn:livabstr1} can be followed by the sequence $\rho$, which gives:
\small
\begin{align}
\Gamma_M; \Delta_M; {\mathcal E}_M\vdash \ta{M}{C} \enspace .
\label{eqn:livabstr2}
\end{align}
\normalsize
Lemma~\ref{lemma:structural-properties-WALL}, point~\ref{lemma:structural-properties-WALL-1}, applied to \eqref{eqn:livabstr2}, implies the existence of $\Gamma_M,\Gamma_N; \Delta_M,\Delta_N; {\mathcal E}_M$
\\$\sqcup{\mathcal E}_N\vdash \ta{M}{C}$, which, followed by $\sigma$, becomes $\Gamma; \Delta; {\mathcal E}\vdash \ta{M}{C}$.
\par
The cases where $R$ is $\li I$ or $\li I_{\$}$ can be proved by following an analogous schema.
\par
Once proved the existence of the deduction corresponding to the redex, we can prove point~\ref{lemma:substitution-property-1}, through~\ref{lemma:substitution-property-4} of the current statement.
\par
\textbf{Point~\ref{lemma:substitution-property-1}} and~\textbf{\ref{lemma:substitution-property-2}}
hold by observing that $\Pi_{M\subs{N}{x}}$ coincides to $\Pi_{M}$ which does not contain the whole deduction $\Pi_N$. So, moving from $\Pi_{(\bs x.M)N}$ to $\Pi_{M}$, we might erase the component of $\Pi_{(\bs x.M)N}$ that determines the value of $\dpth{\Pi_{(\bs x.M)N}}$ or $\wdth{1}{\Pi_{(\bs x.M)N}}$.
\par
\textbf{Point~\ref{lemma:substitution-property-3}}
holds because $\Pi_{(\bs x.M)N}$ has an application more than $\Pi_M$.
\par
\textbf{Point~\ref{lemma:substitution-property-4}} holds because
$\psz{d}{\Pi_{M\subs{N}{y}}}
=\psz{d}{\Pi_{M}}
=\psz{d}{\Pi_{M}}+0\cdot\psz{d}{\Pi_N}$,
for every $0< d\leq \dpth{\Pi_{(\bs x.M)N}}$.

\par{\textbf{Second case.}}
We assume $\nocc{x}{M}=1$, $N\in\text{\PTV}$, and $N\in\text{\PTT}$.
\par
As \textbf{\textit{first hypothesys}} we let \eqref{eqn:genericabstr} had $R_I$ equal to $\li I_{!}$, with assumption:
\small
\begin{eqnarray}
\Pi_{M}\rhd
\Gamma_M; \Delta_M; {\mathcal E}'''_M,(\emptyset;\ta{x}{A})
  \vdash \ta{M}{C} \enspace .
\label{eqn:!abstr-sc}
\end{eqnarray}
\normalsize
So, we must have:
\small
\begin{eqnarray}
\Pi_N\rhd\Gamma_N;\Delta_N;{\mathcal E}_N\vdash\ta{N}{\, !A}
\label{eqn:!abstr-arg-sc}
\end{eqnarray}
\normalsize
on which Lemma~\ref{lemma:structural-properties-WALL}, point~\ref{lemma:structural-properties-WALL-3.1}, implies $\FV{N}\subseteq\dom{\mathcal E_N}$.
Lemma~\ref{lemma:subst-vs-wght}, point~\ref{lemma:subst-vs-wght-02},
applied to \eqref{eqn:!abstr-sc} and \eqref{eqn:!abstr-arg-sc} implies:
\small
\begin{eqnarray}
\Pi'_{M\subs{N}{x}}
\rhd
\Gamma_M,\Gamma_N; \Delta_M,\Delta_N;
{\mathcal E}'''_M\sqcup
{\mathcal E}_N\vdash\ta{M\subs{N}{x}}{C}
\label{eqn:!abstr-subs}
\end{eqnarray}
\normalsize
such that:
(i) 
$\dpth{\Pi'_{M\subs{N}{x}}}
=\max\{\dpth{\Pi_{M}},\dpth{\Pi_{N}}\}$;
(ii)
$\wdth{d}{\Pi'_{M\subs{N}{x}}}
=\wdth{d}{\Pi_{M}}$,
for every $0\leq d\leq 1$;
(iii)
$\psz{0}{\Pi'_{M\subs{N}{x}}}
=\psz{0}{\Pi_{M}}$;
(iv)
$\psz{d}{\Pi'_{M\subs{N}{x}}}
\leq\psz{d}{\Pi_{M}}+\psz{d}{\Pi_{N}}$,
for every $d\geq 1$.
\par
Now, the sequence $\sigma$ and $\rho$ can be applied to
\eqref{eqn:!abstr-subs} to get
\small
\begin{eqnarray}
\Pi_{M\subs{N}{x}}
\rhd
\Gamma; \Delta;{\mathcal E}\vdash\ta{M\subs{N}{x}}{C}
\label{eqn:!abstr-subs-compl} 
\end{eqnarray}
\normalsize

\par
\textbf{Point~\ref{lemma:substitution-property-1}} is:
\small
\begin{align}
\dpth{\Pi_{M\subs{N}{x}}}
&=\dpth{\Pi'_{M\subs{N}{x}}}
\label{align:07-02-18-01}\\
&\leq\max\{\dpth{\Pi_{M}},\dpth{\Pi_{N}}\}
\nonumber\\
&=\max\{\dpth{\Pi_{\bs x.M}},\dpth{\Pi_{N}}\}
 =\dpth{\Pi_{(\bs x.M)N}}
\nonumber
\end{align}
\normalsize
where \eqref{align:07-02-18-01} holds because the rules in $\sigma$ and $\rho$ have a single premise and are different from $!$ and $\$$.

\par
\textbf{Point~\ref{lemma:substitution-property-2}} holds because both $\wdth{0}{\Pi}=0$, for every $\Pi$, and:
\small
\begin{align}
\wdth{1}{\Pi_{M\subs{N}{x}}}
&=\wdth{1}{\Pi'_{M\subs{N}{x}}}+r+s
\label{align:07-02-18-03}\\
&<\wdth{1}{\Pi_{M}}+r+s+1
\label{align:07-02-18-04}\\
&=\wdth{1}{\Pi_{\bs x.M}}+r+1
\label{align:07-02-18-06}\\
&=\wdth{1}{\Pi_{(\bs x.M)N}}+r
\label{align:07-02-18-05}\\
&=\wdth{1}{\Pi_{(\bs x.M)N}}
\label{align:07-02-18-07}
\end{align}
\normalsize
where \eqref{align:07-02-18-03} holds because we apply $\rho$ and $\sigma$ to get \eqref{eqn:!abstr-subs-compl} from \eqref{eqn:!abstr-subs};
\eqref{align:07-02-18-04} holds using point (ii) here above;
\eqref{align:07-02-18-06} holds because we know that \eqref{eqn:genericabstr} is followed by $\rho$;
\eqref{align:07-02-18-05} holds by definition of width that counts an arrow elimination;
\eqref{align:07-02-18-07} holds because we know that \eqref{eqn:genericappl} is followed by $\sigma$.
\par
\textbf{Point~\ref{lemma:substitution-property-3}}
holds from (iii) above, by observing that $\psz{0}{\Pi_M}<\psz{0}{\Pi_{(\bs x.M)N}}$ and that the partial size at level 0 does not count contractions and universal quantifications.
\par
\textbf{Point~\ref{lemma:substitution-property-4}}, for every $d\geq 1$, is:
\small
\begin{align}
\psz{d}{\Pi_{M\subs{N}{x}}}
&=
\psz{d}{\Pi'_{M\subs{N}{x}}}+r+s
\label{align:07-26-05-01}
\\
&\leq
\psz{d}{\Pi_{M}}+\psz{d}{\Pi_{N}}+r+s
\label{align:07-26-05-02}
\\
&<
\psz{d}{\Pi_{M}}+r+s+1+\psz{d}{\Pi_N}+\nocc{x}{M}\psz{d}{\Pi_{N}}
\nonumber
\\
&=
\psz{d}{\Pi_{(\bs x.M)N}}+\nocc{x}{M}\psz{d}{\Pi_{N}}
\enspace .
\nonumber
\end{align}
\normalsize
Step~\eqref{align:07-26-05-01} holds because~\eqref{eqn:!abstr-subs-compl} is obtained from~\eqref{eqn:!abstr-subs} by applying the sequences of rules $\sigma$ and $\rho$.
Step~\eqref{align:07-26-05-02} follows from (iv) here above.

\par
As a \textbf{\textit{second hypothesys}} we let \eqref{eqn:genericabstr} had $R_I$ equal to $\liv I$, with assumption:
\small
\begin{eqnarray}
\Pi_{M}\rhd
\Gamma_M; \Delta_M; {\mathcal E}'''_M,(\Theta_M,\ta{x}{A};\emptyset)
  \vdash \ta{M}{C} \enspace .
\label{eqn:livabstr-sc}
\end{eqnarray}
\normalsize
So, we must have:
\small
\begin{eqnarray}
\Pi_N\rhd\emptyset;\emptyset;{\mathcal E}_N\vdash\ta{N}{\$A} \enspace .
\label{eqn:livabstr-arg-sc}
\end{eqnarray}
\normalsize
We observe that ${\mathcal E}_N\subseteq\{(\Theta_N;\emptyset)\}$, since $\Pi_N$ is the secondary premise of \eqref{eqn:genericappl}, which must be an instance of $\liv E$. 
Lemma~\ref{lemma:subst-vs-wght}, point~\ref{lemma:subst-vs-wght-05}, applied to
\eqref{eqn:livabstr-sc} and \eqref{eqn:livabstr-arg-sc} implies 
the existence of:
\small
\begin{align}
\Pi'_{M\subs{N}{x}}
\rhd
\Gamma_M; \Delta_M;
{\mathcal E}'''_M,
\{(\Theta_M;\emptyset)\}\sqcup
{\mathcal E}_N\vdash\ta{M\subs{N}{x}}{C}
\label{eqn:livabstr-subs}
\end{align}
\normalsize
such that:
(i) 
$\dpth{\Pi'_{M\subs{N}{x}}}
=\max\{\dpth{\Pi_{M}},\dpth{\Pi_{N}}\}$;
(ii)
$\wdth{d}{\Pi'_{M\subs{N}{x}}}
=\wdth{d}{\Pi_{M}}+\wdth{d}{\Pi_{N}}$, with $d\geq0$;
(iii)
$\psz{0}{\Pi'_{M\subs{N}{x}}}
=\psz{0}{\Pi_{M}}$;
(iv)
$\psz{d}{\Pi'_{M\subs{N}{x}}}
\leq\psz{d}{\Pi_{M}}+\psz{d}{\Pi_{N}}$,
for every $d\geq1$.
\par
The sequence $\sigma$ and $\rho$ can be applied to \eqref{eqn:livabstr-subs} to get:
\small
\begin{eqnarray}
\Pi_{M\subs{N}{x}}
\rhd
\Gamma; \Delta; {\mathcal E}\vdash\ta{M\subs{N}{x}}{C}
\enspace .
\label{eqn:livabstr-subs-compl}
\end{eqnarray}
\normalsize
Under the current hypothesis, \textbf{Points~\ref{lemma:substitution-property-1}}, \textbf{\ref{lemma:substitution-property-2}},
\textbf{\ref{lemma:substitution-property-3}}, and~\textbf{\ref{lemma:substitution-property-4}}
can be obtained from \eqref{eqn:livabstr-subs-compl}, \eqref{eqn:livabstr-arg-sc}, and \eqref{eqn:livabstr-sc} in place of 
\eqref{eqn:!abstr-subs-compl}, 
\eqref{eqn:!abstr-arg-sc}, 
and 
\eqref{eqn:!abstr-sc}, 
respectively, following what we have done, for the analogous points, under the first hypothesis above.
\par
The cases where~\eqref{eqn:genericabstr} had $R_I$ equal to $\li I_{\$}$ or $\li I$, are simpler than those just detailed out.

\par{\textbf{Third case.}}
We assume $\nocc{x}{M}>1$ and $N\in$\PTT. This implies that 
\eqref{eqn:genericabstr} has $R_I$ equal to $\li I_!$ with premise:
\small
\begin{eqnarray}
\Pi_{M}\rhd
\Gamma_M; \Delta_M; 
{\mathcal E}'''_M,(\emptyset;\ta{x}{A})\vdash \ta{M}{C} \enspace .
\label{eqn:gabstr} 
\end{eqnarray}
\normalsize
Since $(\bs x.M)N$ is a redex, and $N\in$\PTT, we have $N\in$\PTV, $\FV{N}\subseteq\{y\}$, and
\small
\begin{eqnarray}
\Pi_N\rhd\Gamma_N;\Delta_N;{\mathcal E}_N\vdash\ta{N}{\, !A}\enspace .
\label{eqn:gabstr-arg-sc} 
\end{eqnarray}
\normalsize
Lemma~\ref{lemma:structural-properties-WALL}, point~\ref{lemma:structural-properties-WALL-3.1}, applied to \eqref{eqn:gabstr-arg-sc}, implies that $\FV{N}\subseteq\dom{{\mathcal E}_N}$. The assumption on $\FV{N}$ allows to have only one between ${\mathcal E}_N \subseteq \{(\ta{y}{D};\emptyset)\}$ and ${\mathcal E}_N \subseteq \{(\emptyset;\ta{y}{D})\}$. In fact, only the second case is allowed, as Lemma~\ref{lemma:structural-properties-WALL}, point~\ref{lemma:structural-properties-WALL-3.1}, says that \eqref{eqn:gabstr-arg-sc}
terminates by an instance of the rule $!$ whose precondition forces the unique free variable of $N$ to be a polynomial variable.
So, we assume ${\mathcal E}_N=\{(\emptyset;\ta{y}{D})\}$, the case with ${\mathcal E}=\emptyset$ being analogous.
If we apply Lemma~\ref{lemma:subst-vs-wght} point~\ref{lemma:subst-vs-wght-04} to \eqref{eqn:gabstr} and \eqref{eqn:gabstr-arg-sc} we get:
\small
\begin{align}
\Pi'_{M\subs{N}{x}}
\rhd
\Gamma_M,\Gamma_N; \Delta_M,\Delta_N;
{\mathcal E}'''_M\sqcup{\mathcal E}_N
\vdash\ta{M\subs{N}{x}}{C} \enspace .
\label{eqn:!!abstr-subs} 
\end{align}
\normalsize
such that:
(i) $\dpth{\Pi'_{M\subs{N}{x}}}=\max\{\dpth{\Pi_{M}},\dpth{\Pi_{N}}\}$;
(ii) $\wdth{d}{\Pi'_{M\subs{N}{x}}}\leq\wdth{d}{\Pi_{M}}$, with $0\leq d\leq 1$;
(iii) $\psz{0}{\Pi'_{M\subs{N}{x}}}=\psz{0}{\Pi_{M}}$;
(iv) $\psz{d}{\Pi'_{M\subs{N}{x}}}\leq\psz{d}{\Pi_{M}}+\nocc{x}{M}\psz{d}{\Pi_{N}}$,
for every $d\geq 1$.
\par
We can now apply the sequences of rules $\rho$ and $\sigma$ with $s+r$ instances of $C$ to \eqref{eqn:!!abstr-subs} and get:
\small
\begin{eqnarray}
\Pi_{M\subs{N}{x}}
\rhd
\Gamma; \Delta;{\mathcal E}\vdash\ta{M\subs{N}{x}}{C}\enspace .
\label{eqanarray:07-02-18-10}
\end{eqnarray}
\normalsize
\textbf{Point~\ref{lemma:substitution-property-1}} is:
\small
\begin{align}
\dpth{\Pi_M\subs{N}{x}}
&=\dpth{\Pi'_M\subs{N}{x}}
\label{align:07-02-18-20}\\
&=\max\{\dpth{\Pi_M},\dpth{\Pi_N}\}\enspace .
\label{align:07-02-18-21}
\end{align}
\normalsize
\eqref{align:07-02-18-20} holds because from \eqref{eqn:!!abstr-subs} to \eqref{eqanarray:07-02-18-10} we apply the sequences of rules $\sigma$ and $\rho$ that do not change the depth. \eqref{align:07-02-18-21} holds thanks to point (i) above.
\par
\textbf{Point~\ref{lemma:substitution-property-2}} holds because both $\wdth{0}{\Pi}=0$, for every $\Pi$, and:
\small
\begin{align}
\wdth{1}{\Pi_M\subs{N}{x}}
&=\wdth{1}{\Pi'_M\subs{N}{x}}+s+r
\label{align:07-02-18-23}\\
&<\wdth{1}{\Pi_M}+s+\wdth{1}{\Pi_N}+1+r
\label{align:07-02-18-24}\\
&=\wdth{1}{\Pi_{\bs x.M}}+\wdth{1}{\Pi_M}+1+r
\label{align:07-02-18-26}\\
&=\wdth{1}{\Pi_{(\bs x.M)N}}
\label{align:07-02-18-27}
\enspace .
\end{align}
\normalsize
\eqref{align:07-02-18-23} holds because from \eqref{eqn:!!abstr-subs} to \eqref{eqanarray:07-02-18-10} we apply $s+r$ instances of $C$ and some instances of 
$\forall I, \forall E$. 
\eqref{align:07-02-18-24} holds thanks to point (ii) above.
\eqref{align:07-02-18-26} holds because \eqref{eqn:genericabstr} 
is followed by $s$ instances of $C$, before producing $\Pi_{\bs x.M}$. \eqref{align:07-02-18-27} holds thanks to the definition of width at depth 1 and because \eqref{eqn:genericappl}
contains one instance of an arrow elimination and is followed by $r$ instances of $C$.

\par
\textbf{Point~\ref{lemma:substitution-property-3}} is:
\small
\begin{align}
\psz{0}{\Pi_M\subs{N}{x}}
&=\psz{0}{\Pi'_M\subs{N}{x}}
\label{align:07-02-18-30}\\
&=\psz{0}{\Pi_M}
\label{align:07-02-18-31}\\
&<\psz{0}{\Pi_{(\bs x.M)N}}
\nonumber
\enspace .
\end{align}
\normalsize
\eqref{align:07-02-18-30} holds because from \eqref{eqn:!!abstr-subs} to \eqref{eqanarray:07-02-18-10} we apply the sequences $\sigma$ and $\rho$ of rules not counted as part of the size at depth $0$.
\eqref{align:07-02-18-31} holds thanks to point (iii) above. 

\par
\textbf{Point~\ref{lemma:substitution-property-4}}, for every $d\geq 1$, is:
\small
\begin{align}
\psz{d}{\Pi_M\subs{N}{x}}
&=\psz{d}{\Pi'_M\subs{N}{x}}+s+r
\label{align:07-02-18-40}\\
&\leq\psz{d}{\Pi_M}+\nocc{x}{M}\psz{d}{\Pi_N}+s+r
\label{align:07-02-18-41}\\
&<\psz{d}{\Pi_M}+s+\psz{d}{\Pi_N}+1+r+\nocc{x}{M}\psz{d}{\Pi_N}
\nonumber\\
&=\psz{d}{\Pi_{(\bs x.M)N}}+\nocc{x}{M}\psz{d}{\Pi_N}
\nonumber
\enspace .
\end{align}
\normalsize
\eqref{align:07-02-18-40} holds because from \eqref{eqn:!!abstr-subs} to \eqref{eqanarray:07-02-18-10} we apply the sequences of rules $\sigma$ and $\rho$ with $s+r$ instances of $C$.
\eqref{align:07-02-18-41} holds thanks to points (iv) above.
\paragraph*{Proof of Theorem~\ref{theorem:substitution-property}
(Subject reduction.)}
The assumption $(\bs x.P)Q\red P\subs{Q}{x}$ at depth $d$ in $\Pi_{M}$ implies the existence of $\Pi_{(\bs x.P)Q}\preceq\Pi_M$ such that $\Pi_{(\bs x.P)Q}$ is a depth $d$ in
${\Pi_M}$.
\par
Let us assume $d=0$, and proceed by induction on $\Pi_M$. The assumption $d=0$ excludes that $\Pi_M$ could conclude by $R\in\{\$,!\}$.
The base case is $\Pi_M$ equal to $\Pi_{(\bs x.P)Q}$ and, forcefully
$\Pi_N$ equal to $\Pi_{P\subs{Q}{x}}$. 
\par
Then, Lemma~\ref{lemma:substitution-property} holds and, in particular we have the following correspondences:
(i) points~\ref{lemma:substitution-property-1} and~\ref{lemma:substitution-property-2} of Lemma~\ref{lemma:substitution-property}
are points~\ref{theorem:substitution-property-1} and~\ref{theorem:substitution-property-2} of Theorem~\ref{theorem:substitution-property}, while
(ii) points~\ref{lemma:substitution-property-3} and~\ref{lemma:substitution-property-4} of Lemma~\ref{lemma:substitution-property} become
points~\ref{theorem:substitution-property-4} and~\ref{theorem:substitution-property-5} of Theorem~\ref{theorem:substitution-property}, respectively.
Finally, point~\ref{theorem:substitution-property-3} of Theorem~\ref{theorem:substitution-property} vacuously holds 
because $i$ has to be both smaller and greater than $0$.
\par
Since the last rule of $\Pi_M$ cannot be neither $\$$, nor $!$, the inductive application of Lemma~\ref{lemma:substitution-property} is routine, when $\Pi_{(\bs x.P)Q}$ is strictly a subderivation of $\Pi_M$, with $d=0$. In particular, if an instance of $C$ occurs below the conclusion of $(\bs x.P)Q$ we can write:
\small
\begin{align*}
\wdth{1}{\Pi_N}
&=\wdth{1}{\Pi_{P\subs{Q}{x}}(C,\Pi'_{P\subs{Q}{x}})}
 \leq\wdth{1}{\Pi'_{P\subs{Q}{x}})}+1\\
&\leq\wdth{1}{\Pi'_{(\bs x.P)Q})}+1
 =\wdth{1}{\Pi_{(\bs x.P)Q}(C,\Pi'_{(\bs x.P)Q})}=\wdth{1}{\Pi_M}
\enspace .
\end{align*}
\normalsize
\par
Let us now assume $d>0$, and proceed again by induction on $\Pi_M$. 
The assumption $d>0$ implies that the last rule of $\Pi_M$ cannot be $A$ and that $\Pi_M$ must contain at least one instance of the rules $\$$ and $!$.
We develop the details of the case $\Pi_M(R,\Pi'_M)$, with $R\in\{\$,!\}$, the other cases routinely applying the induction.
We observe that the redex reduces at depth $d$ in $\Pi_M$, namely, by definition of depth, at depth $d-1$ in $\Pi'_M$, originating $\Pi'_N$ which is $\Pi_N$ but the last rule $R$.
\begin{enumerate}
\item 
By induction, $\dpth{\Pi'_{N}}\leq\dpth{\Pi'_{M}}$ holds. This implies
$\dpth{\Pi'_{N}}+1\leq\dpth{\Pi'_{M}}+1$, equivalent, by definition, to
$\dpth{\Pi_{N}}\leq\dpth{\Pi_{M}}$.

\item 
By induction, $\wdth{i}{\Pi'_{N}}\leq\wdth{i}{\Pi'_{M}}$ holds for every $0\leq i\leq d$, since the redex occurs at depth $d$ in $\Pi_{M}$, hence at $d-1$ in $\Pi'_{M}$.
This implies 
$\wdth{i+1}{\Pi_{N}(R,\Pi'_N)}$
\\$\leq\wdth{i+1}{\Pi_{M}(R,\Pi'_M)}$, for every $0\leq i\leq d$, namely
$\wdth{i}{\Pi_{N}(R,\Pi'_N)}\leq\wdth{i}{\Pi_{M}(R,\Pi'_M)}$ for every $1\leq i\leq d+1$. Since $\wdth{0}{\Pi}=0$, for every $\Pi$, 
$\wdth{i}{\Pi_{N}(R,\Pi'_N)}\leq\wdth{i}{\Pi_{M}(R,\Pi'_M)}$ holds for every $0\leq i\leq d+1$.

\item
By induction, $\psz{i}{\Pi'_{N}}=\psz{i}{\Pi'_{M}}$ holds for every $0\leq i< d-1$.
This implies $\psz{i}{\Pi_{N}}=\psz{i}{\Pi_{M}}$, for every $1\leq i< d$, by definition. Since the reduction of the redex modifies $\Pi_M$ at depth $d>0$, the depth $0$ of $\Pi_M$ and $\Pi_N$ is preserved. So, $\psz{i}{\Pi_{N}}=\psz{i}{\Pi_{M}}$, for every $0\leq i< d$.

\item
By induction, $\psz{d-1}{\Pi'_{N}}<\psz{d-1}{\Pi'_{M}}$ holds. This implies
$\psz{d}{\Pi_{N}}<\psz{d}{\Pi_{M}}$ by definition.

\item
By induction, we have 
$\psz{i}{\Pi'_{N}}
\leq
\psz{i}{\Pi'_M}+\nocc{x}{P}\psz{i}{\Pi_Q}$
for every $d-1< i\leq \dpth{\Pi'_M}$. So, we can write:
\small
\begin{align*}
\psz{i+1}{\Pi_{N}(R,\Pi'_N)}
&=\psz{i}{\Pi'_{N}}\\
&\leq\psz{i}{\Pi'_M}+\nocc{x}{P}\psz{i}{\Pi_Q}\\
&=\psz{i+1}{\Pi_M(R,\Pi'_M)}+\nocc{x}{P}\psz{i+1}{\Pi_Q}
\enspace ,
\end{align*}
\normalsize
for every $d-1< i\leq \dpth{\Pi'_M}$. Namely:
\small
\begin{align*}
\psz{i}{\Pi_{N}(R,\Pi'_N)}
\leq
\psz{i}{\Pi_M(R,\Pi'_M)}+\nocc{x}{P}\psz{i}{\Pi_{Q}}
\enspace ,
\end{align*}
\normalsize
for every $d< i\leq \dpth{\Pi'_M}+1=\dpth{\Pi_M}$.
\end{enumerate}
\paragraph*{Proof of Proposition~\ref{proposition:Dynamics of the successor on strings}}
For every $n$, let $\bs fx. f^n\,x$ shorten $\UNum{n}$.
We show the statement, proceeding by cases on $n$.
\small
\begin{align*}
\USucc\, \UNum{0}
&\equiv (\bs nf.(\bs zx.f(z\,x))(n\,f))\,(\bs fx.x)\\
&\red^+\bs f.(\bs zx.f(z\,x))((\bs fx.x)\,f)
 \red^+\bs fx.f((\bs x.x)\,x)\red^+\bs fx.fx\equiv\UNum{1}
\\
\USucc\, \UNum{n}
&\equiv (\bs nf.(\bs zx.f(z\,x))(n\,f))\,(\bs fx.f^n\,x)\\
&\red^+\bs f.(\bs zx.f(z\,x))(\bs x.f^n\,x)
 \red^+\bs fx.f((\bs x.f^n\,x)\,x)\red^+\bs fx.f^{n+1}x\equiv\UNum{n+1}
\end{align*}
\normalsize
\paragraph*{Proof of Proposition~\ref{proposition:Typing rules relative to words}}
As a \textbf{first case} to show that
$\emptyset;\emptyset;\emptyset\vdash
\ta{\StepMkCompZ}
{(\alpha\li\alpha)\li(\BoolT_2\ten\alpha)\li(\BoolT_2\ten\alpha)}
$ suitably derive and compose the following judgments:
\small
\begin{align*}
&
\ta{p}{\BoolT_2};\emptyset;\emptyset\vdash
\ta
{p\lan\lan\pi_0^2,\pi_0^2\ran,\lan\pi_1^2,\pi_1^2\ran\ran}
{\BoolT_2\otimes\BoolT_2}
\\
&
\ta{x}{\alpha\li\alpha},\ta{r}{\alpha};\emptyset;\emptyset\vdash
\ta
{\bs \lan p_1\,p_2\ran.
 \lan p_1,p_2\lan \bs x.x,x\ran r\ran}
{(\BoolT_2\otimes\BoolT_2)\li(\BoolT_2\otimes\alpha)}
\end{align*}
As a \textbf{second case}, to show
$\emptyset;\emptyset;\emptyset\vdash\ta{\MkCompact}{\BIntT\li\BIntT}$
compose the following judgments:
\scriptsize
\begin{align*}
&
\emptyset;\emptyset;\emptyset
\vdash
\ta{\bs zy. (\bs\lan x\,y\ran.y)(z(\BaseMkComp\, y))}
   {\$((\BoolT_2\otimes\alpha)\li(\BoolT_2\otimes\alpha))
    \li
    \$(\alpha\li\alpha)}
\\
&
\ta{n}{\BIntT};\emptyset;
\{(\emptyset;\ta{0}{\alpha\li\alpha})\}
\vdash
\ta{n(\StepMkCompZ\, 0)}
   {\,!((\BoolT_2\otimes\alpha)\li(\BoolT_2\otimes\alpha))
    \li
    \$((\BoolT_2\otimes\alpha)\li(\BoolT_2\otimes\alpha))}
\\
&
\emptyset;\emptyset;
\{(\emptyset;\ta{1}{\alpha\li\alpha})\}
\vdash
\ta{\StepMkCompO\, 1}
   {\,!((\BoolT_2\otimes\alpha)\li(\BoolT_2\otimes\alpha))}
\end{align*}
\normalsize
\paragraph*{Proof of Proposition~\ref{proposition:Dynamics of combinators relative to words}}
For every $n$, let $\bs 01x. \{0,1\}^{n}\,x$ represent $\BNum{n}$, since $\{0,1\}^{n}\,x$ stands for the correct sequence of 0s and 1s that encode $n$ in binary.
\par
As a \textbf{first case} we show the statement relative to $\BSuccZ$  proceeding by cases on $n$.
\small
\begin{align}
\BSuccZ\, \BNum{n}
&\red^+
\MkCompact\,(\bs 01.(\bs zy.0(zy))((\bs 01x. \{0,1\}^n\,x)01))
\nonumber\\
&\red^+
\MkCompact\,(\bs 01y.0((\bs x. \{0,1\}^n\,x)y))
\nonumber\\
&\red^+
\MkCompact\,(\bs 01y.0(\{0,1\}^n\,y))
\label{eqnarray:MkCompact}
\end{align}
\normalsize
If $n=0$, then:
\scriptsize
\begin{align*}
\eqref{eqnarray:MkCompact}
&\equiv
\MkCompact\,(\bs 01y.0y)\\
&\red^+
\bs 01y.(\bs\lan x\,y\ran.y)((\StepMkCompZ\,0)(\BaseMkComp\,y))\\
&\red^+
\bs 01y.(\bs\lan x\,y\ran.y)
        ((\bs \lan p\,r\ran.
          (\bs \lan p_1\,p_2\ran.
           \lan p_1,p_2\lan \bs x.x,0\ran r\ran)
          (p\lan\lan\pi_0^2,\pi_0^2\ran,\lan\pi_1^2,\pi_1^2\ran\ran))
	 \lan \pi^2_0,y\ran)\\
&\red^+
\bs 01y.(\bs\lan x\,y\ran.y)
        (
         (\bs \lan p_1\,p_2\ran.
          \lan p_1,p_2\lan \bs x.x,0\ran y\ran)
         (\pi^2_0\lan\lan\pi_0^2,\pi_0^2\ran,\lan\pi_1^2,\pi_1^2\ran\ran)
	)\\
&\red^+
\bs 01y.(\bs\lan x\,y\ran.y)
        (
         (\bs \lan p_1\,p_2\ran.
          \lan p_1,p_2\lan \bs x.x,0\ran y\ran)
         (\lan\pi_0^2,\pi_0^2\ran)
	)\\
&\red^+
\bs 01y.(\bs\lan x\,y\ran.y)
         (\lan \pi_0^2,\pi_0^2\lan \bs x.x,0\ran y\ran)\\
&\red^+
\bs 01y.\pi_0^2\lan \bs x.x,0\ran y\\
&\red^+
\bs 01y.(\bs x.x) y\red^+ \bs 01y.y\equiv\BNum{0}
\end{align*}
\normalsize
If $n>0$, before proceeding, let us focus on some observations about the behavior of $\MkCompact$:
\small
\begin{align*}
\MkCompact \, 
  (\bs 01y.
   \nu_0(\cdots(\nu_{m-1}(1(0(\cdots(0y)\cdots))))\cdots) )
 &\red^+
  (\bs 01y.
   \nu_0(\cdots(\nu_{m-1}(1y))\cdots) )
\end{align*}
\normalsize
for every $m\geq 0$ and $\nu_m\in\{0,1\}$. Namely, $\MkCompact$ erases any occurrence of the variable name 0 to the right of the most significant bit of its argument, which, by convention, is 1. This is obtained by iterating $\StepMkCompZ 0$ and $\StepMkCompO 1$, starting from $\BaseMkComp y$.
$(\StepMkCompZ \, 0)\lan\pi^2_0,M\ran$ evaluates to $\lan\pi^2_0, M\ran$ when, as effect of the iteration, $\StepMkCompZ \, 0$ is replaced for an occurrence of 0 to the right of the most significant bit. If, on the contrary, $\StepMkCompZ \, 0$ is replaced for an occurrence of 0 to the left of the most significant bit, then $(\StepMkCompZ \, 0)\lan\pi^2_1,M\ran$ evaluates to  $\lan\pi^2_1,0 M\ran$. Finally, $(\StepMkCompO \, 1)\lan\pi^2_0,M\ran$ always evaluates to  $\lan\pi^2_1,1 M\ran$. Therefore, for some $n'$:
\small
\begin{align*}
\eqref{eqnarray:MkCompact}
&\equiv
\MkCompact\,(\bs 01y.0(\{0,1\}^n\,y))\\
&\red^+
\bs 01y.(\bs\lan x\,y\ran.y)
        (
	 (\StepMkCompZ\,0)
	  (\{\StepMkCompZ\,0,\StepMkCompO\,1\}^{n'}
	    (
	     (\StepMkCompO\,1)
	     (\BaseMkComp\,y)
	    )
	  )
	)
	\\
&\red^+
\bs 01y.(\bs\lan x\,y\ran.y)
        (
	 (\StepMkCompZ\,0)
	  (\{\StepMkCompZ\,0,\StepMkCompO\,1\}^{n'}
	    (
	     (\bs \lan p\,r\ran.\lan\pi^2_1,1\,r\ran)
	     \lan \pi^2_0,y\ran
	    )
	  )
	)
\\
&\red^+
\bs 01y.(\bs\lan x\,y\ran.y)
        (
	 (\StepMkCompZ\,0)
	  (\{\StepMkCompZ\,0,\StepMkCompO\,1\}^{n'}
	     \lan\pi^2_1,1\,y\ran
	  )
	)
\\
&\red^+
\bs 01y.(\bs\lan x\,y\ran.y)
        \lan
	 \pi^2_1,
	 0(\{0,1\}^{n'}(1\,y))
	\ran
\\
&\red^+
\bs 01y. 0(\{0,1\}^{n}\,y) \red^+ \BNum{2n}
\end{align*}
\normalsize
\par
As a \textbf{second case} we show the statement relative to
$\Branch\,\BNum{n}\,\BNum{a}\,\BNum{b}$ proceeding by cases on $\BNum{n}$:
\scriptsize
\begin{align*}
\Branch\,\BNum{0}\,\BNum{a}\,\BNum{b}
&\red^+
\bs 01. 
  (\bs w. \bs z_1 z_2.
   w\,\pi^2_0\, \lan z_1, z_2 \ran 
  )(\BNum{0}\,(\bs x.\pi^2_1)\,(\bs x.\pi^2_1))
   (\BNum{a}\,0\,1)(\BNum{b}\,0\,1)
\\
&\red^+
\bs 01. 
  (\bs w. \bs z_1 z_2.
   w\,\pi^2_0\, \lan z_1, z_2 \ran 
  )(\bs y.y)
   (\bs y.\{0,1\}^{a}\,y)(\bs y.\{0,1\}^{b}\,y)
\\
&\red^+
\bs 01. 
   (\bs y.y)\,\pi^2_0\, \lan \bs y.\{0,1\}^{a}\,y
                           , \bs y.\{0,1\}^{b}\,y
			\ran 
\\
&\red^+
\bs 01. 
   \pi^2_0\, \lan \bs y.\{0,1\}^{a}\,y
                , \bs y.\{0,1\}^{b}\,y\ran 
\red^+ \bs 01y.\{0,1\}^{a}\,y\equiv\BNum{a}
\\
\Branch\,\BNum{2n+i}\,\BNum{a}\,\BNum{b}
&\red^+
\bs 01. 
  (\bs w. \bs z_1 z_2.
   w\,\pi^2_0\, \lan z_1, z_2 \ran 
  )(\BNum{2n+i}\,(\bs x.\pi^2_1)\,(\bs x.\pi^2_1))
   (\BNum{a}\,0\,1)(\BNum{b}\,0\,1)
\\
&\red^+
\bs 01. 
  (\bs w. \bs z_1 z_2.
   w\,\pi^2_0\, \lan z_1, z_2 \ran 
  )(\bs y. \{\bs x.\pi^2_1\}^{2n+i}\,y)
   (\bs y.\{0,1\}^{a}\,y)(\bs y.\{0,1\}^{b}\,y)
\\
&\red^+
\bs 01. 
   (\bs y. \{\bs x.\pi^2_1\}^{2n+i}\,y)\,
   \pi^2_0\, \lan \bs y.\{0,1\}^{a}\,y
                           , \bs y.\{0,1\}^{b}\,y
			\ran 
\\
&\red^+
\bs 01. 
   \pi^2_1\,
   \lan \bs y.\{0,1\}^{a}\,y
                           , \bs y.\{0,1\}^{b}\,y
			\ran 
\red^+ \bs 01y.\{0,1\}^{b}\,y\equiv\BNum{b}
\end{align*}
\normalsize
\paragraph*{Proof of Proposition~\ref{proposition:Typing the embeddings}}
The last rule of the derivation that gives type to $\BEmbed{n}{M}$ is an instance of $\liv I$ whose premise is $\emptyset;\emptyset;\{(\ta{x}{\$^{n-1}L};\emptyset)\}\vdash\ta{Mx}{\$^{m+n}A}$, which requires $n\geq 1$.
\par
The last rule of the derivation that gives type to $\LEmbed{n}{p}{M}$ is an instance of $\li I_{\$}$. It is preceded by a sequence of $n\geq 0$ instances of the rule $\$$. If $n\geq 1$, then the last rules proves
$\emptyset;
\ta{x_1}{\$^{n-1}L_1},\ldots,\ta{x_p}{\$^{n-1}L_p};
\emptyset
\vdash\ta{Mx_1\ldots x_p}{\$^{m+n}A}$.
\par
The derivation giving type to $\EEmbed{n}{p}{q}{M}$ is obtained using the following judgments:
\scriptsize
\begin{align*}
\emptyset;
\emptyset;
\{(\{\ta{z_1}{\$^{m+n-1}L_1},\ldots,\ta{z_q}{\$^{m+n-1}L_q}\}
  ;\emptyset)\}
&
\vdash
\ta{\bs w_1\ldots w_p.Mw_1\ldots w_p z_1\ldots z_q}
   {(\liv_{i=1}^{p}\$^{n+1}L_i)\liv\$^{m+n}A}
\\
\emptyset;
\emptyset;
\{(\{\ta{w_i}{L_i}\};\emptyset)\}
&
\vdash\ta{\BEmbed{1}{\Coerc^n}}{\$^{n+1}L_i}
\qquad\qquad\qquad
(1\leq i\leq p)
\end{align*}
\normalsize
\paragraph*{Proof of Proposition~\ref{proposition:Typing the diagonals}}
As a \textbf{first case}, to show that
$\emptyset;\emptyset;\emptyset
\vdash
\ta{\nabla^{}_{n}}
   {\BIntT\li
    \$(\bigotimes_{i=1}^{n}\BIntT)
   }
$, with $n\geq 1$, let assume $B$ be $\bigotimes^{n}_{i=1}\BIntT$, and suitably derive and compose the following judgments:
\small
\begin{align*}
&
\emptyset;\emptyset;\emptyset\vdash
\ta
{\bs \lan x_1\ldots x_{n}\ran.
\lan\BSuccZ\, x_1, \ldots,\BSuccZ\, x_{n}\ran}
{\,!(B\li B)}
\\
&
\emptyset;\emptyset;\emptyset\vdash
\ta
{\bs \lan x_1\ldots x_{n}\ran.
\lan\BSuccO\, x_1,\ldots,\BSuccO\, x_{n}\ran}
{\,!(B\li B)}
\\
&
\ta{w}{\BIntT};\emptyset;\emptyset\vdash
w\, \!\!\begin{array}[t]{l}
		         (\bs \lan x_1\ldots x_{n}\ran.
			  \lan
			  \BSuccZ\, x_1, 
		          \ldots, 
		          \BSuccZ\, x_{n}
			  \ran
		         )\\
		         (\bs \lan x_1\ldots x_{n}\ran.
			  \lan
			  \BSuccO\, x_1, 
		          \ldots, 
		          \BSuccO\, x_{n}
			  \ran
	             )\!:\!\$(B\li B)
	                  \end{array}
\\
&
\emptyset;\emptyset;\emptyset\vdash
\ta
{\bs z. z\,\overbrace{\lan\BNum{0},\ldots,\BNum{0}\ran}^{n}}
{\$(B\li B)\li\$B}
\end{align*}
\normalsize
\par
As a \textbf{second case}, to show
$\emptyset;\emptyset;\emptyset
\vdash
\ta{\nabla^{m}_{n}}
   {\BIntT\li
    \$(\bigodot_{i=1}^{n}\$^{m}\BIntT)
   }
$, with $m,n\geq 1$, let assume $B$ be $\bigodot^{n}_{i=1}\$^{m}\BIntT$, and suitably derive and compose the following judgments:
\scriptsize
\begin{align}
\label{align:proof-type-eager-diagonal-1}
&
\emptyset;\emptyset;\{(\ta{x_i}{\$^{m-1}\BIntT};\emptyset)\}\vdash
\ta
{\BEmbed{m}{\BSuccZ}\, x_i}
{\$^{m}\BIntT}
&i\in\{1,\ldots,n\}
\\
\label{align:proof-type-eager-diagonal-2}
&
\emptyset;\emptyset;\{(\ta{x_i}{\$^{m-1}\BIntT};\emptyset)\}\vdash
\ta
{\BEmbed{m}{\BSuccO}\, x_i}
{\$^{m}\BIntT}
&i\in\{1,\ldots,n\}
\\
\nonumber
&
\emptyset;\emptyset;\emptyset\vdash
\ta
{\bs \elan x_1\ldots x_{n}\eran.
\elan\BEmbed{m}{\BSuccZ}\, x_1
    , \ldots
    ,\BEmbed{m}{\BSuccZ}\, x_{n}\eran}
{\,!(B\li B)}
\\
\nonumber
&
\emptyset;\emptyset;\emptyset\vdash
\ta
{\bs \elan x_1\ldots x_{n}\eran.
\elan\BEmbed{m}{\BSuccO}\, x_1
    ,\ldots
    ,\BEmbed{m}{\BSuccO}\, x_{n}\eran}
{\,!(B\li B)}
\\
\nonumber
&
\ta{w}{\BIntT};\emptyset;\emptyset\vdash
w\, \!\!\begin{array}[t]{l}
		         (\bs \elan x_1\ldots x_{n}\eran.
			  \elan
			  \BEmbed{m}{\BSuccZ}\, x_1, 
		          \ldots, 
		          \BEmbed{m}{\BSuccZ}\, x_{n}
			  \eran
		         )\\
		         (\bs \elan x_1\ldots x_{n}\eran.
			  \elan
			  \BEmbed{m}{\BSuccO}\, x_1, 
		          \ldots, 
		          \BEmbed{m}{\BSuccO}\, x_{n}
			  \eran
	             )\!:\!\$(B\li B)
	                  \end{array}
\\
\nonumber
&
\emptyset;\emptyset;\emptyset\vdash
\ta
{\bs z. z\,\overbrace{\elan\BNum{0}
                          ,\ldots
			  ,\BNum{0}
		      \eran}^{n}}
{\$(B\li B)\li\$B}
\end{align}
\normalsize
Observe that every
$\emptyset;\emptyset;\{(\ta{x_i}{\$^{m-1}\BIntT};\emptyset)\}
\vdash\ta{x_i}{\$^{m}\BIntT}$, argument of $\BEmbed{m}{\BSuccZ}$ and
$\BEmbed{m}{\BSuccO}$ in \eqref{align:proof-type-eager-diagonal-1} and \eqref{align:proof-type-eager-diagonal-2}, is obtained by $m\geq 1$ applications of $\$$ to
$\ta{x_i}{\BIntT};\emptyset;\emptyset\vdash\ta{x_i}{\BIntT}$.
\paragraph*{Proof of Proposition~\ref{proposition:Dynamics of the diagonals}}
For every $n$, let $\{M_0,M_1\}^{n}\,x$ represent the term $\nu_1(\ldots(\nu_m\,x)\ldots)$,
where $\nu_i\equiv M_0$ if the $i^{\text{th}}$ digit in the binary representation of $n$ is $0$, and $\nu_i\equiv M_1$ if it is $1$. 
We show
$\nabla^{}_{2}\, \BNum{a}
 \red^+\underbrace{\lan\BNum{a},\ldots,\BNum{a}\ran}_{n}$, with $n=2$,
to keep things readable:
\begin{align*}
\nabla^{}_{n}\, \BNum{a}
 &\red
  (\bs z.z\lan\BNum{0},\BNum{0}\ran)
   (\BNum{a}(\bs \lan x_1\,x_2\ran.\lan\BSuccZ\,x_1,\BSuccZ\,x_2\ran)
            (\bs \lan x_1\,x_2\ran.\lan\BSuccO\,x_1,\BSuccO\,x_2\ran))
\\
 &\red^+
  (\bs z.z\lan\BNum{0},\BNum{0}\ran)
   (\bs y.\{\BSuccZ,\BSuccO\}^{a}\,y)
\\
 &\red^+
  \{\BSuccZ,\BSuccO\}^{a}\lan\BNum{0},\BNum{0}\ran
\\
 &\red^+
  \{\BSuccZ,\BSuccO\}^{a-1}
    \lan\BSuccO\,\BNum{0},\BSuccO\,\BNum{0}\ran
\\
 &\red^+
  \{\BSuccZ,\BSuccO\}^{a-i}
    \lan\{\BSuccZ,\BSuccO\}^{i}\,\BNum{0}
       ,\{\BSuccZ,\BSuccO\}^{i}\,\BNum{0}\ran
  \red^+\lan\BNum{a},\BNum{a}\ran
\end{align*}
\par
Proceed analogously for $\nabla^{m}_{n}\, \BNum{a}$. The only observation is that the reduction of $\nabla^{m}_{n}\, \BNum{a}$ generates tuples, for example, like
$\elan\BEmbed{m}{\BSuccZ}\,\BNum{b}
      ,\ldots
      ,\BEmbed{m}{\BSuccZ}\,\BNum{b}\eran$
which evaluate to 
$\elan\BSuccZ\,\BNum{b}
      ,\ldots
      ,\BSuccZ\,\BNum{b}\eran$, for some $\BNum{b}$.
\paragraph*{Proof of Proposition~\ref{proposition:Typing the recasting combinators}}
To show
$\emptyset;
\emptyset;
\emptyset;
\vdash
\ta{\UInttoList}{\$^2A\liv\UIntT\li\ListT\, \$A}
$ just compose the two following judgments:
\small
\begin{align*}
&
\emptyset;\emptyset;\emptyset\vdash
\ta{\bs zx.z\,(\bs f.x)\,\Id}
   {\$(((\delta\li\delta)\li\alpha)
       \li
        (\delta\li\delta)\li\alpha
      )\li\$(\alpha\li\alpha)}
\\
&
\emptyset;\emptyset;
\{(\ta{k}{\$A};\ta{c}{\$A\liv\alpha\li\alpha})\}\vdash
\ta{\bs lf.c\,k\,(l\,\Id)}{\,!(((\delta\li\delta)\li\alpha)
                               \li(\delta\li\delta)\li\alpha)}
\end{align*}
\normalsize
using
$
\ta{n}{\UIntT};\emptyset;
\{(\ta{k}{\$A};\ta{c}{\$A\liv\alpha\li\alpha})\}\vdash
\ta{n(\bs lf.c\,k\,(l\,\Id))}{\,\$(((\delta\li\delta)\li\alpha)
                                  \li(\delta\li\delta)\li\alpha)}
$.
\paragraph*{Proof of Proposition~\ref{proposition:Dynamics of the recasting combinators}}
For every $n$, let $\{N\}^{n}\,x$ represent the term $N(\ldots(N\,x)\ldots)$,
We show the details of the reduction relative to $\UInttoList$, with $n\geq 0$:
\scriptsize
\begin{align*}
\UInttoList\,M\,\UNum{n}
&\red^+
\bs c.(\bs zx. z\,(\bs f.x)\Id)(\UNum{n}(\bs lf. c\,M\,(l\,\Id)))
& (M \textbf{ closed value})
\\
&\red^+
\bs c.(\bs zx. z\,(\bs f.x)\Id)(\bs y.(\bs lf. c\,M\,(l\,\Id))^{n}\,y)
\\
&\red^+
\bs c.(\bs x. (\bs y.(\bs lf. c\,M\,(l\,\Id))^{n}\,y)\,(\bs f.x)\Id)
\\
&\red^+
\bs c.(\bs x. (\bs lf. c\,M\,(l\,\Id))^{n}\,(\bs f.x)\,\Id)
\\
&\red^+
\bs c.(\bs x. (\bs lf. c\,M\,(l\,\Id))^{n-1}\,
      ((\bs lf. c \,M\,(l\,\Id))\, (\bs f.x))\,\Id)
\\
&\red^+
\bs c.(\bs x. (\bs lf. c\,M\,(l\,\Id))^{n-1}\,
      (\bs f. c\,M\,x)\,\Id)
\\
&\red^+
\bs c.(\bs x. (\bs f. \{c\,M\}^{n}\,x)\,\Id)
\red^+ \bs cx. \{c\,M\}^{n}\,x
\equiv
\underbrace{[M,\ldots,M]}_{n}
\end{align*}
\normalsize
\paragraph*{Proof of Propositions~\ref{proposition:Typing the configurations}, and~\ref{proposition:Typing the pre-configurations}}
The important point to notice is that,
if we assume to reconstruct upward the deductions that give type to the \textit{configurations} or to the \textit{pre-configurations}, we have to use a suitable number of instances of the contraction rule $C$, just before the use of $\$$. All the rest is standard.
\paragraph*{Proof of Proposition~\ref{proposition:Typing the transition function}}
As a \textbf{first case}, to show 
$\emptyset
;\emptyset
;\emptyset
\vdash
\ta{\BaseTransFunc^m}{\alpha\li\ST[\alpha,\delta;\$^m\BIntT]}$
suitably exploit the judgment:
\scriptsize
\begin{align*}
&\ta{x}
    {(\$^m\BIntT\li\alpha\li\alpha)\li
     (\$^m\BIntT\li\$^m\BIntT)\li
     \$^m\BIntT\liv
     ((\delta\li\delta)\li\alpha)\li\beta},
\ta{y}{\alpha};\emptyset;\emptyset
\\
&\phantom{\ta{x}
   (\$^m\BIntT\li\alpha\li\alpha)\li
    (\$^m\BIntT\li\$^m\BIntT)\li
    \$^m\BIntT\liv
   }\quad
\vdash
\ta{x\, (\bs xy.y)\, \LEmbed{m}{1}{\Id}\,\BNum{0}\,(\bs f.y)}
   {\beta}
\end{align*}
\normalsize
As a \textbf{second case}, to show 
$\emptyset
;\emptyset
;\emptyset
\vdash
\ta{\StepTransFunc^m[G]}
   {(\$^m\BIntT \li\alpha\li\alpha)\li
     \$^m\BIntT\liv
     \ST[\alpha,\delta;\$^m\BIntT]\li
     \ST[\alpha,\delta;\$^m\BIntT]}$,
starting from $\emptyset;\emptyset;\emptyset\vdash \ta{G}{\BIntT\li\BIntT}$,
suitably derive and compose the judgments:
\scriptsize
\begin{align*}
&
\ta{c}{\$^m\BIntT \li\alpha\li\alpha},
\ta{t}{\ST[\alpha;\delta;\$^m\BIntT]},
\\
&
\ta{x}{(\$^m\BIntT\li\alpha\li\alpha)\li
       (\$^m\BIntT\li\$^m\BIntT)\li
       \$^m\BIntT\liv
       ((\delta\li\delta)\li\alpha)\li\beta}
;\emptyset
;\{(\ta{a}{\$^{m-1}\BIntT};\emptyset)\}
\\
&
\qquad\qquad\qquad\qquad\qquad\qquad\qquad\qquad\qquad\qquad\qquad\qquad\qquad\qquad
\vdash
\ta{x\, c\, \LEmbed{m}{1}{G}\, a\, 
   (\bs f.t(\bs cgal. c(g\,a)(l\,\Id)))}
   {\beta}
\\
&
\ta{t}{\ST[\alpha,\delta;\$^m\BIntT]}
;\emptyset
;\emptyset
\vdash
\ta{\bs f.t(\bs cgal. c(g\,a)(l\,\Id))}{(\delta\li\delta)\li\alpha}
\\
&
\ta{c}{\$^m\BIntT \li\alpha\li\alpha},
\ta{x}{(\$^m\BIntT\li\alpha\li\alpha)\li
       (\$^m\BIntT\li\$^m\BIntT)\li
       \$^m\BIntT\liv
       ((\delta\li\delta)\li\alpha)\li\beta}
;\emptyset;
\\
&
\qquad\qquad\qquad\qquad\qquad\qquad\qquad\qquad\qquad\qquad\qquad
\{(\ta{a}{\$^{m-1}\BIntT};\emptyset)\}
\vdash
\ta{x\, c\, \LEmbed{m}{1}{G}\, a}
   {((\delta\li\delta)\li\alpha)\li\beta}
\end{align*}
\normalsize
As a \textbf{third case}, We show 
$\emptyset;\emptyset;\emptyset\vdash
 \ta{\NextConf_{1+\Intg{n};\Intg{s}}[F']}
    {
 \SbcConfT[\alpha_0\ldots\alpha_{\Intg{n}+\Intg{s}},\delta;m]\li
  (\$^{m}\BIntT\liv
   (\li^{\Intg{n}}_{i=0}\alpha_i)\li
   (\li^{\Intg{n}+\Intg{s}}_{j=\Intg{n}+1}\alpha_j)\li
   \gamma
  )\li
  \gamma}$.
To that purpose, it is worth defining the following type:
\scriptsize
\begin{align*}
\SV[a,a;\gamma;\$^{m_a}\BIntT]&\equiv
\SU[\alpha_{m_a},\delta;\$^{m_a}\BIntT]
\li
(\$^{m_a}\BIntT\li
(\li^{\Intg{n}}_{i=0}\alpha_i)\li
(\li^{\Intg{n}+\Intg{s}}_{j=\Intg{n}+1}\alpha_j)
 \li\gamma)
\li\gamma
\\
\SV[b,a;\gamma;\$^{m_b}\BIntT]&\equiv
\SU[\alpha_{m_b},\delta;\$^{m_b}\BIntT]\li
   \SV[b+1,a;\gamma;\$^{m_{b+1}}\BIntT]
\qquad(b<a)
\end{align*}
\normalsize
and suitably compose the following judgments:
\scriptsize
\begin{align*}
&
\ta{d_0}{\$\BIntT\li
           \alpha_{0}\li
	   \alpha_{0}},
\ldots,
\ta{d_{\Intg{n}}}{\$\BIntT\li
                    \alpha_{\Intg{n}}\li
		    \alpha_{\Intg{n}}},
\\
&
\ta{e_1}{\$^{m}\BIntT\li
               \alpha_{\Intg{n}+1}\li
	       \alpha_{\Intg{n}+1}},
\ldots,
\ta{e_{\Intg{s}}}{\$^{m}\BIntT\li
                        \alpha_{\Intg{n}+\Intg{s}}\li
			\alpha_{\Intg{n}+\Intg{s}}},
\\
&
\ta{f_0}{\$\BIntT\li\$\BIntT},
\ldots,
\ta{f_{\Intg{n}}}{\$\BIntT\li\$\BIntT},
\\
&
\ta{g_1}{\$^{m}\BIntT\li\$^{m}\BIntT},
\ldots,
\ta{g_{\Intg{s}}}{\$^{m}\BIntT\li\$^{m}\BIntT},
\\
&
\ta{n^t_0}{\alpha_{0}},
\ldots,
\ta{n^t_{\Intg{n}}}{\alpha_{\Intg{n}}},
\ta{s^t_1}{\alpha_{\Intg{n}+1}},
\ldots,
\ta{s^t_{\Intg{s}}}{\alpha_{\Intg{n}+\Intg{s}}}
;\emptyset;
\\
&
\{(\ta{n_0}{\BIntT},
   \ldots,
   \ta{n_{\Intg{n}}}{\BIntT},
   \ta{s_1}{\$^{m-1}\BIntT},
   \ldots,
   \ta{s_{\Intg{s}}}{\$^{m-1}\BIntT},
   \ta{r}{\$^{m-1}\BIntT}
  ;\emptyset)\}
\\
&
\qquad
\vdash
    \bs x. x\,
    (F'\,n_0\,n_1\ldots n_{\Intg{n}}
	    \,s_1\ldots s_{\Intg{s}}\,r)\,
    (n^t_0\,\Id)\,
    (n^t_1\,\Id)\ldots (n^t_{\Intg{n}}\,\Id)\,
    (s^t_1\,\Id)\ldots (s^t_{\Intg{s}}\,\Id)
   \!:\!
\\
&\qquad\qquad\qquad\qquad\qquad\qquad\qquad\qquad
   (\$^{m}\BIntT\liv
     (\li^{\Intg{n}}_{i=0}\alpha_i)\li
     (\li^{\Intg{n}+\Intg{s}}_{j=\Intg{n}+1}\alpha_j)\li
     \gamma
    )\li\gamma
\\
&
\emptyset;\emptyset;\{(\ta{r}{\$^{m-1}\BIntT};\emptyset)\}
\vdash
\ta{H}
   {\begin{array}[t]{l}
    (\li^{\Intg{n}}_{i=0}\SU[\alpha_i,\delta;\$\BIntT])\li
    (\li^{\Intg{n}+\Intg{s}}_{j=\Intg{n}+1}
         \SU[\alpha_i,\delta;\$^m\BIntT])\li
	 \\\qquad\qquad\qquad
     (\$^{m}\BIntT\liv
      (\li^{\Intg{n}}_{i=0}\alpha_i)\li
      (\li^{\Intg{n}+\Intg{s}}_{j=\Intg{n}+1}\alpha_j)\li
      \gamma
     )\li\gamma
    \end{array}}
\end{align*}
\normalsize
For every $i\in\{0,\ldots,\Intg{n}\}$:
\scriptsize
\begin{align*}
\ta{t_0}
   {\ST[\alpha_0,\delta;\$\BIntT]
   },
\ldots,
\ta{t_{i}}
   {\ST[\alpha_{i},\delta;\$\BIntT]}
;\emptyset;\{(\ta{r}{\$^{m-1}\BIntT};\emptyset)\}
\vdash
\ta{t_{i}(\ldots(t_1(t_0\,H))\ldots)}
   {\SV[i+1,\Intg{n}+\Intg{s},\delta;\$\BIntT]}
\end{align*}
\normalsize
For every $j\in\{1,\ldots,\Intg{s}\}$:
\scriptsize
\begin{align*}
&\ta{t'_{1}}
   {\ST[\alpha_{\Intg{n}+1},\delta;\$^{m}\BIntT]},
\ldots,
\ta{t'_j}
   {\ST[\alpha_{\Intg{n}+j},\delta;\$^{m}\BIntT]},
\\
&
\ta{t_0}
   {\ST[\alpha_0,\delta;\$\BIntT]
   },
\ldots,
\ta{t_{\Intg{n}}}
   {\ST[\alpha_{\Intg{n}},\delta;\$\BIntT]}
;\emptyset;\{(\ta{r}{\$^{m-1}\BIntT};\emptyset)\}
\\
&\qquad\qquad\qquad\qquad\qquad
\vdash
\ta{t'_{j}
	 (\ldots(t'_1(t_{\Intg{n}}
	          \ldots
	          (t_1(t_0\,H))
	          \ldots
	        ))
         \ldots)}
   {\SV[\Intg{n}+j+1,\Intg{n}+\Intg{s},\delta;\$^m\BIntT]}
\end{align*}
\normalsize
Finally:
\scriptsize
\begin{align*}
&
\ta{x}{\SbcConfT[\alpha_0,\ldots,\alpha_{\Intg{n}+\Intg{s}},\delta;m]}
;\emptyset;\emptyset
\\
&
\vdash
\ta{
x(\bs r t_0 \ldots t_{\Intg{n}}
        t'_1\ldots t'_{\Intg{s}}.
        t'_{\Intg{s}}
	 (\ldots(t'_1(t_{\Intg{n}}
	          \ldots
	          (t_1(t_0\,H))
	          \ldots
	        ))
         \ldots)
        )
}
{(\$^{m}\BIntT\liv
 (\li^{\Intg{n}}_{i=0}\alpha_i)\li
 (\li^{\Intg{n}+\Intg{s}}_{j=\Intg{n}+1}\alpha_j)\li
 \gamma
 )\li\gamma}
\end{align*}
\normalsize
Observe that 
$\emptyset;\emptyset;\{(\ta{r}{\$^{m-1}\BIntT};\emptyset)\}
\vdash\ta{r}{\$^{m}\BIntT}$,
and every
$\emptyset;\emptyset;\{(\ta{s_j}{\$^{m-1}\BIntT};\emptyset)\}
\vdash\ta{s_j}{\$^{m}\BIntT}$, arguments of $F'$, are obtained by $m$ applications of $\$$ to
$\ta{r}{\BIntT};\emptyset;\emptyset\vdash\ta{r}{\BIntT}$, and
$\ta{s_j}{\BIntT};\emptyset;\emptyset\vdash\ta{s_j}{\BIntT}$, respectively.
\par
As a \textbf{fourth case}, to show
\scriptsize
\[
\emptyset;\emptyset;\emptyset\vdash
\ta{
\HeadsandTails_{1+\Intg{n};\Intg{s}}[G]
}
{
\!\!
\begin{array}[t]{l}
\bcConfT[1+\Intg{n};\Intg{s};m]\li
(\li^{\Intg{n}}_{i=0}
  !(\$\BIntT\li\alpha_i\li\alpha_i))
\\\qquad
\li
(\li^{\Intg{n}+\Intg{s}}_{j=\Intg{n}+1}
  !(\$^{m}\BIntT\li\alpha_j\li\alpha_j))
\\\qquad\qquad
\li
\$(
(\li^{\Intg{n}}_{i=0}\alpha_i)\li
(\li^{\Intg{n}+\Intg{s}}_{j=\Intg{n}+1}\alpha_j)
\li\SbcConfT[\alpha_0\ldots\alpha_{\Intg{n}+\Intg{s}},\delta;m])
\end{array}
}
\]
\normalsize
starting from
$\emptyset;\emptyset;\emptyset
 \vdash
 \ta{G}{\BIntT\li \BIntT}$,
suitably derive and compose the following judgments:
\scriptsize
\begin{align*}
&
\ta{x}
   {\bcConfT[1+\Intg{n};\Intg{s};m]};
\emptyset;
\\
&
\{(
\ta{d_0}{\$\BIntT\li\alpha_0\li\alpha_0},
\ta{d_1}{\$\BIntT\li\alpha_{1}\li\alpha_{1}},
\ldots,
\ta{d_{\Intg{n}}}{\$\BIntT\li\alpha_{\Intg{n}}
                          \li\alpha_{\Intg{n}}}
\\
&\phantom{\{ }
,\ta{e_1}{\$^{m}\BIntT\li\alpha_{\Intg{n}+1}
                      \li\alpha_{\Intg{n}+1}},
\ldots,
\ta{e_{\Intg{s}}}{\$^{m}\BIntT\li\alpha_{\Intg{n}+\Intg{s}}
                              \li\alpha_{\Intg{n}+\Intg{s}}}
;\emptyset)\}
\\
&
\qquad
\vdash
x \begin{array}[t]{l}
    (\StepTransFunc^{1}[G]\, d_0)
    (\StepTransFunc^{1}[\Id]\,d_{1})
       \ldots(\StepTransFunc^{1}[\Id]\,d_{\Intg{n}})\\
    \phantom{(\StepTransFunc^{1}[F]\, d_0)}
    (\StepTransFunc^{m}[\Id]\,e_{1})
       \ldots(\StepTransFunc^{m}[\Id]\,e_{\Intg{s}}))
    \end{array}
\\
&
\qquad\qquad
\$((\li^{\Intg{n}}_{i=0}\ST[\alpha_i,\delta;\$\BIntT])\li
   (\li^{\Intg{s}}_{j=1}\ST[\alpha_j,\delta;\$^m\BIntT])\li
   \SbcConfT[\alpha_0\ldots\alpha_{\Intg{n}+\Intg{s}},\delta;m])
\\
&
\emptyset;\emptyset;
\{(\ta{d_0}{\$\BIntT\li\alpha_0\li\alpha_0};\emptyset)\}
\vdash
\ta{\StepTransFunc^1[G]\, d_0}
   {\,!(\$\BIntT\liv\ST[\alpha_0,\delta;\$\BIntT]
                \li \ST[\alpha_0,\delta;\$\BIntT])}
\end{align*}
\normalsize
For every $i\in\{1,\ldots,\Intg{n}\}$:
\scriptsize
\begin{align*}
\emptyset;\emptyset;
\{(\ta{d_i}{\$\BIntT\li\alpha_i\li\alpha_i};\emptyset)\}
\vdash
\ta{\StepTransFunc^1[\Id]\, d_i}
   {\,!(\$\BIntT\liv\ST[\alpha_i,\delta;\$\BIntT]
                \li \ST[\alpha_i,\delta;\$\BIntT])}
\end{align*}
\normalsize
and for every $j\in\{\Intg{n}+1
      ,\ldots
      ,\Intg\li^{\Intg{n}}_{i=1}{n}+\Intg{s}\}$:
\scriptsize
\begin{align*}
\emptyset;\emptyset;
\{(\ta{e_j}{\$^{m}\BIntT\li\alpha_j\li\alpha_j};\emptyset)\}
\vdash
\ta{\StepTransFunc^m[\Id]\, e_j}
   {\,!(\$^{m}\BIntT\liv\ST[\alpha_j,\delta;\$^m\BIntT]
                    \li \ST[\alpha_j,\delta;\$^m\BIntT])}
\end{align*}
\normalsize
Finally:
\small
\begin{align*}
&
\ta{b}
   {(\li^{\Intg{n}}_{i=0}\ST[\alpha_i,\delta;\$\BIntT])\li
    (\li^{\Intg{n}+\Intg{s}}_{j=\Intg{n}+1}
        \ST[\alpha_j,\delta;\$^m\BIntT])\li
    \SbcConfT[\alpha_0\ldots\alpha_{\Intg{n}+\Intg{s}},\delta;m]},
\\
&
\ta{w_0}{\alpha_0},
\ldots,
\ta{w_{\Intg{n}}}{\alpha_{\Intg{n}}},
\ta{z_1}{\alpha_{\Intg{n}+1}},
\ldots,
\ta{z_{\Intg{s}}}{\alpha_{\Intg{n}+\Intg{s}}},
\emptyset;\emptyset
\\
&
\qquad
\vdash   
b \begin{array}[t]{l}
	(\BaseTransFunc^1\, w_0)
	(\BaseTransFunc^1\, w_{1})
           \ldots (\BaseTransFunc^1\, w_{\Intg{n}})\\
	\phantom{(\BaseTransFunc^1\, w_0)}
	(\BaseTransFunc^m\, z_{1}) 
           \ldots (\BaseTransFunc^m\, z_{\Intg{s}})
      \!:\!\SbcConfT[\alpha_0\ldots\alpha_{\Intg{n}+\Intg{s}},\delta;m]
  \end{array}
\end{align*}
\normalsize
As a \textbf{fifth case}, to show 
$\emptyset;\emptyset;\emptyset;\vdash
 \ta{\TransFunc_{1+\Intg{n};\Intg{s}}[F,F']}
    {\bcConfT[1+\Intg{n};\Intg{s};m]\li
     \bcConfT[1+\Intg{n};\Intg{s};m]}$, starting from
$\emptyset;\emptyset;\emptyset
 \vdash
 \ta{F}{\BIntT\li \BIntT}$ and
$\emptyset;\emptyset;\emptyset
 \vdash
 \ta{F'}
    {\$\BIntT\liv
     (\liv^{\Intg{n}}_{i=1}\$\BIntT)\liv
     (\liv^{\Intg{s}}_{i=1}\$^{m}\BIntT)\liv
     \$^{m}\BIntT\liv
     \$^{m}\BIntT}$,
suitably derive and compose the following judgments:
\scriptsize
\begin{align*}
&
\ta{x}{\bcConfT[1+\Intg{n};\Intg{s};m]};\emptyset;
\\
&
\{
(\emptyset;\ta{d_0}{\$\BIntT\li\alpha_0\li\alpha_0}),
(\emptyset;\ta{d_{1}}{\$\BIntT\li\alpha_{1}\li\alpha_{1}}),
\ldots,
(\emptyset;\ta{d_{\Intg{n}}}{\$\BIntT\li
                             \alpha_{\Intg{n}}\li
		             \alpha_{\Intg{n}}}),
\\
&\phantom{\{ }
(\emptyset;
\ta{e_{1}}{\$^{m}\BIntT\li
           \alpha_{\Intg{n}+1}\li
	   \alpha_{\Intg{n}+1}}),
\ldots,
(\emptyset;
\ta{e_{\Intg{s}}}{\$^{m}\BIntT\li
                  \alpha_{\Intg{n}+\Intg{s}}\li
	          \alpha_{\Intg{n}+\Intg{s}}})
\}
\\
&
\qquad
\vdash
\HeadsandTails[F]\,x\,d_0\,d_1\ldots d_{\Intg{n}}\,
                           e_1\ldots e_{\Intg{s}}
\!:\!
\$(
(\li^{\Intg{n}}_{i=0}\alpha_i)\li
(\li^{\Intg{n}+\Intg{s}}_{j=\Intg{n}+1}\alpha_j)\li
\SbcConfT[\alpha_0\ldots\alpha_{\Intg{n}+\Intg{s}},\delta;m]
)
\\
&
\ta{b}
   {(\li^{\Intg{n}}_{i=0}\alpha_i)\li
    (\li^{\Intg{n}+\Intg{s}}_{j=\Intg{n}+1}\alpha_j)\li
    \SbcConfT[\alpha_0\ldots\alpha_{1+\Intg{n}+\Intg{s}},\delta;m]
   },
\\
&
\ta{w_0}{\alpha_0},
\ta{w_1}{\alpha_1},
\ldots,
\ta{w_{\Intg{n}}}{\alpha_{\Intg{n}}},
\ta{z_1}{\alpha_{\Intg{n}+1}},
\ldots,
\ta{z_{\Intg{s}}}{\alpha_{\Intg{n}+\Intg{s}}},
\emptyset;\emptyset
\\
&
\qquad
\vdash
\NextConf[F']\,(b\, w_0\, w_1\ldots w_{\Intg{n}}\,
                          z_1\ldots z_{\Intg{s}})
\!:\!
(\$^{m}\BIntT\liv
 (\li^{\Intg{n}}_{i=1}\alpha_i)\li
 (\li^{\Intg{n}+\Intg{s}}_{j=\Intg{n}+1}\alpha_j)\li
 \gamma)\li\gamma
\end{align*}
\normalsize
\paragraph*{Proof of Proposition~\ref{proposition:Dynamics of the transition function}}
Dynamics of $\HeadsandTails_{1+\Intg{n};\Intg{s}}[F]$:
\small
\begin{align*}
&
\HeadsandTails_{1+\Intg{n};\Intg{s}}[F]\,
\lan\!\lan
 \BNum{r},
 [\BNum{a_{1}},\ldots,\BNum{a_{\Intg{r}}}]
 \!\!\!
 \begin{array}[t]{l}
  ,[\BNum{n_{11}},\ldots,\BNum{n_{1\Intg{r}}}]
  ,\ldots
  ,[\BNum{n_{\Intg{n}1}},\ldots,\BNum{n_{\Intg{n}\Intg{r}}}]\\
  ,[\BNum{s_{11}},\ldots,\BNum{s_{1\Intg{r}}}]
  ,\ldots
  ,[\BNum{s_{\Intg{s}1}},\ldots,\BNum{s_{\Intg{s}\Intg{r}}}]
  \ran\!\ran
  \red^+
 \end{array}
\\&
\bs d_0 d_1\ldots d_{\Intg{n}}
        e_1\ldots e_{\Intg{s}}.
\\
&
(\bs b. \bs w_0 w_1\ldots w_{\Intg{n}}
                z_1\ldots z_{\Intg{s}}.
	        b \begin{array}[t]{l}
	           (\BaseTransFunc^{1}\, w_0)
	           (\BaseTransFunc^{1}\, w_{1}) 
		     \ldots (\BaseTransFunc^{1}\, w_{\Intg{n}})\\
		   \phantom{(\BaseTransFunc^{1}\, w_0)}
		   (\BaseTransFunc^{1}\, z_{1}) 
		     \ldots (\BaseTransFunc^{1}\, z_{\Intg{s}})
                  \end{array}
\\
&
)(\lan\!\lan
 \BNum{r},
 [\BNum{a_{1}},\ldots,\BNum{a_{\Intg{r}}}]
 \!\!\!
 \begin{array}[t]{l}
  ,[\BNum{n_{11}},\ldots,\BNum{n_{1\Intg{r}}}]
  ,\ldots
  ,[\BNum{n_{\Intg{n}1}},\ldots,\BNum{n_{\Intg{n}\Intg{r}}}]\\
  ,[\BNum{s_{11}},\ldots,\BNum{s_{1\Intg{r}}}]
  ,\ldots
  ,[\BNum{s_{\Intg{s}1}},\ldots,\BNum{s_{\Intg{s}\Intg{r}}}]
  \ran\!\ran
  \\
  (\StepTransFunc^{1}[F]\, d_0)
  (\StepTransFunc^{1}[\Id]\,d_{1})
     \ldots(\StepTransFunc^{1}[\Id]\,d_{\Intg{n}})\\
  \phantom{(\StepTransFunc^{1}[F]\, d_0)}
  (\StepTransFunc^{m}[\Id]\,e_{1})
     \ldots(\StepTransFunc^{m}[\Id]\,e_{\Intg{s}}))
     \red^+
  \end{array}
\\&
\bs d_0 d_1\ldots d_{\Intg{n}}
        e_1\ldots e_{\Intg{s}}.
\bs w_0 w_1\ldots w_{\Intg{n}}
        z_1\ldots z_{\Intg{s}}.
\bs x.
\\
&
 \begin{array}[t]{l}
  x\,\BNum{r}\,
  ((\StepTransFunc^{1}[F]\, d_0)\,\BNum{a_1}
   (\cdots
    ((\StepTransFunc^{1}[F]\, d_0)\,\BNum{a_{\Intg{r}}}
       (\BaseTransFunc^{1}\,w_{0}))
    \cdots)
  )
  \\
  \phantom{x\,\BNum{r}\,}
  ((\StepTransFunc^{1}[\Id]\, d_1)\,\BNum{n_{11}}
   (\cdots
    ((\StepTransFunc^{1}[\Id]\, d_1)\,\BNum{n_{1\Intg{r}}}
       (\BaseTransFunc^{1}\,w_{1}))
    \cdots)
  )
  \\
  \phantom{x\,\BNum{r}\,}
  \ldots
  ((\StepTransFunc^{1}[\Id]\, d_{\Intg{n}})\,\BNum{n_{\Intg{n}1}}
   (\cdots
    ((\StepTransFunc^{1}[\Id]\,
     d_{\Intg{n}})\,\BNum{n_{\Intg{n}\Intg{r}}}
       (\BaseTransFunc^{1}\,w_{\Intg{n}}))
    \cdots)
  )
  \\
  \phantom{x\,\BNum{r}\,}
  ((\StepTransFunc^{m}[\Id]\, e_1)\,\BNum{s_{11}}
   (\cdots
    ((\StepTransFunc^{m}[\Id]\, e_1)\,\BNum{s_{1\Intg{r}}}
       (\BaseTransFunc^{1}\,z_{1}))
    \cdots)
  )
  \\
  \phantom{x\,\BNum{r}\,}
  \ldots
  ((\StepTransFunc^{m}[\Id]\, e_{\Intg{s}})\,\BNum{s_{\Intg{s}1}}
   (\cdots
    ((\StepTransFunc^{m}[\Id]\,
     e_{\Intg{s}})\,\BNum{s_{\Intg{s}\Intg{r}}}
       (\BaseTransFunc^{1}\,z_{\Intg{s}}))
    \cdots)
  )
  \red^+
  \end{array}
\\&
\bs d_0 d_1\ldots d_{\Intg{n}}
        e_1\ldots e_{\Intg{s}}.
\bs w_0 w_1\ldots w_{\Intg{n}}
        z_1\ldots z_{\Intg{s}}.
\bs x.
\\
&
 \begin{array}[t]{l}
  x\,\BNum{r}\,
  (\bs x.
   x\,d_0\,\LEmbed{1}{1}{F}\,\BNum{a_1}
   (\bs f. d_0\,(\LEmbed{1}{1}{F}\,\BNum{a_2})
                (\cdots(d_0\,\BNum{a'_{\Intg{r}}}\,w_0)\cdots)
   ))
  \\
  \phantom{x\,\BNum{r}\,}
  (\bs x.
   x\,d_1\,\LEmbed{1}{1}{\Id}\,\BNum{n_{11}}
   (\bs f. d_1\,(\LEmbed{1}{1}{\Id}\,\BNum{n_{12}})
                (\cdots(d_1\,\BNum{n_{1\Intg{r}}}\,w_1)\cdots)
   ))
  \\
  \phantom{x\,\BNum{r}\,}
  \ldots
  (\bs x.
   x\,d_{\Intg{n}}
    \,\LEmbed{1}{1}{\Id}\,\BNum{n_{\Intg{n}1}}
    (\bs f. d_{\Intg{n}}
            \,(\LEmbed{1}{1}{\Id}\,\BNum{n_{\Intg{n}2}})
              (\cdots(d_{\Intg{n}\Intg{r}}\,\BNum{n_{\Intg{n}\Intg{r}}}
	                                  \,w_{\Intg{n}})
	       \cdots)
   ))
  \\
  \phantom{x\,\BNum{r}\,}
  (\bs x.
   x\,e_1\,\LEmbed{m}{1}{\Id}\,\BNum{s_{11}}
   (\bs f. e_1\,(\LEmbed{m}{1}{\Id}\,\BNum{s_{12}})
                (\cdots(e_1\,\BNum{s_{1\Intg{r}}}\,z_1)\cdots)
   ))
  \\
  \phantom{x\,\BNum{r}\,}
  \ldots
  (\bs x.
   x\,e_{\Intg{s}}
    \,\LEmbed{m}{1}{\Id}\,\BNum{s_{\Intg{s}1}}
    (\bs f. e_{\Intg{s}}
            \,(\LEmbed{m}{1}{\Id}\,\BNum{s_{\Intg{s}2}})
              (\cdots(e_{\Intg{s}}\,\BNum{s_{\Intg{s}\Intg{r}}}
	                          \,z_{\Intg{s}})
	       \cdots)
   ))
  \red^+
  \end{array}
\\&
\bs d_0 d_1\ldots d_{\Intg{n}}
        e_1\ldots e_{\Intg{s}}.
\bs w_0 w_1\ldots w_{\Intg{n}}
        z_1\ldots z_{\Intg{s}}.
\bs x.
\\
&
 \begin{array}[t]{l}
  x\,\BNum{r}\,
  (\bs x.
   x\,d_0\,\LEmbed{1}{1}{F}\,\BNum{a_1}
   (\bs f. d_0\,\BNum{a'_2}
                (\cdots(d_0\,\BNum{a'_{\Intg{r}}}\,w_0)\cdots)
   ))
  \\
  \phantom{x\,\BNum{r}\,}
  (\bs x.
   x\,d_1\,\LEmbed{1}{1}{\Id}\,\BNum{n_{11}}
   (\bs f. d_1\,\BNum{n_{12}}
                (\cdots(d_1\,\BNum{n_{1\Intg{r}}}\,w_1)\cdots)
   ))
  \\
  \phantom{x\,\BNum{r}\,}
  \ldots
  (\bs x.
   x\,d_{\Intg{n}}
    \,\LEmbed{1}{1}{\Id}\,\BNum{n_{\Intg{n}1}}
    (\bs f. d_{\Intg{n}}
            \,\BNum{n_{\Intg{n}2}}
              (\cdots(d_{\Intg{n}}\,\BNum{n_{\Intg{n}\Intg{r}}}
	                          \,w_{\Intg{n}})
	       \cdots)
   ))
  \\
  \phantom{x\,\BNum{r}\,}
  (\bs x.
   x\,e_1\,\LEmbed{m}{1}{\Id}\,\BNum{s_{11}}
   (\bs f. e_1\,\BNum{s_{12}}
                (\cdots(e_1\,\BNum{s_{1\Intg{r}}}\,z_1)\cdots)
   ))
  \\
  \phantom{x\,\BNum{r}\,}
  \ldots
  (\bs x.
   x\,e_{\Intg{s}}
    \,\LEmbed{m}{1}{\Id}\,\BNum{s_{\Intg{s}1}}
    (\bs f. e_{\Intg{s}}
            \,\BNum{s_{\Intg{s}2}}
              (\cdots(e_{\Intg{s}}\,\BNum{s_{\Intg{s}\Intg{r}}}
	                          \,z_{\Intg{s}})
	       \cdots)
   ))
  \red^+
  \end{array}
\end{align*}
\begin{align*}
&
\bs d_0 d_1\ldots d_{\Intg{n}}e_1\ldots e_{\Intg{s}}.
\bs w_0 w_1\ldots w_{\Intg{n}}z_1\ldots z_{\Intg{s}}.
\\&
\nonumber
(\!\!
\begin{array}[t]{l}
\lan\!\lan
\BNum{r},
\lan\BNum{a_{1}},[\BNum{a'_{2}},\ldots,\BNum{a'_{\Intg{r}}}]\ran\\
\phantom{\lan\!\BNum{r},}
    ,\lan\BNum{n_{11}},[\BNum{n_{12}},\ldots,\BNum{n_{1\Intg{r}}}]\ran
    ,\ldots
    ,\lan\BNum{n_{\Intg{n}1}},
       [\BNum{n_{\Intg{n}2}},
           \ldots,\BNum{n_{\Intg{n}\Intg{r}}}]\ran\\
    \phantom{\lan\!\BNum{r},}
    ,\lan\BNum{s_{11}},
       [\BNum{s_{12}},\ldots,\BNum{s_{1\Intg{r}}}]\ran
    ,\ldots
    ,\lan\BNum{s_{\Intg{s}1}},
        [\BNum{s_{\Intg{s}2}},
	   \ldots,\BNum{s_{\Intg{s}\Intg{r}}}]\ran
    \ran\!\ran_{\scriptsize
	         \begin{array}[t]{l}
	          c_{1} \ldots c_{\Intg{r}}\\
	          d_{11}\ldots d_{1\Intg{r}}
		    \cdots\cdots
	          d_{\Intg{n}1}\ldots d_{\Intg{n}\Intg{r}}\\
	          e_{11}\ldots e_{1\Intg{r}}
		    \cdots\cdots
	          e_{\Intg{s}1}\ldots e_{\Intg{s}\Intg{r}}\\
	          w_{0} w_{1} \ldots w_{\Intg{n}}
		        z_{1} \ldots z_{\Intg{s}}\\
	         \end{array}
		\normalsize
	       }
\end{array}
\\
&
\nonumber
  )\{
     \begin{array}[t]{l}
     ^{d_0}\!/_{c_1}\ldots^{d_0}\!/_{c_{\Intg{r}}}
     \\
     ^{d_{1}}\!/_{d_{11}}
            \ldots 
     ^{d_{1}}\!/_{d_{1\Intg{r}}}
     \cdots
     ^{d_{\Intg{n}}}\!/_{d_{\Intg{n}1}}
            \ldots
     ^{d_{\Intg{n}}}\!/_{d_{\Intg{n}\Intg{r}}}
     \
     ^{e_{1}}\!/_{e_{11}}
            \ldots
     ^{e_{1}}\!/_{e_{1\Intg{r}}}
     \cdots
     ^{e_{\Intg{s}}}\!/_{e_{\Intg{s}1}}
            \ldots
     ^{e_{\Intg{s}}}\!/_{e_{\Intg{s}\Intg{r}}}
     \\
     ^{w_{0}}/_{w_{0}}\ldots^{w_{\Intg{n}}}/_{w_{\Intg{n}}}
     \
     ^{z_{1}}/_{z_{1}}\ldots^{z_{\Intg{s}}}/_{z_{\Intg{s}}}
     \}
     \end{array}
\end{align*}
\normalsize
For the dynamics of $\NextConf_{1+\Intg{n};\Intg{s}}[F']$, we define
$\sigma$ as the substitution:
\small
\[
\begin{array}[t]{l}
^{d_0}\!/_{c_1}\ldots^{d_0}\!/_{c_{\Intg{r}}}
\\
^{d_{1}}\!/_{d_{11}}
	\ldots 
^{d_{1}}\!/_{d_{1\Intg{r}}}
\cdots
^{d_{\Intg{n}}}\!/_{d_{\Intg{n}1}}
	\ldots
^{d_{\Intg{n}}}\!/_{d_{\Intg{n}\Intg{r}}}
\
^{e_{1}}\!/_{e_{11}}
	\ldots
^{e_{1}}\!/_{e_{1\Intg{r}}}
\cdots
^{e_{\Intg{s}}}\!/_{e_{\Intg{s}1}}
	\ldots
^{e_{\Intg{s}}}\!/_{e_{\Intg{s}\Intg{r}}}
\\
^{w_{0}}/_{w_{0}}\ldots^{w_{\Intg{n}}}/_{w_{\Intg{n}}}
\
^{z_{1}}/_{z_{1}}\ldots^{z_{\Intg{s}}}/_{z_{\Intg{s}}}
\end{array}
\]
\normalsize
and we recall that $H$ is:
\small
\[
\begin{array}[t]{l}
  \bs d_0 f_0 n_0 n^t_0.
  \\
  \bs d_1 f_1 n_1 n^t_1.
  \ldots
  \bs d_{\Intg{n}} f_{\Intg{n}} n_{\Intg{n}} n^t_{\Intg{n}}.
  \\
  \bs e_1 g_1 e_1 s^t_1.
  \ldots
  \bs e_{\Intg{s}} g_{\Intg{s}} s_{\Intg{s}} s^t_{\Intg{s}}.
  \\
  \bs x. x\,
         (F'\,n_0\,n_1\ldots n_{\Intg{n}}
	         \,s_1\ldots s_{\Intg{s}}\,r)\,
	 (n^t_0\,\Id)\,
	 (n^t_1\,\Id)\ldots (n^t_{\Intg{n}}\,\Id)\,
	 (s^t_1\,\Id)\ldots (s^t_{\Intg{s}}\,\Id)
\end{array}
\]
\normalsize
Then:
\small
\begin{align*}
&
\begin{array}[t]{l}
\bs d_0 d_1\ldots d_{\Intg{n}}e_1\ldots e_{\Intg{s}}.
\bs w_0 w_1\ldots w_{\Intg{n}}z_1\ldots z_{\Intg{s}}.
\\
(\NextConf_{1+\Intg{n};\Intg{s}}[F']\,
    \lan\!\lan
    \BNum{r},
    \lan\BNum{a_{1}},[\BNum{a_{2}},\ldots,\BNum{a_{\Intg{r}}}]\ran\\
    \phantom{\NextConf_{1+\Intg{n};\Intg{s}}[F']\,\lan\!\lan
             \BNum{r},}
    ,\lan\BNum{n_{11}},[\BNum{n_{12}},\ldots,\BNum{n_{1\Intg{r}}}]\ran
    ,\ldots
    ,\lan\BNum{n_{\Intg{n}1}},
         [\BNum{n_{\Intg{n}2}},
	     \ldots,\BNum{n_{\Intg{n}\Intg{r}}}]\ran\\
    \phantom{\NextConf_{1+\Intg{n};\Intg{s}}[F']\,\lan\!\lan
             \BNum{r},}
    ,\lan\BNum{s_{11}},
         [\BNum{s_{12}},\ldots,\BNum{s_{1\Intg{r}}}]\ran
    ,\ldots
    ,\lan\BNum{s_{\Intg{s}1}},
        [\BNum{s_{\Intg{s}2}},\ldots,\BNum{s_{\Intg{s}\Intg{r}}}]\ran
    \ran\!\ran_{\scriptsize
	         \begin{array}[t]{l}
	          c_{1} \ldots c_{\Intg{r}}\\
	          d_{11}\ldots d_{1\Intg{r}}
		    \cdots\cdots
	          d_{\Intg{n}1}\ldots d_{\Intg{n}\Intg{r}}\\
	          e_{11}\ldots e_{1\Intg{r}}
		    \cdots\cdots
	          e_{\Intg{s}1}\ldots e_{\Intg{s}\Intg{r}}\\
	          w_{0} w_{1} \ldots w_{\Intg{n}}
		        z_{1} \ldots z_{\Intg{s}}\\
	         \end{array}
		\normalsize}
\\
  )\,\sigma\red^+
   \end{array}
\\
&
\begin{array}[t]{l}
\bs d_0 d_1\ldots d_{\Intg{n}}e_1\ldots e_{\Intg{s}}.
\bs w_0 w_1\ldots w_{\Intg{n}}z_1\ldots z_{\Intg{s}}.
\\
(\NextConf_{1+\Intg{n};\Intg{s}}[F']\,
 (\bs x.
  x\,\BNum{r}\,
  (\bs x.
   x\,c_1\,\LEmbed{1}{1}{F}\,\BNum{a_1}
   (\bs f. c_2\,\BNum{a'_2}
                (\cdots(c_{\Intg{r}}\,\BNum{a'_{\Intg{r}}}\,w_0)\cdots)
   ))
  \\
  \phantom{\NextConf_{1+\Intg{n};\Intg{s}}[F']\,(\bs x.x\,\BNum{r}\,}
  (\bs x.
   x\,d_{11}\,\LEmbed{1}{1}{\Id}\,\BNum{n_{11}}
   (\bs f. d_{12}\,\BNum{n_{12}}
                (\cdots(d_{1\Intg{r}}\,\BNum{n_{1\Intg{r}}}\,w_1)\cdots)
   ))
  \\
  \phantom{\NextConf_{1+\Intg{n};\Intg{s}}[F']\,(\bs x.x\,\BNum{r}\,}
  \ldots
  (\bs x.
   x\,d_{\Intg{n}1}
    \,\LEmbed{1}{1}{\Id}\,\BNum{n_{\Intg{n}1}}
    (\bs f. d_{\Intg{n}2}
            \,\BNum{n_{\Intg{n}2}}
              (\cdots(d_{\Intg{n}}\,\BNum{n_{\Intg{n}\Intg{r}}}
	                          \,w_{\Intg{n}})
	       \cdots)
   ))
  \\
  \phantom{\NextConf_{1+\Intg{n};\Intg{s}}[F']\,(\bs x.x\,\BNum{r}\,}
  (\bs x.
   x\,e_{11}\,\LEmbed{m}{1}{\Id}\,\BNum{s_{11}}
   (\bs f. e_{12}\,\BNum{s_{12}}
                   (\cdots(e_{1\Intg{r}}\,\BNum{s_{1\Intg{r}}}
		                        \,z_1)\cdots)
   ))
  \\
  \phantom{\NextConf_{1+\Intg{n};\Intg{s}}[F']\,(\bs x.x\,\BNum{r}\,}
  \ldots
  (\bs x.
   x\,e_{\Intg{s}1}
    \,\LEmbed{m}{1}{\Id}\,\BNum{s_{\Intg{s}1}}
    (\bs f. e_{\Intg{s}2}
            \,\BNum{s_{\Intg{s}2}}
              (\cdots(e_{\Intg{s}\Intg{r}}\,\BNum{s_{\Intg{s}\Intg{r}}}
	                                  \,z_{\Intg{s}})
	       \cdots)
   )))
)\,\sigma
 \red^+
\end{array}
\\
&
\begin{array}[t]{l}
\bs d_0 d_1\ldots d_{\Intg{n}}e_1\ldots e_{\Intg{s}}.
\bs w_0 w_1\ldots w_{\Intg{n}}z_1\ldots z_{\Intg{s}}.
\\
((\bs r t_0 t_1 \ldots t_{\Intg{n}}
             t'_1 \ldots t'_{\Intg{s}}.
  t'_{\Intg{s}}(\ldots(t'_1(t_{\Intg{n}}(\ldots(t_0\,H)\ldots)))\ldots)
  )
  \\\phantom{(}
  \BNum{r}\,
  (\bs x.
   x\,c_1\,\LEmbed{1}{1}{F}\,\BNum{a_1}
   (\bs f. c_2\,\BNum{a'_2}
                (\cdots(c_{\Intg{r}}\,\BNum{a'_{\Intg{r}}}\,w_0)\cdots)
   ))
  \\
  \phantom{(\BNum{r}\,}
  (\bs x.
   x\,d_{11}\,\LEmbed{1}{1}{\Id}\,\BNum{n_{11}}
   (\bs f. d_{12}\,\BNum{n_{12}}
                (\cdots(d_{1\Intg{r}}\,\BNum{n_{1\Intg{r}}}\,w_1)\cdots)
   ))
  \\
  \phantom{(\BNum{r}\,}
  \ldots
  (\bs x.
   x\,d_{\Intg{n}1}
    \,\LEmbed{1}{1}{\Id}\,\BNum{n_{\Intg{n}1}}
    (\bs f. d_{\Intg{n}2}
            \,\BNum{n_{\Intg{n}2}}
              (\cdots(d_{\Intg{n}}\,\BNum{n_{\Intg{n}\Intg{r}}}
	                          \,w_{\Intg{n}})
	       \cdots)
   ))
  \\
  \phantom{(\BNum{r}\,}
  (\bs x.
   x\,e_{11}\,\LEmbed{m}{1}{\Id}\,\BNum{s_{11}}
   (\bs f. e_{12}\,\BNum{s_{12}}
                   (\cdots(e_{1\Intg{r}}\,\BNum{s_{1\Intg{r}}}
		                        \,z_1)\cdots)
   ))
  \\
  \phantom{(\BNum{r}\,}
  \ldots
  (\bs x.
   x\,e_{\Intg{s}1}
    \,\LEmbed{m}{1}{\Id}\,\BNum{s_{\Intg{s}1}}
    (\bs f. e_{\Intg{s}2}
            \,\BNum{s_{\Intg{s}2}}
              (\cdots(e_{\Intg{s}\Intg{r}}\,\BNum{s_{\Intg{s}\Intg{r}}}
	                                  \,z_{\Intg{s}})
	       \cdots)
   ))
)\,\sigma
 \red^+
\end{array}
\\
&
\begin{array}[t]{l}
\bs d_0 d_1\ldots d_{\Intg{n}}e_1\ldots e_{\Intg{s}}.
\bs w_0 w_1\ldots w_{\Intg{n}}z_1\ldots z_{\Intg{s}}.
\\
((\bs t_1 \ldots t_{\Intg{n}}
      t'_1 \ldots t'_{\Intg{s}}.
  t'_{\Intg{s}}(\ldots(t'_1(t_{\Intg{n}}(\ldots
       (t_1\,
        (\!\!\!
	 \begin{array}[t]{l}
	  \bs d_1 f_1 n_1 n^t_1
	      \ldots
	      \bs d_{\Intg{n}} f_{\Intg{n}} 
	          n_{\Intg{n}} n^t_{\Intg{n}}.
	  \\
	  \bs e_1 g_1 s_1 s^t_1
	      \ldots
	      \bs e_{\Intg{s}} g_{\Intg{n}} 
	          s_{\Intg{s}} s^t_{\Intg{n}}.
         \\
	 \bs x.x
	     \,(F'\,\BNum{a_1}\,\BNum{n_{1}}\ldots\BNum{n_{\Intg{n}}}\,
	            \BNum{s_{1}}\ldots\BNum{s_{\Intg{s}}}\,\BNum{r}
         \\\phantom{\bs x.x\,}
	       )(c_2\,\BNum{a_2}(\ldots(c_{\Intg{r}}\,\BNum{a_{\Intg{r}}}\,w_0)\ldots))
	 \\\phantom{\bs x.x\,)}
	   (n^t_1\,\Id)\ldots(n^t_{\Intg{n}}\,\Id)
	   (s^t_1\,\Id)\ldots(s^t_{\Intg{s}}\,\Id)
        ))\ldots)))\ldots)
         \end{array}
  \\
  \phantom{()}
  )(\bs x.
    x\,d_{11}\,\LEmbed{1}{\Id}\,\BNum{n_{11}}
    (\bs f. d_{12}\,\BNum{n_{12}}
                (\cdots(d_{1\Intg{r}}\,\BNum{n_{1\Intg{r}}}\,w_1)\cdots)
    ))
  \\
  \phantom{(\BNum{r}\,}
  \ldots
  (\bs x.
   x\,d_{\Intg{n}1}
    \,\LEmbed{1}{1}{\Id}\,\BNum{n_{\Intg{n}1}}
    (\bs f. d_{\Intg{n}2}
            \,\BNum{n_{\Intg{n}2}}
              (\cdots(d_{\Intg{n}\Intg{r}}\,\BNum{n_{\Intg{n}\Intg{r}}}
	                          \,w_{\Intg{n}})
	       \cdots)
   ))
  \\
  \phantom{(\BNum{r}\,}
  (\bs x.
   x\,e_{11}\,\LEmbed{m}{1}{\Id}\,\BNum{s_{11}}
   (\bs f. e_{12}\,\BNum{s_{12}}
                   (\cdots(e_{1\Intg{r}}\,\BNum{s_{1\Intg{r}}}
		                        \,z_1)\cdots)
   ))
  \\
  \phantom{(\BNum{r}\,}
  \ldots
  (\bs x.
   x\,e_{\Intg{s}1}
    \,\LEmbed{m}{1}{\Id}\,\BNum{s_{\Intg{s}1}}
    (\bs f. e_{\Intg{s}2}
            \,\BNum{s_{\Intg{s}2}}
              (\cdots(e_{\Intg{s}\Intg{r}}\,\BNum{s_{\Intg{s}\Intg{r}}}
	                                  \,z_{\Intg{s}})
	       \cdots)
   ))
)\,\sigma
 \red^+
\end{array}
\end{align*}
\begin{align*}
&
\begin{array}[t]{l}
\bs d_0 d_1\ldots d_{\Intg{n}}e_1\ldots e_{\Intg{s}}.
\bs w_0 w_1\ldots w_{\Intg{n}}z_1\ldots z_{\Intg{s}}.
\\
(
\bs x.x
\,(F'\,\BNum{a_1}\,
	\BNum{n_{11}}\ldots\BNum{n_{\Intg{n}1}}\,
	\BNum{s_{11}}\ldots\BNum{s_{\Intg{s}1}}\,\BNum{r}
)(c_2\,\BNum{a_2}(\ldots(c_{\Intg{r}}\,\BNum{a_{\Intg{r}}}\,w_0)\ldots))
\\\phantom{(\bs x.x\,(F'\,\BNum{a_1}\,
	                 \BNum{n_{11}}\ldots\BNum{n_{\Intg{n}1}}\,
	                 \BNum{s_{11}}\ldots\BNum{s_{\Intg{s}1}}\,\BNum{r})}
  (d_{12}\,\BNum{n_{12}}
          (\cdots(d_{1\Intg{r}}\,\BNum{n_{1\Intg{r}}}\,w_1)\cdots))
\\\phantom{(\bs x.x\,(F'\,\BNum{a_1}\,
	                 \BNum{n_{11}}\ldots\BNum{n_{\Intg{n}1}}\,
	                 \BNum{s_{11}}\ldots\BNum{s_{\Intg{s}1}}\,\BNum{r})\,}
  \ldots
  (d_{\Intg{n}2}
            \,\BNum{n_{\Intg{n}2}}
              (\cdots(d_{\Intg{n}\Intg{r}}\,\BNum{n_{\Intg{n}\Intg{r}}}
	                          \,w_{\Intg{n}})
	       \cdots))
\\\phantom{(\bs x.x\,(F'\,\BNum{a_1}\,
	                 \BNum{n_{11}}\ldots\BNum{n_{\Intg{n}1}}\,
	                 \BNum{s_{11}}\ldots\BNum{s_{\Intg{s}1}}\,\BNum{r})\,}
   (e_{12}\,\BNum{s_{12}}
                   (\cdots(e_{1\Intg{r}}\,\BNum{s_{1\Intg{r}}}
		                        \,z_1)\cdots))
\\\phantom{(\bs x.x\,(F'\,\BNum{a_1}\,
	                 \BNum{n_{11}}\ldots\BNum{n_{\Intg{n}1}}\,
	                 \BNum{s_{11}}\ldots\BNum{s_{\Intg{s}1}}\,\BNum{r})\,}
  \ldots
  (e_{\Intg{s}2}
            \,\BNum{s_{\Intg{s}2}}
              (\cdots(e_{\Intg{s}\Intg{r}}\,\BNum{s_{\Intg{s}\Intg{r}}}
	                                  \,z_{\Intg{s}})
	       \cdots)
  )
)\,\sigma
 \red^+
\end{array}
\\
&
\bs d_0 d_1\ldots d_{\Intg{n}}e_1\ldots e_{\Intg{s}}.
\bs w_0 w_1\ldots w_{\Intg{n}}z_1\ldots z_{\Intg{s}}.
\\
&
\bs x.x
\,\BNum{r'}
\,(d_0\,\BNum{a_2}(\ldots(d_0\,\BNum{a_{\Intg{r}}}\,w_0)\ldots))
\\
&\phantom{\bs x.x\,\BNum{r'}}
  (d_{1}\,\BNum{n_{12}}
          (\cdots(d_{1}\,\BNum{n_{1\Intg{r}}}\,w_1)\cdots))
  \ldots
  (d_{\Intg{n}}
            \,\BNum{n_{\Intg{n}2}}
              (\cdots(d_{\Intg{n}}\,\BNum{n_{\Intg{n}\Intg{r}}}
	                          \,w_{\Intg{n}})
	       \cdots))
\\
&\phantom{\bs x.x\,\BNum{r'}}
   (e_{1}\,\BNum{s_{12}}
                   (\cdots(e_{1}\,\BNum{s_{1\Intg{r}}}
		                        \,z_1)\cdots))
  \ldots
  (e_{\Intg{s}}
            \,\BNum{s_{\Intg{s}2}}
              (\cdots(e_{\Intg{s}}\,\BNum{s_{\Intg{s}\Intg{r}}}
	                                  \,z_{\Intg{s}})
	       \cdots)
  )
 \equiv
\\&
    \lan\!\lan
    \BNum{r'}
    ,[\BNum{a'_{2}},\ldots,\BNum{a'_{\Intg{r}}}]
    ,[\BNum{n_{12}},\ldots,\BNum{n_{1\Intg{r}}}]
    ,\ldots
    ,[\BNum{n_{\Intg{n}2}},\ldots,\BNum{n_{\Intg{n}\Intg{r}}}]
    ,[\BNum{s_{12}},\ldots,\BNum{s_{1\Intg{r}}}]
    ,\ldots
    ,[\BNum{s_{\Intg{s}2}},\ldots,\BNum{s_{\Intg{s}\Intg{r}}}]
    \ran\!\ran
\end{align*}
\normalsize
For the dynamics of $\TransFunc_{1+\Intg{n};\Intg{s}}[F,F']$ we, again, define $\sigma$ as the substitution
\\
\small
$
^{d_0}\!/_{c_1}\ldots^{d_0}\!/_{c_{\Intg{r}}}
\,
^{d_{1}}\!/_{d_{11}}
	\ldots 
^{d_{1}}\!/_{d_{1\Intg{r}}}
\cdots
^{d_{\Intg{n}}}\!/_{d_{\Intg{n}1}}
	\ldots
^{d_{\Intg{n}}}\!/_{d_{\Intg{n}\Intg{r}}}
\,
^{e_{1}}\!/_{e_{11}}
	\ldots
^{e_{1}}\!/_{e_{1\Intg{r}}}
\cdots
^{e_{\Intg{s}}}\!/_{e_{\Intg{s}1}}
	\ldots
^{e_{\Intg{s}}}\!/_{e_{\Intg{s}\Intg{r}}}
^{w_{0}}/_{w_{0}}\ldots^{w_{\Intg{n}}}/_{w_{\Intg{n}}}
\
^{z_{1}}/_{z_{1}}\ldots^{z_{\Intg{s}}}/_{z_{\Intg{s}}}$
\normalsize.
Then:
\scriptsize
\begin{align*}
\lefteqn{
\TransFunc_{1+\Intg{n};\Intg{s}}[F,F']\,
\lan\!\lan
 \BNum{r},
 [\BNum{a_{1}},\ldots,\BNum{a_{\Intg{r}}}]
 \!\!\!
 \begin{array}[t]{l}
  ,[\BNum{n_{11}},\ldots,\BNum{n_{1\Intg{r}}}]
  ,\ldots
  ,[\BNum{n_{\Intg{n}1}},\ldots,\BNum{n_{\Intg{n}\Intg{r}}}]\\
  ,[\BNum{s_{11}},\ldots,\BNum{s_{1\Intg{r}}}]
  ,\ldots
  ,[\BNum{s_{\Intg{s}1}},\ldots,\BNum{s_{\Intg{s}\Intg{r}}}]
  \ran\!\ran
  \red^+
 \end{array}
}
\\
&&
\bs d_0\ldots d_{\Intg{n}} e_1\ldots e_{\Intg{s}}.
\\
&&
(\bs b.\bs w_0\ldots w_{\Intg{n}} z_1\ldots z_{\Intg{s}}.
 \NextConf_{1+\Intg{n};\Intg{s}}[F'](b\,w_0\ldots w_{\Intg{n}}
                                      \,z_1\ldots z_{\Intg{s}}
                                    )
\\
&&
)(\HeadsandTails_{1+\Intg{n};\Intg{s}}[F]
  \begin{array}[t]{l}
  \lan\!\lan
   \BNum{r}
  ,[\BNum{a_{1}},\ldots,\BNum{a_{\Intg{r}}}]
  ,[\BNum{n_{11}},\ldots,\BNum{n_{1\Intg{r}}}]
  ,\ldots
  ,[\BNum{n_{\Intg{n}1}},\ldots,\BNum{n_{\Intg{n}\Intg{r}}}]
  \\\phantom{\lan\!\lan\BNum{r}
                      ,[\BNum{a_{1}},\ldots,\BNum{a_{\Intg{r}}}]}
  ,[\BNum{s_{11}},\ldots,\BNum{s_{1\Intg{r}}}]
  ,\ldots
  ,[\BNum{s_{\Intg{s}1}},\ldots,\BNum{s_{\Intg{s}\Intg{r}}}]
  \ran\!\ran
  \,d_0\ldots d_{\Intg{n}}\,e_1\ldots e_{\Intg{s}})
  \red^+
 \end{array}
\\
&&
\bs d_0\ldots d_{\Intg{n}} e_1\ldots e_{\Intg{s}}.
\\
&&
(\bs b.\bs w_0\ldots w_{\Intg{n}} z_1\ldots z_{\Intg{s}}.
 \NextConf_{1+\Intg{n};\Intg{s}}[F'](b\,w_0\ldots w_{\Intg{n}}
                                      \,z_1\ldots z_{\Intg{s}}
                                    )
\\
&&
)
( \bs w_0\ldots w_{\Intg{n}} z_1\ldots z_{\Intg{s}}.
  \bs x.
  \!\!\!
  \begin{array}[t]{l}
  x\,\BNum{r}\,
  (\bs x.
   x\,d_0\,\LEmbed{1}{1}{F}\,\BNum{a_1}
   (\bs f. d_0\,\BNum{a'_2}
                (\cdots(d_0\,\BNum{a'_{\Intg{r}}}\,w_0)\cdots)
   ))
  \\
  \phantom{x\,\BNum{r}\,}
  (\bs x.
   x\,d_1\,\LEmbed{1}{1}{\Id}\,\BNum{n_{11}}
   (\bs f. d_1\,\BNum{n_{12}}
                (\cdots(d_1\,\BNum{n_{1\Intg{r}}}\,w_1)\cdots)
   ))
  \\
  \phantom{x\,\BNum{r}\,}
  \ldots
  (\bs x.
   x\,d_{\Intg{n}}
    \,\LEmbed{1}{1}{\Id}\,\BNum{n_{\Intg{n}1}}
    (\bs f. d_{\Intg{n}}
            \,\BNum{n_{\Intg{n}2}}
              (\cdots(d_{\Intg{n}}\,\BNum{n_{\Intg{n}\Intg{r}}}
	                          \,w_{\Intg{n}})
	       \cdots)
   ))
  \\
  \phantom{x\,\BNum{r}\,}
  (\bs x.
   x\,e_1\,\LEmbed{m}{1}{\Id}\,\BNum{s_{11}}
   (\bs f. e_1\,\BNum{s_{12}}
                (\cdots(e_1\,\BNum{s_{1\Intg{r}}}\,z_1)\cdots)
   ))
  \\
  \phantom{x\,\BNum{r}\,}
  \ldots
  (\bs x.
   x\,e_{\Intg{s}}
    \,\LEmbed{m}{1}{\Id}\,\BNum{s_{\Intg{s}1}}
    (\bs f. e_{\Intg{s}}
            \,\BNum{s_{\Intg{s}2}}
              (\cdots(e_{\Intg{s}}\,\BNum{s_{\Intg{s}\Intg{r}}}
	                          \,z_{\Intg{s}})
	       \cdots)
   ))
)
  \red^+
  \end{array}
\\
&&
\bs d_0\ldots d_{\Intg{n}} e_1\ldots e_{\Intg{s}}.
\bs w_0\ldots w_{\Intg{n}} z_1\ldots z_{\Intg{s}}.
\\
&&
(\NextConf_{1+\Intg{n};\Intg{s}}[F']\,
    \lan\!\lan
    \BNum{r},
    \lan\BNum{a_{1}},[\BNum{a'_{2}},\ldots,\BNum{a'_{\Intg{r}}}]\ran
\\
&&  \phantom{\NextConf_{1+\Intg{n};\Intg{s}}[F']\,\lan\!\lan
             \BNum{r},}
    ,\lan\BNum{n_{11}},[\BNum{n_{12}},\ldots,\BNum{n_{1\Intg{r}}}]\ran
    ,\ldots
    ,\lan\BNum{n_{\Intg{n}1}},
         [\BNum{n_{\Intg{n}2}},
	     \ldots,\BNum{n_{\Intg{n}\Intg{r}}}]\ran
\\
&&  \phantom{\NextConf_{1+\Intg{n};\Intg{s}}[F']\,\lan\!\lan
             \BNum{r},}
    ,\lan\BNum{s_{11}},
         [\BNum{s_{12}},\ldots,\BNum{s_{1\Intg{r}}}]\ran
    ,\ldots
    ,\lan\BNum{s_{\Intg{s}1}},
        [\BNum{s_{\Intg{s}2}},\ldots,\BNum{s_{\Intg{s}\Intg{r}}}]\ran
    \ran\!\ran_{\scriptsize
	         \begin{array}[t]{l}
	          c_{1} \ldots c_{\Intg{r}}\\
	          d_{11}\ldots d_{1\Intg{r}}
		    \cdots\cdots
	          d_{\Intg{n}1}\ldots d_{\Intg{n}\Intg{r}}\\
	          e_{11}\ldots e_{1\Intg{r}}
		    \cdots\cdots
	          e_{\Intg{s}1}\ldots e_{\Intg{s}\Intg{r}}\\
	          w_{0} w_{1} \ldots w_{\Intg{n}}
		        z_{1} \ldots z_{\Intg{s}}\\
	         \end{array}
		\normalsize}
\\
&&
  )\,\sigma\red^+
   \begin{array}[t]{l}
    \lan\!\lan
    \BNum{r'}
    ,[\BNum{a'_{2}},\ldots,\BNum{a'_{\Intg{r}}}]
    ,[\BNum{n_{12}},\ldots,\BNum{n_{1\Intg{r}}}]
    ,\ldots
    ,[\BNum{n_{\Intg{n}2}},\ldots,\BNum{n_{\Intg{n}\Intg{r}}}]
    ,[\BNum{s_{12}},\ldots,\BNum{s_{1\Intg{r}}}]
    ,\ldots
    ,[\BNum{s_{\Intg{s}2}},\ldots,\BNum{s_{\Intg{s}\Intg{r}}}]
    \ran\!\ran
   \end{array}
\end{align*}
\normalsize
\paragraph*{Proof of Proposition~\ref{proposition:Typing the iterator}}
As a \textbf{first case}, to show 
$\emptyset;\emptyset;\emptyset\vdash
 \ta{\ListstoConf_{1+\Intg{n};\Intg{s}}}{
 \ListT\, \$\BIntT \li 
 (\li^{\Intg{n}}_{i=1} \ListT\, \$\BIntT)\li 
 (\li^{\Intg{s}}_{j=1} \ListT\, \$^{m} \BIntT) \li
 \bcConfT[1+\Intg{n};\Intg{s}};m]$
suitably derive and compose the judgments:
\scriptsize
\begin{align*}
&
\ta{l_{i}}{\ListT\,\$\BIntT};
\emptyset;
\{(\emptyset;\ta{d_{i}}{\$\BIntT\liv\alpha_{i}\li\alpha_{i}})\}
\vdash\ta{l_{i} d_{i}}{\$(\alpha_{i}\li\alpha_{i})}
\qquad(i\in\{0,\ldots,\Intg{n}\})
\\&
\ta{l_{\Intg{n}+j}}{\ListT\,\$^{m}\BIntT},
\emptyset;
\{
(\emptyset;
\ta{e_{j}}{\$^{m}\BIntT\liv\alpha_{\Intg{n}+j}
              \li\alpha_{\Intg{n}+j}}
)\}
\vdash\ta{l_{\Intg{n}+j} e_{j}}
         {\$(\alpha_{\Intg{n}+j}\li\alpha_{\Intg{n}+j})}
\qquad(j\in\{1,\ldots,\Intg{s}\})
\\
&
\ta{b_0}{\alpha_0\li\alpha_0},
\ta{b_{1}}{\alpha_{1}\li\alpha_{1}}
  ,\ldots,\ta{b_{\Intg{n}}}{\alpha_{\Intg{n}}\li\alpha_{\Intg{n}}},
\\&
\ta{c_{1}}{\alpha_{\Intg{n}+1}\li\alpha_{\Intg{n}+1}}
  ,\ldots,
  \ta{c_{\Intg{s}}}{\alpha_{\Intg{n}+\Intg{s}}
                             \li
			     \alpha_{\Intg{n}+\Intg{s}}},
\\&
\ta{w_0}{\alpha_0}
  ,\ldots,\ta{w_{\Intg{n}}}{\alpha_{\Intg{n}}},
\ta{z_{1}}{\alpha_{\Intg{n}+1}}
  ,\ldots,
  \ta{z_{\Intg{s}}}{\alpha_{\Intg{n}+\Intg{s}}},
\\
&
\ta{x}
   {\$^{m}\BIntT\liv
    (\li^{\Intg{n}}_{i=0}\alpha_i)\li
    (\li^{\Intg{n}+\Intg{s}}_{j=\Intg{n}+1}\alpha_j)\li
    \gamma
   }
;\emptyset;\emptyset
\\&\qquad\qquad\qquad
\vdash
   x\,
   \BNum{0}\,
   (b_0\, w_0)
   (b_1\, w_1)
   \ldots
   (b_{\Intg{n}}\, w_{\Intg{n}})\,
   (c_1\, z_{1})
   \ldots
   (c_{\Intg{s}}\, z_{\Intg{s}})\!:\!\gamma
\end{align*}

\begin{align*}
&
\ta{l_{0}}{\ListT\,\$\BIntT}
  ,\ldots,\ta{l_{\Intg{n}}}{\ListT\,\$\BIntT},
\ta{l_{\Intg{n}+1}}{\ListT\,\$^{m}\BIntT}
  ,\ldots,
  \ta{l_{\Intg{n}+\Intg{s}}}{\ListT\,\$^{m}\BIntT};\emptyset;
\\&
\{
(\emptyset;\ta{d_0}{\$\BIntT\liv\alpha_0\li\alpha_0)}
  ,\ldots,(\emptyset;\ta{d_{\Intg{n}}}{\$\BIntT\liv
                                       \alpha_{\Intg{n}}\li
			               \alpha_{\Intg{n}}}),
\\&\phantom{\{ }
(\emptyset;\ta{e_{1}}{\$^{m}\BIntT\liv
                      \alpha_{\Intg{n}+1}\li
	              \alpha_{\Intg{n}+1}})
  ,\ldots,
  (\emptyset;\ta{e_{\Intg{s}}}{\$^{m}\BIntT\liv
                               \alpha_{\Intg{n}+\Intg{s}}\li
	                       \alpha_{\Intg{n}+\Intg{s}})}\}
\\&\qquad
\vdash
  (\bs b_0 b_1\ldots b_{\Intg{n}} c_1\ldots c_{\Intg{s}}.
   \bs w_0 w_1\ldots w_{\Intg{n}} z_1\ldots z_{\Intg{s}}.
\\&\qquad\quad\phantom{\vdash}
  \bs x.
   x\,
   \BNum{0}\,
   (b_0\, w_0)
   (b_1\, w_1)
   \ldots
   (b_{\Intg{n}}\, w_{\Intg{n}})\,
   (c_1\, z_{1})
   \ldots
   (c_{\Intg{s}}\, z_{\Intg{s}})
)\,(l_0\,d_0)\!\!\!
        \begin{array}[t]{l}
         (l_1\, d_1) \ldots (l_{\Intg{n}}\, d_{\Intg{n}})\\
         (l_{\Intg{n}+1}\, e_{1}) 
	  \ldots (l_{\Intg{n}+\Intg{s}}\,e_{\Intg{s}})
	 \!:\!
        \end{array}
\\&\qquad\qquad
   \$((\li_{i=0}^{\Intg{n}}\alpha_i)\li
      (\li_{j=\Intg{n}+1}^{\Intg{n}+\Intg{s}}\alpha_j)\li
      \forall\gamma.((\$^m\BIntT\liv
                     (\li_{i=0}^{\Intg{n}}\alpha_i)\li
		     (\li_{j=\Intg{n}+1}^{\Intg{n}+\Intg{s}}\alpha_j)\li
		     \gamma)\li\gamma))
\end{align*}
\normalsize
As a \textbf{second case}, to show 
$\emptyset;\emptyset;\emptyset\vdash
 \ta{\BInttoBCconf_{1+\Intg{n};\Intg{s}}}
    {(\liv^{\Intg{n}}_{i=1}\$^{3}\BIntT)
     \liv
     (\liv^{\Intg{s}}_{j=1}\$^{m+2}\BIntT)
$
\\
$\liv
    \BIntT\li
     \$\bcConfT[1+\Intg{n};\Intg{s};m]}$
suitably derive and compose the following judgments:
\scriptsize
\begin{align*}
&
\ta{w}{\BIntT};\emptyset;\emptyset\vdash
\ta{\nabla_{1+\Intg{n}+\Intg{s}}\, w}
   {\$(\bigotimes_{i=1}^{1+\Intg{n}+\Intg{s}}\BIntT)}
\\
&
\emptyset;\emptyset;
\{(\ta{n_1}{\$\BIntT}
   ,\ldots,
   \ta{n_{\Intg{n}}}{\$\BIntT},
   \ta{s_1}{\$^{m}\BIntT}
   ,\ldots,
   \ta{s_{\Intg{s}}}{\$^{m}\BIntT}
  ;\emptyset)\}
\\
&
\qquad
\vdash
\bs k_0\,k_1\ldots k_{\Intg{n}}\,
         h_1\ldots h_{\Intg{s}}.
\ListstoConf_{1+\Intg{n};\Intg{s}}\,
  (\UInttoList\, \BNum{0}\, (\USucc\, (\BInttoUInt\, k_0)))
\\
&\qquad
 \phantom{\vdash
          \bs k_0\,k_1\ldots k_{\Intg{n}}\,
                   h_1\ldots h_{\Intg{s}}.
           \ListstoConf_{1+\Intg{n};\Intg{s}}\,}
  (\UInttoList\, n_{1}\,
                 (\USucc\, (\BInttoUInt\, k_1)))
  \ldots
  (\UInttoList\, n_{\Intg{n}}\,
                 (\USucc\, (\BInttoUInt\, k_{\Intg{n}})))
\\
&\qquad
 \phantom{\vdash
          \bs k_0\,k_1\ldots k_{\Intg{n}}\,
                   h_1\ldots h_{\Intg{s}}.
           \ListstoConf_{1+\Intg{n};\Intg{s}}\,}
  (\UInttoList\, s_{1}\,
                 (\USucc\, (\BInttoUInt\, h_1)))
  \ldots
  (\UInttoList\, s_{\Intg{s}}\,
                 (\USucc\, (\BInttoUInt\, h_{\Intg{s}})))
\!:\!
\\
&
\qquad
 \phantom{\vdash \bs}
 (\li_{i=1}^{1+\Intg{n}+\Intg{s}}\BIntT)\li
\bcConfT[1+\Intg{n};\Intg{s};m]
\\
&
\ta{k_0}{\BIntT};\emptyset;\emptyset\vdash
 \ta{\UInttoList\, \BNum{0}\, (\USucc\, (\BInttoUInt\, k_0))}
    {\ListT \$\BIntT}
\\
&
\ta{k_i}{\BIntT};
\emptyset;
\{(\ta{n_i}{\$\BIntT};\emptyset)\}
\vdash
\ta{\UInttoList\, n_{i}\,(\USucc\,(\BInttoUInt\, k_i))}
   {\ListT \$\BIntT}
   \qquad(i\in\{1,\ldots,\Intg{n}\})
\\
&
\ta{h_j}{\BIntT};
\emptyset;
\{(\ta{s_j}{\$^{m}\BIntT};\emptyset)\}
\vdash
\ta{\UInttoList\, s_{j}\,(\USucc\,(\BInttoUInt\, h_j))}
    {\ListT \$^{m}\BIntT}
    \qquad(j\in\{1,\ldots,\Intg{s}\})
\end{align*}
\normalsize
Observe that every
$\emptyset;\emptyset;\{(\ta{h_j}{\$^{m}\BIntT};\emptyset)\}
\vdash\ta{s_j}{\$^{m+1}\BIntT}$, argument of
$\UInttoList$, is obtained by $m+1$ applications of $\$$ to
$\ta{h_j}{\BIntT};\emptyset;\emptyset\vdash\ta{h_j}{\BIntT}$.
\par
As a \textbf{third case}, to show 
$\emptyset;\emptyset;\emptyset\vdash
 \ta{\BCconftoFConf_{1+\Intg{n};\Intg{s}}}
    {\bcConfT[1+\Intg{n};\Intg{s};m]\li
     \FbcConfT[1+\Intg{n};\Intg{s};m]}$
suitably derive and compose the following judgments:
\small
\begin{align*}
&
\ta{c}{\bcConfT[1+\Intg{n};\Intg{s};m]};\emptyset;
\\
&
\{
(\emptyset;\ta{d_0}{\$\BIntT}\liv\alpha_{1}\li\alpha_{1}),
\ldots,
(\emptyset;\ta{d_{\Intg{n}}}{\$\BIntT}\liv\alpha_{1}\li\alpha_{1}),
\\
&\phantom{\{}
(\emptyset;\ta{e_1}{\$^{m}\BIntT}\liv
           \alpha_{\Intg{n}+1}\li\alpha_{\Intg{n}+1}),
\ldots,
(\emptyset;\ta{e_{\Intg{s}}}{\$^{m}\BIntT}\liv
           \alpha_{\Intg{n}+\Intg{s}}\li\alpha_{\Intg{n}+\Intg{s}})
\}
\\
&\phantom{\{}
\vdash
\ta{c\,d_0\ldots d_{\Intg{n}}\,e_1\ldots e_{\Intg{s}}}
   {\$((\li^{\Intg{n}+\Intg{s}}_{i=0}\alpha_{i})\li
       (\forall \gamma.
        (\$^{m}\BIntT\liv
	 (\li^{\Intg{n}+\Intg{s}}_{i=0}\alpha_i)\li\gamma)\li\gamma)
      )}
\\
&
\ta{b}
   {(\li^{\Intg{n}+\Intg{s}}_{i=0}\alpha_i)\li
    (\forall \gamma.
        (\$^{m}\BIntT\liv
	 (\li^{\Intg{n}+\Intg{s}}_{i=1}\alpha_i)\li
	 \gamma)\li\gamma)}
;\emptyset;\emptyset
\\
&
\qquad\qquad\qquad
\vdash
\ta{b\, w_0\ldots w_{\Intg{n}}
     \, z_1\ldots z_{\Intg{s}}
   }
   {(\$^{m}\BIntT\liv
    (\li^{\Intg{n}+\Intg{s}}_{i=0}\alpha_i)\li
    \gamma)\li\gamma}
\end{align*}
\normalsize
As a \textbf{fourth case}, to show 
$\emptyset;\emptyset;\emptyset\vdash
 \ta{\BCconftoBInt_{1+\Intg{n};\Intg{s}}}
    {\FbcConfT[1+\Intg{n};\Intg{s};m]\li\$^{m+1}\BIntT}$
suitably derive and compose the following judgments:
\scriptsize
\begin{align*}
&
\ta{c}{\bcConfT[1+\Intg{n};\Intg{s};m]}
;\emptyset;\emptyset\vdash
\ta{c\,\overbrace{\bs xy.x
                  \cdots 
	 	  \bs xy.x}^{1+\Intg{n}+\Intg{s}}}
   {\$((\li^{\Intg{n}}_{i=0}\$\BIntT)\li
       (\li^{\Intg{n}+\Intg{s}}_{j=\Intg{n}+1}\$^{m}\BIntT)\li
\\
&
\qquad\qquad\qquad\qquad\qquad\qquad\qquad\qquad\qquad\quad
        (\$^{m}\BIntT\liv
	 (\li^{\Intg{n}}_{i=0}\$\BIntT)\li
	 (\li^{\Intg{n}+\Intg{s}}_{j=\Intg{n}+1}\$^{m}\BIntT)\li
	 \$^{m}\BIntT)\li\$^{m}\BIntT
      )}
\\
&
\ta{b}
   {(\li^{\Intg{n}}_{i=0}\$\BIntT)\li
    (\li^{\Intg{n}+\Intg{s}}_{j=\Intg{n}+1}\$^{m}\BIntT)\li
        (\$^{m}\BIntT\liv
	 (\li^{\Intg{n}}_{i=0}\$\BIntT)\li
	 (\li^{\Intg{n}+\Intg{s}}_{j=\Intg{n}+1}\$^{m}\BIntT)\li
	 \$^{m}\BIntT)\li\$^{m}\BIntT}
;\emptyset;\emptyset
\\
&
\qquad\qquad\qquad
\vdash
\ta{b\,\underbrace{\BNum{0}
                   \cdots
		   \BNum{0}}_{1+\Intg{n}+\Intg{s}}\,
           (\bs r x_0\ldots x_{\Intg{n}+\Intg{s}}. r)
   }
   {\$^{m}\BIntT}
\\
&
\emptyset;
\ta{x_0}{\BIntT},
\ldots,
\ta{x_{\Intg{n}}}{\BIntT},
\ta{x_{\Intg{n}+1}}{\$^{m-1}\BIntT},
\ldots,
\ta{x_{\Intg{n}+\Intg{s}}}{\$^{m-1}\BIntT};
\{(\ta{r}{\$^{m-1}\BIntT};\emptyset)\}
\vdash
\ta{r}
   {\$^{m}\BIntT}
\end{align*}
\normalsize
Observe that the last judgment comes from applying $m$ applications of $\$$ to $\ta{r}{\BIntT};\emptyset;\emptyset\vdash\ta{r}{\BIntT}$, the last one introducing $\ta{x_0}{\BIntT},
\ldots,
\ta{x_{\Intg{n}}}{\BIntT},
\ta{x_{\Intg{n}+1}}{\$^{m-1}\BIntT},
\ldots,$\\
$\ta{x_{\Intg{n}+\Intg{s}}}{\$^{m-1}\BIntT}$ as fake assumptions.
\par
As a \textbf{fifth case}, to show 
$\emptyset;\emptyset;\emptyset
 \vdash
 \ta{\Iter{1+\Intg{n}}{\Intg{s}}{G_0}{G_1}{G_2}}
    {\$\BIntT\liv
     (\liv^{\Intg{n}}_{i=1}\$\BIntT)\liv
     (\liv^{\Intg{s}}_{i=1}\$^{m+4}\BIntT)\liv
     \$^{m+4}\BIntT
    }
$, starting from
$\emptyset;\emptyset;\emptyset
 \vdash
 \ta{G_k}
    {\$\BIntT\liv
     (\liv^{\Intg{n}}_{i=1}\$\BIntT)\liv
     (\liv^{\Intg{s}}_{j=1}\$^{m}\BIntT)\liv
     \$^{m}\BIntT\liv
     \$^{m}\BIntT
    }
$, with  $k\in\{0,1,2\}$,
suitably derive and compose the following judgments:
\small
\begin{align*}
&
\emptyset;\emptyset;\{(\ta{n}{\BIntT};\emptyset)\}
\vdash
\ta{\LEmbed{1}{1}{\nabla_2}\,n}
   {\$^2(\BIntT\otimes\BIntT)}
\\
&
\emptyset;\emptyset;\{(\ta{n_i}{\BIntT};\emptyset)\}
\vdash
\ta{\LEmbed{1}{1}{\Coerc^4}\,n_i}
   {\$^5\BIntT}
\qquad
(i\in\{1,\ldots,\Intg{n}\})
\\
&
\emptyset;\emptyset;\emptyset\vdash
\ta{\EEmbed{1}{0}{1+\Intg{n}+\Intg{s}}{H}}
   {\begin{array}[t]{l}
    \$^2(\BIntT\otimes\BIntT)\liv
    \\\quad
    (\liv^{\Intg{n}}_{i=1}\$^5\BIntT)\liv
    (\liv^{\Intg{n}+\Intg{s}}_{j=\Intg{n}+1}\$^{m+4}\BIntT)\liv
    \$^{m+4}\BIntT
    \end{array}
   }
\\
&
\emptyset;\emptyset;
\{(\ta{n_{1}}{\$^2\BIntT},\ldots,\ta{n_{\Intg{n}}}{\$^2\BIntT},
   \ta{s_{1}}{\$^{m+1}\BIntT},
   \ldots,
   \ta{s_{\Intg{s}}}{\$^{m+1}\BIntT};\emptyset)\}
\\
&
\qquad
\vdash
\begin{array}[t]{l}
\bs ab. H'\,
 (a\,\TransFunc_{1+\Intg{n};\Intg{s}}[\BSuccZ,G_0]
   \,\TransFunc_{1+\Intg{n};\Intg{s}}[\BSuccO,G_1])
\\
\phantom{\bs ab.H'\,}
	 (\BInttoBCconf\,n_1\ldots n_{\Intg{n}}
	               \,s_1\ldots s_{\Intg{s}}
		       \,b)
\!:\!\BIntT\li\BIntT\li\$^{m+2}\BIntT
\end{array}
\\
&
\ta{b}{\BIntT};\emptyset;
\{(\ta{n_{1}}{\$^2\BIntT},\ldots,\ta{n_{\Intg{n}}}{\$^2\BIntT},
   \ta{s_{1}}{\$^{m+1}\BIntT},
   \ldots,
   \ta{s_{\Intg{s}}}{\$^{m+1}\BIntT};\emptyset)\}
\\
&
\qquad\qquad\qquad\qquad\qquad\qquad\qquad
\vdash \BInttoBCconf\,n_1\ldots n_{\Intg{n}}
	            \,s_1\ldots s_{\Intg{s}}
		    \,b
\!:\!\$\bcConfT[1+\Intg{n};\Intg{s};m]
\\
&
\ta{a}{\BIntT};\emptyset;\emptyset
\vdash 
a\,\TransFunc_{1+\Intg{n};\Intg{s}}[\BSuccZ,G_0]
 \,\TransFunc_{1+\Intg{n};\Intg{s}}[\BSuccO,G_1]
\!:\!\$(\bcConfT[1+\Intg{n};\Intg{s};m]\li
        \bcConfT[1+\Intg{n};\Intg{s};m])
\\
&
\emptyset;\emptyset;\emptyset
\vdash
\ta{H'}
{\$(\bcConfT[1+\Intg{n};\Intg{s};m]\li
    \bcConfT[1+\Intg{n};\Intg{s};m])\li
 \$\bcConfT[1+\Intg{n};\Intg{s};m]\li
 \$^{m+2}\BIntT}
\end{align*}
\normalsize
where:
\small
\begin{align*}
H&\equiv
   \bs t
       n_1\ldots n_{\Intg{n}}
       s_1\ldots s_{\Intg{s}}.
\\
&
   t\, (\bs ab.
        H'\,(a\,\TransFunc_{1+\Intg{n};\Intg{s}}[\BSuccZ,G_0]
              \,\TransFunc_{1+\Intg{n};\Intg{s}}[\BSuccO,G_1])
	 (\BInttoBCconf\,n_1\ldots n_{\Intg{n}}
	               \,s_1\ldots s_{\Intg{s}}
		       \,b)
       )
\\
H'&\equiv
\bs zy.
	 \BCconftoBInt_{1+\Intg{n};\Intg{s}}\, 
         (\BCconftoFConf_{1+\Intg{n};\Intg{s}}
	   \,(z\,(\TransFunc_{1+\Intg{n};\Intg{s}}[\Id,G_2]\,y)
	     ))
\end{align*}
\normalsize
Observe that every
$\emptyset;\emptyset;\{(\ta{n_i}{\$^{2}\BIntT};\emptyset)\}
\vdash\ta{n_i}{\$^{3}\BIntT}$, and
$\emptyset;\emptyset;\{(\ta{s_j}{\$^{m+1}\BIntT};\emptyset)\}
\vdash\ta{s_j}{\$^{m+2}\BIntT}$, 
arguments of $\BInttoBCconf$, are obtained by a suitable number of applications of $\$$ to
$\ta{n_i}{\BIntT};\emptyset;\emptyset\vdash\ta{n_i}{\BIntT}$, and
$\ta{s_j}{\BIntT};\emptyset;\emptyset\vdash\ta{s_j}{\BIntT}$, respectively.
\paragraph*{Proof of Proposition~\ref{proposition:Dynamics of the combinators for the iterator}}
The dynamics of $\ListstoConf_{1+\Intg{n};\Intg{s}}$ is:
\scriptsize
\begin{align*}
&
\ListstoConf_{1+\Intg{n};\Intg{s}}\,
 [\BNum{a_{1}},\ldots,\BNum{a_{\Intg{r}}}]\,
 [\BNum{n_{11}},\ldots,\BNum{n_{1\Intg{r}}}]
 \ldots
 [\BNum{n_{\Intg{n}1}},\ldots,\BNum{n_{\Intg{n}\Intg{r}}}]\,
 [\BNum{s_{11}},\ldots,\BNum{s_{1\Intg{r}}}]
 \ldots
 [\BNum{s_{\Intg{s}1}},\ldots,\BNum{s_{\Intg{s}\Intg{r}}}]
\red^+
\\
& 
\bs d_0 d_1\ldots d_{\Intg{n}}e_1\ldots e_{\Intg{s}}.
\\
&
(\bs b_0 b_1\ldots b_{\Intg{n}}c_1\ldots c_{\Intg{s}}.
 \bs w_0 w_1\ldots w_{\Intg{n}}z_1\ldots z_{\Intg{s}}.
 \bs x.x\, \BNum{0} (b_0\, w_0)(b_1\, w_1)\ldots(b_{\Intg{n}}\, w_{\Intg{n}})
                               (c_1\, z_1)\ldots(c_{\Intg{s}}\, z_{\Intg{s}})
\\
&
)
 ([\BNum{a_{1}},\ldots,\BNum{a_{\Intg{r}}}]\, d_0)
 ([\BNum{n_{11}},\ldots,\BNum{n_{1\Intg{r}}}]\, d_1)
 \ldots
 ([\BNum{n_{\Intg{n}1}},\ldots,\BNum{n_{\Intg{n}\Intg{r}}}]\, d_{\Intg{n}})
 ([\BNum{s_{11}},\ldots,\BNum{s_{1\Intg{r}}}]\, s_1)
 \ldots
 ([\BNum{s_{\Intg{s}1}},\ldots,\BNum{s_{\Intg{s}\Intg{r}}}]\, s_{\Intg{s}})
\red^+
\\ 
& 
\bs d_0 d_1\ldots d_{\Intg{n}}e_1\ldots e_{\Intg{s}}.
\\
&
(\bs b_0 b_1\ldots b_{\Intg{n}}c_1\ldots c_{\Intg{s}}.
 \bs w_0 w_1\ldots w_{\Intg{n}}z_1\ldots z_{\Intg{s}}.
 \bs x.x\, \BNum{0} (b_0\, w_0)(b_1\, w_1)\ldots(b_{\Intg{n}}\, w_{\Intg{n}})
                               (c_1\, z_1)\ldots(c_{\Intg{s}}\, z_{\Intg{s}})
\\
&
)
 (\bs x. d_0\, \BNum{a_{1}}(\ldots(d_0\, \BNum{a_{\Intg{r}}}\, x)\ldots))
\\
&
\phantom{)}
 (\bs x. d_1\, \BNum{n_{11}}(\ldots(d_1\,\BNum{n_{1\Intg{r}}}\, x)\ldots))
 \ldots
 (\bs x. d_\Intg{n}\, \BNum{n_{\Intg{n}1}}(\ldots(d_\Intg{n}\,\BNum{n_{\Intg{n}\Intg{r}}}\, x)\ldots))
\\&
\phantom{)}
 (\bs x. e_1\, \BNum{s_{11}}(\ldots(e_1\,\BNum{s_{1\Intg{r}}}\, x)\ldots))
 \ldots
 (\bs x. e_\Intg{s}\, \BNum{s_{\Intg{s}1}}(\ldots(s_\Intg{s}\,\BNum{s_{\Intg{s}\Intg{r}}}\, x)\ldots))
\red^+
\\ 
& 
\bs d_0 d_1\ldots d_{\Intg{n}}e_1\ldots e_{\Intg{s}}.
\\
&
(\bs b_0 b_1\ldots b_{\Intg{n}}c_1\ldots c_{\Intg{s}}.
 \bs w_0 w_1\ldots w_{\Intg{n}}z_1\ldots z_{\Intg{s}}.
 \bs x.x\, \BNum{0} (b_0\, w_0)(b_1\, w_1)\ldots(b_{\Intg{n}}\, w_{\Intg{n}})
                               (c_1\, z_1)\ldots(c_{\Intg{s}}\, z_{\Intg{s}})
\\
&
)\bs x.x\, \BNum{0}\,
 (d_0\, \BNum{a_{1}}(\ldots(d_0\, \BNum{a_{\Intg{r}}}\, w_0)\ldots))
\\
&
\phantom{)\bs x.x\, \BNum{0}\,}
 (d_1\, \BNum{n_{11}}(\ldots(d_1\,\BNum{n_{1\Intg{r}}}\, w_1)\ldots))
 \ldots
 (d_\Intg{n}\, \BNum{n_{\Intg{n}1}}(\ldots(d_\Intg{n}\,\BNum{n_{\Intg{n}\Intg{r}}}\, w_{\Intg{n}})\ldots))
\\
&
\phantom{)\bs x.x\, \BNum{0}\,}
 (e_1\, \BNum{s_{11}}(\ldots(e_1\,\BNum{s_{1\Intg{r}}}\, z_1)\ldots))
 \ldots
 (e_\Intg{s}\, \BNum{s_{\Intg{s}1}}(\ldots(s_\Intg{s}\,\BNum{s_{\Intg{s}\Intg{r}}}\, z_{\Intg{s}})\ldots))
\red^+
\\ 
&
\bcConf{
 \BNum{0} 
,[\BNum{a_{1}},\ldots,\BNum{a_{\Intg{r}}}]
,[\BNum{n_{11}},\ldots,\BNum{n_{1\Intg{r}}}]
,\ldots
,[\BNum{n_{\Intg{n}1}},\ldots,\BNum{n_{\Intg{n}\Intg{r}}}]
,[\BNum{s_{11}},\ldots,\BNum{s_{1\Intg{r}}}]
,\ldots
,[\BNum{s_{\Intg{s}1}},\ldots,\BNum{s_{\Intg{s}\Intg{r}}}]
}
\end{align*}
\normalsize
The dynamics of $\BInttoBCconf_{1+\Intg{n};\Intg{s}}$ is:
\small
\begin{align*}
&
\BInttoBCconf_{1+\Intg{n};\Intg{s}}\,
 \BNum{n_{1}} \ldots \BNum{n_{\Intg{n}}}\,
 \BNum{s_{1}} \ldots \BNum{s_{\Intg{s}}}\,
 \BNum{n}
\red^+
\\
&
(\bs t.t(\bs k_0 k_1\ldots k_{\Intg{n}} h_1\ldots h_{\Intg{s}}.
\\
&
\phantom{(\bs t.t(}
         \ListstoConf_{1+\Intg{n};\Intg{s}}\,(\UInttoList\,\BNum{0}(\USucc(\BInttoUInt\, k_0)))
\\
&
\phantom{(\bs t.t(\ListstoConf_{1+\Intg{n};\Intg{s}}\,}
         (\UInttoList\,\BNum{n_{1}}(\USucc(\BInttoUInt\, k_1)))
         \ldots
         (\UInttoList\,\BNum{n_{\Intg{n}}}(\USucc(\BInttoUInt\, k_{\Intg{n}})))
\\
&
\phantom{(\bs t.t(\ListstoConf_{1+\Intg{n};\Intg{s}}\,}
         (\UInttoList\,\BNum{s_{1}}(\USucc(\BInttoUInt\, h_1)))
         \ldots
         (\UInttoList\,\BNum{s_{\Intg{s}}}(\USucc(\BInttoUInt\, h_{\Intg{s}}))))
\\
&
)(\nabla_{1+\Intg{n}+\Intg{s}}\, \BNum{n})
\red^{+}
\\ 
& 
(\bs t.t(\bs k_0 k_1\ldots k_{\Intg{n}} h_1\ldots h_{\Intg{s}}.
\\
&
\phantom{(\bs t.t(}
         \ListstoConf_{1+\Intg{n};\Intg{s}}\,(\UInttoList\,\BNum{0}(\USucc(\BInttoUInt\, k_0)))
\\
&
\phantom{(\bs t.t(\ListstoConf_{1+\Intg{n};\Intg{s}}\,}
         (\UInttoList\,\BNum{n_{1}}(\USucc(\BInttoUInt\, k_1)))
         \ldots
         (\UInttoList\,\BNum{n_{\Intg{n}}}(\USucc(\BInttoUInt\, k_{\Intg{n}})))
\\
&
\phantom{(\bs t.t(\ListstoConf_{1+\Intg{n};\Intg{s}}\,}
         (\UInttoList\,\BNum{s_{1}}(\USucc(\BInttoUInt\, h_1)))
         \ldots
         (\UInttoList\,\BNum{s_{\Intg{s}}}(\USucc(\BInttoUInt\, h_{\Intg{s}}))))
\\
&
)\lan\underbrace{\BNum{n},\ldots,\BNum{n}}_{1+\Intg{n}+\Intg{s}}\ran
\red^{+}
\\ 
& 
(\bs z.z\, \underbrace{\BNum{n}\ldots\BNum{n}}_{1+\Intg{n}+\Intg{s}})
(\bs k_0 k_1\ldots k_{\Intg{n}} h_1\ldots h_{\Intg{s}}.
\\
&
\phantom{(\bs z.z\, \BNum{n}\ldots\BNum{n})(} \ListstoConf_{1+\Intg{n};\Intg{s}}\,(\UInttoList\,\BNum{0}(\USucc(\BInttoUInt\, k_0)))
\\
&
\phantom{(\bs z.z\, \BNum{n}\ldots\BNum{n})
         \ListstoConf_{1+\Intg{n};\Intg{s}}\,(
        }
         (\UInttoList\,\BNum{n_{1}}(\USucc(\BInttoUInt\, k_1)))
         \ldots
         (\UInttoList\,\BNum{n_{\Intg{n}}}(\USucc(\BInttoUInt\, k_{\Intg{n}})))
\\
&
\phantom{(\bs z.z\, \BNum{n}\ldots\BNum{n})
         \ListstoConf_{1+\Intg{n};\Intg{s}}\,(
        }
         (\UInttoList\,\BNum{s_{1}}(\USucc(\BInttoUInt\, h_1)))
         \ldots
         (\UInttoList\,\BNum{s_{\Intg{s}}}(\USucc(\BInttoUInt\, h_{\Intg{s}}))))
\red^{+}
\\ 
& 
\ListstoConf_{1+\Intg{n};\Intg{s}}\,(\UInttoList\,\BNum{0}(\USucc(\BInttoUInt\, \BNum{n})))
\\
&
\phantom{\ListstoConf_{1+\Intg{n};\Intg{s}}\,
        }
         (\UInttoList\,\BNum{n_{1}}(\USucc(\BInttoUInt\, \BNum{n})))
         \ldots
         (\UInttoList\,\BNum{n_{\Intg{n}}}(\USucc(\BInttoUInt\, \BNum{n})))
\\
&
\phantom{\ListstoConf_{1+\Intg{n};\Intg{s}}\,
        }
         (\UInttoList\,\BNum{s_{1}}(\USucc(\BInttoUInt\, \BNum{n})))
         \ldots
         (\UInttoList\,\BNum{s_{\Intg{s}}}(\USucc(\BInttoUInt\, \BNum{n}))))
\red^{+}
\\ 
& 
\ListstoConf_{1+\Intg{n};\Intg{s}}\,(\UInttoList\,\BNum{0}\, \BNum{m+2})
\\
&
\phantom{\ListstoConf_{1+\Intg{n};\Intg{s}}\,
        }
         (\UInttoList\,\BNum{n_{1}}\,\BNum{m+2})
         \ldots
         (\UInttoList\,\BNum{n_{\Intg{n}}}\,\BNum{m+2})
\\
&
\phantom{\ListstoConf_{1+\Intg{n};\Intg{s}}\,
        }
         (\UInttoList\,\BNum{s_{1}}\,\BNum{m+2})
         \ldots
         (\UInttoList\,\BNum{s_{\Intg{s}}}\,\BNum{m+2}))
\red^{+}
\\ 
& 
\ListstoConf_{1+\Intg{n};\Intg{s}}
\,[\overbrace{\BNum{0},\ldots,\BNum{0}}^{m+2}]
\,[\overbrace{\BNum{n_{1}},\ldots,\BNum{n_{1}}}^{m+2}]
  \ldots
  [\overbrace{\BNum{n_{\Intg{n}}},\ldots,\BNum{n_{\Intg{n}}}}^{m+2}]
\,[\overbrace{\BNum{s_{1}},\ldots,\BNum{s_{1}}}^{m+2}]
  \ldots
  [\overbrace{\BNum{s_{\Intg{s}}},\ldots,\BNum{s_{\Intg{s}}}}^{m+2}]
\red^{+}
\\ 
&
\bcConf{
 \BNum{0}
,[\BNum{0},\ldots,\BNum{0}]
,[\BNum{n_{1}},\ldots,\BNum{n_{1}}]
,\ldots
,[\BNum{n_{\Intg{n}}},\ldots,\BNum{n_{\Intg{n}}}]
,[\BNum{s_{1}},\ldots,\BNum{s_{1}}]
,\ldots
,[\BNum{s_{\Intg{s}}},\ldots,\BNum{s_{\Intg{s}}}]
}
\end{align*}
\normalsize
where we assume $m+1$ be the number of binary digits of $\BNum{n}$, namely its length $\size{\BNum{n}}$, so $m+2$ is $\size{\BNum{n}}+1$.
\par
The dynamics of $\BCconftoFConf_{1+\Intg{n};\Intg{s}}$ is:
\scriptsize
\begin{align*}
&
\BCconftoFConf_{1+\Intg{n};\Intg{s}}\,
\lan\!\lan
 \BNum{r},
 [\BNum{a_{1}},\ldots,\BNum{a_{\Intg{r}}}]
  \!\!\!
  \begin{array}[t]{l}
  ,[\BNum{n_{11}},\ldots,\BNum{n_{1\Intg{r}}}]
  ,\ldots
  ,[\BNum{n_{\Intg{n}1}},\ldots,\BNum{n_{\Intg{n}\Intg{r}}}]
  ,[\BNum{s_{11}},\ldots,\BNum{s_{1\Intg{r}}}]
  ,\ldots
  ,[\BNum{s_{\Intg{s}1}},\ldots,\BNum{s_{\Intg{s}\Intg{r}}}]
  \ran\!\ran
\red^+
 \end{array}
\\
& 
\bs d_0\ldots d_{\Intg{n}}e_1\ldots e_{\Intg{s}}.
\\
&
(\bs b  w_0\ldots w_{\Intg{n}}  z_1\ldots z_{\Intg{s}}.
     b\,w_0\ldots w_{\Intg{n}}\,z_1\ldots z_{\Intg{s}})
\\
&
\phantom{(}
(
 \lan\!\lan
  \BNum{r},
  [\BNum{a_{1}},\ldots,\BNum{a_{\Intg{r}}}]
   \!\!\!
   \begin{array}[t]{l}
   ,[\BNum{n_{11}},\ldots,\BNum{n_{1\Intg{r}}}]
   ,\ldots
   ,[\BNum{n_{\Intg{n}1}},\ldots,\BNum{n_{\Intg{n}\Intg{r}}}]
   ,[\BNum{s_{11}},\ldots,\BNum{s_{1\Intg{r}}}]
   ,\ldots
   ,[\BNum{s_{\Intg{s}1}},\ldots,\BNum{s_{\Intg{s}\Intg{r}}}]
   \ran\!\ran\,d_0\ldots d_{\Intg{n}}\,e_1\ldots e_{\Intg{s}})
   \red^+
 \end{array}
\\ 
& 
\bs d_0\ldots d_{\Intg{n}}e_1\ldots e_{\Intg{s}}.
\\
&
(\bs b  w_0\ldots w_{\Intg{n}}  z_1\ldots z_{\Intg{s}}.
     b\,w_0\ldots w_{\Intg{n}}\,z_1\ldots z_{\Intg{s}})
\\
&
\phantom{(}
(
\bs w_0 w_1\ldots w_{\Intg{n}}
        z_1\ldots z_{\Intg{s}}.
\begin{array}[t]{ll}
\bs x.
 x\,
 \BNum{r}\,
 (d_0\,\BNum{a_1}
   (\cdots(d_0\,\BNum{a_{\Intg{r}}}\,w_0)\cdots))\,
 \\\phantom{\bs x.x\,\BNum{r}\,}
 (d_1\,\BNum{n_{11}}
   (\cdots(d_1\BNum{n_{1\Intg{r}}}\,w_1)\cdots))
 \ldots
 (d_{\Intg{n}}\,\BNum{n_{\Intg{n}1}}
   (\cdots
     (d_{\Intg{n}}\,\BNum{n_{\Intg{n}\Intg{r}}}
          \,w_{\Intg{n}})\cdots))\,
 \\\phantom{\bs x.x\,\BNum{r}\,}
 (e_1\BNum{s_{11}}
   (\cdots(e_1\BNum{s_{1\Intg{r}}}z_1)\cdots))
 \ldots
 (e_{\Intg{s}}\BNum{s_{\Intg{s}1}}
   (\cdots
         (e_{\Intg{s}}\BNum{s_{\Intg{s}\Intg{r}}}
	    \,z_{\Intg{s}})\cdots))
\end{array}
\\
&
\phantom{(}
)\red^+
\\ 
& 
\bs d_0\ldots d_{\Intg{n}}e_1\ldots e_{\Intg{s}}.
\bs w_0\ldots w_{\Intg{n}}  z_1\ldots z_{\Intg{s}}.
\begin{array}[t]{ll}
\bs x.
 x\,
 \BNum{r}\,
 (d_0\,\BNum{a_1}
   (\cdots(d_0\,\BNum{a_{\Intg{r}}}\,w_0)\cdots))\,
 \\\phantom{\bs x.x\,\BNum{r}\,}
 (d_1\,\BNum{n_{11}}
   (\cdots(d_1\BNum{n_{1\Intg{r}}}\,w_1)\cdots))
 \ldots
 (d_{\Intg{n}}\,\BNum{n_{\Intg{n}1}}
   (\cdots
     (d_{\Intg{n}}\,\BNum{n_{\Intg{n}\Intg{r}}}
          \,w_{\Intg{n}})\cdots))\,
 \\\phantom{\bs x.x\,\BNum{r}\,}
 (e_1\BNum{s_{11}}
   (\cdots(e_1\BNum{s_{1\Intg{r}}}z_1)\cdots))
 \ldots
 (e_{\Intg{s}}\BNum{s_{\Intg{s}1}}
   (\cdots
         (e_{\Intg{s}}\BNum{s_{\Intg{s}\Intg{r}}}
	    \,z_{\Intg{s}})\cdots))\equiv
\end{array}
\\ 
&
\lan\!\lan
 \BNum{r},
 [\BNum{a_{1}},\ldots,\BNum{a_{\Intg{r}}}]
  ,[\BNum{n_{11}},\ldots,\BNum{n_{1\Intg{r}}}]
  ,\ldots
  ,[\BNum{n_{\Intg{n}1}},\ldots,\BNum{n_{\Intg{n}\Intg{r}}}]
  ,[\BNum{s_{11}},\ldots,\BNum{s_{1\Intg{r}}}]
  ,\ldots
  ,[\BNum{s_{\Intg{s}1}},\ldots,\BNum{s_{\Intg{s}\Intg{r}}}]
  \ran\!\ran
\end{align*}
\normalsize
\par
The dynamics of $\BCconftoBInt_{1+\Intg{n};\Intg{s}}$ is:
\scriptsize
\begin{align*}
&
\BCconftoBInt_{1+\Intg{n};\Intg{s}}\,
\lan\!\lan
 \BNum{r},
 [\BNum{a_{1}},\ldots,\BNum{a_{\Intg{r}}}]
  \!\!\!
  \begin{array}[t]{l}
  ,[\BNum{n_{11}},\ldots,\BNum{n_{1\Intg{r}}}]
  ,\ldots
  ,[\BNum{n_{\Intg{n}1}},\ldots,\BNum{n_{\Intg{n}\Intg{r}}}]
  ,[\BNum{s_{11}},\ldots,\BNum{s_{1\Intg{r}}}]
  ,\ldots
  ,[\BNum{s_{\Intg{s}1}},\ldots,\BNum{s_{\Intg{s}\Intg{r}}}]
  \ran\!\ran
  \red^+
 \end{array}
\\
& 
(\bs b. b\,\underbrace{\BNum{0}
                       \cdots
		       \BNum{0}}_{1+\Intg{n}+\Intg{s}}\,
           (\bs r x_0\ldots x_{\Intg{n}+\Intg{s}}. r)
)(
\lan\!\lan
 \BNum{r},
 [\BNum{a_{1}},\ldots,\BNum{a_{\Intg{r}}}]
  \!\!\!
  \begin{array}[t]{l}
  ,[\BNum{n_{11}},\ldots,\BNum{n_{1\Intg{r}}}]
  ,\ldots
  ,[\BNum{n_{\Intg{n}1}},\ldots,\BNum{n_{\Intg{n}\Intg{r}}}]
   \\
  ,[\BNum{s_{11}},\ldots,\BNum{s_{1\Intg{r}}}]
  ,\ldots
  ,[\BNum{s_{\Intg{s}1}},\ldots,\BNum{s_{\Intg{s}\Intg{r}}}]
  \quad
  \ran\!\ran
  \,\underbrace{(\bs xy.x)\cdots (\bs xy.x)}_{1+\Intg{n}+\Intg{s}})
\red^+
\end{array}
\\ 
& 
(\bs x.
 x\,
 \BNum{r}\,
 ((\bs xy.x)\,\BNum{a_1}
   (\cdots((\bs xy.x)\,\BNum{a_{\Intg{r}}}\,\BNum{0})\cdots))\,
\\
&
\phantom{(\bs x.x\,\BNum{r}\,}
 ((\bs xy.x)\,\BNum{n_{11}}
   (\cdots((\bs xy.x)\BNum{n_{1\Intg{r}}}\,\BNum{0})\cdots))
 \ldots
 ((\bs xy.x)\,\BNum{n_{\Intg{n}1}}
   (\cdots
     ((\bs xy.x)\,\BNum{n_{\Intg{n}\Intg{r}}}
          \,\BNum{0})\cdots))\,
\\
&
\phantom{(\bs x.x\,\BNum{r}\,}
 ((\bs xy.x)\BNum{s_{11}}
   (\cdots((\bs xy.x)\BNum{s_{1\Intg{r}}}\,\BNum{0})\cdots))
 \ldots
 ((\bs xy.x)\BNum{s_{\Intg{s}1}}
   (\cdots
         ((\bs xy.x)\BNum{s_{\Intg{s}\Intg{r}}}
	    \,\BNum{0})\cdots))
\\
&
)(\bs r x_0\ldots x_{\Intg{n}+\Intg{s}}. r)\red^+\BNum{r}
\end{align*}
\normalsize
\paragraph*{Proof of Proposition~\ref{proposition:Dynamics of the iterator}, Point~\ref{proposition:Dynamics of the iterator-part-a}}
We proceed by induction on $m-i$.
\textbf{We assume $m-i=0$.} Namely, $m=i$. We need to prove that:
\small
\[
\TransFunc_{1+\Intg{n};\Intg{s}}[\BSucc{\nu_m},G_{\nu_m}]
\lan\!\lan
 \BNum{a},
 [\BNum{0}]^k
  ,[\BNum{n_{1}}]^k
  ,\ldots
  ,[\BNum{n_{\Intg{n}}}]^k
  ,[\BNum{s_{1}}]^k
  ,\ldots
  ,[\BNum{s_{\Intg{s}}}]^k
\ran\!\ran
\]
\normalsize
rewrites to:
\small
\[
\lan\!\lan
 r[0,a,n_1,\ldots,n_{\Intg{n}},s_1,\ldots,s_{\Intg{s}}],
 \left[ \,\BNum{\sum_{j=0}^{0}2^{0} \nu_{m}}\,\right]^{k-1}
\begin{array}[t]{l}
  ,[\BNum{n_{1}}]^{k-1}
  ,\ldots
  ,[\BNum{n_{\Intg{n}}}]^{k-1}\\
  ,[\BNum{s_{1}}]^{k-1}
  ,\ldots
  ,[\BNum{s_{\Intg{s}}}]^{k-1}
\quad 
\ran\!\ran
\end{array}
\]
\normalsize
which is 
$\lan\!\lan
 r[0,a,n_1,\ldots,n_{\Intg{n}},s_1,\ldots,s_{\Intg{s}}],
 [\BNum{\nu_{m}}]^{k-1}
,[\BNum{n_{1}}]^{k-1}
,\ldots
,[\BNum{n_{\Intg{n}}}]^{k-1}
,[\BNum{s_{1}}]^{k-1}
,\ldots
,$ \\ $[\BNum{s_{\Intg{s}}}]^{k-1}
\ran\!\ran
$.
Since, by definition of words, $\nu_m$ is $1$, we have to prove that:
\\
$\TransFunc_{1+\Intg{n};\Intg{s}}[\BSuccO,G_{1}]
\lan\!\lan
 \BNum{a},
 [\BNum{0}]^k
  ,[\BNum{n_{1}}]^k
  ,\ldots
  ,[\BNum{n_{\Intg{n}}}]^k
  ,[\BNum{s_{1}}]^k
  ,\ldots
  ,[\BNum{s_{\Intg{s}}}]^k
\ran\!\ran
$ rewrites to
$\lan\!\lan
 r[0,a,n_1,$
 \\
$\ldots,n_{\Intg{n}},s_1,\ldots,s_{\Intg{s}}],
 [\BNum{1}]^{k-1}
,[\BNum{n_{1}}]^{k-1}
,\ldots
,[\BNum{n_{\Intg{n}}}]^{k-1}
,[\BNum{s_{1}}]^{k-1}
,\ldots
,[\BNum{s_{\Intg{s}}}]^{k-1}
\ran\!\ran
$.
But this can be obtained using Proposition~\ref{proposition:Dynamics of the transition function}, with $F\equiv \BSuccO$ and $F'\equiv G_{1}$, Proposition~\ref{proposition:Dynamics of combinators relative to words}, and the assumption on $G_1$.
\par
\textbf{We assume $m-i>0$.} Namely, $m>i$. By induction, we have that:
\scriptsize
\begin{align}
\label{Dyn-iter-part-a-01}
\TransFunc_{1+\Intg{n};\Intg{s}}[\BSucc{\nu_{i-1}},G_{\nu_{i-1}}]
(\ldots
(
\TransFunc_{1+\Intg{n};\Intg{s}}[\BSucc{\nu_m},G_{\nu_m}]
\lan\!\lan
 \BNum{a},
 [\BNum{0}]^k
  ,[\BNum{n_{1}}]^k
  ,\ldots
  ,[\BNum{n_{\Intg{n}}}]^k
  ,[\BNum{s_{1}}]^k
  ,\ldots
  ,[\BNum{s_{\Intg{s}}}]^k
\ran\!\ran
)\ldots)
\end{align}
\normalsize
rewrites to the word:
\small
\[
\lan\!\lan
 r[m-(i+1),a,n_1,\ldots,n_{\Intg{n}},s_1,\ldots,s_{\Intg{s}}]
\begin{array}[t]{l}
   ,\left[ \,\BNum{\sum_{j=0}^{m-(i+1)}2^{m-(i+1)-j} \nu_{m-j}}
    \,\right]^{k-(m-i-1)-1}
   \\
  ,[\BNum{n_{1}}]^{k-(m-i-1)-1}
  ,\ldots
  ,[\BNum{n_{\Intg{n}}}]^{k-(m-i-1)-1}\\
  ,[\BNum{s_{1}}]^{k-(m-i-1)-1}
  ,\ldots
  ,[\BNum{s_{\Intg{s}}}]^{k-(m-i-1)-1}
\quad 
\ran\!\ran
\end{array}
\]
\normalsize
If we apply 
$\TransFunc_{1+\Intg{n};\Intg{s}}[\BSucc{\nu_{i}},G_{\nu_{i}}]$ to \eqref{Dyn-iter-part-a-01}, we end up to calculate:
\scriptsize
\begin{align}
\label{Dyn-iter-part-a-02}
\TransFunc_{1+\Intg{n};\Intg{s}}[\BSucc{\nu_{i}},G_{\nu_{i}}]
\lan\!\lan
 r[m-(i+1),a,n_1,\ldots,n_{\Intg{n}},s_1,\ldots,s_{\Intg{s}}]
\begin{array}[t]{l}
   ,\left[ \,\BNum{\sum_{j=0}^{m-(i+1)}2^{m-(i+1)-j} \nu_{m-j}}
    \,\right]^{k-(m-i-1)-1}
   \\
  ,[\BNum{n_{1}}]^{k-(m-i-1)-1}
  ,\ldots
  ,[\BNum{n_{\Intg{n}}}]^{k-(m-i-1)-1}\\
  ,[\BNum{s_{1}}]^{k-(m-i-1)-1}
  ,\ldots
  ,[\BNum{s_{\Intg{s}}}]^{k-(m-i-1)-1}
\quad 
\ran\!\ran
\end{array}
\end{align}
\normalsize
Now, Proposition~\ref{proposition:Dynamics of the transition function}, with $F\equiv \BSucc{\nu_i}$ and $F'\equiv G_{\nu_i}$, the assumption:
\scriptsize
\begin{align*}
&G_{\nu_i}\,
 \BNum{\left(\sum_{j=0}^{m-(i+1)} 2^{m-(i+1)-j} \nu_{m-j}\right)}\,
 \BNum{n_{1}}\,\ldots\,\BNum{n_{\Intg{n}}}\,
 \BNum{s_{1}}\,\ldots\,\BNum{s_{\Intg{s}}}\,
 r[m-(i+1),a,n_1,\ldots,n_{\Intg{n}},s_1,\ldots,s_{\Intg{s}}]
\\
&
\qquad\qquad\qquad\qquad\qquad\qquad\qquad\qquad\qquad\qquad\qquad\qquad
\qquad\qquad\qquad
\red^+
r[m-i,a,n_1,\ldots,n_{\Intg{n}},s_1,\ldots,s_{\Intg{s}}]
\end{align*}
\normalsize
and Proposition~\ref{proposition:Dynamics of combinators relative to words}, which yields:
\scriptsize
\begin{align*}
\BSucc{\nu_i} \BNum{\left(\sum_{j=0}^{m-(i+1)} 2^{m-(i+1)-j} \nu_{m-j}\right)}
\red^+
\BNum{2\left(\sum_{j=0}^{m-(i+1)} 2^{m-(i+1)-j} \nu_{m-j}\right)+\nu_i}
=
\BNum{\sum_{j=0}^{m-i} 2^{m-i-j} \nu_{m-j}}
\enspace ,
\end{align*}
\normalsize
imply that \eqref{Dyn-iter-part-a-02} rewrites to:
\small
\[
\lan\!\lan
 r[m-1,a,n_1,\ldots,n_{\Intg{n}},s_1,\ldots,s_{\Intg{s}}]
\begin{array}[t]{l}
,\left[\BNum{\sum_{j=0}^{m-i} 2^{m-i-j} \nu_{m-j}}\right]^{k-(m-i-1)-1-1}
\\
,[\BNum{n_{1}}]^{k-(m-i-1)-1-1}
,\ldots
,[\BNum{n_{\Intg{n}}}]^{k-(m-i-1)-1-1}
\\
,[\BNum{s_{1}}]^{k-(m-i-1)-1-1}
,\ldots
,[\BNum{s_{\Intg{s}}}]^{k-(m-i-1)-1-1}
\ran\!\ran
\end{array}
\]
\normalsize
where $k-(m-i-1)-1-1$ is exactly $k-(m-i)-1$.
\paragraph*{Proof of Proposition~\ref{proposition:Typing the composition}}
\small
\begin{align*}
&
\comp{\Intg{n}}
     {\sum_{i=1}^{\Intg{s}'} \Intg{s}_i}
     {\Intg{n}'}
     {\Intg{s}'}
     {F,G_1,\ldots,G_{\Intg{n}'},H_1,\ldots,H_{\Intg{s}'}}
	\BNum{n_{1}}\,\ldots\,\BNum{n_{\Intg{n}}}\,
	\BNum{s_{11}}\,\ldots\,\BNum{s_{1\Intg{s}_{1}}}
	\ldots\ldots
        \BNum{s_{\Intg{s}'1}}\,\ldots\,\BNum{s_{\Intg{s}'s_{\Intg{s}'}}}
\\
&
\red^{+}
\EEmbed{2}
       {0}{\Intg{n}+\sum_{i=1}^{\Intg{s}'} \Intg{s}_i}
       {G}
(\LEmbed{1}
        {1}
        {\nabla^1_{\Intg{n}'+\Intg{s}'}}\,\BNum{n_{1}})
\ldots
(\LEmbed{1}
        {1}
        {\nabla^1_{\Intg{n}'+\Intg{s}'}}\,\BNum{n_{\Intg{n}}})
\BNum{s_{11}}
\ldots
\BNum{s_{1\Intg{s}_1}}
\ldots\ldots
\BNum{s_{\Intg{s}'1}}
\ldots
\BNum{s_{\Intg{s}'\Intg{s}_{\Intg{s}'}}}
\\
&\red^{+}
\EEmbed{2}
       {0}{\Intg{n}+\sum_{i=1}^{\Intg{s}'} \Intg{s}_i}
       {G}
(\nabla^1_{\Intg{n}'+\Intg{s}'}\,\BNum{n_{1}})
\ldots
(\nabla^1_{\Intg{n}'+\Intg{s}'}\,\BNum{n_{\Intg{n}}})
\,
\BNum{s_{11}}
\ldots
\BNum{s_{1\Intg{s}_1}}
\ldots\ldots
\BNum{s_{\Intg{s}'1}}
\ldots
\BNum{s_{\Intg{s}'\Intg{s}_{\Intg{s}'}}}
\\
&\red^{+}
\EEmbed{2}
       {0}{\Intg{n}+\sum_{i=1}^{\Intg{s}'} \Intg{s}_i}
       {G}
 \elan 
 \underbrace{\BNum{n_{1}} 
             \ldots 
             \BNum{n_{1}}}_{\Intg{n}'+\Intg{s}'}
 \eran
\ldots
 \elan 
 \underbrace{\BNum{n_{\Intg{n}}} 
             \ldots 
             \BNum{n_{\Intg{n}}}}_{\Intg{n}'+\Intg{s}'}
 \eran
\,
\BNum{s_{11}}
\ldots
\BNum{s_{1\Intg{s}_1}}
\ldots\ldots
\BNum{s_{\Intg{s}'1}}
\ldots
\BNum{s_{\Intg{s}'\Intg{s}_{\Intg{s}'}}}
\\
&\red^{+}
G\,
 \elan 
 \BNum{n_{1}}\ldots\BNum{n_{1}}
 \eran
\ldots
 \elan 
 \BNum{n_{\Intg{n}}} \ldots \BNum{n_{\Intg{n}}}
 \eran
\,
\BNum{s_{11}}
\ldots
\BNum{s_{1\Intg{s}_1}}
\ldots\ldots
\BNum{s_{\Intg{s}'1}}
\ldots
\BNum{s_{\Intg{s}'\Intg{s}_{\Intg{s}'}}}
\\
&\red^{+}
   \EEmbed{m-1}{0}{\Intg{n}'+\Intg{s}'}{F}\,
\\&\qquad
   (G_1\, \BNum{n_{1}}\ldots\BNum{n_{\Intg{n}}})
   \ldots
   (G_{\Intg{n}'}\, \BNum{n_{1}}\ldots\BNum{n_{\Intg{n}}})
\\&\qquad
   (\EEmbed{m-1}{0}{\Intg{n}+\Intg{s}_1}{H_{1}}\,
     (\LEmbed{1}{1}{\Coerc^{m-1}}\, \BNum{n_{1}})
     \ldots 
   (\LEmbed{1}{1}{\Coerc^{m-1}}\,
      \BNum{n_{\Intg{n}}})\BNum{s_{11}}\ldots \BNum{s_{1\Intg{s}_1}})
\\&\qquad
\nonumber
    \ldots
   (\EEmbed{m-1}{0}{\Intg{n}+\Intg{s}_1}{H_{\Intg{s}'}}\,
     (\LEmbed{1}{1}{\Coerc^{m-1}}\, \BNum{n_{1}})
     \ldots 
   (\LEmbed{1}{1}{\Coerc^{m-1}}\,\BNum{n_{\Intg{n}}})
      \BNum{s_{\Intg{s}'1}}
      \ldots 
      \BNum{s_{\Intg{s}'\Intg{s}_{\Intg{s}'}}})
\\
&\red^{+}
\EEmbed{m-1}{0}{\Intg{n}'+\Intg{s}'}{F}\,
\BNum{g_1}\,\ldots\,\BNum{g_{\Intg{n}'}}\,
\BNum{h_1}\,\ldots\,\BNum{h_{\Intg{s}'}}
\red^{+}\BNum{f}
\qquad\qquad\qquad\qquad\text{(by the assumptions)}
\end{align*}
\normalsize
\paragraph*{Proof of Proposition~\ref{proposition:Dynamics of the composition}}
To show
$\emptyset;\emptyset;\emptyset\vdash
\comp{\Intg{n}}
{\sum_{i=1}^{\Intg{s}'} \Intg{s}_i}
{\Intg{n}'}
{\Intg{s}'}
{F,G_1,\ldots,G_{\Intg{n}'},H_1,\ldots,H_{\Intg{s}'}}
\!:\!
(\liv^{\Intg{n}}_{i=1}\$\BIntT)\liv
$ \\
$(\liv^{\sum_{i=1}^{\Intg{s}'} \Intg{s}_i}_{i=1}\$^{2m+1}\BIntT)\liv
\$^{2m+1}\BIntT$ derive and suitably compose the following judgments:
\small
\begin{align*}
&
\emptyset;\emptyset;
\{(\ta{n_i}{\BIntT};\emptyset)\}
\vdash
\ta{
\LEmbed{1}
       {1}
{\nabla^{1}_{\Intg{n}'+\Intg{s}'}}
\,n_i
}{\$^{2}\left(\bigodot_{i=1}^{\Intg{n}'+\Intg{s}'}\$\BIntT\right)}
\qquad (1\leq i\leq \Intg{n})
\\
&
\emptyset;\emptyset;
\{(\ta{y_{ij}}{\BIntT};\emptyset)\}
\vdash
\ta{
\LEmbed{1}{1}{\Coerc^{m-1}}\, y_{ij}
}
{\$^{m}\BIntT}
\qquad(1\leq i\leq \Intg{s}', 1\leq j\leq \Intg{n})
\\
&
\emptyset;\emptyset;
\{(\ta{w_{ik}}{\$^{2m-2}\BIntT};\emptyset)\}
\vdash\ta{w_{ik}}{\$^{2m-1}\BIntT}
\qquad(1\leq i\leq \Intg{s}', 1\leq k\leq \Intg{s}_i)
\\
&
\emptyset;\emptyset;
\{(\ta{x_{ik}}{\BIntT};\emptyset)\}
\vdash\ta{x_{ik}}{\$\BIntT}
\qquad(1\leq i\leq \Intg{n}', 1\leq k\leq \Intg{n})
\\
&
\emptyset;\emptyset;
\{(
\begin{array}[t]{l}
 \ta{y_{i1}}{\BIntT}
,\ldots,
 \ta{y_{i\Intg{n}}}{\BIntT},
\\
 \ta{w_{i1}}{\$^{2m-2}\BIntT}
,\ldots,
 \ta{w_{i\Intg{s}_i}}{\$^{2m-2}\BIntT}
;\emptyset)\}
\\
\vdash
\EEmbed{m-1}{0}{\Intg{n}+\Intg{s}_i}{H_{i}}\,
(\LEmbed{1}{1}{\Coerc^{m-1}}\, y_{i1})
\ldots
(\LEmbed{1}{1}{\Coerc^{m-1}}\,y_{i\Intg{n}})\,
\\
\phantom{\vdash\EEmbed{m-1}{0}{\Intg{n}+\Intg{s}_i}{H_{i}}\,}
w_{i1}\ldots w_{i\Intg{s}_i}
\!:\!
\$^{2m-1}\BIntT
\end{array}
\qquad(1\leq i\leq \Intg{s}')
\\
&
\emptyset;\emptyset;
\{(
 \ta{x_{i1}}{\BIntT}
,\ldots,
 \ta{x_{i\Intg{n}}}{\BIntT}
;\emptyset)\}
\vdash
\ta{G_i\, x_{i1} \ldots x_{i\Intg{n}}}{\$^{m}\BIntT}
\qquad(1\leq i\leq \Intg{n}')
\\
&
\emptyset;\emptyset;
\{(
 \begin{array}[t]{l}
 \ta{x_{11}}{\BIntT},\ldots,\ta{x_{1\Intg{n}}}{\BIntT}
 \ldots\ldots
 \ta{x_{\Intg{n}'1}}{\BIntT},\ldots,\ta{x_{\Intg{n}'\Intg{n}}}{\BIntT}
 \\
 \ta{y_{11}}{\BIntT},\ldots,\ta{y_{1\Intg{n}}}{\BIntT},
 \ta{w_{11}}{\$^{2m-2}\BIntT},\ldots,\ta{w_{1\Intg{s}_1}}{\$^{2m-2}\BIntT}
 \\
 \ldots
 \\
 \ta{y_{\Intg{s}'1}}{\BIntT},\ldots,\ta{y_{\Intg{s}'\Intg{n}}}{\BIntT},
 \ta{w_{\Intg{s}'1}}{\$^{2m-2}\BIntT},
 \ldots,
 \ta{w_{\Intg{s}'\Intg{s}_{\Intg{s}'}}}{\$^{2m-2}\BIntT}
;\emptyset)
\}
\vdash
\ta{H}{\$^{2m-1}\BIntT}
\end{array}
\end{align*}
\normalsize
$H$ being the term in the definition of the composition of Section~\ref{subsubsection:Composition}
\paragraph*{Proof of Theorem~\ref{theorem:QlSRN is a subsystems of WALT}, point~\ref{theo:realizPRNQ3}}
We proceed by induction on the structure of the \textit{closed} term. The base case coincides with $t=0$. We focus on the inductive cases only,
\par
\textbf{First case.}
Let $t$ be $f(t_1,\ldots,t_k,u_1,\ldots,u_l)$.
Then,
$\emptyset;\emptyset;\emptyset\vdash
 \ta{\srtw{t}{}}
    {\$^{v}\BIntT}$ with 
$v=\max\{u-1+m,q_1,\ldots,q_l\}$, and
$u=\max\{m,p_1,\ldots,p_k\}$, for some 
$m,p_1,\ldots,p_k,q_1,$ \\ $\ldots,q_l$.
If we develop the relations between $u$ and $v$, we get
$v=\max\{\max\{m,p_1,$ \\ $\ldots,p_k\}-1+m,q_1,\ldots,q_l\}
  =\max\{2m-1,p_1-1+m,\ldots,p_k-1+m,q_1,\ldots,q_l\}
  \leq\max\{2m,2p_1,\ldots,2p_k,q_1,\ldots,q_l\}$.
Since, by induction, $m\leq\wght{f}{}$, $p_i\leq\wght{t_i}{}$ with $1\leq i\leq k$, and $q_j\leq\wght{u_j}{}$ with $1\leq j\leq l$, we get
$\max\{2m,2p_1,\ldots,2p_k,q_1,$ \\ $\ldots,q_l\}
\leq 2\max\{\wght{f}{},\wght{t_1}{},\ldots,\wght{t_k}{},\wght{u_1}{},\ldots,\wght{u_l}{}\}$ which is exactly $\wght{f(t_1,\ldots,t_k,u_1,\ldots,u_l)}{}$.
\par
\textbf{Second case.}
Let $t$ be 
$\comp{k}
      {\sum_{i=1}^{l'}l_i}
      {k'}
      {l'}
      {f,g_1,\ldots,g_{k'},h_1,\ldots,h_{l'}}$.
Then,
$\emptyset;\emptyset;\emptyset\vdash
 \ta{\srtw{t}{}}
    {\$^{2p+1}\BIntT}$ with 
$p=\max\{m,m_1,\ldots,$ \\ $m_{k'}, n_1,\ldots,n_{l'}\}$, for some 
$m,m_1,\ldots,m_{k'},n_1,\ldots,n_{l'}$.
We have
$2p+1
=2\max\{m,$ $m_1,\ldots,m_{k'},n_1,\ldots,n_{l'}\}+1
\leq 3\max\{m,m_1,\ldots,m_{k'},n_1,\ldots,n_{l'}\}$.
The induction $m\leq\wght{f}{}$, $m_i\leq\wght{g_i}{}$ with $1\leq i\leq k'$, and $n_j\leq\wght{h_j}{}$ with $1\leq j\leq l'$, implies
$3\max\{m,m_1,\ldots,m_{k'},n_1,\ldots,n_{l'}\}
\leq 3\max\{\wght{f}{},\wght{g_1}{},\ldots,\wght{g_{k'}}{},$ \\ $
                       \wght{h_1}{},\ldots,\wght{h_{l'}}{}\}
=
\wght{\comp{k}
       {\sum_{i=1}^{l'}l_i}
       {k'}
       {l'}
       {f,g_1,\ldots,g_{k'},h_1,\ldots,h_{l'}}}{}$.
\par
\textbf{Third case.}
Let $t$ be $\rec{k+1}{l}{g,h_0,h_1}$.
Then,
$\emptyset;\emptyset;\emptyset\vdash
 \ta{\srtw{t}{}}
    {\$^{p+1}\BIntT}$ with 
$p=\max\{m_0,$ \\ $m_1,m_{2}\}$, for some $m_0,m_1,,m_2$.
We have
$p+1
=\max\{m_0+1,m_1+1,m_2+1\}
\leq 2\max\{m_0,m_1,m_2\}$.
The induction $m_0\leq\wght{g}{}$, $m_i\leq\wght{h_i}{}$ with $0\leq i\leq 1$,
implies
$2\max\{m_0,m_1,m_2\}
\leq 2\max\{\wght{g}{},\wght{h_0}{},\wght{h_1}{}\}
=$ \\ $\wght{\rec{k+1}{l}{g,h_0,h_1}}{}$.
\paragraph*{Proof of Theorem~\ref{theorem:QlSRN is a subsystems of WALT}, point~\ref{theo:realizPRNQ4}}
By induction on the structure of $n$.
If $n=0$, then, by definition,
$\srtw{0}{}
=\LEmbed{1}{0}{\BNum{0}}$, which is $\BNum{0}$. Otherwise,
let $n\geq 0$. Then, $n$ can be written as 
$\sum^{m}_{j=0}2^{m-j} \nu_{m-j}
=2(\sum^{m-1}_{j=0}2^{m-j-1}\nu_{m-j})+\nu_{0}
=2m+\nu_{0}$, for some $m$.
Then, by the inductive hypothesis, we have
$\srtw{n}{}
=\srtw{\mathtt{s}_{\nu_0}(\ldots(\mathtt{s}_{\nu_{m-1}}(\Suco\,\zero{0}{0}))\ldots)}{}
=$ \\ $\etp{\mathtt{s}_{\nu_0}}
 (\srtw{\mathtt{s}_{\nu_1}(\ldots(\mathtt{s}_{\nu_{m-1}}(\Suco\,\zero{0}{0}))\ldots)}{})
\red^{*}\etp{\mathtt{s}_{\nu_0}}\,\BNum{m}
=\BSucc{\nu_0}\,\BNum{m}\red^{+}\BNum{2m+\nu_0}=\BNum{n}
$.
\paragraph*{Proof of Theorem~\ref{theorem:QlSRN is a subsystems of WALT}, point~\ref{theo:realizPRNQ5}}
By structural induction on the structure of $f$. We develop the details in the case $f$ be a recursive scheme, the most interesting one. So, let $f$ be $\rec{k+1}{l}{g,h_0,h_1}$. We have to consider two cases. 
\par
\textbf{First case.} It has
$\rec{k+1}{l}{g,h_0,h_1}(n_0,\seq{n}{1}{k},\seq{s}{1}{l})=n$ with $n_0=0$, for some $n$.
By definition
\\
$\rec{k+1}{l}{g,h_0,h_1}(n_0,\seq{n}{1}{k},\seq{s}{1}{l})
= g(\seq{n}{1}{k},\seq{s}{1}{l})=n$, which, by induction, implies
$\srtw{g(\seq{n}{1}{k},\seq{s}{1}{l})}{}\red^+\BNum{n}$.
The statement we need to prove is
$\srtw{\rec{k+1}{l}{g,h_0,h_1}(n_0,\seq{n}{1}{k},\seq{s}{1}{l})}{}\red^{+}\BNum{n}$.
By definition:
\small
\begin{align}
\nonumber
&
\srtw{\rec{k+1}{l}{g,h_0,h_1}(n_0,\seq{n}{1}{k},\seq{s}{1}{l})}{}
\\
\nonumber
&=
\mathtt{Ev}^{v-u+1-p+4}_{0;l}
       [\begin{array}[t]{l}
        \EEmbed{u-1}{0}{k+l}{\etp{\rec{k+1}{l}{g,h_0,h_1}}{}}\\
        (\LEmbed{u-p_{0}}{0}{\srtw{0}{}})
        (\LEmbed{u-p_{1}}{0}{\srtw{n_1}{}})
        \ldots
        (\LEmbed{u-p_{k}}{0}{\srtw{n_k}{}})
       \end{array}
\\
&
\phantom{=\mathtt{Ev}^{v-u+1-p+4}_{0;l}}
\ 
](\LEmbed{v-q_{1}}{0}{\srtw{s_1}{}})
 \ldots
 (\LEmbed{v-q_{l}}{0}{\srtw{s_l}{}})
\label{align:sim001}
\end{align}
\normalsize
where
$\emptyset;\emptyset;\emptyset
\vdash\ta{\etp{\rec{k+1}{l}{g,h_0,h_1}}}
{\$\BIntT\liv(\liv_{i=1}^{k}\$\BIntT)\liv(\liv_{j=1}^{l}\$^{p+4}\BIntT)\liv\$^{p+4}\BIntT}
$,
$\emptyset;\emptyset;\emptyset
\vdash\ta{\srtw{0}{}}{\$^{p_0}\BIntT}
$,
$\emptyset;\emptyset;\emptyset
\vdash\ta{\srtw{n_i}{}}{\$^{p_i}\BIntT}
$, with $1\leq i\leq k$,
$\emptyset;\emptyset;\emptyset
\vdash\ta{\srtw{s_j}{}}{\$^{q_j}\BIntT}
$, with $1\leq j\leq l$, 
$v=\max\{u-1+p+4,q_1,\ldots,q_l\}$, and
$u=\max\{p+4,p_1,\ldots,p_k\}$.
By the definition of linear embedding, and point~\ref{theo:realizPRNQ4} of Theorem~\ref{theorem:QlSRN is a subsystems of WALT}:
\small
\begin{align}
\nonumber
\eqref{align:sim001}
&\red^{+}
\EEmbed{v-u+1-p+4}
       {0}
       {l}
       {
        \EEmbed{u-1}{0}{k+l}{\etp{\rec{k+1}{l}{g,h_0,h_1}}{}}
        \, \BNum{0} \,\BNum{n_1} \ldots \BNum{n_k}
       }
       \,\BNum{s_1} \ldots \BNum{s_l}
\\
&
\equiv
\EEmbed{v-u+1-p+4}
       {0}
       {l}
       {
        \EEmbed{u-1}
               {0}
               {k+l}
               {
                 \Iter{1+\Intg{k}}{l}{\etp{H_0}}{\etp{H_1}}{\etp{G}}
               }
        \, \BNum{0} \,\BNum{n_1} \ldots \BNum{n_k}
       }
       \,\BNum{s_1} \ldots \BNum{s_l}
\label{align:sim002}
\end{align}
where
$G\equiv\EEmbed{p-m_g}
               {k+1}{l+1}
               {\bs n_0\ldots n_k\,s_1\ldots s_l\,r.
                \etp{g}\,n_1\ldots n_k\,s_1\ldots s_l}$,
$\emptyset;\emptyset;\emptyset
\vdash\ta{\etp{g}}
{(\liv_{i=1}^{k}\$\BIntT)\liv(\liv_{j=1}^{l}\$^{m_g}\BIntT)\liv\$^{m_g}\BIntT}
$,
$H_i\equiv\EEmbed{p-m_i}
                 {k+1}{l+1}
                 {\etp{h_i}}$,
$\emptyset;\emptyset;\emptyset
\vdash\ta{\etp{h_i}}
{\$\BIntT\liv(\liv_{i=1}^{k}\$\BIntT)\liv(\liv_{j=1}^{l}\$^{m_i}\BIntT)\liv\$^{m_i}\BIntT}
$,
and
$p=\max\{m_g,m_0,m_1\}$, with $i\in\{0,1\}$.
\par
Using the definition of the embeddings and point~\ref{proposition:Dynamics of the iterator-part-b} of Proposition~\ref{proposition:Dynamics of the iterator}, we have that $\eqref{align:sim002}\red^{+}\BNum{a}$ if
$G \BNum{0}\,\BNum{n_1}\ldots\BNum{n_k}\,\BNum{s_1}\ldots \BNum{s_l}\,\BNum{0}\red^{+}\BNum{a}$, for some $a$.
We observe that, by the definition of $G$, we have 
$G \BNum{0}\,\BNum{n_1}\ldots\BNum{n_k}\,\BNum{s_1}$ $\ldots \BNum{s_l}\,\BNum{0}\red^{+}
\etp{g} \BNum{0}\,\BNum{n_1}\ldots\BNum{n_k}\,\BNum{s_1}\ldots \BNum{s_l}\,\BNum{0}
$, so we are left to prove
$\etp{g} \BNum{0}\,\BNum{n_1}\ldots\BNum{n_k}\,\BNum{s_1}\ldots \BNum{s_l}\,\BNum{0}
\red^{+}\BNum{a}$, for some $a$, and that $a$, in fact, is $n$.
\par
To prove this, we start by the induction. It implies
$\srtw{g(\seq{n}{1}{k},\seq{s}{1}{l})}{}\red^{+}\BNum{n}$. Then, we observe that, by the definition of the embeddings, and thanks to Point~\ref{theo:realizPRNQ4} of Theorem~\ref{theorem:QlSRN is a subsystems of WALT}:
\begin{align}
\nonumber
\srtw{g(\seq{n}{1}{k},\seq{s}{1}{l})}{}
&=
\mathtt{Ev}^{v'-u'+1-m_g}_{0;l}
       [\EEmbed{u'-1}{0}{k+l}{\etp{g}}
        (\LEmbed{u'-p_{1}}{0}{\srtw{n_1}{}})
        \ldots
        (\LEmbed{u'-p_{k}}{0}{\srtw{n_k}{}})
\\
&
\phantom{=\mathtt{Ev}^{v-u+1-p+4}_{0;l}}
\quad
](\LEmbed{v'-q_{1}}{0}{\srtw{s_1}{}})
 \ldots
 (\LEmbed{v'-q_{l}}{0}{\srtw{s_l}{}})
\nonumber
\\
&\equiv
\mathtt{Ev}^{v'-u'+1-m_g}_{0;l}
       [\EEmbed{u'-1}{0}{k+l}{\etp{g}}
        \srtw{n_1}{}
        \ldots
        \srtw{n_k}{}
       ]\srtw{s_1}{}\ldots\srtw{s_l}{}
\nonumber
\\
&\red^{+}
\mathtt{Ev}^{v'-u'+1-m_g}_{0;l}
       [\EEmbed{u'-1}{0}{k+l}{\etp{g}}
        \BNum{n_1}
        \ldots
        \BNum{n_k}
       ]\BNum{s_1}\ldots\BNum{s_l}
\nonumber
\\
&\red^{+}
\etp{g}
\BNum{n_1}
\ldots
\BNum{n_k}\,
\BNum{s_1}\ldots\BNum{s_l}
\label{align:sim003}
\end{align}
The sequences of rewritings that lead to \eqref{align:sim003} is unique up to the ordering in which we obtain every of the arguments $\BNum{n_1},\ldots,\BNum{n_k},
\BNum{s_1},\ldots,\BNum{s_l}$.
So, thanks to $\srtw{g(\seq{n}{1}{k},\seq{s}{1}{l})}{}\red^{+}\BNum{n}$
it must be $\eqref{align:sim003}\red^{+}\BNum{n}$ showing that $a$ is, in fact, $n$.
\par
\textbf{Second case.} 
It has
$\rec{k+1}{l}{g,h_0,h_1}(n_0,\seq{n}{1}{k},\seq{s}{1}{l})=n$ with $n_0=\sum_{j=0}^{m}2^{m-j}\nu_{m-j}>0$, for some $n$.
By definition:
\small
\begin{align*}
&\rec{k+1}{l}{g,h_0,h_1}(\sum_{j=0}^{m}2^{m-j}\nu_{m-j},\seq{n}{1}{k},\seq{s}{1}{l})
\\
&=
h_{\nu_{0}}
(\sum_{j=0}^{m-1}2^{m-1-j}\nu_{m-j},
 \seq{n}{1}{k},
 \seq{s}{1}{l},
 h_{\nu_{1}}
 (\sum_{j=0}^{m-2}2^{m-2-j}\nu_{m-j},
  \seq{n}{1}{k},
  \seq{s}{1}{l},\ldots
\\
&
\qquad\qquad
  \ldots,
   h_{\nu_{m}}
   (0,
    \seq{n}{1}{k},
    \seq{s}{1}{l},
    g(0,\seq{n}{1}{k},\seq{s}{1}{l})
   )
 )
 \ldots
 )
)
\\
&=
h_{\nu_{0}}
(\sum_{j=0}^{m-1}2^{m-1-j}\nu_{m-j},
 \seq{n}{1}{k},
 \seq{s}{1}{l},
 h_{\nu_{1}}
 (\sum_{j=0}^{m-2}2^{m-2-j}\nu_{m-j},
  \seq{n}{1}{k},
  \seq{s}{1}{l},\ldots
\\
&
\qquad\qquad
  \ldots,
   h_{\nu_{m}}
   (0,
    \seq{n}{1}{k},
    \seq{s}{1}{l},
    a
   )
 )
 \ldots
 )
)
\\
&=
h_{\nu_{0}}
(\sum_{j=0}^{m-1}2^{m-1-j}\nu_{m-j},
 \seq{n}{1}{k},
 \seq{s}{1}{l},
 h_{\nu_{1}}
 (\sum_{j=0}^{m-2}2^{m-2-j}\nu_{m-j},
  \seq{n}{1}{k},
  \seq{s}{1}{l},\ldots,
  v_m
 )
 \ldots
 )
)
\\
&=
h_{\nu_{0}}
(\sum_{j=0}^{m-1}2^{m-1-j}\nu_{m-j},
 \seq{n}{1}{k},
 \seq{s}{1}{l},
 v_{1}
) = v_0
\end{align*}
\normalsize
for some sequence of natural numbers $n=v_0,v_1,\ldots,v_m,a$.
\par
Since $h_{\nu_i}\in\{h_0,h_1\}$, for every $1\leq i\leq m$, we can apply the induction to
$h_{\nu_m}(0,\seq{n}{1}{k},\seq{s}{1}{l},a)=v_m$, which implies:
\small
\begin{align}
\label{align:sim004}
\srtw{h_{\nu_m}(0,\seq{n}{1}{k},\seq{s}{1}{l},a)}{}\red^{+}\BNum{v_m}
\enspace ,
\end{align}
\normalsize
to
$h_{\nu_i}(\sum_{j=0}^{m-(i+1)}2^{m-(i+1)-j}\nu_{m-j},
           \seq{n}{1}{k},
           \seq{s}{1}{l},v_{i+1})=v_i$, for every $0\leq i\leq m-1$,
which implies:
\small
\begin{align}
\label{align:sim005}
\srtw{h_{\nu_{i}}(\sum_{j=0}^{m-(i+1)}2^{m-(i+1)-j}\nu_{m-j},
                  \seq{n}{1}{k},
                  \seq{s}{1}{l},v_{i+1})}{}
\red^{+}\BNum{v_i}
\enspace ,
\end{align}
\normalsize
and, finally, to
$g(\seq{n}{1}{k},\seq{s}{1}{l})=a$, which implies:
\small
\begin{align}
\label{align:sim006}
\srtw{g(\seq{n}{1}{k},\seq{s}{1}{l})}{}
\red^{+}\BNum{a}
\enspace .
\end{align}
\normalsize
The statement we need to prove is
$\srtw{\rec{k+1}{l}{g,h_0,h_1}
(n_0,\seq{n}{1}{k},\seq{s}{1}{l})}{}\red^{+}\BNum{n}$
with
$n_0=\sum_{j=0}^{m}2^{m-j}\nu_{m-j}$.
\par
By the definitions and point~\ref{theo:realizPRNQ4} of Theorem~\ref{theorem:QlSRN is a subsystems of WALT}, we get:
\small
\begin{align}
\nonumber
&\srtw{\rec{k+1}{l}{g,h_0,h_1}(\seq{n}{0}{k},\seq{s}{1}{l})}{}
\\
&
\red^{+}
\Iter{1+\Intg{k}}
     {l}
     {\etp{H_0}}
     {\etp{H_1}}
     {\etp{G}}
     \BNum{n_0}\,
     \BNum{n_1} \ldots \BNum{n_k}\,
     \BNum{s_1} \ldots \BNum{s_l}
\label{align:sim007}
\end{align}
\normalsize
where
$G\equiv\EEmbed{p-m_g}
               {k+1}{l+1}
               {\bs n_0\ldots n_k\,s_1\ldots s_l\,r.
                \etp{g}\,n_1\ldots n_k\,s_1\ldots s_l}$,
$\emptyset;\emptyset;\emptyset
\vdash\ta{\etp{g}}
{(\liv_{i=1}^{k}\$\BIntT)\liv(\liv_{j=1}^{l}\$^{m_g}\BIntT)\liv\$^{m_g}\BIntT}
$,
$H_i\equiv\EEmbed{p-m_i}
                 {k+1}{l+1}
                 {\etp{f_i}}$,
$\emptyset;\emptyset;\emptyset
\vdash\ta{\etp{f_i}}
{\$\BIntT\liv(\liv_{i=1}^{k}\$\BIntT)\liv(\liv_{j=1}^{l}\$^{m_i}\BIntT)\liv\$^{m_i}\BIntT}
$,
and
$p=\max\{m_g,m_0,m_1\}$, with $i\in\{0,1\}$.
By point~\ref{proposition:Dynamics of the iterator-part-b} of Proposition~\ref{proposition:Dynamics of the iterator}, we have that $\eqref{align:sim007}\red^{+}
r[m,n,n_1,$ $\ldots,n_{k},s_1,\ldots, s_{l}]$ 
if:
\scriptsize
\begin{align*}
&
G\,\BNum{0}\,
   \BNum{n_1}\ldots\BNum{n_k}\,
   \BNum{s_1}\ldots \BNum{s_l}\,
   \BNum{0}
\red^{+}
\BNum{n_G}
\\
&
H_{\nu_m}
   \BNum{0}\,
   \BNum{n_1}\ldots\BNum{n_k}\,
   \BNum{s_1}\ldots \BNum{s_l}\,
   \BNum{n_G}
\red^{+}
r[0,n_G,n_1,\ldots,n_{k},s_1,\ldots,s_{l}]
\\
&
H_{\nu_i}
\left(\BNum{\sum_{j=0}^{m-(i+1)}2^{m-(i+1)-j}\nu_{m-j}}\right)
\BNum{n_1}\ldots\BNum{n_k}\,
\BNum{s_1}\ldots \BNum{s_l}
r[m-(i+1),n_G,n_1,\ldots,n_{k},s_1,\ldots,s_{l}]
\\&
\qquad\qquad\qquad\qquad\qquad\qquad
\qquad\qquad\qquad\qquad\qquad\qquad\qquad\qquad
\red^{+}
r[m-i,n_G,n_1,\ldots,n_{k},s_1,\ldots,s_{l}]
\enspace ,
\end{align*}
\normalsize
every
$r[m-(i+1),n_G,n_1,\ldots,n_{k},s_1,\ldots,s_{l}]$ being a word and every $\nu_i\in\{0,1\}$, with $0\leq i\leq i-1$.
So, we are left to prove:
\small
\begin{align}
\label{eqn:sim010'}
r[0,n_G,n_1,\ldots,n_{k},s_1,\ldots,s_{l}]&=\BNum{v_m}
\\
\label{eqn:sim011'}
r[m-(i+1),n_G,n_1,\ldots,n_{k},s_1,\ldots,s_{l}]&=\BNum{v_{i+1}}
\\
\label{eqn:sim012'}
n_G&=a
\end{align}
\normalsize
for every $0\leq i\leq m-1$. In fact, the here above three equations can be a consequence of proving:
\small
\begin{align}
\label{eqn:sim010}
H_{\nu_{m}}\,
  \BNum{0}\,
  \BNum{n_1}\ldots\BNum{n_k}\,
  \BNum{s_1}\ldots \BNum{s_l}\,
  \BNum{a}
&\red^{+}\BNum{v_m}
\\
\label{eqn:sim011}
H_{\nu_{i}}\,
  \BNum{\sum_{j=0}^{m-(i+1)}2^{m-(i+1)-j}\nu_{m-j}}\,
  \BNum{n_1}\ldots\BNum{n_k}\,
  \BNum{s_1}\ldots \BNum{s_l}\,
  \BNum{v_{i+1}}
&\red^{+}\BNum{v_i}
\\
\label{eqn:sim012}
G \BNum{0}\,
  \BNum{n_1}\ldots\BNum{n_k}\,
  \BNum{s_1}\ldots \BNum{s_l}\,
  \BNum{0}
&\red^{+}\BNum{a}
\enspace ,
\end{align}
\normalsize
for every $0\leq i\leq m-1$.
For proving \eqref{eqn:sim010}, \eqref{eqn:sim011}, and \eqref{eqn:sim012}, we start by the results \eqref{align:sim004}, \eqref{align:sim005}, and \eqref{align:sim006} the inductive hypothesis gives us.
\eqref{eqn:sim012} can be proved exactly as we did in the \textbf{first case} above, when $n_0=0$.
\eqref{eqn:sim011} holds by developing the definitions and observing that \eqref{align:sim004} reduces to 
$\etp{h_{\nu_{m}}}\,
  \BNum{0}\,
  \BNum{n_1}\ldots\BNum{n_k}\,
  \BNum{s_1}\ldots$ \\ $\BNum{s_l}\,
  \BNum{a}
$ in a unique way, up to the order of the evaluation of 
$\srtw{0}{}\red^{+}0,
\srtw{n_i}{}\red^{+}\BNum{n_i}$, and $\srtw{s_j}{}\red^{+}\BNum{s_j}$, for every
$1\leq i\leq k$ and $1\leq j\leq l$.
So, it must also be
$\etp{h_{\nu_{m}}}\,
  \BNum{0}\,
  \BNum{n_1}\ldots\BNum{n_k}\,
  \BNum{s_1}\ldots \BNum{s_l}\,
  \BNum{a}\red^{+}\BNum{v_m}$, getting \eqref{eqn:sim010'}.
We can proceed analogously to prove \eqref{eqn:sim011'}, for every
$0\leq i\leq m-1$.
\paragraph*{Proof of Corollary~\ref{corollary:The embedding of QlSRN into WALT is sound}}
By structural induction on the structure of $t$.
If $t$ is a variable we use point~\ref{theo:realizPRNQ4} of Theorem~\ref{theorem:QlSRN is a subsystems of WALT}.
Let $t$ be
$f(t_1,\ldots,t_k,u_1,\ldots,u_l)$. The assumption
$f(t_1,\ldots,t_k,u_1,\ldots,u_l)=n$ implies 
both
$t_i$ is closed, and $t_i=n_i$, for some $n_i$, 
with $0\leq i\leq k$,
and, analogously, 
both
$u_j$ is closed, and $u_j=s_j$, for some $s_j$, 
with $0\leq j\leq l$.
So, the points~\ref{theo:realizPRNQ4} and ~\ref{theo:realizPRNQ5} of Theorem~\ref{theorem:QlSRN is a subsystems of WALT} allow to write
$
\srtw{f(t_1,\ldots,t_k,u_1,\ldots,u_l)}{\rho}
\red^{*}
\etp{f}\,\BNum{n_1}\,\ldots\,\BNum{n_k}\,\BNum{s_1}\,\ldots\,\BNum{s_l}
\red^{*}
\BNum{n}$.
\end{document}